\documentclass[aap]{stylefile}

\usepackage{hyperref}
\usepackage{bbm, enumerate, mathtools, subcaption}
\usepackage{ulem}
\usepackage{dsfont}

\usepackage{cancel}
\usepackage{blindtext}

\usepackage{mathtools}
\usepackage{amsmath,amssymb}
\usepackage{hyperref}
\usepackage{amsthm}
\usepackage{arydshln}
\usepackage{enumitem}
\usepackage{xcolor}
\usepackage{threeparttable}
\usepackage{tikz}
\usetikzlibrary{shapes.arrows}

\usepackage{pgfplots}
\pgfplotsset{compat=1.15}
\usepackage{adjustbox}
\usepackage{bm}
\usepackage{gincltex}
\usepackage{caption}
\usepackage{subcaption}
\usepgfplotslibrary{fillbetween}
\usepgfplotslibrary{groupplots}
\usetikzlibrary{plotmarks}
\usetikzlibrary{patterns}
\pgfplotsset{
tick label style={font=\footnotesize},
label style={font=\footnotesize},
legend style={font=\footnotesize},
}

\usepackage{placeins}

\usepackage{amsthm}

\theoremstyle{plain}
\newtheorem{theorem}{Theorem}
\newtheorem{corollary}{Corollary}

\newtheorem{lemma}{Lemma}
\newtheorem{claim}{Claim}

\theoremstyle{definition}
\newtheorem{assumption}{Assumption}
\newtheorem{definition}{Definition}
\newtheorem{example}{Example}

\theoremstyle{remark}
\newtheorem{remark}{Remark}

\newtheorem{apptheorem}{Theorem}[section]

\let\plainqed\qedsymbol
\newcommand{\claimqed}{$\lrcorner$}

\usepackage[utf8]{inputenc}

\newcommand{\ti}{(t) }

\newcommand{\q}{{\bf Q}}

\allowdisplaybreaks

\begin{document}

\begin{frontmatter}
\title{Load Balancing Policies in Heterogeneous Systems: Non-Monotone Stability and Heavy-Traffic Optimality}

\begin{aug}
\author{\fnms{Yishun} \snm{Luo}\ead[label=e2]{luo00329@umn.edu}}
\and
\author{\fnms{Martin} \snm{Zubeldia}\ead[label=e1]{zubeldia@umn.edu}}
\address{
University of Minnesota, Minneapolis, USA.}

\end{aug}

\begin{abstract}
We consider a discrete-time queueing system with \(n\) heterogeneous parallel single-server queues. Jobs arrive at a central dispatcher and must be assigned immediately to one of the queues. We develop a unified framework for a broad family of load-balancing policies, including Join the Shortest Queue (JSQ), Join the Shortest Expected Delay (JSED), and Power-of-\(d\) Choices (Po\(d\)). In this framework, the dispatcher updates queue-length information periodically, possibly at arbitrarily long intervals, and dispatches jobs based on the sampled permutation of scaled queue lengths and the servers' service rates. Leveraging this structure, we derive a closed-form, easily verifiable sufficient condition for stability. We further show that, for general policies, stability above the induced threshold need not be monotone in the arrival rate, and we obtain an exact characterization under a persistent bottleneck dominance condition. When the stability condition holds strictly, we prove state-space collapse and heavy-traffic delay optimality. We also show that the steady-state queue-length vector converges in distribution to a deterministic vector scaled by an exponential random variable in heavy traffic. Methodologically, we extend Lyapunov-drift and transform techniques to a cycle-based analysis with multi-step updates. Our results connect the policy-induced dispatch fractions and sampled permutations to stability, delay, and distributional performance, providing guidance for designing scalable load-balancing schemes with limited queue-length information.
\end{abstract}

\begin{keyword}
\kwd{Load balancing}
\kwd{Stability}
\kwd{Heavy-traffic}
\end{keyword}

\end{frontmatter}

\setcounter{tocdepth}{2}
\tableofcontents

\section{Introduction}
We consider a discrete-time queueing system with $n$ heterogeneous parallel single-server queues, each maintaining its own queue with an infinite buffer. In each time slot, a central dispatcher receives a random number of jobs and immediately assigns them to a queue based on a load balancing policy. A wide range of policies and their variants have been studied, including Random Routing (RAND) policy, Weighted Random Routing (WRAND) policy, Round Robin (RR) policy, Join the Shortest Queue (JSQ) policy, Join the Shortest Expected Delay (JSED) policy, Power-of-$d$ Choices (Po$d$) policy, Join the Idle Queue (JIQ) policy, etc. 

To evaluate such policies, the first fundamental performance metric is the system's stability region. A policy is considered throughput optimal (Definition \ref{def:throughput_optimal}) when it maximizes this region. Although JSQ is well known to be throughput optimal, policies such as Po$d$ are throughput optimal in homogeneous systems, but not necessarily in heterogeneous systems \cite{gardner2021scalable-podJIQgeneral-fastslow, hurtado2021throughput-Pod-stability-heterogeneous, abdul2022general-Pod-GeneralmodifiedPod}. Moreover, even under JSQ, routing errors that misdirect jobs to longer queues can shrink the stability region \cite{moyal2022stability-generalhomogeneous}. Theoretically, \cite{zhou2017designing-JBT, zhou2019flexible-MultiDimensionSSC, zhou2018degree-generalPpolicy} show that a policy is throughput optimal if it prefers to send jobs to shorter queues relative to RAND. However, such criteria require a specific policy design and are not applicable to general policies. In summary, the existing literature predominantly focuses on a limited set of well-structured policies and their variants. These policies are typically designed with stability guarantees. As a result, the underlying relationship between load balancing policies and stability remains unclear, and a general characterization of the stability region for general policies is still lacking.

For policies that are throughput optimal, the second fundamental performance metric is the average job delay. JSQ and JSED are known to be delay optimal in certain settings \cite{winston1977optimality-JSQoptimal, weber1978optimal-JSQoptimal, whitt1986deciding-JSED-JSQnotoptimal}. Since exact analysis for a policy is often intractable, several asymptotic regimes have been explored in the literature \cite{der2022scalable-literaturereview}, including the classical heavy-traffic regime and the many-server heavy-traffic regime. In some settings, JSQ and Po$d$ have been shown to be delay optimal if they are throughput optimal \cite{hurtado2021throughput-Pod-stability-heterogeneous, eryilmaz2012asymptotically-driftmethod} and their queue-length distributions have been characterized in heavy traffic \cite{hurtado2021throughput-Pod-stability-heterogeneous, hurtado2022heavy-JSQoptimalEquation, hurtado2022load-manyserver-JSQ}. Furthermore, a policy can be delay optimal in heavy traffic if it satisfies certain conditions that prioritize sending jobs to shorter queues \cite{zhou2017designing-JBT, zhou2019flexible-MultiDimensionSSC, zhou2018degree-generalPpolicy}. These conditions also depend on tailored policy designs and are not applicable to general policies. It remains an open question how to determine whether a general policy is delay optimal in heavy traffic. 

In practice, policy design is often constrained by communication budgets and imperfect state information. For example, JSQ and JSED can achieve excellent delay performance, but they typically require frequent system-wide queries and accurate queue-length information, leading to high communication overhead (e.g., 2$n$ messages per request). This motivates information-efficient alternatives such as Po$d$, pull-based schemes (e.g., JIQ \cite{lu2011join-JIQ}, JID \cite{gamarnik2022stability-memoryenhanced-JoinIDQ} and JBT-$r$ \cite{zhou2018heavy-JBT-r}), and policies based on historical information or estimation (e.g., RR-type variants \cite{ye2025optimal-optimalRRpolicy} and CARE \cite{mendelson2026load-Estimationtriggered-general}). Pull-based policies and CARE-type schemes can attain near-optimal delay under sparse communication through event-triggered server-to-dispatcher updates; CARE further maintains queue-length estimates by combining communication and approximation modules. However, in some deployments, pull-based signaling is impractical and accurate state estimation is infeasible or unreliable. Moreover, even when queue-state updates are available, they may only be refreshed intermittently. Consequently, the dispatcher must assign jobs under sparse and stale information.

When a policy must rely on stale queue-length information, aggressively exploiting the shortest queue may lead to poor delay performance. \cite{mitzenmacher1997useful-T-JSQ-herd-simulation} shows that naively applying JSQ can trigger herding effects, where consecutive arrivals are sent to the same server, leading to transient overload. Despite their practical relevance, dispatching policies with intermittently refreshed updates remain comparatively underexplored, particularly with respect to stability and delay performance \cite{lipshutz2019open-openproblem}.

\subsection{Our Contribution}
In this paper, we study load-balancing policies with periodic information updates in heterogeneous service systems. We focus on a broad family of policies whose dispatching decisions are based on the ordering of possibly scaled queue lengths, including RAND, WRAND, JSQ, JSED, and Po$d$. Existing analytical tools are limited in this information-constrained setting: they are typically policy-specific, rely on current queue-length information, or involve stationary quantities that are difficult to characterize in heterogeneous systems. To address these limitations, we develop a unified analytical framework for analyzing stability and heavy-traffic delay performance across this broad policy class.

Our analysis reveals fundamental mechanisms behind stability and heavy-traffic delay performance. At a high level, the behavior of ordering-based policies can be understood through the geometry of their ordering rules and the majorization relations they induce among queues. This perspective reduces complex dispatching dynamics to explicit stability and delay characterizations, and further guides the design of new policies that remain effective under sparse communication and stale queue-length information.

Our main contributions are as follows:
\begin{enumerate}
    \item \textbf{Unified policy framework.}
    We formulate a general policy class $\Pi$ (Section \ref{subsec:general_framework}) for dispatching under periodic queue-length updates. Within each update interval, policies in $\Pi$ make assignment decisions based only on the induced permutation of scaled queue lengths. This abstraction provides a common representation for analyzing and comparing a broad range of load-balancing policies in heterogeneous service systems.

    \item \textbf{Stability characterization, throughput optimality, and nonmonotone stability.}
    For any policy \(\pi\in\Pi\), we derive a closed-form sufficient condition for positive recurrence based on partial-sum load comparisons over the longest \(\boldsymbol{\gamma}\)-scaled queues (Theorem~\ref{thm_our:stable_withrate_withratio}). This condition induces a static sufficient-stability threshold \(h^*\), determined by the service rates and the dispatch fractions induced by each sampled permutation. The same partial-sum structure also yields a majorization-based sufficient condition for throughput optimality (Corollary~\ref{co_our:condition_withrate_withratio}). We then show that, for the general policy class \(\Pi\), the threshold \(h^*\) need not be necessary and stability above \(h^*\) need not be monotone in the arrival rate. To classify such load-dependent behavior, we introduce two fixed-load properties, persistent bottleneck dominance (PBD) and sequential catch-up drainage (SCD), which certify transient and stable load windows, respectively (Theorem~\ref{thm_our:fixed_load_stability_certificates} and Corollary~\ref{cor_our:load_window_certificate_classification}). Finally, under a global persistent bottleneck dominance condition, \(h^*\) becomes a necessary threshold, yielding an exact stability characterization (Theorem~\ref{thm_our:stable_withrate_withratio_necessary_general}).
    
    \item \textbf{Heavy-traffic delay optimality and characterization.}
    Under the strict majorization condition, we establish state-space collapse and derive an upper bound on the steady-state queue length (theorems~\ref{thm_our:SSC_withrate} and~\ref{thm_our:SSCUpper_withrate_withratio}). These results imply asymptotic delay optimality in heavy traffic (Corollary~\ref{co_our:delayoptimality}) and clarify the role of stale information: the error terms introduced by periodic updates are asymptotically negligible when the update interval is fixed or grows sufficiently slowly. We further characterize the steady-state queue-length distribution in the heavy-traffic limit (Theorem~\ref{thm_our:queuelength_distribution}).
    
    \item \textbf{Policy design under sparse communication.}
    The framework quantifies how communication can be reduced without sacrificing heavy-traffic delay optimality. In particular, if the interval between information updates grows more slowly than $1/\sqrt{\epsilon}$, where $\epsilon$ denotes the capacity slack defined in Section~\ref{subsec:DTQS}, then heavy-traffic delay optimality is preserved while communication per time slot and per arriving job can vanish. We apply the framework to Po$d$-ED and characterize its stability, throughput optimality, and heavy-traffic delay optimality (Corollary~\ref{co_our:Pod-ED}). We also propose $k$-WRAND-SLQ-$d$, a communication-efficient policy family that achieves both throughput optimality and heavy-traffic delay optimality (Corollary~\ref{co_our:k-WRAND-SLQ-d}). Our simulations show that, under stale information, \(k\)-WRAND-SLQ-\(d\) with appropriate parameters can yield lower delays than stale JSQ under the same communication budget.
    
    \item \textbf{Methodology.}
    Methodologically, we extend the drift-based heavy-traffic approach of \cite{eryilmaz2012asymptotically-driftmethod} and the multi-step Lyapunov drift method of \cite{Yishun2025k-slq-d} to general policies with periodic queue-length updates. The main challenge is that dispatching decisions within a dispatching cycle are based on the ordering sampled at the beginning of the cycle, which need not coincide with the current ordering as queue lengths evolve. We develop a \(T\)-step drift framework that connects these periodic decisions to the policy-induced dispatch fractions \(f_{l,\eta}\) defined in Section~\ref{subsec:NotationfromFramework}, which enables the stability and state-space-collapse analyses. To handle load-dependent stability behavior, we also combine moving-barrier supermartingale arguments with fluid limit Lyapunov arguments: the former certifies persistent bottleneck growth, while the latter certifies stability through sequential catch-up and finite-time draining of fluid limits. For the limiting distribution, we extend the transform method of \cite{hurtado2020transform-Transformmethod} from one-step to \(T\)-step identities.
\end{enumerate}

\subsection{Basic Notation} \label{subsec:basic_notation}
We use $\mathbb R$ to denote the set of real numbers and $\mathbb R_+$ to denote the set of nonnegative real numbers. We let $\mathbb N=\{0,1,2,\ldots\}$ denote the set of natural numbers, and $\mathbb R^n$ denote the set of $n$-dimensional real vectors. We use bold letters to denote vectors and, for any vector $\mathbf x$, we use $x_i$ to denote the $i$th coordinate of $\mathbf x$. The standard inner product of two vectors $\mathbf x,\mathbf y\in\mathbb R^n$ is defined as $\langle \mathbf x,\mathbf y\rangle=\mathbf x^T\mathbf y=\sum_{i=1}^n x_i y_i$. For any vector $\mathbf x\in\mathbb R^n$, the $\ell_2$-norm is denoted by $\|\mathbf x\|_2=\sqrt{\langle \mathbf x,\mathbf x\rangle}$, and the $\ell_1$-norm is denoted by $\|\mathbf x\|_1=\sum_{i=1}^n |x_i|$. For a positive vector $\boldsymbol\gamma\in\mathbb R_+^n$, we define $\gamma_{\min}:=\min_{i\in[n]}\gamma_i$ and $\gamma_{\max}:=\max_{i\in[n]}\gamma_i$. We also define the $\boldsymbol\gamma$-weighted inner product by $\langle \mathbf x,\mathbf y\rangle_\gamma:=\sum_{i=1}^n \gamma_i x_i y_i$, and the corresponding weighted norm by $\|\mathbf x\|_\gamma^2:=\langle \mathbf x,\mathbf x\rangle_\gamma=\sum_{i=1}^n\gamma_i x_i^2$. For any set $\mathcal A$, $\mathbf 1_{\mathcal A}$ denotes the indicator random variable for $\mathcal A$. We use $\lfloor\cdot\rfloor$ and $\lceil\cdot\rceil$ to denote the floor and ceiling functions, respectively. We let $[n]=\{1,2,\ldots,n\}$ and let $\mathcal S_n$ denote the set of all permutations of $[n]$. Let $e_i\in\{0,1\}$ for all $i\in[n]$ denote an indicator variable.

\section{Literature Review}

\subsection{Load Balancing Policies} \label{lr:policy}
Load balancing policies have been extensively studied in $n$-server parallel systems that operate in both homogeneous and heterogeneous settings. These policies exhibit distinct trade-offs in information requirements, communication overhead, and delay performance. 

In homogeneous systems, among the simplest are the RAND and RR \cite{roundRobin1991-optimalnocommunication}, both of which require neither system state information nor communication. However, this simplicity comes at the cost of poor delay performance. At the other end of the spectrum, JSQ achieves delay optimality \cite{winston1977optimality-JSQoptimal, weber1978optimal-JSQoptimal}, but it relies on system-wide state information and thus incurs substantial communication overhead. To balance performance with communication efficiency, Po$d$ \cite{vvedenskaya1996queueing-Pod-Po2-original, mitzenmacher2002power-Po2first, mukherjee2018universality-Pod-manyserver} has been proposed. Po$d$ significantly improves delay performance over RAND and RR, while requiring considerably less communication than JSQ. To further reduce communication overhead, the dispatcher can utilize local state information. When equipped with local memory, the dispatcher can implement memory-enhanced policies such as JIQ \cite{lu2011join-JIQ, badonnel2008dynamic-JIQ}, Join the Open Queue (JOQ) policy \cite{van2020zero-memoryenhanced-JOQ}, Po$d$ with memory \cite{mitzenmacher2002load-podwithMemory, shah2002use-memoryenhanced-podm-bin, anselmi2020power-Podwithmemory}. More generally, \cite{gamarnik2018delay-memoryenhanced-pullbased} indicates that delay performance is closely related to the availability of system resources, specifically memory capacity and message exchange rate.

In heterogeneous systems, policies such as RAND, RR and even Po$d$ may be unstable \cite{gardner2021scalable-podJIQgeneral-fastslow}, because slow servers may receive more jobs than they can handle. To address these issues, load balancing policies can take advantage of current state information, service rate information and memory to improve performance and reduce communication frequency. On one hand, some policies incorporate service rate information, current queue lengths and system status when selecting servers or assigning jobs. Policies such as WRAND, Balanced Po$d$ \cite{chen2012asymptotic-Pod-Heterogeneous-Proportional}, and Balanced RR \cite{ye2025optimal-optimalRRpolicy} select servers proportional to their service rates. JSED \cite{adan1996shortest-JSED-earlang} and Speed-Aware JSQ (SA-JSQ) \cite{bhambay2022asymptotic-speedawareJSQ} leverage service rates in the assignment process, with the aim of minimizing the expected delay (waiting time in queue). Along the same lines, power-of-$d$ variants have been proposed. For example, the Po$d$-Expected Delay (Po$d$-ED) policy \cite{gardner2021scalable-podJIQgeneral-fastslow, abdul2022general-Pod-GeneralmodifiedPod} uses service-rate information, whereas the Po$d$-Expected Sojourn Time (Po$d$-ET) policy \cite{abdul2022general-Pod-GeneralmodifiedPod, selen2016steady-JSED} uses both service-rate and job-age information in the assignment decision. Furthermore, as shown in the general frameworks of Po$d$ and JIQ \cite{gardner2021scalable-podJIQgeneral-fastslow, abdul2022general-Pod-GeneralmodifiedPod}, service-rate information is used during both server selection and job assignment. On the other hand, memory-enhanced policies have also been developed, such as JIQ \cite{stolyar2015pull-JIQ-manyserver}, Join the stored ID (JID) policy \cite{gamarnik2022stability-memoryenhanced-JoinIDQ}, Join Below Threshold (JBT-$r$) policy \cite{zhou2018heavy-JBT-r}, Idle-One-First (I1F) policy \cite{gupta2019load-JSQ-NSQ-IDle-1-First}, Persistent Idle (PI) policy \cite{atar2020persistent-JIQ-memoryenhanced-Drift}, Local Shortest Queue (LSQ) policy \cite{vargaftik2020lsq-LSQ}, Local Estimation Driven (LED) policy \cite{zhou2021asymptotically-LED-meomory}, $k$-Skip-the-$d$-Longest-Queues ($k$-SLQ-$d$) policy \cite{Yishun2025k-slq-d} and RR variants that use arrival and service history (RR-AC, RR-SC) \cite{ye2025optimal-optimalRRpolicy}. Moreover, when the dispatcher and the servers can approximate system state from observations (e.g., queue lengths, service rates, or workload) and use these approximations to make assignment decisions and report status, approximation-enhanced policies such as Synchronized-Updates or Asynchronized-Updates Join the Shortest Queue (SUJSQ and AUJSQ) policy \cite{van2019hyper-Estimation-JSQ}, CARE \cite{mendelson2026load-Estimationtriggered-general} can be employed.

Among all load balancing policies, we focus on a family of policies that make dispatching decisions based on the order of (possibly scaled) queue lengths, including but not limited to, JSQ, JSED, Po$d$, Balanced Po$d$ and $k$-SLQ-$d$. Our general framework also accommodates RAND, RR, WRAND and Balanced RR, as special cases. To the best of our knowledge, we provide the first unified analytical characterization of this family of policies in heterogeneous systems.

\subsection{Stability of Load Balancing Policies}\label{lr:stability}
In queueing theory, a system becomes unstable when the arrival rate to any server exceeds its service rate over time, leading to unbounded growth of the queue. Policies such as WRAND aim to achieve throughput optimality by controlling the load directed to each server. Meanwhile, JSQ is known to be throughput optimal in general settings \cite{bramson2011stability-JSQoptimal}, indicating that systems are stable if jobs are sent to the shortest queue. Other policies like Po$d$ and JIQ combine these two ideas in such a way that they tend to send more jobs to shorter queues while sending fewer jobs to longer queues than RAND. 

In homogeneous systems, \cite{moyal2022stability-generalhomogeneous} generalizes RAND by introducing $\mathbf{p}$-allocation policies, where the $l$th element specifies the probability of assigning an arrival to the $l$th shortest queue. They show that policies such as JSQ, Po$d$, and other order-based policies are throughput optimal if their corresponding $\mathbf{p}$-allocation vector is dominated by the uniform vector in generalized Schur-convex order. In addition, the authors show that policies that assign more jobs even to the second shortest queue may not be throughput optimal.

In heterogeneous systems, \cite{zhou2017designing-JBT} characterizes the throughput optimality of a policy by its tendency to assign more jobs to shorter queues than to longer ones, relative to WRAND. This condition requires a specifically tailored policy design. \cite{zhou2018degree-generalPpolicy} shows that a policy is throughput-optimal if, in steady-state, it outperforms WRAND by exhibiting a strong preference for shorter queues. However, verifying this condition requires evaluating, under the stationary distribution, the expected service rate of the server occupying the $l$th position in the queue length ordering, which is not available in closed form and can be difficult to check. \cite{gamarnik2022stability-memoryenhanced-JoinIDQ} establishes a relationship between memory requirements, message exchange rate, and throughput optimality. For Po$d$, \cite{hurtado2021throughput-Pod-stability-heterogeneous} provides a necessary and sufficient condition for stability in heterogeneous systems. 

In summary, for heterogeneous systems, it remains unclear whether a general load balancing policy is throughput optimal, and how to characterize its stability region when it is not. We derive a sufficient condition for stability and show that it becomes necessary under a drift-dominance assumption. We further provide a majorization-type partial-sum condition for throughput optimality.
 
\subsection{Heavy-traffic Delay Optimality}\label{lr:delay}
One of the most common asymptotic regimes in the literature is the classical heavy-traffic regime \cite{kingman1962queues-classic-heavytraffic, halfin1981heavy-heavytraffic-region}. We focus on this regime, where the number of servers is fixed and the system load approaches one.

The key behavior required for a policy to achieve heavy-traffic delay optimality is known as state-space collapse, where the multidimensional queue-length vector concentrates on a single dimension or a lower-dimensional subspace. This property has been used to analyze the heavy-traffic delay optimality of JSQ \cite{eryilmaz2012asymptotically-driftmethod, Foschini1978-SSC-JSQ} and Po$d$ \cite{hurtado2020transform-Transformmethod}. \cite{zhou2019flexible-MultiDimensionSSC} also shows that a policy can achieve heavy-traffic delay optimality under a multidimensional state-space collapse.

Among all policies, JSQ has been established to be delay optimal in the classical heavy-traffic regime, under both continuous-time \cite{winston1977optimality-JSQoptimal, weber1978optimal-JSQoptimal} and discrete-time settings \cite{eryilmaz2012asymptotically-driftmethod, hurtado2022heavy-JSQoptimalEquation, hurtado2020transform-Transformmethod}. Furthermore, under JSQ, the queue-length vector is shown to converge to a multidimensional exponential random vector with identical elements \cite{hurtado2020transform-Transformmethod}. For other policies, Po$d$ has been shown to be delay optimal in heavy traffic as long as it is throughput optimal and its queue length converges to the same distribution as that of JSQ \cite{ hurtado2021throughput-Pod-stability-heterogeneous, hurtado2020transform-Transformmethod, maguluri2014heavy-Pod-Po2-switch}. JIQ is not delay optimal in heavy traffic \cite{zhou2017designing-JBT}, whereas JBT-$r$ has been shown to achieve delay optimality in this regime \cite{zhou2018heavy-JBT-r}. $k$-SLQ-$d$ achieves delay optimality by choosing an appropriate $d$ \cite{Yishun2025k-slq-d}. RR-AC and RR-SC have been shown to be delay optimal in heavy traffic with appropriate parameters \cite{ye2025optimal-optimalRRpolicy}. In addition, two conditions have been provided for a policy to be delay optimal in heavy traffic, one is that the policy satisfies a "Long-term Dispatching Preference" condition \cite{zhou2018degree-generalPpolicy}, and the other is that the policy strictly prefers sending more jobs to shorter queues than WRAND \cite{zhou2019flexible-MultiDimensionSSC}. However, only policies with a specific design satisfy these two conditions. For policies using memory, such as LED, it is delay optimal in heavy traffic if the local estimation of the queue length is close to its true value \cite{zhou2021asymptotically-LED-meomory}.

Prior work has established heavy-traffic delay optimality for specific policies or carefully designed families of policies. How to verify heavy-traffic delay optimality for a general policy remains open. We provide an easily verifiable sufficient condition that guarantees heavy-traffic delay optimality and derive the asymptotic queue-length distribution.

\subsection{Policy Performance with Stale Information}\label{lr:stale_information}
Policies such as JSQ and Po$d$ rely on timely state information from the dispatcher. When queue-length information is available at each dispatching epoch, \cite{zhou2018degree-generalPpolicy} proposes a measure termed the degree of dispatching preference and shows via simulation that policies with a stronger preference for shorter queues tend to yield better delay performance. However, when queue-length information is stale, a stronger preference based on the available queue-length estimates need not translate into better delay performance.

In practice, queue-length information may be stale due to periodic updates, communication latency or physical distance between dispatcher and servers. Such staleness can induce herding effects: many arrivals make correlated decisions based on the same outdated snapshot and temporarily overload a subset of servers. \cite{mitzenmacher1997useful-T-JSQ-herd-simulation} documents this effect for periodically updated JSQ and Po$d$, and \cite{zhu2020racksched-T-JSQ-herd-implentation-serverlevel} reports similar behavior in a system with delayed updates. For settings with delayed queue-length information, \cite{atar2021heavy-JSEstimatedQ} proposes the Join the Shortest Estimated Queue (JSEstQ) policy and derives its heavy-traffic diffusion limit. \cite{He2025-RJSQ-physicaldistance} proposes the Randomized Join the Shortest Queue (RJSQ) policy for settings with physical distance and shows that RJSQ is asymptotically delay optimal in heavy traffic.

Our framework analyzes policies in settings where queue-length information is updated periodically. We establish the existence of herding effects through simulation and design a new policy that achieves better delay performance than JSQ and Po$d$ under the same communication budget.

\subsection{Methodology}\label{lr:methodology}
Several approaches have been used to establish stability and delay optimality in heavy traffic, both in the continuous-time and discrete-time settings. In the continuous-time setting, \cite{ye2025optimal-optimalRRpolicy} uses the diffusion limit method to study RR-AC and RR-SC. In the discrete-time setting, the drift method, based on Foster-Lyapunov theory, is widely used \cite{hajek1982hitting-lyapunovupperbound, dai2020processing-ForstLyapunovThm}, including in single-period models \cite{hurtado2021throughput-Pod-stability-heterogeneous, zhou2017designing-JBT, zhou2019flexible-MultiDimensionSSC, eryilmaz2012asymptotically-driftmethod, zhou2018heavy-JBT-r, hurtado2020transform-Transformmethod} and in multi-period models \cite{zhou2018degree-generalPpolicy,zhou2021asymptotically-LED-meomory, Yishun2025k-slq-d}. Moreover, the transform method \cite{hurtado2021throughput-Pod-stability-heterogeneous, hurtado2022load-manyserver-JSQ, hurtado2020transform-Transformmethod}, Stein's method \cite{hurtado2022load-manyserver-JSQ, gaunt2020stein-Stein-single-server, zhou2020note-steinsmethod}, and the basic adjoint relationship (BAR) \cite{guang2025steady-BAR-continuous-distribution} method have been used to characterize the queue-length distribution in heavy traffic.

In this paper, we extend the multi-period drift method in \cite{Yishun2025k-slq-d} to general policies with periodic queue-length updates. To analyze load-dependent stability behavior, we combine moving-barrier supermartingale arguments with fluid limit Lyapunov arguments, which are used to establish the persistent bottleneck dominance and sequential catch-up drainage certificates. We also extend the transform method in \cite{hurtado2020transform-Transformmethod} to the multi-period setting.

\section{Queueing Model and General Dispatching Framework }
In this section, we present the discrete-time queueing model (Section~\ref{subsec:DTQS}), introduce a general dispatching framework (Section~\ref{subsec:general_framework}), and establish key notation for this framework (Section~\ref{subsec:NotationfromFramework}), which will be used throughout the analysis in sections~\ref{sec:stability_throughputoptimality} and~\ref{sec:heavytraffic_queuelengthdistribution}.

\subsection{Discrete-time Queueing Model} \label{subsec:DTQS}

We consider a discrete-time queueing system consisting of $n$ single-server FIFO queues with infinite buffers. Let $\q(t)=(Q_1(t),\dots,Q_n(t))$ be the queue-length vector at the beginning of time slot $t$, where $Q_l(t)$ is the length of the $l$th queue. 

Jobs arrive in the system in each time slot according to a process $ \{A(t) \}_{t\geq0}$, where $A(t)$ denotes the number of arrivals in time slot $t$. Upon the arrival of $A(t)$ jobs at time $t$, a dispatching policy (specified later) assigns all of them to a single queue with index $I^*(t) \in [n]$. After this, server $l$ can serve up to $S_l(t)$ jobs in slot $t$; let $\mathbf S(t)=(S_1(t),\ldots,S_n(t))$ denote the vector of potential services. We impose the following assumptions on the arrival and potential service processes.

\begin{assumption}[Arrivals]\label{ass:arrivals}
The arrival process $\{A(t)\}_{t\ge 0}$ is i.i.d.\ across time slots and takes values in $\mathbb{N}$, with $\mathbb{E}[A(1)]=n\lambda$, $\text{Var}(A(1))=n\sigma_\lambda^2$, and $A(1)\le nA_{\max}$ a.s.
\end{assumption}

\begin{assumption}[Potential Services]\label{ass:service}
For each server $l\in[n]$, the potential service process $\{S_l(t)\}_{t\ge 0}$ is i.i.d.\ across time slots and takes values in $\mathbb{N}$, with $\mathbb{E}[S_l(1)]=\mu_l$, $\text{Var}(S_l(1))=\sigma_l^2$, and $S_l(1)\le S_{\max}$ a.s. Moreover, the processes $\{S_l(t)\}_{t\ge 0}$ are independent across $l\in[n]$.
\end{assumption}
Without loss of generality, we order the service rates so that $\mu_1 \leq ... \leq \mu_n$. We let $\mu_{\min}:=\mu_1$ and $\mu_{\max}:=\mu_n$.

\begin{assumption}[Independence]\label{ass:indep}
The arrival process $\{A(t)\}_{t\geq0}$ is independent of $\{S_l(t)\}_{t\geq0}$ for all $l\in[n]$.
\end{assumption}
These assumptions are commonly adopted in the literature on discrete-time settings \cite{zhou2017designing-JBT, zhou2019flexible-MultiDimensionSSC, zhou2018degree-generalPpolicy, eryilmaz2012asymptotically-driftmethod, hurtado2022heavy-JSQoptimalEquation, zhou2018heavy-JBT-r, hurtado2020transform-Transformmethod}. In this discrete-time formulation, the dispatcher makes one assignment decision per time slot; hence, all jobs arriving within the same slot are treated as a batch and assigned to a single queue. This model can be interpreted as a discretization of a continuous-time system in which assignment decisions are updated only at discrete decision epochs. Alternatively, one may view each slot as the arrival of a single job that consists of multiple subtasks, all of which must be assigned to the same server.


To describe the queue dynamics, we denote the dispatching action chosen by the dispatcher by ${\bf Z}(t) \in \{0,1\}^n$ such that $Z_{j}(t) = 1$ for $j = I^*(t) $ and $Z_{j}(t) = 0$ otherwise. For each $l \in[n]$, let
\[ A_l(t):=A(t)Z_l(t) \]
denote the number of arrivals routed to queue $l$ at time $t$. Then, for each $l\in[n]$,
\begin{equation*}
    Q_l(t+1)=\left[ Q_l(t)+A_l(t)-S_l(t) \right ]^+.
\end{equation*} 
Equivalently, the queue length process is given by
\begin{equation}
\label{eq: jsq_lindley}
    \q(t+1) = [\q(t) +A(t){\bf Z}(t) -{\bf S}(t) ]^+ = \q(t) +A(t){\bf Z}(t) -{\bf S}(t) + {\bf U}(t) ,
\end{equation}
where ${\bf U}(t)$ represents the unused services that arise because there might not be enough jobs to serve. Note that the unused-service term $U_i\ti $ is positive only if $Q_i(t+1) =0$, which implies $Q_i(t+1)U_i(t) =0$ for all $i$, or simply $\langle \q(t+1),{\bf U}(t) \rangle =0$ for any $t\geq0$. Also, the unused services cannot be larger than the service itself, so we have $0\leq U_i(t)\leq S_i(t) \leq S_{\max}$. We drop the dependence on $t$ to denote the variables in steady state, i.e., $\q$ follows the steady-state distribution of the queue length process $\{\q(t)\}_{t=0}^\infty$.

We define the capacity slack as 
\begin{align*}
    \epsilon := \sum_{l=1}^{n} \mu_l - n\lambda,
\end{align*}
and we study the asymptotic behavior of the steady-state queue length in the heavy-traffic regime as $\epsilon \downarrow 0 $.

\subsection{General Dispatching Policy Framework} \label{subsec:general_framework} 
We construct a general framework $\Pi$ that includes a broad family of policies such as RAND, RR, JSQ, Po$d$, JSED, and others (see Example~\ref{ex:framework_policy}). 
\begin{definition}[Dispatching policy framework $\Pi$]\label{def:policy_framework}
    A policy $\pi\in\Pi$ is characterized by a tuple
    \[
        (T,\boldsymbol{\gamma},L(\cdot),Y(\cdot)),
    \]
    which specifies the dispatching cycle, the sorting function, and the decision function as follows.
    \begin{enumerate}
        \item Dispatching cycle.
        \begin{enumerate}
            \item A positive integer $T$, which specifies the dispatching cycle. \\ 
            The parameter $T$ means that queue-length information is sampled once per cycle, and the decision sequence generated from that sample is used for the next $T$ time slots.
        \end{enumerate}
        \item Sorting function.
        \begin{enumerate}
            \item A vector with all positive components $\boldsymbol{\gamma} \in \mathbb{R}_+^n$. Define the $\boldsymbol{\gamma}$-scaled queue-length vector at time $r T$ for all $r \in \mathbb{N}$ by
            \[ \mathbf{Q}^{(\boldsymbol{\gamma})}(rT):= \left( \frac{Q_l(rT)}{\gamma_l} \right)_{l=1}^n. \]
            Existing policies commonly adopt $\boldsymbol{\gamma}=\mathbf{1}$ and $\boldsymbol{\gamma}=\boldsymbol{\mu}$. The $\mathbf{1}$-scaled queue length coincides with the raw queue length, while the $\boldsymbol{\mu}$-scaled queue length $Q_l/\mu_l$ is often interpreted as the expected waiting time at queue $l$ in heterogeneous systems. In this paper, delay refers to the waiting time in queue.
    
            \item Randomization variables $\{V_r\}_{r\geq 0}$, which are i.i.d., Unif$[0,1]$.
            \item A permutation
            \begin{align*} \label{eq:eta_k}
                \eta_r := L\left (\mathbf{Q}^{(\boldsymbol{\gamma})}(rT),V_r \right)
            \end{align*} 
            where $L(\mathbf{Q}^{(\boldsymbol{\gamma})}(rT),V_r) :\mathbb{R}_+^n \times [0,1] \rightarrow \mathcal{S}_n$ is such that
            \[ Q^{(\boldsymbol{\gamma})}_{\eta_r(1)}(rT) \geq ... \geq Q^{(\boldsymbol{\gamma})}_{\eta_r(n)}(rT), \]
            Thus, $\eta_r$ records the ranking of queues based on the sampled $\boldsymbol{\gamma}$-scaled queue lengths, from the largest to the smallest.
        \end{enumerate}
        \item Decision function.
        \begin{enumerate}
            \item Randomization variables $\{W_r\}_{r\geq 0}$, which are i.i.d., Unif$[0,1]$.
            \item Decision vectors $\phi(r) := Y(\eta_r,\boldsymbol{\mu},W_r)$ where 
            \[
                Y(\eta_r,\boldsymbol{\mu},W_r):\mathcal{S}_n \times (0,\infty)^n \times [0,1]\rightarrow \left\{ (e_1,\ldots,e_n)\in \{0,1\}^n: \sum_{i=1}^n e_i=1 \right\}^{T}
            \]
            The decision function $Y$ maps the permutation $\eta_r$, service rates $\boldsymbol{\mu}$, and randomization $W_r$ to a sequence of dispatching decisions used during cycle $r$. Specifically, $\phi(r)=(\phi_1(r),\ldots,\phi_T(r))$, where each $\phi_j(r)$ is an $n$-dimensional $0$-$1$ vector with exactly one component equal to one. The component equal to one indicates the queue selected in the $j$-th time slot of cycle $r$. Thus, for all $t$, $ \left \| \phi_{t-\lfloor t/T \rfloor T+1}(\lfloor t/T \rfloor) \right \|_1=1 $.  \hfill $\square$
        \end{enumerate} 
    \end{enumerate} 
\end{definition}
\begin{remark}[Normalization of the scaling vector] \label{rmk_our:normalized_scaled_vector}
    The scaling vector $\boldsymbol{\gamma}$ is defined only up to a positive multiplicative constant. Indeed, for any constant $c>0$, replacing $\boldsymbol{\gamma}$ by $c\boldsymbol{\gamma}$ gives
    \[
        \frac{Q_l(t)}{c\gamma_l} = \frac{1}{c}\frac{Q_l(t)}{\gamma_l},
        \qquad \forall \, l\in[n],
    \]
    and therefore does not change the ordering of the scaled queue lengths. Hence, the induced permutation $\eta_r$ remains unchanged. Without loss of generality, throughout the paper we normalize the scaling vector so that $\|\boldsymbol{\gamma}\|_1=1$. \hfill \(\square\)
\end{remark}
Within this framework, at time $t$, incoming jobs are sent to queue
\[
    I^*(t) := \arg\max_{i\in[n]} \left\{ \left[ \phi_{t-\lfloor t/T\rfloor T+1}(\lfloor t/T\rfloor) \right]_i \right\}.
\]
Then, the corresponding decision vector $\mathbf{Z}(t)$ is
\[ \mathbf{Z}(t)=\phi_{t- \lfloor t/T \rfloor T +1}(\lfloor t/T \rfloor). \]
At each sampling epoch $t=rT$, $r\in\mathbb{N}$, the dispatcher observes the queue length information and applies the sorting function $L(\cdot)$ to the $\boldsymbol{\gamma}$-scaled queues. Based on the resulting permutation $\eta_r$, the decision function $Y(\cdot)$ generates the $T$-slot decision sequence $\{ I^*(rT),\ldots,I^*((r+1)T-1) \}$, which is applied over slots $t\in\{rT,\ldots,(r+1)T-1\}$. 

This framework also allows us to quantify the communication overhead induced by a policy $\pi$. Let $m(\pi)$ denote the total number of messages exchanged between the dispatcher and servers at each sampling epoch. We define the average number of messages per time slot and per job as
\begin{align*}
    M^{\pi}_{\text{slot}} := \frac{m(\pi)}{T}, \quad M^{\pi}_{\text{job}} := \frac{m(\pi)}{n\lambda T},
\end{align*}
where $n\lambda$ is the expected number of arrivals per time slot. For instance, under JSQ the dispatcher communicates with all $n$ servers at each sampling epoch, so $m(\mathrm{JSQ}) = 2n$. Under Po$d$, the dispatcher queries $d$ servers and receives $d$ responses, so $m(\mathrm{Po}d) = 2d$.

\begin{example}[Representation of Common Policies within the Framework]\label{ex:framework_policy}
    Table \ref{tab:policies} illustrates how this framework can describe several common policies. JSQ, JSED and $\mathrm{Po}d$ typically set $T=1$ to update information every time slot while $k$-SLQ-$d$ policy sets $T=k(n-d)\geq 1$ to update information periodically. JSQ, Po$d$ and $k$-SLQ-$d$ use $\boldsymbol{\gamma}=(1,...,1)$ and JSED uses $\boldsymbol{\gamma} = (\mu_1,...,\mu_n)$. Although Po$d$ observes only $d$ servers at each decision epoch, it is naturally interpreted as operating with the global ordering of all servers induced by their queue lengths: it samples $d$ servers and selects the one that is smallest with respect to this global order. 
    
    It is also important to note that dispatching decisions differ across policies. For RAND, WRAND and RR, the assignment decision is independent of the service rate and the scaled queue length, while WRAND incorporates the service rate into the decision function. JSQ, Po$d$, $k$-SLQ-$d$ and JSED make their decisions based on the (scaled) queue lengths. \hfill $\square$
    
    \begin{table}[h!]
        \centering
        \footnotesize
        \caption{Components of common policies under the general framework}
        \label{tab:policies}
        \begin{tabular}{lllll}
            \hline
            \textbf{Policy} & \textbf{$T$} & \textbf{$\boldsymbol{\gamma}$} & \textbf{Sorting Function} & \textbf{Decision Function} \\
            \hline
            RAND & 1 & $(1,\ldots,1)$ & -- & assign uniformly \\
            WRAND & 1 & $(1,\ldots,1)$ & -- & assign proportionally to service rates \\
            RR & $n$ & $(1,\ldots,1)$ & -- & assign in round-robin \\
            JSQ & 1 & $(1,\ldots,1)$ & sorted by queue length & assign to shortest queue \\
            Po$d$ & 1 & $(1,\ldots,1)$ & sorted by queue length & assign to shortest queue among $d$ choices \\
            JSED & 1 & $(\mu_1,\ldots,\mu_n)$ & sorted by scaled queue length & assign to shortest scaled queue \\
            $k$-SLQ-$d$ & $k(n-d)$ & $(1,\ldots,1)$ & sorted by queue length & assign in round-robin except the $d$ longest queues \\
            \hline
        \end{tabular}    
        
        \vspace{0.5em}
        \begin{minipage}{0.95\textwidth}
            \footnotesize
            \textit{Note.} RAND, WRAND, and RR do not require queue-length information. They are included in the framework by viewing their decision rules as generated over the specified cycle: RAND and WRAND make one-slot randomized decisions, while RR follows a deterministic $n$-slot round-robin cycle.
        \end{minipage}
    \end{table}
\end{example}

\subsection{Markov Chain Representation under the Framework $\Pi$}
After representing a policy $\pi$ within this framework, we model the resulting system as a discrete-time Markov chain $\{\mathbf{X}(t)\}_{t\geq 0}$ such that
\begin{equation}
    \mathbf{X}(t) := \Big(\mathbf{Q}(t),\,\, \eta_{\lfloor t/T \rfloor}, \,\,\phi( \lfloor t / T \rfloor ),\,\, t - \lfloor t / T \rfloor T \Big).  \label{def:Mkchain}
\end{equation}
Here, $\mathbf{Q}(t)\in\mathbb{N}^n$ is the queue-length vector at time $t$, $\eta_{\lfloor t/T \rfloor}\in\mathcal{S}_n$ is the permutation generated from the queue-length information sampled at the most recent observation epoch $\lfloor t/T \rfloor T$, and $\phi(\lfloor t/T \rfloor)$ is the sequence of decision vectors used during the current dispatching cycle. The last component $t-\lfloor t/T \rfloor T\in\{0,1,\ldots,T-1\}$ records the position of time slot $t$ within the current cycle. In particular, the decision vector used at time $t$ is $ \phi_{t-\lfloor t/T \rfloor T+1} (\lfloor t/T \rfloor )$. We also consider the embedded queue-length process
\begin{align*}
    \tilde{\mathbf{Q}}(t) := \mathbf{Q}\big( Tt\big), \qquad t\ge 0.    
\end{align*}
Since $\eta_{\lfloor t/T \rfloor}$, $\phi(\lfloor t/T \rfloor)$, and $t-\lfloor t/T \rfloor T$ keep track of the intermediate ranking, decision sequence, and position within the cycle, and since a new ranking and decision sequence are generated from the updated queue-length vector at the beginning of each cycle, $\{\tilde{\mathbf{Q}}(t)\}_{t\ge 0}$ is also a Markov chain.


\subsection{Policy-Induced Dispatch Fractions} \label{subsec:NotationfromFramework}
Building on the unified framework in Section \ref{subsec:general_framework}, we introduce several additional pieces of notation that will be used in our main theorem. Consider a load balancing policy $\pi \in \Pi$. After sampling the queue length at times $rT$ for all $r\in \mathbb{N}$, let $\eta_r \in \mathcal{S}_n $ be the permutation specified in step 2(c) of Section \ref{subsec:general_framework}. Define the random vector $ \mathbf{N}_{r,\eta_r}$ by
\begin{align*}
    N_{r, \eta_r}(l) := \sum_{t=rT}^{(r+1)T-1} Z_{\eta_r(l)}(t), \, \forall \, l \in [n].
\end{align*}
Here, $ N_{r, \eta_r}(l) $ denotes the total number of time slots during $t\in \{rT,...,(r+1)T-1 \}$ in which arrivals are sent to the $l$th longest queue, as determined by the $\boldsymbol{\gamma}$-scaled queue-length ordering at time $rT$. Note that, if $\eta_r = \eta_{r^{\prime}}$, then $ \mathbf{N}_{r,\eta_r} $ and $ \mathbf{N}_{r^{\prime},\eta_{r^{\prime}}} $ have the same distribution.
Consequently, conditional on a fixed permutation $\eta$, the policy applies the same dispatching rule to assign jobs. 

With this in mind, we define the conditional mean $f_{l,\eta}$ and variance $\tau_{l,\eta}$ of dispatch fractions as follows:
\begin{definition}[Conditional mean and variance of dispatch fractions] \label{def:ftau} 
Fix a load balancing policy $\pi\in\Pi$. For each permutation $\eta\in\mathcal S_n$ and for all $l\in[n]$, define
\begin{align*}
    f_{l,\eta} := \mathbb E\!\left[\left.\frac{N_{r,\eta_r}(l)}{T}\right|\,\eta_r=\eta\right], \qquad 
    \tau^2_{l,\eta} := \mathrm{Var}\!\left(\left.\frac{N_{r,\eta_r}(l)}{T}\right|\,\eta_r=\eta\right),
\end{align*}
where $ f_{l,\eta} $ and $ \tau^2_{l,\eta} $ represent the expected value and variance of the fraction of jobs sent to the $l$th longest $\boldsymbol{\gamma}$-scaled queue when the permutation is $\eta$. Let $\mathbf{f}_{\eta}:=\{ f_{1,\eta},...,f_{n,\eta} \}$ and $ \boldsymbol{\tau}_{\eta}:= \{\tau_{1,\eta},...,\tau_{n,\eta}\}$ denote the corresponding vectors. \hfill $\square$
\end{definition}


\begin{example}[Characterization of the Vector \(\mathbf{f}_{\eta}\)]
    Table \ref{tab:policies_fvector} shows the vector \(\mathbf{f}_{\eta}\) for several policies. RAND, RR, JSQ, Po$d$ and JSED have the same $\mathbf{f}_{\eta}$ for all permutations $\eta \in \mathcal{S}_n$. Also note that RAND and RR have the same vector $\mathbf{f}_{\eta}$, but their variance vectors $ \boldsymbol{\tau}_{\eta} $ differ. Under $k$-SLQ-$d$, arriving jobs are dispatched in a round-robin fashion among the $n-d$ shortest queues. 

    For WRAND, the probability vector $\mathbf{f}_{\eta}$ varies with $\eta$. This is because WRAND assigns jobs to servers in proportion to their service rates. Specifically, if, under permutation $\eta$, server $l$ is the $j$th longest queue, then the $j$th longest queue receives an arrival with probability $ f_{\eta,j} = \mu_l / (\sum_{r=1}^n \mu_r) $. \hfill $\square$

    \begin{table}[h!]
        \centering
        \footnotesize
        \caption{Vector $\mathbf{f}_{\eta}$ of policies}
        \label{tab:policies_fvector}
        \begin{tabular}{l p{11.5cm}}
            \hline
            \textbf{Policy} & \textbf{$\mathbf{f}_{\eta}$} \\
            \hline
            RAND & $\mathbf{f}_{\eta} = (1/n,\ldots,1/n), \ \forall \eta \in \mathcal{S}_n$ \\
            WRAND & $\mathbf{f}_{\eta} = \left(\mu_{\eta(1)}/\sum_{l=1}^n \mu_l,\ldots,\mu_{\eta(n)}/\sum_{l=1}^n \mu_l\right), \ \forall \eta \in \mathcal{S}_n$ \\
            RR & $\mathbf{f}_{\eta} = (1/n,\ldots,1/n), \ \forall \eta \in \mathcal{S}_n$ \\
            JSQ & $\mathbf{f}_{\eta} = (0,\ldots,1), \ \forall \eta \in \mathcal{S}_n$ \\
            Po$d$ & $\mathbf{f}_{\eta} = \left(0,\ldots,\binom{d-1}{d-1}/\binom{n}{d},\ldots,\binom{n-1}{d-1}/\binom{n}{d}\right), \ \forall \eta \in \mathcal{S}_n$ \\
            JSED & $\mathbf{f}_{\eta} = (0,\ldots,1), \ \forall \eta \in \mathcal{S}_n$ \\
            $k$-SLQ-$d$ & $\mathbf{f}_{\eta} = (0,\ldots,0,1/(n-d),\ldots,1/(n-d)), \ \forall \eta \in \mathcal{S}_n$ \\
            \hline
        \end{tabular}
    \end{table}
\end{example}

\section{Stability Region and Throughput Optimality}
\label{sec:stability_throughputoptimality}

In this section, we study stability and throughput optimality for policies in the framework \(\Pi\). We first derive a partial-sum stability condition that yields a sufficient threshold \(h^*\) for the arrival rate (Section~\ref{subsec:stability_sufficient}), and then use the same structure to obtain a majorization-based criterion for throughput optimality (Section~\ref{subsec:throughput_optimality}). We next show through an example that \(h^*\) need not be a necessary threshold for general policies: stability above \(h^*\) may be load-dependent and non-monotone. To explain this behavior, we introduce two fixed-load certificates, persistent bottleneck dominance (PBD) and sequential catch-up drainage (SCD), which certify transient and stable load windows, respectively (Section~\ref{subsec:load_dependent_certificates}). Finally, we identify a global PBD condition under which \(h^*\) becomes the exact stability threshold 
(Section~\ref{subsec:stability_necessary}).

\subsection{A Sufficient Condition for Stability} \label{subsec:stability_sufficient}
For a general load balancing policy \(\pi\in\Pi\), the following theorem provides a sufficient condition for stability. Throughout, we adopt the convention that \(x/0:=\infty\) for all \(x>0\).

\begin{theorem}[A Sufficient Condition for Stability]
\label{thm_our:stable_withrate_withratio}
    For any load balancing policy \(\pi\in\Pi\), if
    \begin{align}
        n\lambda < \min_{\eta\in\mathcal S_n} \left\{ \min_{m\in[n]} \left\{ \frac{ \sum_{l=1}^m\mu_{\eta(l)} }{ \sum_{l=1}^m f_{l,\eta} } \right\} \right\},
        \label{eq:feasibleregion_withrate}
    \end{align}
    then the Markov chain \(\{\mathbf X(t)\}_{t\ge0}\) is positive recurrent.
\end{theorem}
Here, \(f_{l,\eta}\) is the expected fraction of jobs sent to the \(l\)th longest \(\boldsymbol\gamma\)-scaled queue when the sampled permutation is \(\eta\), as defined in Definition~\ref{def:ftau}. The proof uses a multi-period Lyapunov drift argument, adapting the drift-based approach of \cite{eryilmaz2012asymptotically-driftmethod} and the multi-step Lyapunov framework of \cite{Yishun2025k-slq-d} to the general policy class \(\Pi\) through the policy-induced dispatch fractions \(f_{l,\eta}\). The detailed proof is given in Appendix~\ref{prf:thm_stable_withrate_withratio}.

Equivalently, the condition in Equation~\eqref{eq:feasibleregion_withrate} can be written as
\[
    n\lambda\sum_{l=1}^m f_{l,\eta} < \sum_{l=1}^m\mu_{\eta(l)}, \qquad \forall m\in[n],\ \eta\in\mathcal S_n.
\]
Theorem~\ref{thm_our:stable_withrate_withratio} can be interpreted as a collection of partial-sum load conditions. Conditional on a permutation \(\eta\), the first \(m\) positions correspond to the \(m\) longest \(\boldsymbol\gamma\)-scaled queues. The term \(\sum_{l=1}^m f_{l,\eta}\) is the fraction of arrivals assigned to this subset, so the aggregate arrival rate to this potential bottleneck subset is \(n\lambda\sum_{l=1}^m f_{l,\eta}\). The corresponding aggregate service capacity is \(\sum_{l=1}^m\mu_{\eta(l)}\). Thus, the condition in Equation~\eqref{eq:feasibleregion_withrate} requires that no such potential bottleneck subset is overloaded under any permutation.

The sufficient condition naturally induces a static threshold. For any \(m\in[n]\) and \(\eta\in\mathcal S_n\), define
\[
    h(\eta,m) := \frac{\sum_{l=1}^m\mu_{\eta(l)}}{\sum_{l=1}^m f_{l,\eta}} .
\]
This is the function minimized in the condition in Equation~\eqref{eq:feasibleregion_withrate}. Let
\[
    h^* := \min_{\eta\in\mathcal S_n}\min_{m\in[n]} h(\eta,m).
\]
Then Theorem~\ref{thm_our:stable_withrate_withratio} implies positive recurrence whenever \(n\lambda<h^*\). The threshold \(h^*\) is obtained by fixing a permutation \(\eta\) and a prefix length \(m\), and comparing the arrival fraction assigned to the first \(m\) positions with their aggregate service capacity. Thus \(h^*\) is a static sufficient-stability threshold. Before examining whether this threshold is necessary for general policies, we first record a direct throughput-optimality implication of the same partial-sum condition.

\subsection{Throughput Optimality}
\label{subsec:throughput_optimality}

We first give a formal definition of throughput optimality.

\begin{definition}[Throughput Optimal]
\label{def:throughput_optimal}
    A load balancing policy \(\pi\) is throughput optimal if the Markov chain \(\{\mathbf X(t)\}_{t\ge0}\) is positive recurrent for all \(n\lambda\in C\), where
    \[
        C := \left\{ \lambda\in\mathbb R_+: n\lambda<\sum_{l=1}^n\mu_l \right\}.
    \]
\end{definition}
Based on Theorem~\ref{thm_our:stable_withrate_withratio}, the following corollary provides a majorization-based condition for a policy to be throughput optimal.

\begin{corollary}[A Sufficient Condition for Throughput Optimality]
\label{co_our:condition_withrate_withratio}
    For any load balancing policy \(\pi\in\Pi\), if
    \[
        \sum_{l=1}^{m} f_{l,\eta} \le \frac{ \sum_{l=1}^{m}\mu_{\eta(l)} }{ \sum_{l=1}^{n}\mu_{\eta(l)} },
        \qquad \forall m\in[n],\ \eta\in\mathcal S_n,
    \]
    then policy \(\pi\) is throughput optimal.
\end{corollary}
The proof is given in Appendix~\ref{prf:condition_withrate_withratio}.

An interpretation of Corollary \ref{co_our:condition_withrate_withratio} is that the cumulative fraction $\sum_{l=1}^m f_{l,\eta}$ is required to be no larger than the corresponding fraction of aggregate service capacity $\big(\sum_{l=1}^m \mu_{\eta(l)}\big)\big/\big(\sum_{l=1}^n \mu_{\eta(l)}\big)$ for all $m$ and $\eta$. Equivalently, the partial sums of $\mathbf{f}_{\eta}$ are dominated by the corresponding partial sums of normalized service-rate vector, which prevents bottlenecks among the longest queues and implies throughput optimality.
\begin{remark}[Comparison with WRAND]\label{rmk:weightedrand_boundary}
    Under WRAND, the partial sums satisfy 
    \begin{align*}
        \sum_{l=1}^{m} f^{\text{WRAND}}_{l,\eta} = \frac{\sum\limits_{l=1}^{m}\mu_{\eta(l)}}{\sum\limits_{l=1}^{n}\mu_{\eta(l)}}, \quad \forall\, m\in[n], \eta\in S_n.
    \end{align*}
\end{remark}
In view of Remark \ref{rmk:weightedrand_boundary}, another interpretation is that, if a policy sends (weakly) fewer jobs to the longer $\boldsymbol{\gamma}$-scaled queues than WRAND, that is, $\sum_{l=1}^m f_{l,\eta} \leq \sum_{l=1}^m f^{\text{WRAND}}_{l,\eta}$ for all $m$ and $\eta$, then the policy satisfies the condition of Corollary \ref{co_our:condition_withrate_withratio} and is therefore throughput optimal. 

This condition generalizes the results of \cite{moyal2022stability-generalhomogeneous} from homogeneous to heterogeneous systems, without requiring a specific policy design as in \cite{zhou2017designing-JBT, zhou2018degree-generalPpolicy}. Here $f_{l,\eta}$ captures the policy-induced expected fraction of arrivals routed to the $l$th longest $\boldsymbol{\gamma}$-scaled queue under permutation $\eta$, and thus need not coincide with a fixed routing probability.

\begin{example}[Visualization of Corollary \ref{co_our:condition_withrate_withratio}]
    We illustrate Corollary \ref{co_our:condition_withrate_withratio} for a system with $n=3$ servers. In the homogeneous case, $\boldsymbol{\mu}=(1,1,1)$; in the heterogeneous case, $\boldsymbol{\mu}=(1,3,6)$. We set $\boldsymbol{\gamma}=\mathbf 1$. Let $\eta\in\mathcal S_3$ denote the permutation that orders the queues by (scaled) queue length. For each $\eta$, define $\mathbf v_\eta=(v_{1,\eta},v_{2,\eta},v_{3,\eta})$ by \( v_{m,\eta} := \sum_{l=1}^{m}\mu_{\eta(l)} / \sum_{l=1}^{n}\mu_{\eta(l)} \) for all $m \in [n]$.
    Here, $v_{m,\eta}$ is the normalized partial sum of service rates under ordering $\eta$. In the homogeneous case, $\mathbf v_\eta=(1/3,\,2/3,\,1)$ for all $\eta$, whereas in the heterogeneous case $\mathbf v_\eta$ depends on $\eta$. We consider JSQ and RAND (WRAND in the heterogeneous case). Let $F_{m,\eta}:=\sum_{l=1}^m f_{l,\eta}$.
    
    \noindent\textbf{Homogeneous case}. Under JSQ, $\mathbf{F}^{\text{JSQ}}_{\eta}=(F^{\text{JSQ}}_{1,\eta},F^{\text{JSQ}}_{2,\eta},F^{\text{JSQ}}_{3,\eta})=(0,0,1)$ for all $\eta$. Under RAND, $\mathbf{F}^{\text{RAND}}_\eta=(1/3,2/3,1)$ for all $\eta$. Figure \ref{fig:homogeneous} depicts the throughput-optimality region in Corollary \ref{co_our:condition_withrate_withratio}: a policy $\pi$ is throughput optimal if, for every $\eta$, the vector $(F_{1,\eta},F_{2,\eta},F_{3,\eta})$ lies in the shaded set (including its boundary).

    \noindent\textbf{Heterogeneous case}. JSQ still yields $\mathbf{F}^{\text{JSQ}}_\eta=(0,0,1)$ for all $\eta$. For WRAND, the cumulative fractions satisfy $\mathbf{F}^{\text{WRAND}}_\eta=\mathbf v^{\text{WRAND}}_\eta$ for each $\eta\in\mathcal S_3$. Figure \ref{fig:heterogeneous} plots JSQ and WRAND with $\{\mathbf v_\eta:\eta\in\mathcal S_3\}$ across all permutations. Corollary \ref{co_our:condition_withrate_withratio} then requires that, for each $\eta$, the cumulative-fraction vector $(F_{1,\eta},F_{2,\eta},F_{3,\eta})$ lies in the region bounded by the JSQ curve and the benchmark vector $\mathbf v_\eta$. \hfill $\square$

    \begin{figure}[h!]
        \centering
    
        \begin{subfigure}{0.49\textwidth}
            \centering
            \includegraphics[width=\textwidth]{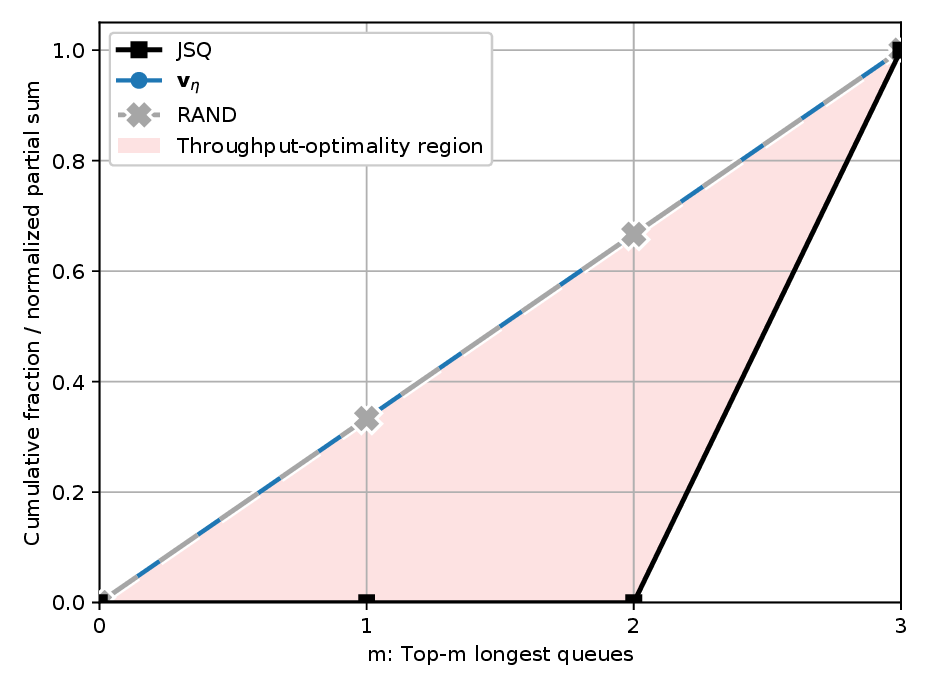}
            \caption{Homogeneous system: $\boldsymbol{\mu} = (1,1,1)$.}
            \label{fig:homogeneous}
        \end{subfigure}
        \hfill
        \begin{subfigure}{0.49\textwidth}
            \centering
            \includegraphics[width=\textwidth]{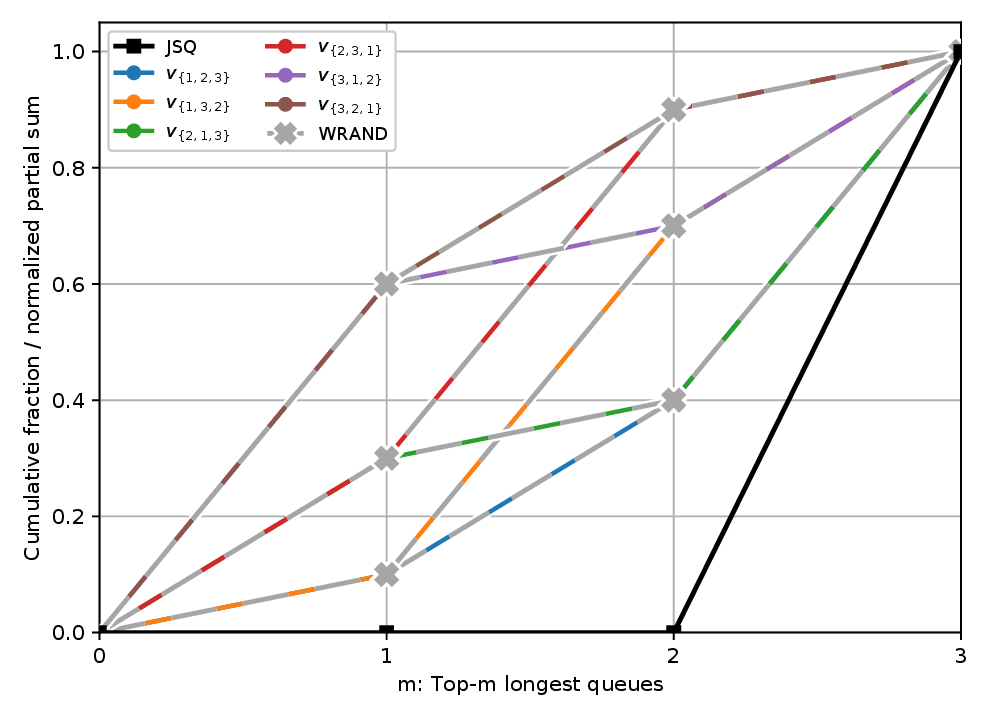}
            \caption{Heterogeneous system: $\boldsymbol{\mu} = (1,3,6)$.}
            \label{fig:heterogeneous}
        \end{subfigure}
    
        \caption{Throughput optimality region in Corollary \ref{co_our:condition_withrate_withratio}: homogeneous and heterogeneous systems.}
        \label{fig:homogeneous_heterogeneous}
    
        \vspace{0.5em}
        \begin{minipage}{0.95\textwidth}
            \footnotesize
            \textit{Note.} The $x$-axis is $m\in\{0,1,2,3\}$, the number of longest queues considered among $n=3$ servers. The $y$-axis plots cumulative fractions/normalized partial sums in $[0,1]$. Without loss of generality, set $v_{0,\eta}=F_{0,\eta}=0$ for all $\eta$. For $\boldsymbol{\mu}=(1,1,1)$, $\mathbf v_\eta=(1/3,2/3,1)$ for all $\eta\in\mathcal S_3$. For $\boldsymbol{\mu}=(1,3,6)$, the pairs $(\mathbf v_\eta,\eta)$ are $((0.1,0.4,1),(1,2,3))$, $((0.1,0.7,1),(1,3,2))$, $((0.3,0.4,1),(2,1,3))$, $((0.3,0.9,1),(2,3,1))$, $((0.6,0.7,1),(3,1,2))$, and $((0.6,0.9,1),(3,2,1))$. Under RAND in the homogeneous case and WRAND in the heterogeneous case, $F_{\eta}=\mathbf v_\eta$ for all $\eta\in\mathcal S_3$.
        \end{minipage}
    \end{figure}
\end{example}

Theorem~\ref{thm_our:stable_withrate_withratio} and  Corollary~\ref{co_our:condition_withrate_withratio} show that the sufficient stability and throughput-optimality certificates are determined by the policy-induced dispatch fractions \(f_{l,\eta}\). For a fixed sampling interval \(T\), these certificates have no additional explicit dependence on \(T\) beyond its effect on \(f_{l,\eta}\). Thus, reusing sampled queue-length information over multiple slots can preserve the same static certificate, provided that the induced dispatch fractions satisfy the required partial-sum conditions. However, for general policies, the threshold \(h^*\) from Theorem~\ref{thm_our:stable_withrate_withratio} is only sufficient. We next show that stability above \(h^*\) can be load-dependent and non-monotone.

\subsection{Non-monotonicity of the Stability Region}

The next example shows that, for general policies in \(\Pi\), the static threshold \(h^*\) need not fully characterize stability. Two policies can have the same value of \(h^*\), yet behave differently above \(h^*\). Stability may alternate across load intervals, and a policy that sends more arrivals to the shortest scaled queue may have a smaller stability region.

\begin{example}[Non-monotone Stability Region]
\label{ex:nonmonotone_stability_region}
Consider a four-server system with \(\boldsymbol{\mu}=(1,2,3,20)\), \(T=1\), and \(\gamma_i=1\) for all \(i\in[4]\). The total service capacity is \(\sum_{i=1}^4\mu_i=26\). We compare two policies with the same static threshold \(h^*=5\). Policy A is given in Table~\ref{tab:nonmonotone_policy_A}: if server \(i\) is the longest scaled queue, row \(i\) gives the dispatching probabilities to actual servers \(1,2,3,4\).
\begin{table}[h!]
    \centering
    \footnotesize
    \caption{Dispatching probabilities for Policy A}
    \label{tab:nonmonotone_policy_A}
    \begin{tabular}{lllll}
        \hline
        \textbf{Longest server $i$} & \textbf{Server 1} & \textbf{Server 2} & \textbf{Server 3} & \textbf{Server 4} \\
        \hline
        1 & 0.20 & 0.35 & 0    & 0.45 \\
        2 & 0    & 0.25 & 0.35 & 0.40 \\
        3 & 0    & 0    & 0.25 & 0.75 \\
        4 & 0    & 0    & 0    & 1    \\
        \hline
    \end{tabular}
\end{table}
Policy B is defined as follows. Let \(\eta(1)\) and \(\eta(4)\) denote the longest and shortest scaled queues, respectively. If \(\eta(1)=1\), Policy B dispatches arrivals to server 1 with probability \(0.2\) and to \(\eta(4)\) with probability \(0.8\). If \(\eta(1)\neq1\), Policy B dispatches arrivals to \(\eta(4)\) with probability \(1\). Compared with Policy A, Policy B sends more arrivals to the current shortest scaled queue under every permutation.


Figure~\ref{fig:policy_nonmonotone} reports the time-average queue length under the two policies for $n\lambda \in \{4,5.5,7,9,11,14\}$, where $n\lambda=4$ is included as a below-threshold baseline and the remaining loads exceed the common threshold $h^*=5$. For Policy A, the curves with \(n\lambda=5.5,9,14\) grow rapidly, while the curves with \(n\lambda=7,11\) remain bounded over the simulation horizon. For Policy B, all displayed loads above \(h^*\) exhibit unstable behavior. \hfill \(\square\)

\begin{figure}[h!]
    \centering

    \begin{subfigure}{0.49\textwidth}
        \centering
        \includegraphics[width=\textwidth]{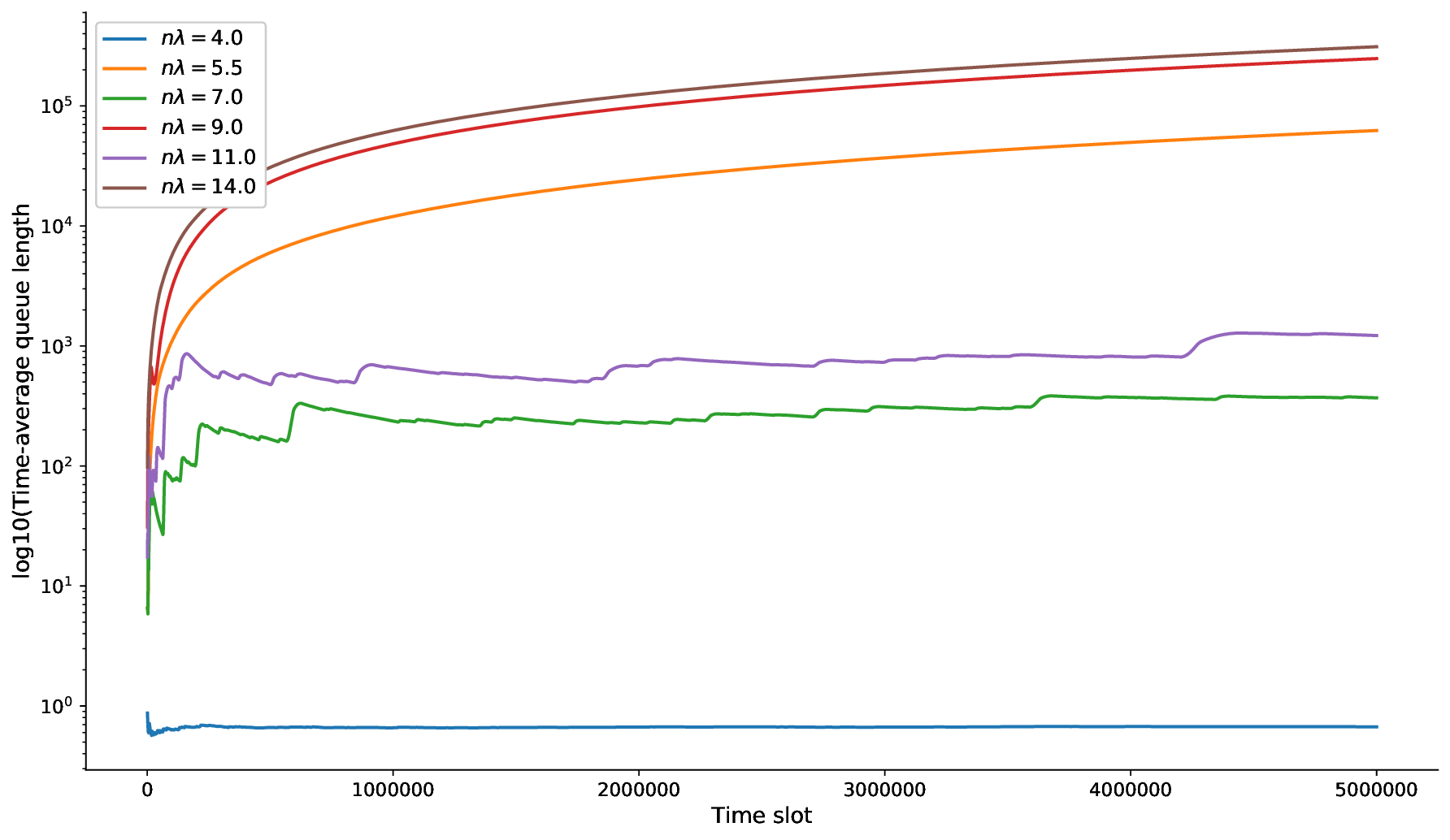}
        \caption{Policy A.}
        \label{fig:policy_A}
    \end{subfigure}
    \hfill
    \begin{subfigure}{0.49\textwidth}
        \centering
        \includegraphics[width=\textwidth]{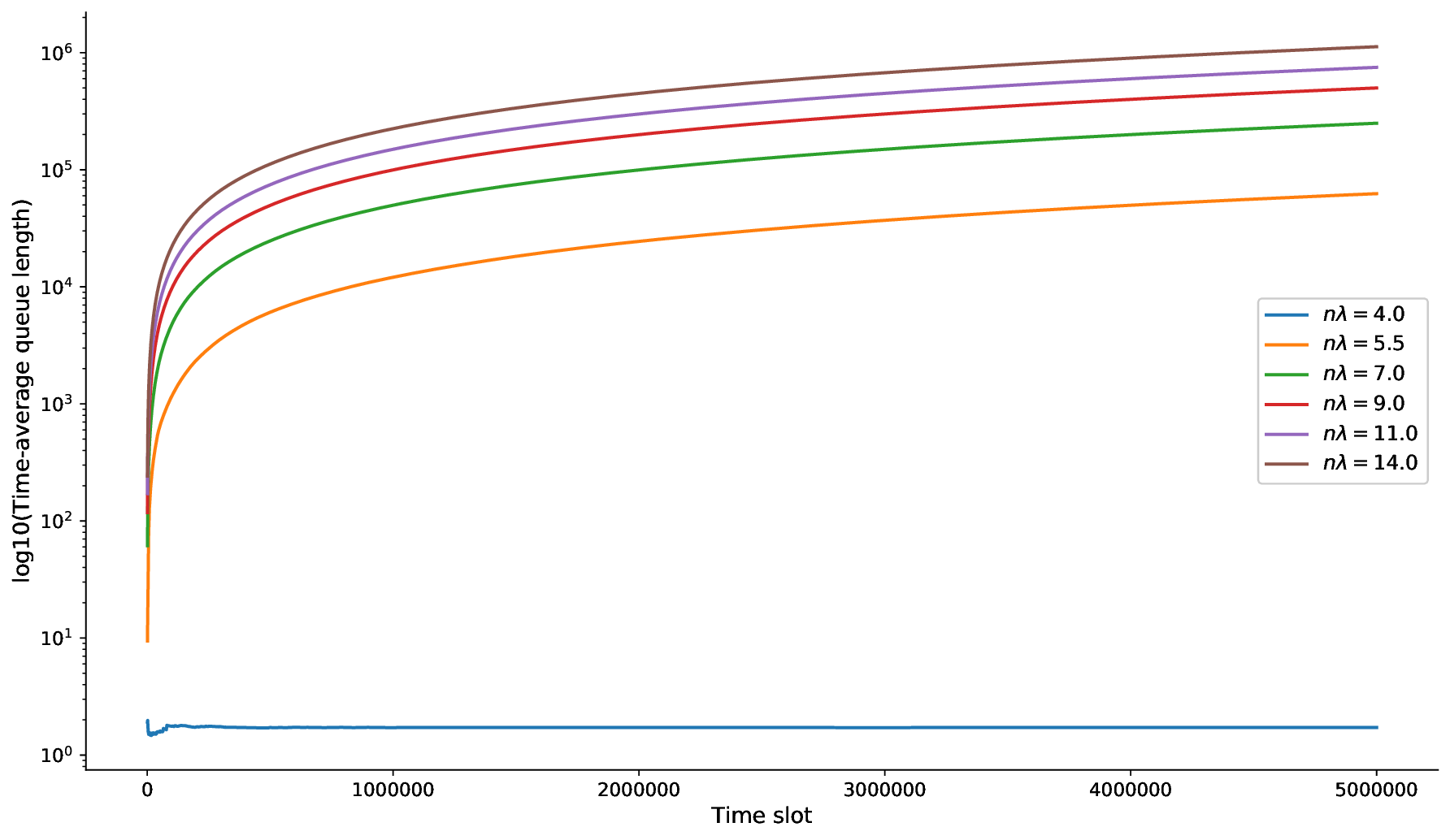}
        \caption{Policy B.}
        \label{fig:policy_B}
    \end{subfigure}

    \caption{Time-average queue length under Policies A and B}
    \label{fig:policy_nonmonotone}

    \vspace{0.5em}
    \begin{minipage}{0.95\textwidth}
        \footnotesize
        \textit{Note.} The \(x\)-axis is the simulation time slot, and the \(y\)-axis is the time-average scaled queue length on a base-10 logarithmic scale. Each curve corresponds to a different load \(n\lambda\). For Policy A, the curves with \(n\lambda=5.5,9,14\) grow rapidly, while the curves with \(n\lambda=7,11\) remain bounded. Policy B exhibits unstable behavior for all simulated loads above \(h^*\).
    \end{minipage}
\end{figure}
\end{example}
This example illustrates two points. First, the static threshold \(h^*\) is not a necessary stability threshold for general policies in \(\Pi\). Second, stability above \(h^*\) need not be monotone in the load, as Policy A appears unstable at some lower and higher loads but stable at intermediate loads. The contrast with Policy B indicates that shortest-queue preference alone is not enough to explain stability behavior above \(h^*\). The next subsection introduces fixed-load certificates that provide a theoretical classification of these load windows.

\subsection{Load-Dependent Stability Properties}
\label{subsec:load_dependent_certificates}

Motivated by Example~\ref{ex:nonmonotone_stability_region}, we next introduce two fixed-load certificates for stability behavior above the static threshold \(h^*\). Persistent bottleneck dominance \(\mathrm{PBD}(n\lambda)\) captures loads at which a locally overloaded subset remains separated under the sampled permutations and certifies transience. Sequential catch-up drainage \(\mathrm{SCD}(n\lambda)\) captures loads at which changes in the sampled permutation eventually lead the fluid trajectory to a draining phase and certifies positive recurrence.

\paragraph{Persistent Bottleneck Dominance.} For a fixed load \(n\lambda\in(h^*,\sum_{i=1}^n\mu_i)\) and a nonempty proper set \(B\subset[n]\), let \(m:=|B|\), \(\bar B:=[n]\setminus B\), and
\[
    \mathcal P(B) := \left\{ \eta\in\mathcal S_n: \{\eta(1),...,\eta(m)\}=B \right\}.
\]
Define
\[
    \underline{\delta}(B,n\lambda) := \min_{\eta\in\mathcal P(B)} \min_{a\in\{1,...,m\}} \left\{ \frac{ T\left( n\lambda\sum_{l=m-a+1}^{m}f_{l,\eta} - \sum_{l=m-a+1}^{m}\mu_{\eta(l)} \right) }{ \sum_{l=m-a+1}^{m}\gamma_{\eta(l)} } \right\},
\]
and
\[
    \overline{\delta}(\bar B,n\lambda) := \max_{\eta\in\mathcal P(B)} \max_{b\in\{1,...,n-m\}} \left\{ \frac{ T\left( n\lambda\sum_{l=m+1}^{m+b}f_{l,\eta} - \sum_{l=m+1}^{m+b}\mu_{\eta(l)} \right) }{ \sum_{l=m+1}^{m+b}\gamma_{\eta(l)} } \right\}.
\]
The quantity \(\underline{\delta}(B,n\lambda)\) is a uniform lower bound on the normalized drift of any suffix of the candidate bottleneck set \(B\), over all sampled permutations in \(\mathcal P(B)\). The quantity \(\overline{\delta}(\bar B,n\lambda)\) is a uniform upper bound on the normalized drift of any prefix of the complement \(\bar B\), over the same set of sampled permutations.

The following definition formalizes the bottleneck mechanism. It requires a candidate subset \(B\) to have positive normalized drift and to grow faster than any leading part of its complement.
\begin{definition}[Persistent Bottleneck Dominance]
\label{def:fixed_load_pbd}
    A policy \(\pi\in\Pi\) satisfies persistent bottleneck dominance at load \(n\lambda\), denoted by \(\mathrm{PBD}(n\lambda)\), if there exists a nonempty proper set \(B\subset[n]\) such that
    \begin{align}
        \underline{\delta}(B,n\lambda) > \max\left\{ \overline{\delta}(\bar B,n\lambda),0 \right\}.
        \label{eq:bottleneck_condition}
    \end{align}
\end{definition}
Definition~\ref{def:fixed_load_pbd} is a load-specific drift-dominance property. It requires the bottleneck side to have strictly positive normalized drift and to grow faster than the fastest-growing part of its complement.

\paragraph{Sequential catch-up drainage.}
We next define the stability-side property. Unlike persistent bottleneck dominance, which describes a separated set that continues to grow, this property captures the opposite mechanism: a leading set may initially have positive drift, but another set catches it in finite fluid time, after which the sampled permutation changes and the trajectory eventually enters a draining regime.

Fix a load \(n\lambda<\sum_{i=1}^n\mu_i\). For \(i\in[n]\) and \(\eta\in\mathcal S_n\), let \(\ell_\eta(i)\) be the unique position such that \(\eta(\ell_\eta(i))=i\). For a state \(\mathbf q\in\mathbb R_+^n\), let \(J(\mathbf q):=\{i\in[n]:q_i>0\}\). Let \(\mathfrak W(\mathbf q)\) be the set of probability vectors \(\mathbf w=(w_\eta)_{\eta\in\mathcal S_n}\) such that \(w_\eta\ge0\), \(\sum_{\eta\in\mathcal S_n}w_\eta=1\), and \(w_\eta=0\) whenever \(q_{\eta(1)}^{(\gamma)}\ge\cdots\ge q_{\eta(n)}^{(\gamma)}\) does not hold. Thus \(\mathfrak W(\mathbf q)\) represents possible fluid-scale mixtures of sampled permutations consistent with the weak permutation of \(\mathbf q^{(\gamma)}\).

For \(\mathbf w\in\mathfrak W(\mathbf q)\), define the reflected fluid drift of the \(\boldsymbol{\gamma}\)-scaled queue of server \(i\) by
\[
    b_i^{J(\mathbf q)}(\mathbf w,n\lambda) :=
    \begin{cases}
        \dfrac{ T\left( n\lambda \sum_{\eta\in\mathcal S_n} w_\eta f_{\ell_\eta(i),\eta} - \mu_i \right) }{\gamma_i}, & i\in J(\mathbf q),\\
        \max\left\{ \dfrac{ T\left( n\lambda \sum_{\eta\in\mathcal S_n} w_\eta f_{\ell_\eta(i),\eta} - \mu_i \right) }{\gamma_i}, 0 \right\}, & i\notin J(\mathbf q).
    \end{cases}
\]
The first case applies when queue \(i\) is positive. The second case captures zero-boundary reflection, where negative averaged nominal drift is absorbed by unused service and positive drift makes the queue leave the boundary.

Let
\[ \mathfrak D_{n\lambda}(\mathbf q):= \{b^{J(\mathbf q)}(\mathbf w,n\lambda):\mathbf w\in\mathfrak W(\mathbf q)\} \]
be the set of admissible reflected fluid drifts of the \(\boldsymbol{\gamma}\)-scaled state at \(\mathbf q\). For a local Lipschitz function \(H:\mathbb R_+^n\to\mathbb R\) and a vector \(\mathbf v\in\mathbb R^n\), define \(D^+H(\mathbf x,\mathbf v):=\limsup_{h\downarrow0}[H(\mathbf x+h\mathbf v)-H(\mathbf x)]/h\).

The following definition formalizes the stability-side mechanism. Each nonterminal phase closes a catch-up gap at a uniform rate, and the final phase drains the fluid state.

\begin{definition}[Sequential Catch-up Drainage]
\label{def:fixed_load_scd}
    A policy \(\pi\in\Pi\) satisfies sequential catch-up drainage at load \(n\lambda\), denoted by \(\mathrm{SCD}(n\lambda)\), if there exist an integer \(R\ge0\), nonnegative locally Lipschitz functions \(G_1,\ldots,G_R,V_D:\mathbb R_+^n\to\mathbb R_+\), constants \(\delta_1,\ldots,\delta_R,\delta_D>0\), and constants \(\kappa_1,\ldots,\kappa_R,\kappa_D<\infty\), such that the following property holds.

    Consider any absolutely continuous path \(\mathbf q(\cdot)\) whose \(\boldsymbol{\gamma}\)-scaled version satisfies \(\frac{d}{dt}\mathbf q^{(\gamma)}(t)\in \mathfrak D_{n\lambda}(\mathbf q(t))\) for almost every \(t\) before it reaches the origin. Let \(\tau_{\mathrm{hit}}:=\inf\{t\ge0:\mathbf q(t)=\mathbf 0\}\), with \(\tau_{\mathrm{hit}}=\infty\) if the origin is never reached. Then there exist switching times \(0=s_0\le s_1\le\cdots\le s_R\le \tau_{\mathrm{hit}}<\infty\) such that:
    \begin{enumerate}
        \item For each \(r=1,\ldots,R\), either the \(r\)-th catch-up phase is skipped, \(s_r=s_{r-1}\), or on \([s_{r-1},s_r)\) we have \(G_r(\mathbf q^{(\gamma)}(t))>0\) and \(D^+G_r \left (\mathbf q^{(\gamma)}(t), \frac{d}{dt}\mathbf q^{(\gamma)}(t) \right )\le-\delta_r\) for almost every \(t\). In addition, \(G_r(\mathbf q^{(\gamma)}(s_{r-1}))\le \kappa_r\|\mathbf q^{(\gamma)}(s_{r-1})\|_1\).

        \item If \(s_R<\tau_{\mathrm{hit}}\), then on \([s_R,\tau_{\mathrm{hit}})\) we have \(V_D(\mathbf q^{(\gamma)}(t))>0\) and \(D^+V_D \left (\mathbf q^{(\gamma)}(t), \frac{d}{dt}\mathbf q^{(\gamma)}(t) \right) \le-\delta_D\) for almost every \(t\). Moreover, \(V_D(\mathbf q^{(\gamma)}(s_R))\le \kappa_D\|\mathbf q^{(\gamma)}(s_R)\|_1\), and \(V_D(\mathbf x)=0\) if and only if \(\mathbf x=\mathbf 0\). \hfill $\square$
    \end{enumerate}
\end{definition}
Definition~\ref{def:fixed_load_scd} is a stability-side drift-structure property. Each nonterminal phase has a catch-up gap that closes at a uniform rate under the admissible fluid drift. After finitely many catch-up phases, the trajectory enters a terminal phase in which a draining workload has uniformly negative drift.

The next theorem connects these two fixed-load properties to the stochastic behavior of the queueing system.
\begin{theorem}[Fixed-load Stability Certificates]
\label{thm_our:fixed_load_stability_certificates}
    Fix a load \(n\lambda\).
    \begin{enumerate}
        \item Suppose that \(h^*<n\lambda<\sum_{i=1}^n\mu_i\) and \(\mathrm{PBD}(n\lambda)\) holds. Then the Markov chain \(\{\mathbf X(t)\}_{t\ge0}\) is transient from some initial state.

        \item Suppose that \(n\lambda<\sum_{i=1}^n\mu_i\) and \(\mathrm{SCD}(n\lambda)\) holds. Then the Markov chain \(\{\mathbf X(t)\}_{t\ge0}\) is positive recurrent.
    \end{enumerate}
\end{theorem}

Theorem~\ref{thm_our:fixed_load_stability_certificates} translates the two fixed-load certificates into stochastic stability conclusions. Under \(\mathrm{PBD}(n\lambda)\), a persistently separated bottleneck yields a transient communicating class. Under \(\mathrm{SCD}(n\lambda)\), finite-time fluid drainage yields positive recurrence. The proof is given in Appendix~\ref{prf:fixed_load_stability_certificates}. The transience part uses a moving-barrier construction, and the positive-recurrence part uses a fluid-stability argument.

\subsection{Load-Window Stability Classification}
\label{subsec:load_window_classification}

Having established the fixed-load certificates in Theorem~\ref{thm_our:fixed_load_stability_certificates}, we now apply them to the load range above the static threshold \(h^*\). Specifically, we classify \((h^*,\sum_{i=1}^n\mu_i)\) into load windows on which either persistent bottleneck dominance certifies transience or sequential catch-up drainage certifies positive recurrence. To capture the possibility that stability may switch across load windows, we assume that the load range above \(h^*\) can be partitioned into intervals on which one of the two fixed-load certificates holds pointwise for every \(n\lambda\in I_k\).

\begin{assumption}[Certified Load-Window Partition]
\label{assump:load_window_certificates}
    There exist constants
    \[
        h^*=\xi_0<\xi_1<\cdots<\xi_K<\xi_{K+1}:=\sum_{i=1}^n\mu_i
    \]
    and, for each interval \(I_k:=(\xi_k,\xi_{k+1})\), a label \(\chi_k\in\{\mathrm T,\mathrm S\}\). If \(\chi_k=\mathrm T\), then \(\mathrm{PBD}(n\lambda)\) holds for every \(n\lambda\in I_k\). If \(\chi_k=\mathrm S\), then \(\mathrm{SCD}(n\lambda)\) holds for every \(n\lambda\in I_k\).
\end{assumption}
The label \(\mathrm T\) denotes a window certified by persistent bottleneck dominance and hence associated with transience. The label \(\mathrm S\) denotes a window certified by sequential catch-up drainage and hence associated with positive recurrence.

The following corollary is an immediate consequence of Theorem~\ref{thm_our:fixed_load_stability_certificates} applied on each certified interval.

\begin{corollary}[Stability Classification over Load Windows]
\label{cor_our:load_window_certificate_classification}
    Suppose that policy \(\pi\in\Pi\) satisfies Assumption~\ref{assump:load_window_certificates}. Then, for each interval \(I_k=(\xi_k,\xi_{k+1})\), the following statements hold.
    
    \begin{enumerate}
        \item If \(\chi_k=\mathrm T\), then for every \(n\lambda\in I_k\), the Markov chain \(\{\mathbf X(t)\}_{t\ge0}\) is transient from some initial state.
    
        \item If \(\chi_k=\mathrm S\), then for every \(n\lambda\in I_k\), the Markov chain \(\{\mathbf X(t)\}_{t\ge0}\) is positive recurrent.
    \end{enumerate}
\end{corollary}

The proof is given in Appendix~\ref{prf:load_window_certificate_classification}. On a T-window, PBD\((n\lambda)\) holds at every load and gives transience from some initial state. On an S-window, SCD\((n\lambda)\) holds at every load and gives positive recurrence. Thus, once different windows carry different labels, the stability behavior of a general policy in \(\Pi\) need not be summarized by a single load threshold.

\paragraph{Example~\ref{ex:nonmonotone_stability_region} revisited.}
For Policy A, the transient and stable windows above \(h^*\) are \(\mathcal I_{\mathrm T}^A=(5,20/3)\cup(8,10)\cup(12,26)\) and \(\mathcal I_{\mathrm S}^A=(20/3,8)\cup(10,12)\). For every \(n\lambda\in\mathcal I_{\mathrm T}^A\), \(\mathrm{PBD}(n\lambda)\) holds. For every \(n\lambda\in\mathcal I_{\mathrm S}^A\), \(\mathrm{SCD}(n\lambda)\) holds. In contrast, Policy B satisfies \(\mathrm{PBD}(n\lambda)\) for every \(n\lambda\in(5,26)\). Hence, the fixed-load certificates provide a theoretical classification of the behavior observed in Figure~\ref{fig:policy_nonmonotone}. The verification of these load-window classifications for Policies A and B is given in Appendix~\ref{app:example_A_theoretical_analysis}.

\subsection{A Necessary Condition under Persistent Bottleneck Dominance}
\label{subsec:stability_necessary}

The preceding subsection shows that a general policy in \(\Pi\) may exhibit non-monotone stability behavior. We now identify an important special case in which the static threshold \(h^*\) becomes a necessary stability threshold: persistent bottleneck dominance holds throughout the entire nontrivial overload region.

\begin{assumption}[Persistent bottleneck dominance]
    \label{assump:persistent_bottleneck_dominance}
    A policy \(\pi\in\Pi\) satisfies persistent bottleneck dominance if \(\mathrm{PBD}(n\lambda)\) holds for every \(n\lambda\in(h^*,\sum_{i=1}^n\mu_i)\).
\end{assumption}
Assumption~\ref{assump:persistent_bottleneck_dominance} is the all-\(\mathrm T\) case of Assumption~\ref{assump:load_window_certificates}. It requires \(\mathrm{PBD}(n\lambda)\) to hold for every \(n\lambda\in(h^*,\sum_{i=1}^n\mu_i)\), thereby ruling out stable windows generated by changes in the sampled permutation. The next corollary gives a simple sufficient structure under which this condition holds.

\begin{corollary}[A Sufficient Structure for Persistent Bottleneck Dominance]
\label{cor_our:rank_invariant_pbd}
    Suppose that the policy admits permutation-invariant dispatch fractions: there exists a vector \(\mathbf p=(p_1,\ldots,p_n)\), with \(p_l\ge0\) and \(\sum_{l=1}^n p_l=1\), such that
    \[
        f_{l,\eta}=p_l, \qquad \forall l\in[n],\ \forall \eta\in\mathcal S_n.
    \]
    Then Assumption~\ref{assump:persistent_bottleneck_dominance} holds.
\end{corollary}
This corollary covers policies such as RR, RAND, Po\(d\), and \(k\)-SLQ-\(d\). For these policies, the dispatch fraction \(f_{l,\eta}\) depends only on the position \(l\) in the sampled permutation and not on the identity of the server occupying that position. In addition, JSQ, JSED, WRAND, and Balanced-Po\(d\) satisfy Assumption~\ref{assump:persistent_bottleneck_dominance} directly. The proof is given in Appendix~\ref{prf:rank_invariant_pbd}.

Under Assumption~\ref{assump:persistent_bottleneck_dominance}, we obtain the following necessary condition for stability.

\begin{theorem}[A Necessary Condition for Stability under Persistent Bottleneck Dominance]
    \label{thm_our:stable_withrate_withratio_necessary_general}
    For any load balancing policy \(\pi\in\Pi\) satisfying Assumption~\ref{assump:persistent_bottleneck_dominance}, if either \(h^*<n\lambda<\sum_{i=1}^n\mu_i\) or \(n\lambda>\sum_{i=1}^n\mu_i\), then the Markov chain \(\{\mathbf X(t)\}_{t\ge0}\) is transient.
\end{theorem}
Theorem~\ref{thm_our:stable_withrate_withratio_necessary_general} shows that, under Assumption~\ref{assump:persistent_bottleneck_dominance}, the static threshold \(h^*\) is also a necessary stability threshold. Indeed, if \(h^*<n\lambda<\sum_i\mu_i\), then Assumption~\ref{assump:persistent_bottleneck_dominance} implies \(\mathrm{PBD}(n\lambda)\), and Part~1 of Theorem~\ref{thm_our:fixed_load_stability_certificates} gives transience. If \(n\lambda>\sum_i\mu_i\), transience follows directly from total-capacity overload. The proof is given in Appendix~\ref{prf:stable_withrate_withratio_necessary_general}.

Together, theorems~\ref{thm_our:stable_withrate_withratio} and \ref{thm_our:stable_withrate_withratio_necessary_general} identify the stability region for any policy within our framework that satisfies Assumption~\ref{assump:persistent_bottleneck_dominance}. In particular, stability is governed by the policy's allocation to the scaled longest queues. This characterization is also straightforward to verify, as it reduces to checking explicit partial-sum conditions.

\begin{example}[Stability regions of common policies]
    Table~\ref{tab:stability_region} summarizes the stability region for common
    policies. The stability region for Po\(d\) is aligned with Remark 4 in
    \cite{hurtado2021throughput-Pod-stability-heterogeneous}. \hfill \(\square\)

    \begin{table}[h!]
        \centering
        \footnotesize
        \caption{Stability regions of common policies}
        \label{tab:stability_region}
        \begin{tabular}{lll}
            \hline
            \textbf{Policy} & \textbf{Assumption~\ref{assump:persistent_bottleneck_dominance}} & \textbf{Stability Region $C_{\pi}$} \\
            \hline
            RAND & \(\checkmark\) & \(\left\{\lambda \in \mathbb{R}_+ : n\lambda < n\mu_1 \right\}\) \\
            WRAND & \(\checkmark\) & \(\left\{\lambda \in \mathbb{R}_+ : n\lambda < \sum_{l=1}^n \mu_l \right\}\) \\
            RR & \(\checkmark\) & \(\left\{\lambda \in \mathbb{R}_+ : n\lambda < n\mu_1 \right\}\) \\
            JSQ & \(\checkmark\) & \(\left\{\lambda \in \mathbb{R}_+ : n\lambda < \sum_{l=1}^n \mu_l \right\}\) \\
            Po\(d\) & \(\checkmark\) & \(\left\{\lambda \in \mathbb{R}_+ : n\lambda < \min_{j \in \{d,\ldots,n\}} \left\{ \left[\binom{n}{d}\sum_{l=1}^j \mu_l \right]/\binom{j}{d} \right\} \right\}\) \\
            JSED & \(\checkmark\) & \(\left\{\lambda \in \mathbb{R}_+ : n\lambda < \sum_{l=1}^n \mu_l \right\}\) \\
            \(k\)-SLQ-\(d\) & \(\checkmark\) & \(\left\{\lambda \in \mathbb{R}_+ : n\lambda < \min_{j \in \{d+1,\ldots,n\}} \left\{ \left[(n-d)\sum_{l=1}^j \mu_l \right]/(j-d) \right\} \right\}\) \\
            \hline
        \end{tabular}
    \end{table}
\end{example}

\section{Heavy-Traffic Optimality and Queue-Length Distribution} \label{sec:heavytraffic_queuelengthdistribution}
In this section, we establish heavy-traffic performance guarantees for policies in $\Pi$. We first give a sufficient condition under which a policy exhibits state-space collapse (Section~\ref{subsec:SSC}). We then use this collapse result to derive a sufficient condition for heavy-traffic delay optimality (Section~\ref{subsec:delay_optimality}) and characterize the steady-state queue-length distribution in heavy traffic (Section~\ref{subsec:qlength_distribution}).

The key object is the $\boldsymbol{\gamma}$-scaled queue-length vector
$\mathbf Q^{(\gamma)}(t)=(Q_l(t)/\gamma_l)_{l=1}^n$. In heavy traffic, the desired collapse direction is the one in which the scaled queue lengths are approximately equal across servers. Equivalently, the original queue-length vector is close to the $1$-dimensional cone
\[
    \mathcal C_\gamma := \{x\boldsymbol{\gamma}:x\geq0\}.
\]
Thus, in the scaled coordinates $\mathbf Q^{(\gamma)}(t)$, the collapse cone corresponds to the one-dimensional subspace spanned by the all-one vector $\mathbf 1$.

Using the $\boldsymbol{\gamma}$-weighted inner product and norm defined in Section~\ref{subsec:basic_notation}, the projection of $\mathbf Q^{(\gamma)}(t)$ onto $\mathrm{span}\{\mathbf 1\}$ is
\[
    \mathbf Q^{(\gamma)}_{\parallel}(t) = \frac{\langle \mathbf Q^{(\gamma)}(t),\mathbf 1\rangle_\gamma}{\langle \mathbf 1,\mathbf 1\rangle_\gamma}\mathbf 1 = \frac{\|\mathbf Q(t)\|_1}{\|\boldsymbol{\gamma}\|_1}\mathbf 1.
\]
The corresponding perpendicular component is
\[
    \mathbf Q^{(\gamma)}_{\perp}(t) := \mathbf Q^{(\gamma)}(t)-\mathbf Q^{(\gamma)}_{\parallel}(t),
\]
and hence
\[
    \|\mathbf Q^{(\gamma)}_{\perp}(t)\|_\gamma^2 = \sum_{l=1}^n \gamma_l \left( \frac{Q_l(t)}{\gamma_l} - \frac{\|\mathbf Q(t)\|_1}{\|\boldsymbol{\gamma}\|_1} \right)^2.
\]
This quantity measures the deviation of the scaled queue-length vector from the collapse subspace. The next subsection formalizes state-space collapse by proving a uniform steady-state bound on this weighted perpendicular component.

\subsection{State-space Collapse} \label{subsec:SSC}
The following theorem gives a sufficient condition under which the weighted perpendicular component of the $\boldsymbol{\gamma}$-scaled queue-length vector is uniformly bounded in steady state.

\begin{theorem}[State-Space Collapse of Scaled Queue Lengths]
    \label{thm_our:SSC_withrate}
    For any load balancing policy $\pi\in\Pi$, suppose that
    \begin{align} \label{eq:strictlyless_withrate}
        \sum_{l=1}^m f_{l,\eta} < \frac{\sum_{l=1}^m \mu_{\eta(l)}}{\sum_{l=1}^n \mu_l},
        \qquad \forall\, m\in[n]\setminus\{n\},\ \eta\in\mathcal S_n.
    \end{align}
    Then, for all sufficiently small $\epsilon$, there exists a function $N_\perp(n,T)$, independent of $\epsilon$, such that
    \[
        \mathbb E\left[ \|\mathbf Q^{(\gamma)}_{\perp}\|_\gamma^2 \right] \leq N_\perp^2(n,T).
    \]
    Furthermore, for all $n,T\geq1$ and $\boldsymbol{\gamma}\in\mathbb R_+^n$ satisfying $\|\boldsymbol{\gamma}\|_1=1$,
    \[
        N_\perp(n,T) \in \Theta\left( \frac{\gamma_{\max}}{\gamma_{\min}\sqrt{\gamma_{\min}}} Tn^{11/2} \right).
    \]
\end{theorem}
Note that the policy $\pi$ is throughput optimal when the inequalities in Equation~\eqref{eq:strictlyless_withrate} hold strictly for all $m\in[n]\setminus\{n\}$ and all $\eta\in\mathcal S_n$. The proof is given in Appendix~\ref{prf:thm_SSC_withrate_withratio}. The argument follows the drift-based steady-state state-space collapse approach of \cite{eryilmaz2012asymptotically-driftmethod}, adapted to our multi-period setting.

The state-space collapse bound in Theorem~\ref{thm_our:SSC_withrate} also implies that the raw queue-length vector becomes close to the collapse cone relative to the total queue length.
\begin{corollary}[Cone Collapse Interpretation]
    \label{co_our:relative_SSC}
    Under the assumptions of Theorem~\ref{thm_our:SSC_withrate},
    \[
        \lim_{\epsilon\downarrow0} \frac{ \mathbb E\left[ \left\| \mathbf Q - \frac{\boldsymbol\gamma}{\|\boldsymbol\gamma\|_1} \|\mathbf Q\|_1 \right\|_1 \right] }{ \mathbb E[\|\mathbf Q\|_1] } = 0.
    \]
    Equivalently, in heavy traffic, the queue-length vector $\mathbf Q$ becomes asymptotically close to the $1$-dimensional cone $\mathcal C_\gamma=\{x\boldsymbol\gamma:x\ge0\}$.
\end{corollary}
Corollary~\ref{co_our:relative_SSC} states that, relative to the total queue length, $\mathbf Q$ is asymptotically close to a vector of the form $x\boldsymbol\gamma$. Hence $Q_l\approx x\gamma_l$ for all $l\in[n]$, or equivalently, the scaled queue lengths $Q_l/\gamma_l$ become approximately equal across servers in heavy traffic. The proof is given in Appendix~\ref{prf:relative_SSC}.

\subsection{Asymptotic Delay Optimality} \label{subsec:delay_optimality}
We now exploit the SSC result above to establish the asymptotic delay optimality of policies. First, we provide an upper bound for the expected average queue length in steady state.

\begin{theorem}[An Upper Bound for Average Queue Length] \label{thm_our:SSCUpper_withrate_withratio}
    For any load balancing policy $\pi \in \Pi$, suppose that Equation~\eqref{eq:strictlyless_withrate} holds for all $\eta \in \mathcal{S}_n$ and $m\in [n]\setminus \{n\}$. For all $\epsilon$ small enough, we have
    \begin{align} \label{eq:thm_upper_withrate}
        \epsilon \mathbb{E} \left[ \frac{1}{n} \sum\limits_{l=1}^n Q_l \right ] \leq \frac{n \sigma_{\lambda}^2 + \sum\limits_{l=1}^{n}\sigma_{l}^2 }{2n} + \epsilon^2 \frac{T}{2n} + \epsilon \frac{T S_{\max} \gamma_{\min} + 2 T \| \boldsymbol{\gamma} \|_1 A_{\max} }{2 \gamma_{\min}} + \sqrt{\epsilon} \frac{ \| \boldsymbol{\gamma} \|_1 N_{\perp}(n,T) \sqrt{ S_{\max}} }{\sqrt{\gamma_{\min}}n},
    \end{align}
    where $ N_{\perp}(n,T) $ is defined in Theorem \ref{thm_our:SSC_withrate}.
\end{theorem}
The proof is given in Appendix \ref{prf:thm_upperbound_withrate_withratio}.

Also, for any load balancing policy in $\Pi$, the system is lower bounded by a single-server system with service rate $ \sum_{l=1}^n \mu_l $. Then, Lemma 5 of \cite{eryilmaz2012asymptotically-driftmethod} implies that
\begin{align*}
    \epsilon \mathbb{E} \left[ \frac{1}{n} \sum\limits_{l=1}^n Q_l \right ] \geq \frac{n \sigma_{\lambda}^2 + \sum\limits_{l=1}^{n}\sigma_{l}^2 + \epsilon^2 - nS_{\max} \epsilon }{2 n}.
\end{align*}
Combining this with Theorem \ref{thm_our:SSCUpper_withrate_withratio} and Little's Law, we have the following corollary.

\begin{corollary}[A Sufficient Condition for Heavy-traffic Delay Optimality] \label{co_our:delayoptimality}
    For any load balancing policy $\pi \in \Pi$, suppose that Equation~\eqref{eq:strictlyless_withrate} holds for all $\eta \in \mathcal{S}_n$ and $m\in [n]\setminus \{n\}$. Then $\pi$ is asymptotically delay optimal in heavy traffic.
\end{corollary}
By Remark \ref{rmk:weightedrand_boundary}, WRAND attains equality in the partial-sum condition. Therefore, Corollary \ref{co_our:delayoptimality} implies that any policy $\pi$ that strictly sends fewer jobs to the longer $\boldsymbol{\gamma}$-scaled queues than WRAND is heavy-traffic delay optimal.

\subsection{Queue-length Distribution} \label{subsec:qlength_distribution}
Next, we characterize the steady-state queue-length distribution.
\begin{theorem}[The Steady-state Queue-length Distribution]
\label{thm_our:queuelength_distribution}
    For any load balancing policy $\pi \in \Pi$, suppose that Equation~\eqref{eq:strictlyless_withrate} holds for all $\eta \in \mathcal{S}_n$ and $m\in [n]\setminus \{n\}$. Then $ \epsilon \mathbf{Q} $ converges in distribution to $ \Upsilon ( \gamma_l / \| \boldsymbol{\gamma} \|_1 )_{l=1}^n $ where $\Upsilon$ is an exponential random variable with mean $ \left ( n \sigma_{\lambda}^2 + \sum_{l=1}^n \sigma_l^2 \right ) / 2 $. 
\end{theorem}
The proof is given in Appendix \ref{prf:thm_our:queuelength_distribution}, and we extend the transform method of \cite{hurtado2020transform-Transformmethod} to the multi-period setting.

As shown in Table \ref{tab:policies}, JSQ and Po$d$ use $\mathbf{1}$ as the scaling vector $\boldsymbol{\gamma}$, meaning that the expected queue lengths are asymptotically equal in heavy traffic. This coincides with the results in \cite{hurtado2020transform-Transformmethod, hurtado2021throughput-Pod-stability-heterogeneous}. On the other hand, JSED uses the scaling vector $ \boldsymbol{\gamma} = ( \mu_1,...,\mu_n ) $, and thus the expected queue lengths are asymptotically proportional to their service rates in heavy traffic, or equivalently, the scaled backlogs \(Q_l/\mu_l\) are asymptotically equal across servers.

\section{Applications of the Framework}


In this section, we use our framework to quantify communication overhead under sparse sampling (Section \ref{subsec:sparse_overhead}), and to analyze and design load balancing policies under communication constraints (Section \ref{subsec:policy_design}).

\subsection{Communication Savings under Heavy-Traffic Delay Optimality} \label{subsec:sparse_overhead}


Using the communication metrics introduced in Section~\ref{subsec:general_framework}, we next discuss the trade-off between sparse communication and delay performance. Theorem~\ref{thm_our:SSCUpper_withrate_withratio} and Corollary~\ref{co_our:delayoptimality} imply that, if policy $\pi$ satisfies the strict majorization condition in Equation~\eqref{eq:strictlyless_withrate}, then $\pi$ remains heavy-traffic delay optimal as long as the sampling interval satisfies $T=o(1/\sqrt{\epsilon})$.

This condition reflects the balance between the heavy-traffic queue-size scale and the error introduced by stale information. In heavy traffic, the steady-state queue length grows on the order of $1/\epsilon$. Periodic sampling reduces communication, but it also creates a mismatch between the queue-length ordering observed at the beginning of a cycle and the true queue lengths during the cycle. As $T$ increases, this mismatch can accumulate over the cycle. The requirement $T=o(1/\sqrt{\epsilon})$ ensures that the additional terms caused by stale information in the upper bound of Theorem~\ref{thm_our:SSCUpper_withrate_withratio} are asymptotically negligible relative to the dominant $1/\epsilon$ queue-size scale.

Therefore, the sampling interval can grow as the system approaches heavy traffic, but not too quickly. For example, choosing $T=\epsilon^{-1/4}$ yields
\[
    \lim_{\epsilon\downarrow 0} M^\pi_{\mathrm{slot}} = \lim_{\epsilon\downarrow 0} \frac{m(\pi)}{T} = \lim_{\epsilon\downarrow 0} m(\pi)\epsilon^{1/4} =0.
\]
Thus, heavy-traffic delay optimality is compatible with vanishing communication per time slot, and equivalently with vanishing expected communication per arriving job, provided that the sampling interval grows slowly enough so that the stale-information error remains asymptotically negligible.

\subsection{Applications: Policy Analysis and Design} \label{subsec:policy_design}
Our framework provides a unified, tractable analysis of load balancing policies. We analyze a policy from the literature, then propose a new one and derive parameter conditions for stability and heavy-traffic delay optimality. Finally, we present numerical experiments to evaluate the performance of the policies.

\begin{table}[h!]
    \centering
    \footnotesize
    \caption{Description of Policies under the General Framework.}
    \label{tab:policies_new}
    \begin{tabular}{llll}
        \hline
        Policy & $T$ & $\boldsymbol{\gamma}$ & $\mathbf{f}_{\eta}$ \\
        \hline
        Po$d$-ED & 1 & $\boldsymbol{\mu}$ & $\mathbf{f}_{\eta}= \left ( 0,...,0,\binom{d-1}{d-1}/\binom{n}{d},...,\binom{n-1}{d-1}/\binom{n}{d} \right ), \forall \, \eta \in \mathcal{S}_n$  \\
        $k$-WRAND-SLQ-$d$ & $k$ & $\mathbf{1}$ & $ \mathbf{f}_{\eta} = \left (0,...,0, \mu_{\eta(d+1)} / (\sum_{j=d+1}^n \mu_{\eta(j)} ), ..., \mu_{\eta(n)} / (\sum_{j=d+1}^n \mu_{\eta(j)} ) \right ), \forall \, \eta \in \mathcal{S}_n $  \\
        \hline
    \end{tabular}
\end{table}

\subsubsection{Po$d$-Expected Delay (Po$d$-ED).} 
Po\(d\)-ED scales the queue lengths of the \(d\) sampled servers by their service rates and assigns jobs to the shortest scaled queue. Prior simulations can exhibit poor delay performance, as reported by \cite{gardner2021scalable-podJIQgeneral-fastslow} and \cite{abdul2022general-Pod-GeneralmodifiedPod}.

Po$d$-ED fits into the policy class $\Pi$ by choosing $T=1$ and the scaling vector $\boldsymbol{\gamma}=\boldsymbol{\mu}$. Upon each arrival, Po$d$-ED samples $d$ servers uniformly at random and assigns the job to the sampled server with the shortest $\mu$-scaled queue length. Table \ref{tab:policies_new} reports the corresponding vector $\mathbf{f}_{\eta}$, which is the same for all permutations $\eta \in \mathcal{S}_n$. Moreover, this $\mathbf{f}_\eta$ coincides with that induced by Po$d$.

Similar to the analysis of Po$d$ in \cite{hurtado2021throughput-Pod-stability-heterogeneous}, Po$d$-ED satisfies Assumption \ref{assump:persistent_bottleneck_dominance}. By theorems \ref{thm_our:stable_withrate_withratio} and \ref{thm_our:stable_withrate_withratio_necessary_general}, we obtain the following corollary characterizing stability, throughput optimality and heavy-traffic delay optimality.
\begin{corollary}[Properties for Po$d$-ED] \label{co_our:Pod-ED}
    Under Po$d$-ED, define
    \begin{align*}
        h^* := \min_{j\in\{d,\ldots,n\}} \left\{ \frac{\binom{n}{d}}{\binom{j}{d}} \left ( \sum_{l=1}^j \mu_l \right) \right\}.
    \end{align*}
    (a) If $n\lambda < h^*$, then the Markov chain $\{\mathbf{X}(t)\}_{t\ge 0}$ is positive recurrent; if $n\lambda > h^*$, then it is transient. \\
    (b) If, in addition,
    \begin{align*}
        \frac{ \sum\limits_{l=1}^j \mu_l }{ \sum\limits_{l=1}^n \mu_l } \geq \frac{ \binom{j}{d} }{ \binom{n}{d} }, \qquad \forall j\in\{d,\ldots,n-1\},
    \end{align*}
    then Po$d$-ED is throughput optimal. \\
    (c) If 
    \begin{align*}
        \frac{ \sum\limits_{l=1}^j \mu_l }{ \sum\limits_{l=1}^n \mu_l } > \frac{ \binom{j}{d} }{ \binom{n}{d} }, \qquad \forall j\in\{d,\ldots,n-1\},
    \end{align*}
    then Po$d$-ED is asymptotically delay optimal in heavy traffic. 
\end{corollary}
The proof is given in Appendix~\ref{prf:Pod-ED}.

\subsubsection{$k$-Weighted RAND-Skip the $d$ Longest Queues ($k$-WRAND-SLQ-$d$).}


\cite{Yishun2025k-slq-d} proposes $k$-SLQ-$d$ and shows that avoiding assigning jobs to the longest queues can lead to substantial performance improvements. Inspired by this idea, we introduce $k$-WRAND-SLQ-$d$.

Under $k$-WRAND-SLQ-$d$, we set the sampling interval to $T=k$ and use the scaling vector $\mathbf{1}$. At each sampling epoch $rk$, $r\in\mathbb{N}$, the dispatcher collects the queue lengths of all servers and identifies the $d$ longest queues, breaking ties uniformly at random. Denote this set by $\mathcal{I}_d(r)\subseteq[n]$, with $|\mathcal{I}_d(r)|=d$. During the sampling cycle $t\in\{rk,\ldots,(r+1)k-1\}$, all jobs arriving in slot $t$ are assigned to server $I^*(t)$, where
\[
    \mathbb{P}\left \{ I^*(t)=l \mid \mathcal{I}_d(r) \right \} 
    = \frac{\mu_l}{\sum\limits_{j\in[n]\setminus\mathcal{I}_d(r)}\mu_j}, \qquad l\in[n]\setminus\mathcal{I}_d(r).
\]
When $d=0$, $k$-WRAND-SLQ-$d$ reduces to $k$-WRAND. When $d=n-1$, only the shortest queue at the sampling epoch remains eligible, so $k$-WRAND-SLQ-$d$ coincides with $k$-JSQ, which samples queue lengths every $k$ slots and assigns arrivals during each cycle to the queue that was shortest at the sampling epoch. Table~\ref{tab:policies_new} reports the corresponding vector $\mathbf{f}_{\eta}$. The policy satisfies Assumption~\ref{assump:persistent_bottleneck_dominance}, leading to the following corollary.

\begin{corollary}[Properties for $k$-WRAND-SLQ-$d$] \label{co_our:k-WRAND-SLQ-d}
(a) $k$-WRAND-SLQ-$d$ is throughput optimal for all $d \in \{1,...,n-1\}$. \\
(b) $k$-WRAND-SLQ-$d$ is asymptotically delay optimal in heavy traffic for all $d \in \{1,...,n-1\}$.
\end{corollary}
The proof is given in Appendix~\ref{prf:k_WRAND_SLQ_d}. It is known that WRAND is not heavy-traffic delay optimal. Corollary~\ref{co_our:k-WRAND-SLQ-d} shows that \(k\)-WRAND-SLQ-\(d\) achieves heavy-traffic delay optimality even when the policy skips only the longest queue (\(d=1\)), whereas \(k\)-SLQ-\(d\) may need to skip a certain fraction of servers in heterogeneous systems.

We simulate $k$-WRAND-SLQ-$d$ in a heterogeneous system described in Section~\ref{ex:numerical} and consider $\rho\in\{0.9,0.9995\}$. 
Figure \ref{fig:Simulation_heterogeneous} reports $\ln(\text{average queue length})$ across all 50 queues versus $d \in \{0,...,49 \}$ for $k\in\{1,2,10,20,30,50,100,200,250\}$. First, as $\rho$ increases, moving from $d=0$ (i.e., $k$-WRAND) to $d=1$ yields a substantial reduction in the average queue length. Second, for less frequent updates, i.e., $k\geq2$, the curves become U-shaped, and the minimum is achieved at some $d<49$. We also observe a substantial reduction in the average queue length when moving from $d=49$ (i.e., $k$-JSQ) to $d=d^*<49$, especially for large $k$, suggesting that always assigning jobs to the shortest queue may be suboptimal in terms of delay. Another interesting observation is that the optimal $d^*$ decreases as $k$ increases. Intuitively, larger $k$ makes queue length information less timely, so a smaller $d$ (a larger candidate set) helps mitigate over-concentration induced by outdated information.

\begin{figure}[h!]
    \centering
    \begin{subfigure}{0.49\textwidth}
        \centering
        \includegraphics[width=\hsize]{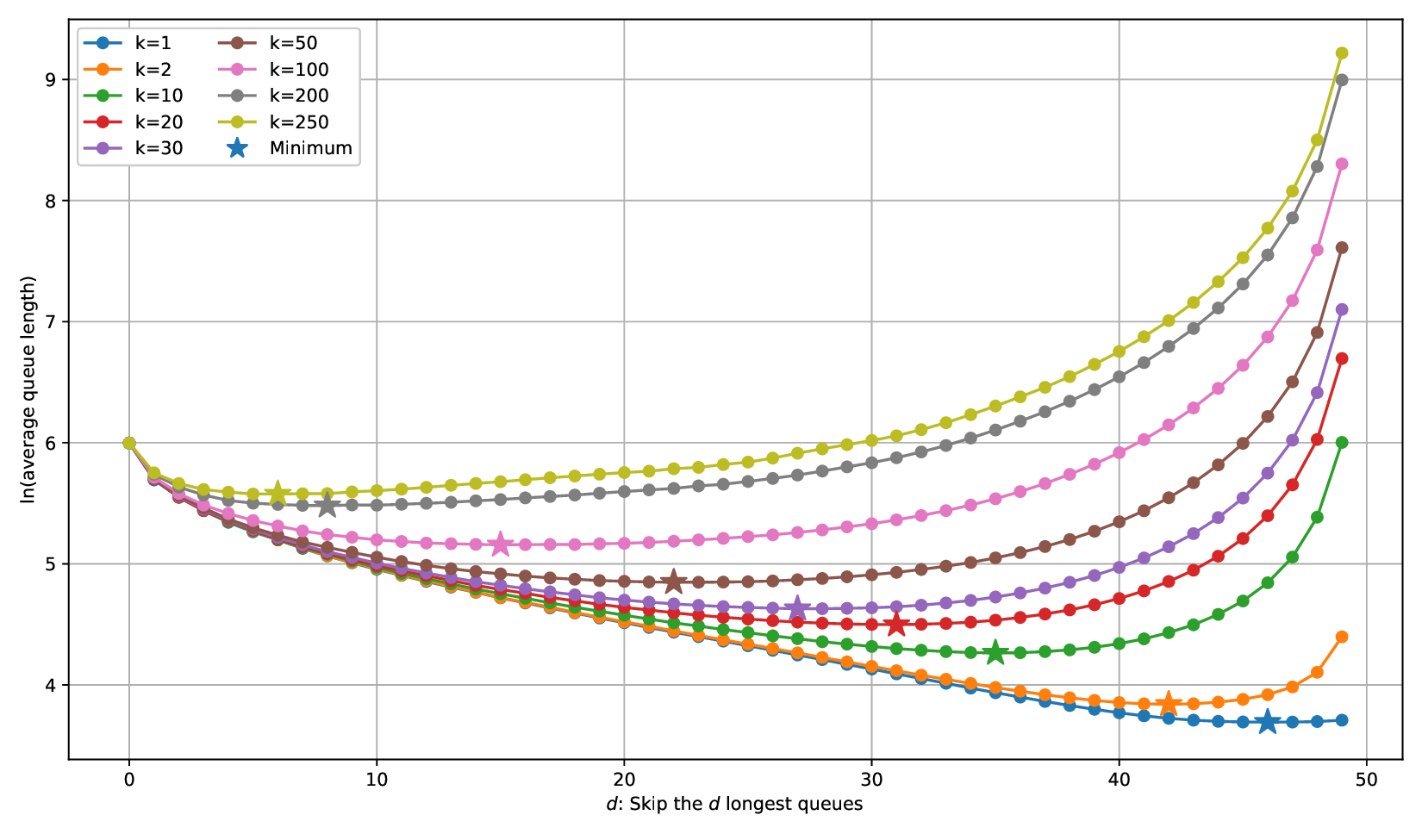}
        \caption{$\rho=0.9$}
        \label{fig:rho-0.9}
    \end{subfigure}
    \hfill
    \begin{subfigure}{0.49\textwidth}
        \centering
        \includegraphics[width=\hsize]{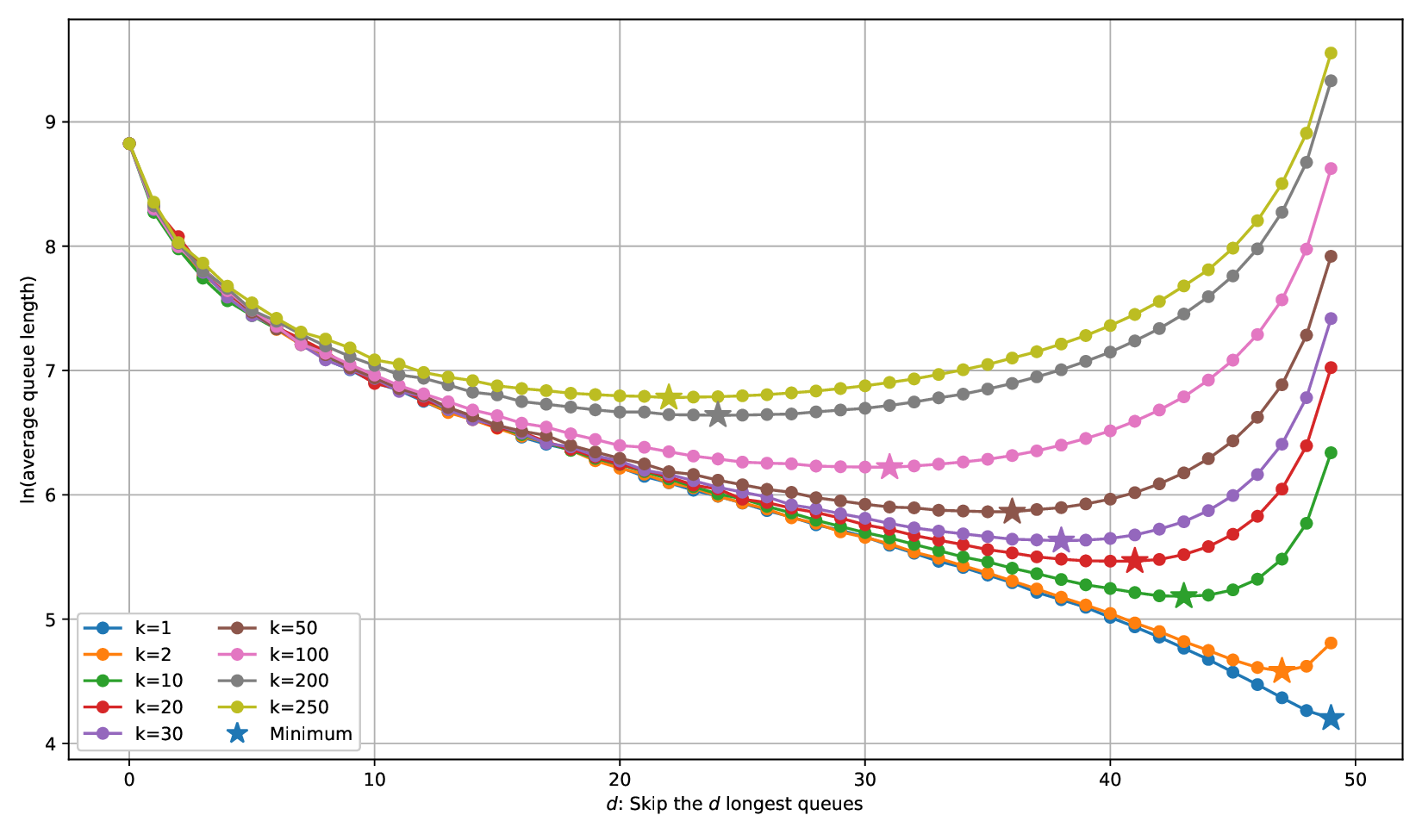}
        \caption{$\rho=0.9995$}
        \label{fig:rho-0.9995}
    \end{subfigure}

    \caption{$\ln(\text{average queue length})$ vs. $d$ under $k$-WRAND-SLQ-$d$.}
    \label{fig:Simulation_heterogeneous}

    \vspace{0.5em}
    \begin{minipage}{0.95\textwidth}
        \footnotesize
        \textit{Note.} The $x$-axis is $d\in\{0,1,\ldots,49\}$, the number of longest queues skipped under $k$-WRAND-SLQ-$d$ (with $n=50$). The $y$-axis reports $\ln(\text{average queue length})$, where the average queue length is computed as the time average over $10^6$ time slots and then averaged across the 50 servers. When $d=0$, the policy becomes $k$-WRAND; when $d=49$, it becomes $k$-JSQ. A $ \star $ marks the minimum average queue length over $d$ for each $k$. For example, the pairs $(k,d^*)$ under $ \rho =0.9995 $ are $(1,49)$, $(2,47)$, $(10,43)$, $(20,41)$, $(30,38)$, $(50,36)$, $(100,31)$, $(200,24)$, $(250,22)$.
    \end{minipage}
\end{figure}

\subsubsection{Numerical Simulation.} \label{ex:numerical}
We consider a system with $n=50$ servers, where $25$ servers have service rate $\mu_l=1$ and the other $25$ servers have $\mu_l=3$ (so $\sum_l \mu_l=100$) with $\sigma_{l}=1$ for all $l \in [n]$. We set $n\lambda = 99.95$ ($\rho=0.9995$) and $ \sqrt{n} \sigma_{\lambda} = 5$. Each simulation starts from empty queues and runs for $10^6$ time slots. To make policies comparable, we select parameters for each candidate policy such that the average number of messages per time slot $M_{\text{slot}}$ is $0.4$. Table \ref{tab:policies_simulations} reports the average queue length under different policies.

Table \ref{tab:policies_simulations} summarizes the simulation results. Among all the policies with the same communication budget ($M_{\text{slot}} = 0.4$), 250-WRAND-SLQ-22 and $10$-SLQ-$25$ achieve the shortest average queue length. Overall, Table~\ref{tab:policies_simulations} demonstrates that our framework enables the design of effective load-balancing policies under communication constraints.

\begin{table}[h!]
    \centering
    \footnotesize
    \caption{Simulation Results: $\rho = 0.9995$.}
    \label{tab:policies_simulations}
    \begin{tabular}{lrrrrrrrr}
        \hline
        Policy & $\mathrm{Mean}(\mathbf{Q})$ & $M_{\text{slot}}$ & $M_{\text{job}}$ \\
        \hline
        WRAND & 6805 & 0.00 & 0.0000 \\
        JSQ & 66 & 100.00 & 1.0004 \\
        Po$2$ & 1191 & 4.00 & 0.0400 \\
        Balanced Po$2$ & 412 & 4.00 & 0.0400 \\ 
        $250$-JSQ & 14108 & 0.40 & 0.0040 \\
        $10$-Po$2$ & 9234 & 0.40 & 0.0040 \\
        $10$-Po$2$-ED & 6428 & 0.40 & 0.0040 \\
        $10$-Balanced Po$2$ & 3688 & 0.40 & 0.0040 \\ 
        $125$-Po$25$ & 8356 & 0.40 & 0.0040 \\
        $10$-SLQ-$25$ & 777 & 0.40 & 0.0040 \\
        $250$-WRAND-SLQ-$1$ & 4248 & 0.40 & 0.0040 \\
        $250$-WRAND-SLQ-$22$ & 880 & 0.40 & 0.0040  \\
        \hline
    \end{tabular}

    \vspace{0.5em}
    \begin{minipage}{0.95\textwidth}
        \footnotesize
        \textit{Note.} $\mathrm{Mean}(\mathbf{Q})$ is computed across all 50 servers using per-server time-averaged queue lengths. Balanced Po$2$ samples $2$ servers with replacement according to probabilities proportional to their service rates \cite{chen2012asymptotic-Pod-Heterogeneous-Proportional}.
    \end{minipage}
\end{table}

\FloatBarrier
\bibliographystyle{elsarticle-num}
\bibliography{references}

\clearpage

\appendix

\section{Auxiliary Notation and Preliminary Theorems}
We begin by introducing additional notation and preliminary theorems that will be used throughout the proofs. We first define several cycle-level quantities induced by a policy in the general framework $\Pi$, and then state the drift-based tools used to establish positive recurrence and moment bounds.

\subsection{Auxiliary Notation} \label{subsec:auxiliary_notation}
\subsubsection{\texorpdfstring{$\sqrt{\boldsymbol{\gamma}}$-scaled Queue Length Vector.}{sqrt-gamma-scaled auxiliary queue length}}
Throughout the proofs, we use the following $\sqrt{\boldsymbol{\gamma}}$-scaled queue-length vector:
\[
    \mathbf O(t) := \left( \frac{Q_l(t)}{\sqrt{\gamma_l}} \right)_{l=1}^n.
\]
This transformation allows us to express the drift of quadratic Lyapunov functions in terms of the $\boldsymbol{\gamma}$-scaled queue lengths \(Q_l(t)/\gamma_l\). Indeed, since \(\|\mathbf O(t)\|_2^2=\sum_{l=1}^n Q_l^2(t)/\gamma_l\), expanding the drift of \(\|\mathbf O(t)\|_2^2\) naturally produces linear terms of the form \(Q_l(t)/\gamma_l\). These are precisely the $\boldsymbol{\gamma}$-scaled queue lengths used to order servers in the policy and to state the partial-sum conditions in the stability and heavy-traffic arguments.

When proving state-space collapse property (Theorem~\ref{thm_our:SSC_withrate}), we additionally use the vector
\[
    \mathbf c := \left( \sqrt{\gamma_l} \right)_{l=1}^n.
\]
Let $\mathbf O_{\parallel}(t)$ be the Euclidean projection of $\mathbf O(t)$ onto
$\operatorname{span}\{\mathbf c\}$, namely
\[
    \mathbf O_{\parallel}(t) := \frac{\langle \mathbf{O}(t), \mathbf{c} \rangle} {\langle \mathbf c,\mathbf c\rangle}\mathbf c = \frac{\|\mathbf Q(t)\|_1}{\|\boldsymbol\gamma\|_1}\mathbf c.
\]
We define the perpendicular component by
\[
    \mathbf O_{\perp}(t) := \mathbf O(t)-\mathbf O_{\parallel}(t).
\]
Equivalently, for each $l\in[n]$,
\[
    O_{\perp, l}(t) = \frac{Q_l(t)}{\sqrt{\gamma_l}} - \frac{\|\mathbf Q(t)\|_1}{\|\boldsymbol\gamma\|_1}\sqrt{\gamma_l} = \sqrt{\gamma_l} \left( \frac{Q_l(t)}{\gamma_l} - \frac{\|\mathbf Q(t)\|_1}{\|\boldsymbol\gamma\|_1} \right).
\]
Thus, we have
\[
    \|\mathbf O_\perp(t)\|_2^2 = \sum_{l=1}^n \gamma_l \left( \frac{Q_l(t)}{\gamma_l} - \frac{\|\mathbf Q(t)\|_1}{\|\boldsymbol\gamma\|_1} \right)^2 = \|\mathbf Q^{(\gamma)}_\perp(t)\|_\gamma^2.
\]
In the proof of Theorem~\ref{thm_our:SSC_withrate} in Appendix~\ref{prf:thm_SSC_withrate_withratio}, controlling \(\mathbf O_\perp(t)\) is equivalent to controlling the deviations of the scaled queue lengths \(Q_l(t)/\gamma_l\) from their common weighted average.

\subsubsection{Conditional Moments of Cumulative Arrivals.}
For a policy $\pi\in \Pi$, over each block of $T$ time slots, the cumulative number of arrivals sent to the $l$th longest server is
\[ \sum\limits_{t=rT}^{(r+1)T-1} A_{\eta_r(l)} = \sum_{t=rT}^{(r+1)T-1} A(t)Z_{\eta_r(l)}(t). \]
Since the dispatching decisions $ \mathbf{Z}(t) $ are independent of the arrival process, we have, for any permutation $\eta \in \mathcal{S}_n$,
\begin{align*}
    \mathbb{E} \left [ \sum\limits_{t=rT}^{(r+1)T-1} A_{\eta_r(l)} \, \middle | \, \eta_r = \eta \right ] & = \mathbb{E} \left [ \sum\limits_{t=rT}^{(r+1)T-1} A(t)Z_{\eta(l)}(t) \, \middle | \, \eta_r = \eta \right ]  \\
    & = \sum\limits_{t=rT}^{(r+1)T-1} \mathbb{E} \left [ A(t) Z_{\eta(l)}(t) \, \middle | \, \eta_r = \eta \right ] \\
    & = \sum\limits_{t=rT}^{(r+1)T-1} \mathbb{E} \left [ \mathbb{E}[ A(t) Z_{\eta(l)}(t) \, | \, \eta_r =\eta, \mathbf{Z}(t) ] \, \middle | \, \eta_r = \eta \right ] \\
    & = \sum\limits_{t=rT}^{(r+1)T-1} \mathbb{E} \left [ \mathbb{E}[ A(t) ] Z_{\eta(l)}(t) \, \middle | \, \eta_r = \eta \right ] \\
    & = \mathbb{E}[ A(1) ] \left ( \sum\limits_{t=rT}^{(r+1)T-1} \mathbb{E} \left [ Z_{\eta(l)}(t) \, \middle | \, \eta_r = \eta \right ] \right ) \\
    & = n\lambda \mathbb{E} \left [ \sum\limits_{t=rT}^{(r+1)T-1}  Z_{\eta(l)}(t) \, \middle | \, \eta_r = \eta \right ] \\
    & \overset{(a)}{=} n\lambda \mathbb{E} \left [ N_{r,\eta}(l) \, \middle | \, \eta_r = \eta \right] \\
    & \overset{(b)}{=} n\lambda T f_{l,\eta}.
\end{align*}
where (a) follows from the definition $ N_{r,\eta}(l) = \sum_{t=rT}^{(r+1)T-1} Z_{\eta(l)}(t) $ and (b) follows from the definition $f_{l,\eta} = \mathbb{E}[ N_{r,\eta_r}(l) / T \, | \, \eta_r=\eta  ]$ in Section 
\ref{subsec:NotationfromFramework}.

The variance is
\begin{align*}
    \mathrm{Var} \left ( \sum\limits_{t=rT}^{(r+1)T-1} A_{\eta_r(l)} \, \middle | \, \eta_r = \eta \right ) &= \mathbb{E}[N_{r,\eta_r}(l) \, | \, \eta_r = \eta] \mathrm{Var}(A(1)) + \mathbb{E}[A(1)]^2 \mathrm{Var}(N_{r,\eta_r}(l)) \\
    &= Tf_{l,\eta}n\sigma_{\lambda}^2 + T^2 n^2\lambda^2 \tau_{l,\eta}^2.
\end{align*} 
The second moment is then
\[ \mathbb{E} \left [ \left ( \sum\limits_{t=rT}^{(r+1)T-1} A_{\eta_r(l)} \right )^2 \, \middle | \, \eta_r = \eta \right ] = Tf_{l,\eta}n\sigma_{\lambda}^2 + T^2 n^2\lambda^2 \left (\tau_{l,\eta}^2 +f_{l,\eta}^2 \right ). \]

\subsection{Preliminary Theorems}
We first introduce the Foster-Lyapunov theorem, which will be used in Appendix~\ref{prf:thm_stable_withrate_withratio} to establish positive recurrence of the discrete-time Markov chain. We restate Theorem C.26 in \cite{dai2020processing-ForstLyapunovThm} in the following form, adapted for our use.
\begin{apptheorem}[Foster-Lyapunov theorem \cite{dai2020processing-ForstLyapunovThm}] \label{thm:Foster-Lyapunov-criteria}
    Let $\{\mathbf{\Xi}(t)\}_{t\ge 0}$ be an irreducible discrete-time Markov chain on a countable state space $\mathcal{S}$. Suppose there exist constants $b>0$ and $\delta_F>0$, a finite set $\mathcal{K}\subset \mathcal{S}$, and a nonnegative function $V:\mathcal{S}\to\mathbb{R}_+$ such that the following conditions hold.
    
    \noindent \textbf{Condition A1}.
    For every $\mathbf{x}\in \mathcal{S}\setminus \mathcal{K}$,
    \begin{align*}
        \mathbb{E} \left[ V(\mathbf{\Xi}(t+1)) - V(\mathbf{\Xi}(t)) \mid \mathbf{\Xi}(t)=\mathbf{x}\right] \leq -\delta_F.    
    \end{align*}
    
    \noindent \textbf{Condition A2}.
    For every $\mathbf{x}\in \mathcal{K}$,
    \begin{align*}
        \mathbb{E}\!\left[V(\mathbf{\Xi}(t+1))-V(\mathbf{\Xi}(t)) \mid \mathbf{\Xi}(t)=\mathbf{x}\right]\le -\delta_F + b.    
    \end{align*}
    Then the Markov chain $\{\mathbf{\Xi}(t)\}_{t\ge 0}$ is positive recurrent.
\end{apptheorem}
We next introduce a Hajek-type exponential drift theorem from Theorem 2.3 of \cite{hajek1982hitting-lyapunovupperbound}. In our analysis, we apply this result in two places. First, in Appendix~\ref{prf:thm_SSC_withrate_withratio}, we apply it with \(H_{t'}=V_{\perp}(\tilde{\mathbf O}(t'))\), where \(V_{\perp}\) is the Lyapunov function used in the state-space-collapse proof, so that \(H_{t'+1}-H_{t'}=\Delta V_{\perp}(\tilde{\mathbf O}(t'))\). This application gives an exponential-moment bound for \(\|\mathbf O_{\perp}\|_2\), which is used to control the perpendicular component in the SSC and transform arguments. Second, in Appendix~\ref{prf:thm_upperbound_withrate_withratio}, we apply the same result with \(H_{t'}=V(\tilde{\mathbf O}(t'))=\|\tilde{\mathbf O}(t')\|_2\), so that \(H_{t'+1}-H_{t'}=\Delta V(\tilde{\mathbf O}(t'))\). This second application establishes exponential integrability of the total \(\sqrt{\boldsymbol\gamma}\)-scaled queue length and, along the heavy-traffic sequence, shows that the moment generating function of \(\epsilon\|\tilde{\mathbf Q}\|_1\) is well defined in a fixed neighborhood of the origin.

\begin{apptheorem}[Hajek-type exponential drift theorem \cite{hajek1982hitting-lyapunovupperbound}] \label{thm:hajek_drift}
    Let $ \{ H_t \}_{t \geq 0}$ be a real-valued stochastic process adapted to $ \{ \mathcal{F}_t \}_{t \geq 0}$. Assume the following two conditions.
    
    \noindent \textbf{Condition B1}.
    There exist constants $a\in\mathbb{R}$ and $\delta_H>0$ such that, for every $t \geq 0$,
    \begin{align*}
        \mathbb{E}\!\left[\,H_{t+1}-H_t+\delta_H;\, H_t > a \,\middle|\, \mathcal{F}_t \right] \leq 0.
    \end{align*}
    
    \noindent
    \textbf{Condition B2}.
    There exists a constant $\theta>0$ and a nonnegative random variable $G$ such that, for every $ t \geq 0$ and every $ u \geq 0$,
    \begin{align*}
        \mathbb{P}\!\left(|H_{t+1}-H_t| > u \,\middle| \, \mathcal{F}_t\right) \leq \mathbb{P}(G > u),
    \end{align*}
    and
    \begin{align*}
        \mathbb{E}\!\left[e^{\theta G}\right]=D<\infty.
    \end{align*}
    Define
    \begin{align*}
        C_H:=\frac{\mathbb{E}[e^{\theta G}]-(1+\theta \mathbb{E}[G])}{\theta^2}.
    \end{align*}
    Choose $\zeta$ such that
    \begin{align*}
        0<\zeta\le \theta, \qquad \zeta<\frac{\delta_H}{C_H},
    \end{align*}
    and define
    \begin{align*}
        \varrho:=1-\delta_H\zeta+ C_H \zeta^2.
    \end{align*}
    Then $0<\varrho<1$, and the following two consequences hold:
    \noindent For every $t \geq 0$,
    \begin{align*}
        \mathbb{E}\!\left[e^{\zeta (H_{t+1}-H_t)};\, H_t > a \,\middle|\, \mathcal{F}_t \right]\leq \varrho.
    \end{align*}
    \noindent For every $t \geq 0$,
    \begin{align*}
        \mathbb{E}\!\left[e^{\zeta (H_{t+1}-a)};\, H_t \le a \,\middle|\, \mathcal{F}_t\right]\le D.
    \end{align*}
    
    \noindent Consequently, for all $t \geq 0$,
    \begin{align*}
        \mathbb{E}\!\left[e^{\zeta H_t}\middle| \mathcal{F}_0\right]
        \leq \varrho^t e^{\zeta H_0} + \frac{1-\varrho^t}{1-\varrho}De^{\zeta a}.
    \end{align*}
\end{apptheorem}

\section{Proof of Theorem \ref{thm_our:stable_withrate_withratio} } \label{prf:thm_stable_withrate_withratio}
Suppose that 
\begin{align*}
    n\lambda < \min\limits_{\eta \in S_n} \left\{ \min\limits_{ m\in [n] } \left \{  \frac{ \sum\limits_{l=1}^m \mu_{\eta(l)} }{ \sum\limits_{l=1}^m f_{l, \eta} } \right \} \right\}.
\end{align*}
We aim to establish that the Markov chain $\{\mathbf X(t)\}_{t\ge 0}$ defined in Equation~\eqref{def:Mkchain} is positive recurrent. To do so, we let $t$ be a multiple of $T$ and define $t' = t/T$. We then consider the chain sampled at the beginning of each dispatching cycle,
\[
    \tilde{\mathbf X}(t') := \mathbf X(Tt'), \qquad t'\in\mathbb{N}.
\]
Its queue-length component is \( \tilde{\mathbf Q}(t') := \mathbf Q(Tt')\) for all $t'\in\mathbb{N}$. We also define the corresponding sampled $\sqrt{\boldsymbol\gamma}$-scaled queue-length vector by
\[
    \tilde{\mathbf O}(t') := \mathbf O(Tt'), \qquad t'\in\mathbb{N},
\]
where $\mathbf O(t)$ is defined in Section~\ref{subsec:auxiliary_notation}.

At the sampling epoch $t=Tt'$, the auxiliary components of $\tilde{\mathbf X}(t')=\mathbf X(Tt')$ take values in finite sets: $\eta_{t'}\in S_n$, $\phi(t')$ belongs to a finite set of $T$-length decision sequences, and the cycle-position component belongs to $\{0,1,\ldots,T-1\}$, taking value zero at cycle boundaries.

This reduction is without loss for the recurrence argument. Since $\boldsymbol\gamma$ has strictly positive components, $\tilde{\mathbf O}(t')$ is a one-to-one deterministic transformation of $\tilde{\mathbf Q}(t')$. Hence, positive recurrence of the sampled $\sqrt{\boldsymbol{\gamma}}$-scaled queue-length process implies positive recurrence of the sampled chain $\{\tilde{\mathbf X}(t')\}_{t'\ge0}$.

Once positive recurrence of the sampled chain is established, positive recurrence of the corresponding recurrent class in the original chain follows from the definition of positive recurrence in terms of finite mean return times. Indeed, the sampled chain is obtained by observing the original chain every $T$ steps, so the return time in the original chain is at most $T$ times the return time in the sampled chain. \\

Define the Lyapunov functions for any $\sqrt{\boldsymbol{\gamma}}$-scaled queue-length vector $\mathbf{o} \in \mathbb R_+^n$. 
\[
    V(\mathbf{o} )=\|\mathbf{o} \|_2, \qquad W(\mathbf{o} )= \|\mathbf{o} \|_2^2.
\]
We next verify Conditions A1 and A2 of Theorem~\ref{thm:Foster-Lyapunov-criteria} for the embedded chain $\{\tilde{\mathbf X}(t')\}_{t'\ge0}$ using a Lyapunov function that depends on $\tilde{\mathbf O}(t')$. \\

\noindent \textbf{Step 1}. We show that the drift of $V(\tilde{\mathbf O}(t'))$ is uniformly absolutely bounded, which verifies Condition A2.

\begin{lemma} \label{lem_our:bd_absolute_delta_vO}
    For any policy $\pi \in \Pi$, we have
    \begin{align}
        | \Delta V(\tilde{\mathbf{O}}(t')) | \leq \frac{1}{\sqrt{\gamma_{\min}}} T ( nA_{\max} + \sqrt{n} S_{\max} ), \qquad \forall \, \tilde{\mathbf{O}}(t^{\prime}) \in \mathbb{R}^n_{+}, t^{\prime} \in \mathbb{N} \notag
    \end{align}
\end{lemma}
The proof is given in Appendix \ref{prf:bd_absolute_delta_vO}.\\

\noindent \textbf{Step 2}. We show that $V(\tilde{\mathbf O}(t'))$ has a negative drift when the queue-length component $\tilde{\mathbf Q}(t')$ is sufficiently large, which verifies Condition A1.

\begin{lemma} \label{lem_our:bd_negative_delta_vO}
    For any policy $\pi \in \Pi$, suppose that the condition in Equation~\eqref{eq:feasibleregion_withrate} holds. Then there exist constants $\epsilon^{\prime} > 0$ and $0< K < \infty$ such that, for any sampled scaled queue-length vector $\mathbf{o}$ with corresponding queue-length vector $\mathbf{q}$, if
    \begin{align*}
        \| \mathbf{q} \|_1  \geq  \frac{TS_{\max}\|\boldsymbol{\gamma}\|_1}{\gamma_{\min}},
    \end{align*}
    then
    \begin{align}
        \mathbb{E} \left[ \Delta V(\tilde{\mathbf{O}}(t')) \,\middle|\, \tilde{\mathbf{O}}(t')=\mathbf{o} \right] \leq -\frac{\epsilon' T \sqrt{\gamma_{\min}}}{\|\boldsymbol{\gamma}\|_1} + \frac{K\sqrt{\gamma_{\max}}}{2\|\mathbf q\|_2}. \notag
    \end{align}
\end{lemma}
The proof is given in Appendix~\ref{prf:bd_negative_delta_vO}.\\

\noindent \textbf{Step 3}. We now establish positive recurrence.

To verify Condition A1 in Theorem~\ref{thm:Foster-Lyapunov-criteria}, define
\[
    \mathcal{K}:= \left\{ \mathbf{q}\in\mathbb{Z}_+^n: \|\mathbf{q}\|_1 < \frac{TS_{\max}\|\boldsymbol{\gamma}\|_1}{\gamma_{\min}} \right\}
    \cup \left\{ \mathbf{q}\in\mathbb{Z}_+^n: \|\mathbf{q}\|_2 \leq \frac{K\sqrt{\gamma_{\max}}\|\boldsymbol{\gamma}\|_1} {\epsilon'T\sqrt{\gamma_{\min}}} \right\}.
\]
Since the state space is countable and both thresholds are finite, the set $\mathcal{K}$ is finite.

Now let $\mathbf{q}\notin\mathcal{K}$, and let $\mathbf{o}$ be the corresponding sampled scaled queue-length vector. Then
\[
    \|\mathbf{q}\|_1 \geq \frac{TS_{\max}\|\boldsymbol{\gamma}\|_1}{\gamma_{\min}}, \qquad 
    \|\mathbf{q}\|_2 > \frac{K\sqrt{\gamma_{\max}}\|\boldsymbol{\gamma}\|_1} {\epsilon'T\sqrt{\gamma_{\min}}}.
\]
Hence,
\[
    \frac{K\sqrt{\gamma_{\max}}}{2\|\mathbf{q}\|_2}
    <
    \frac{\epsilon'T\sqrt{\gamma_{\min}}}{2\|\boldsymbol{\gamma}\|_1}.
\]
By Lemma~\ref{lem_our:bd_negative_delta_vO}, we have
\begin{align*}
    \mathbb{E} \left[ \Delta V(\tilde{\mathbf{O}}(t')) \,\middle|\, \tilde{\mathbf{O}}(t')=\mathbf{o} \right]
    & \leq - \frac{\epsilon'T\sqrt{\gamma_{\min}}}{\|\boldsymbol{\gamma}\|_1} + \frac{\epsilon'T\sqrt{\gamma_{\min}}}{2\|\boldsymbol{\gamma}\|_1} \\
    & = -\frac{\epsilon'T\sqrt{\gamma_{\min}}}{2\|\boldsymbol{\gamma}\|_1}.
\end{align*}
Thus, Condition A1 holds with
\[
    \delta_F := \frac{\epsilon'T\sqrt{\gamma_{\min}}}{2\|\boldsymbol{\gamma}\|_1}.
\]

Condition A2 holds in the finite set $\mathcal{K}$ by the uniform absolute drift bound in Lemma~\ref{lem_our:bd_absolute_delta_vO}. Specifically, choose
\[
    b:= \frac{1}{\sqrt{\gamma_{\min}}} T \left(nA_{\max}+\sqrt{n}S_{\max}\right) +\delta_F
    = \frac{1}{\sqrt{\gamma_{\min}}} T \left(nA_{\max}+\sqrt{n}S_{\max}\right) + \frac{\epsilon'T\sqrt{\gamma_{\min}}}{2\|\boldsymbol{\gamma}\|_1}.
\]
Therefore, Theorem~\ref{thm:Foster-Lyapunov-criteria} applies to the sampled 
\(\sqrt{\boldsymbol{\gamma}}\)-scaled queue-length process on the communicating class under consideration. Hence this sampled process is positive recurrent on that class. Since \(\tilde{\mathbf O}(t')\) is a one-to-one deterministic transformation of \(\tilde{\mathbf Q}(t')\), positive recurrence of \(\{\tilde{\mathbf O}(t')\}_{t'\ge0}\) implies positive recurrence of the embedded queue-length process \(\{\tilde{\mathbf Q}(t')\}_{t'\ge0}\) on the corresponding class. Because the auxiliary components of \(\tilde{\mathbf X}(t')\) take values in finite sets, the embedded chain \(\{\tilde{\mathbf X}(t')\}_{t'\ge0}\) is also positive recurrent. Finally, since the original chain is a finite-cycle extension of the embedded chain, positive recurrence transfers to the corresponding communicating class of \(\{\mathbf X(t)\}_{t\ge0}\).

\subsection{Proof of Lemma~\ref{lem_our:bd_absolute_delta_vO}} \label{prf:bd_absolute_delta_vO}
We establish a uniform absolute bound on the drift of $V(\tilde{\mathbf O}(t'))$:
\begin{align}
    | \Delta V(\tilde{\mathbf{O}}(t')) | & = | \| \tilde{\mathbf{O}}(t'+1) \|_2 - \| \tilde{\mathbf{O}}(t^{\prime}) \|_2 | \notag \\
    & \overset{(a)}{\leq} \| \tilde{\mathbf{O}}(t'+1) - \tilde{\mathbf{O}}(t')  \|_2 \notag \\
    & = \sqrt{ \sum\limits_{l=1}^n \left ( \frac{Q_{\eta_{t'}(l)}(t+T)}{ \sqrt{\gamma_{\eta_{t'}(l)}}}  - \frac{Q_{\eta_{t'}(l)}(t)}{ \sqrt{\gamma_{\eta_{t'}(l)}}} \right )^2 }  \notag \\
    & \leq \frac{1}{\sqrt{\gamma_{\min}}} \sqrt{ \sum\limits_{l=1}^n \left (Q_{\eta_{t'}(l)}(t) + \sum\limits_{j=0}^{T-1}A_{\eta_{t'}(l)}(t+j) - \sum\limits_{j=0}^{T-1}S_{\eta_{t'}(l)}(t+j) + \sum\limits_{j=0}^{T-1}U_{\eta_{t'}(l)}(t+j) - Q_{\eta_{t'}(l)}(t) \right )^2 } \notag \\
    & = \frac{1}{\sqrt{\gamma_{\min}}} \sqrt{ \sum\limits_{l=1}^n \left (\sum\limits_{j=0}^{T-1}A_{\eta_{t'}(l)}(t+j) - \left( \sum\limits_{j=0}^{T-1}S_{\eta_{t'}(l)}(t+j) - \sum\limits_{j=0}^{T-1}U_{\eta_{t'}(l)}(t+j) \right) \right )^2 } \notag \\
    & \overset{(b)}{\leq} \frac{1}{\sqrt{\gamma_{\min}}} \left [ \sqrt{ \sum\limits_{l=1}^n \left (\sum\limits_{j=0}^{T-1}A_{\eta_{t'}(l)}(t+j) \right )^2 } + \sqrt{ \sum\limits_{l=1}^n \left ( \sum\limits_{j=0}^{T-1}S_{\eta_{t'}(l)}(t+j) - \sum\limits_{j=0}^{T-1}U_{\eta_{t'}(l)}(t+j) \right )^2 } \right ] \notag \\
    & \overset{(c)}{\leq} \frac{1}{\sqrt{\gamma_{\min}}} \left [ \sqrt{ \sum\limits_{l=1}^n \left (\sum\limits_{j=0}^{T-1}A_{\eta_{t'}(l)}(t+j) \right )^2 } + \sqrt{ \sum\limits_{l=1}^n \left ( \sum\limits_{j=0}^{T-1}S_{\eta_{t'}(l)}(t+j) \right )^2 } \right ] \notag \\
    & \overset{(d)}{\leq} \frac{1}{\sqrt{\gamma_{\min}}} \left [ \left | \sum\limits_{l=1}^n \left (\sum\limits_{j=0}^{T-1}A_{\eta_{t'}(l)}(t+j) \right ) \right |  + \sqrt{ \sum\limits_{l=1}^n \left ( \sum\limits_{j=0}^{T-1}S_{\eta_{t'}(l)}(t+j) \right )^2 } \right ] \notag \\
    & = \frac{1}{\sqrt{\gamma_{\min}}} \left [ \left | \sum\limits_{j=0}^{T-1} \left (\sum\limits_{l=1}^{n}A_{\eta_{t'}(l)}(t+j) \right ) \right |  + \sqrt{ \sum\limits_{l=1}^n \left ( \sum\limits_{j=0}^{T-1}S_{\eta_{t'}(l)}(t+j) \right )^2 } \right ] \notag \\
    & \overset{(e)}{\leq} \frac{1}{\sqrt{\gamma_{\min}}} \left ( TnA_{\max} + \sqrt{ nT^2S^2_{\max}} \right ) \notag \\
    & = \frac{1}{\sqrt{\gamma_{\min}}} T ( nA_{\max} + \sqrt{n} S_{\max} ), \label{eq:absolutebound_withrate}
\end{align}
where (a) follows from the fact that $\| \mathbf{x} \|_2 - \| \mathbf{y} \|_2 \leq \| \mathbf{x}-\mathbf{y} \|_2$ for any $\mathbf{x},\mathbf{y}\in\mathbb{R}^n $, (b) follows from the triangle inequality, (c) follows from the fact that $ S_l(t) \geq U_l(t)\geq 0,\  \forall \ l,t $, (d) holds because $\| \mathbf{x} \|_2 \leq \| \mathbf{x} \|_1$ and (e) follows from the facts that $ \sum_{l=1}^n A_{\eta(l)}(t) \leq n A_{\max}$ for all $\eta$ and $t$, and $S_l(t)\leq S_{\max}$ for all $l$ and $t$.

\subsection{Proof of Lemma~\ref{lem_our:bd_negative_delta_vO}} \label{prf:bd_negative_delta_vO}
First, $\Delta V(\tilde{\mathbf{O}}(t'))$ can be bounded by $\Delta W(\tilde{\mathbf{O}}(t'))$ as follows:
\begin{align}
    \Delta V(\tilde{\mathbf{O}}(t')) & = \| \tilde{\mathbf{O}}(t'+1) \|_2 - \| \tilde{\mathbf{O}}(t') \|_2  \notag \\
    & = \sqrt{ \| \mathbf{O}(t+T) \|^2_2} - \sqrt{ \| \mathbf{O}(t) \|^2_2} \notag \\
    & \overset{(a)}{\leq} \frac{1}{ 2\| \mathbf{O}(t) \|_2} \Big( \| \mathbf{O}(t+T) \|_2^2 - \| \mathbf{O}(t) \|_2^2 \Big) \notag \\
    & = \frac{ 1 }{2\| \mathbf{O}(t) \|_2}\big[ \Delta W\big(\tilde{\mathbf{O}}(t') \big) \big] \label{eq:delta_V_bound_by_delta_W}
\end{align}
where (a) follows from the fact that $ g(x)=\sqrt{x} $ is concave for $ x \geq 0 $, and thus $ g(y) - g(x)\leq (y-x)g^{\prime}(x) $.

Second, we show that $W(\tilde{\mathbf O}(t'))$ has a negative drift when the total queue length is sufficiently large. Specifically, suppose that
\[
    \|\tilde{\mathbf Q}(t')\|_1  \geq  \frac{TS_{\max}\|\boldsymbol{\gamma}\|_1}{\gamma_{\min}}.
\]
Recall that $\tilde Q_l(t')=Q_l(Tt')=Q_l(t)$. We use $\eta_{t'}$ to denote the ordering of the $\boldsymbol{\gamma}$-scaled queue lengths at time $t$, so that $\eta_{t'}(i)$ is the index of the server with the $i$-th longest $\boldsymbol{\gamma}$-scaled queue length:
\[
    \frac{\tilde Q_{\eta_{t'}(1)}(t')}{\gamma_{\eta_{t'}(1)}} \geq \frac{\tilde Q_{\eta_{t'}(2)}(t')}{\gamma_{\eta_{t'}(2)}} \geq
    ...
    \geq \frac{\tilde Q_{\eta_{t'}(n)}(t')}{\gamma_{\eta_{t'}(n)}} .
\]
The assumption $\|\tilde{\mathbf Q}(t')\|_1  \geq  TS_{\max}\|\boldsymbol{\gamma}\|_1/\gamma_{\min}$ implies that the server with the longest $\boldsymbol{\gamma}$-scaled queue length has at least $TS_{\max}$ jobs. Indeed,
\begin{align}
    \frac{\tilde{Q}_{\eta_{t'}(1)}(t')}{\gamma_{\eta_{t'}(1)}} & = \frac{ \|\boldsymbol{\gamma}\|_1 }{ \|\boldsymbol{\gamma}\|_1 } \frac{\tilde{Q}_{\eta_{t'}(1)}(t')}{\gamma_{\eta_{t'}(1)}} \notag \\
    & = \frac{ 1 }{ \|\boldsymbol{\gamma}\|_1 } \left ( \sum\limits_{l=1}^n\gamma_{\eta_{t'}(l)} \frac{\tilde{Q}_{\eta_{t'}(1)}(t')}{\gamma_{\eta_{t'}(1)}} \right ) \notag \\
    & \geq \frac{ 1 }{ \|\boldsymbol{\gamma}\|_1 } \left ( \sum\limits_{l=1}^n\gamma_{\eta_{t'}(l)} \frac{\tilde{Q}_{\eta_{t'}(l)}(t')}{\gamma_{\eta_{t'}(l)}} \right ) \notag \\
    & \geq \frac{\|\tilde{\mathbf Q}(t')\|_1}{\|\boldsymbol{\gamma}\|_1} \label{eq:scaled_q1_greater_norm} \\
    & \geq \frac{TS_{\max}}{\gamma_{\min}}, \notag
\end{align}
We then define the following partition of rank positions:
\begin{align}
    I_1 &:= \{1\}, \notag \\
    I_2 &:= \left\{
        l\in\{2,\ldots,n\}:
        \frac{\tilde Q_{\eta_{t'}(l)}(t')}{\gamma_{\eta_{t'}(l)}}
        \geq
        \frac{TS_{\max}}{\gamma_{\min}}
    \right\}, \notag \\
    I_3 &:= \left\{
        l\in\{2,\ldots,n\}:
        \frac{\tilde Q_{\eta_{t'}(l)}(t')}{\gamma_{\eta_{t'}(l)}}
        <
        \frac{TS_{\max}}{\gamma_{\min}}
    \right\}. \notag
\end{align}
By construction, for every $l\in I_1\cup I_2$,
\[
    \tilde Q_{\eta_{t'}(l)}(t')
    \geq
    \gamma_{\eta_{t'}(l)}
    \frac{TS_{\max}}{\gamma_{\min}}
    \geq
    TS_{\max}.
\]
Thus, the queues corresponding to rank positions in $I_1\cup I_2$ have at least $TS_{\max}$ jobs. For every $l\in I_3$, we have
\[
    \frac{\tilde Q_{\eta_{t'}(l)}(t')}{\gamma_{\eta_{t'}(l)}} < \frac{TS_{\max}}{\gamma_{\min}},
\]
so the contribution of these $\boldsymbol{\gamma}$-scaled queue lengths can be bounded by a constant.

Using this partition, we bound the drift of $W(\tilde{\mathbf{O}}(t'))$ as follows:
\begin{align}
    & \mathbb{E}[\Delta W(\tilde{\mathbf{O}}(t'))\ | \  \tilde{\mathbf{O}}(t')=\mathbf{o}] \notag \\
    & = \mathbb{E}[ \| \tilde{\mathbf{O}}(t'+1) \|^2_2 - \| \tilde{\mathbf{O}}(t') \|^2_2 \ | \  \tilde{\mathbf{O}}(t')=\mathbf{o}] \notag \\
    & = \mathbb{E}[ \| \mathbf{O}(t+T) \|^2_2 - \| \mathbf{O}(t) \|^2_2 \ | \  \mathbf{O}(t)=\mathbf{o}] \notag \\
    & = \mathbb{E} \left [ \sum\limits_{l=1}^{n} \frac{ \left (Q_{\eta_{t'}(l)}(t) + \sum\limits_{j=0}^{T-1}A_{\eta_{t'}(l)}(t+j) - \sum\limits_{j=0}^{T-1}S_{\eta_{t'}(l)}(t+j) + \sum\limits_{j=0}^{T-1}U_{\eta_{t'}(l)}(t+j) \right )^2 }{\gamma_{\eta_{t'}(l)}} - \sum\limits_{l=1}^{n} \frac{ Q_{\eta_{t'}(l)}(t)^2 }{\gamma_{\eta_{t'}(l)}} \  \middle | \  \mathbf{O}(t)=\mathbf{o} \right ] \notag \\ 
    & = \mathbb{E} \left [ \sum\limits_{l=1}^{n} \left [ \frac{ \left (Q_{\eta_{t'}(l)}(t) + \sum\limits_{j=0}^{T-1}A_{\eta_{t'}(l)}(t+j) - \sum\limits_{j=0}^{T-1}S_{\eta_{t'}(l)}(t+j) + \sum\limits_{j=0}^{T-1}U_{\eta_{t'}(l)}(t+j) \right )^2 - Q_{\eta_{t'}(l)}(t)^2 }{\gamma_{\eta_{t'}(l)}}  \right ] \  \middle | \ \mathbf{O}(t)=\mathbf{o} \right] \notag \\ 
    & \overset{(a)}{=} \mathbb{E} \left[ \frac{ \left (Q_{\eta_{t'}(1)}(t) + \sum\limits_{j=0}^{T-1}A_{\eta_{t'}(1)}(t+j) - \sum\limits_{j=0}^{T-1}S_{\eta_{t'}(1)}(t+j) + \sum\limits_{j=0}^{T-1}U_{\eta_{t'}(1)}(t+j) \right )^2 - Q_{\eta_{t'}(1)}(t)^2 }{\gamma_{\eta_{t'}(1)}}  \  \middle | \ \mathbf{O}(t)=\mathbf{o} \right ] \notag \\
    & \quad + \mathbb{E} \left[ \sum\limits_{ l\in I_2 } \left [ \frac { \left (Q_{\eta_{t'}(l)}(t) + \sum\limits_{j=0}^{T-1}A_{\eta_{t'}(l)}(t+j) - \sum\limits_{j=0}^{T-1}S_{\eta_{t'}(l)}(t+j) + \sum\limits_{j=0}^{T-1}U_{\eta_{t'}(l)}(t+j) \right )^2 - Q_{\eta_{t'}(l)}(t)^2 }{\gamma_{\eta_{t'}(l)}}  \right ] \  \middle | \ \mathbf{O}(t)=\mathbf{o} \right ] \notag \\
    & \quad + \mathbb{E} \left[ \sum\limits_{ l\in I_3 } \left [ \frac{ \left (Q_{\eta_{t'}(l)}(t) + \sum\limits_{j=0}^{T-1}A_{\eta_{t'}(l)}(t+j) - \sum\limits_{j=0}^{T-1}S_{\eta_{t'}(l)}(t+j) + \sum\limits_{j=0}^{T-1}U_{\eta_{t'}(l)}(t+j) \right )^2 - Q_{\eta_{t'}(l)}(t)^2}{\gamma_{\eta_{t'}(l)}} \right ] \  \middle | \ \mathbf{O}(t)=\mathbf{o} \right ] \label{eq:stable_I1I2I3_withrate}
\end{align}
where (a) follows from the definition of the partition. To further bound Equation~\eqref{eq:stable_I1I2I3_withrate}, using the following fact:
\begin{itemize}
    \item[(i)] The queues that have at least $TS_{\max}$ jobs in them at time $t$ will have no unused services for the $T$ time slots, that is, we have $U_{\eta_{t'}(l)}(t+j)=0$ for all $j \in \{ 0,...,T-1 \}$ for all $l\in I_1 \cup I_2$.
    \item[(ii)] The number of unused services is bounded above by the number of potential services, that is, we have $ S_{\eta_{t'}(l)}(t+j) \geq U_{\eta_{t'}(l)}(t+j)$ for all $l\in[n]$ and $j\geq 0$.
    \item[(iii)] For every $l\in I_3$, the queue lengths remain nonnegative, and hence $0\leq Q_{\eta_{t'}(l)}(t+T)\le Q_{\eta_{t'}(l)}(t)+\sum_{j=0}^{T-1}A_{\eta_{t'}(l)}(t+j)$. Therefore, $ Q_{\eta_{t'}(l)}(t+T)^2-Q_{\eta_{t'}(l)}(t)^2 \leq \left(Q_{\eta_{t'}(l)}(t)+\sum_{j=0}^{T-1}A_{\eta_{t'}(l)}(t+j)\right)^2-Q_{\eta_{t'}(l)}(t)^2 $.
\end{itemize}
Then Equation~\eqref{eq:stable_I1I2I3_withrate} can be bounded as follows:
\begin{align}
    & \mathbb{E}[\Delta W(\tilde{\mathbf{O}}(t'))\ | \  \tilde{\mathbf{O}}(t')=\mathbf{o}] \notag \\
    & \leq \mathbb{E} \left[ \frac{ \left (Q_{\eta_{t'}(1)}(t) + \sum\limits_{j=0}^{T-1}A_{\eta_{t'}(1)}(t+j) - \sum\limits_{j=0}^{T-1}S_{\eta_{t'}(1)}(t+j) \right )^2 - Q_{\eta_{t'}(1)}(t)^2 }{\gamma_{\eta_{t'}(1)}} \  \middle | \ \mathbf{O}(t)=\mathbf{o} \right ] \notag \\
    & \quad + \mathbb{E} \left[ \sum\limits_{ l\in I_2 } \left [ \frac{ \left (Q_{\eta_{t'}(l)}(t) + \sum\limits_{j=0}^{T-1}A_{\eta_{t'}(l)}(t+j) - \sum\limits_{j=0}^{T-1}S_{\eta_{t'}(l)}(t+j) \right )^2 - Q_{\eta_{t'}(l)}(t)^2 }{\gamma_{\eta_{t'}(l)}} \right ] \  \middle | \ \mathbf{O}(t)=\mathbf{o} \right ] \notag \\
    & \quad + \mathbb{E} \left[ \sum\limits_{ l\in I_3 } \left [ \frac{ \left (Q_{\eta_{t'}(l)}(t) + \sum\limits_{j=0}^{T-1}A_{\eta_{t'}(l)}(t+j) \right )^2 - Q_{\eta_{t'}(l)}(t)^2 }{\gamma_{\eta_{t'}(l)}} \right ] \  \middle | \ \mathbf{O}(t)=\mathbf{o} \right ] \label{eq:stable_expofwq_withrate}
\end{align}
We now compute the expectation in Equation~\eqref{eq:stable_expofwq_withrate} for each term separately. For the longest queue $ l^{\prime} $ and $l \in I_2$, using the fact that arrivals and potential services are independent from each other and from past queue lengths, we have
\begin{align}
    & \mathbb{E} \left [ \sum\limits_{ l \in \{ 1 \} \cup I_2 } \left [ 
        \frac{ \left (Q_{\eta_{t'}(l)}(t) + \sum\limits_{j=0}^{T-1}A_{\eta_{t'}(l)}(t+j) - \sum\limits_{j=0}^{T-1}S_{\eta_{t'}(l)}(t+j) \right )^2 - (Q_{\eta_{t'}(l)}(t))^2 }{\gamma_{\eta_{t'}(l)}} 
    \right ] \  \middle | \  \mathbf{O}(t)=\mathbf{o} \right ] \notag \\
    & = \mathbb{E} \left [ \sum\limits_{ l \in \{ 1 \} \cup I_2 } \left [ 
        \frac{ 2Q_{\eta_{t'}(l)}(t) \left (\sum\limits_{j=0}^{T-1}A_{\eta_{t'}(l)}(t+j) - \sum\limits_{j=0}^{T-1}S_{\eta_{t'}(l)}(t+j) \right ) }{\gamma_{\eta_{t'}(l)}}
    \right ] \  \middle | \  \mathbf{O}(t)=\mathbf{o} \right ] \notag \\
    & \qquad + \mathbb{E} \left [ \sum\limits_{ l \in \{ 1 \} \cup I_2 } \left [ \frac{ \left (\sum\limits_{j=0}^{T-1}A_{\eta_{t'}(l)}(t+j) - \sum\limits_{j=0}^{T-1}S_{\eta_{t'}(l)}(t+j) \right )^2 }{\gamma_{\eta_{t'}(l)}} 
    \right ] \  \middle | \  \mathbf{O}(t)=\mathbf{o} \right ] \notag \\
    & = \sum\limits_{ l \in \{ 1 \} \cup I_2 } \left [
        2 \frac{q_{\eta_{t'}(l)}}{\gamma_{\eta_{t'}(l)}} T \left ( n \lambda f_{l,\eta_{t'}} - \mu_{\eta_{t'}(l)} \right ) + \frac{T^2 \left ( n \lambda f_{l,\eta_{t'}}-\mu_{\eta_{t'}(l)} \right )^2 + Tf_{l,\eta_{t'}}n\sigma_{\lambda}^2 + T^2 n^2\lambda^2 \tau_{l,\eta_{t'}}^2 + T\sigma_{\eta_{t'}(l)}^2}{\gamma_{\eta_{t'}(l)}}
    \right ]. \label{eq:stable_longest_withrate}
\end{align}
For queues with indices in $ I_3 $, using the fact that their $\boldsymbol{\gamma}$-scaled queue lengths are bounded by $ TS_{\max} / \gamma_{\min} $ by the definition of the partition, we have:
\begin{align}
    & \mathbb{E} \left [ \sum\limits_{ l \in I_3 } \left [ 
        \frac{\left (Q_{\eta_{t'}(l)}(t) + \sum\limits_{j=0}^{T-1}A_{\eta_{t'}(l)}(t+j) \right )^2 - (Q_{\eta_{t'}(l)}(t))^2 }{\gamma_{\eta_{t'}(l)}} 
    \right ] \  \middle | \  \mathbf{O}(t)=\mathbf{o} \right ] \notag \\
    & = \mathbb{E} \left [ \sum\limits_{ l \in I_3 } \left [ 
        \frac{2Q_{\eta_{t'}(l)}(t) \left (\sum\limits_{j=0}^{T-1}A_{\eta_{t'}(l)}(t+j) \right )}{\gamma_{\eta_{t'}(l)}} + \frac{ \left (\sum\limits_{j=0}^{T-1}A_{\eta_{t'}(l)}(t+j) \right )^2 }{\gamma_{\eta_{t'}(l)}} 
    \right ] \  \middle | \  \mathbf{O}(t)=\mathbf{o} \right ] \notag \\
    & = \sum\limits_{ l\in I_3 } \left [
        2\frac{q_{\eta_{t'}(l)}}{\gamma_{\eta_{t'}(l)}} T n\lambda f_{l,\eta_{t'}} + \frac{T^2 n^2\lambda^2 f^2_{l,\eta_{t'}} + Tf_{l,\eta_{t'}}n\sigma_{\lambda}^2 + T^2 n^2\lambda^2 \tau_{l,\eta_{t'}}^2}{\gamma_{\eta_{t'}(l)}}
    \right ] \notag \\
    & \leq \sum\limits_{ l\in I_3 } \left [
        2\frac{TS_{\max}}{\gamma_{\min}} T n\lambda f_{l,\eta_{t'}} + \frac{T^2 n^2\lambda^2 f^2_{l,\eta_{t'}} + Tf_{l,\eta_{t'}}n\sigma_{\lambda}^2 + T^2 n^2\lambda^2 \tau_{l,\eta_{t'}}^2}{\gamma_{\eta_{t'}(l)}}
    \right ] \label{eq:stable_I3_withrate}
\end{align}
Combining equations~\eqref{eq:stable_expofwq_withrate}, \eqref{eq:stable_longest_withrate}, and \eqref{eq:stable_I3_withrate}, we get that
\begin{align}
    & \mathbb{E}[\Delta W(\tilde{\mathbf{O}}(t'))\ | \  \tilde{\mathbf{O}}(t')=\mathbf{o}] \notag \\
    & = 2\frac{q_{\eta_{t'}(1)}}{\gamma_{\eta_{t'}(1)}} T \left ( n \lambda f_{1,\eta_{t'}} - \mu_{\eta_{t'}(1)} \right ) + \frac{ T^2 \left ( n \lambda f_{1,\eta_{t'}}-\mu_{\eta_{t'}(1)} \right )^2 + Tf_{1,\eta_{t'}}n\sigma_{\lambda}^2 + T^2 n^2\lambda^2 \tau_{1,\eta_{t'}}^2 +T\sigma_{\eta_{t'}(1)}^2 }{\gamma_{\eta_{t'}(1)}} \notag \\
    & \quad + \sum\limits_{ l \in I_2 } \left [ 2\frac{q_{\eta_{t'}(l)}}{\gamma_{\eta_{t'}(l)}} T \left ( n \lambda f_{l,\eta_{t'}} - \mu_{\eta_{t'}(l)} \right )+ \frac{ T^2 \left ( n \lambda f_{l,\eta_{t'}}-\mu_{\eta_{t'}(l)} \right )^2 +  Tf_{l,\eta_{t'}}n\sigma_{\lambda}^2 + T^2 n^2\lambda^2 \tau_{l,\eta_{t'}}^2 +T\sigma_{\eta_{t'}(l)}^2 }{\gamma_{\eta_{t'}(l)}} \right ] \notag \\
    & \quad + \sum\limits_{ l\in I_3 } \left [ 2\frac{q_{\eta_{t'}(l)}}{\gamma_{\eta_{t'}(l)}} T n\lambda f_{l,\eta_{t'}} + \frac{ T^2 n^2\lambda^2 f^2_{l,\eta_{t'}} + Tf_{l,\eta_{t'}}n\sigma_{\lambda}^2 + T^2 n^2\lambda^2 \tau_{l,\eta_{t'}}^2 }{\gamma_{\eta_{t'}(l)}} \right ] \notag \\
    & \leq 2 \frac{ q_{\eta_{t'}(1)}}{\gamma_{\eta_{t'}(1)}} T \left ( n \lambda f_{1,\eta_{t'}} - \mu_{\eta_{t'}(1)} \right )+ \frac{ T^2 \left ( n \lambda f_{1,\eta_{t'}}-\mu_{\eta_{t'}(1)} \right )^2 + Tf_{1,\eta_{t'}}n\sigma_{\lambda}^2 + T^2 n^2\lambda^2 \tau_{1,\eta_{t'}}^2 +T\sigma_{\eta_{t'}(1)}^2 }{\gamma_{\eta_{t'}(1)}} \notag \\
    & \quad + \sum\limits_{ l \in I_2 } \left [ 2\frac{q_{\eta_{t'}(l)}}{\gamma_{\eta_{t'}(l)}} T \left ( n \lambda f_{l,\eta_{t'}} - \mu_{\eta_{t'}(l)} \right ) + \frac{ T^2 \left ( n \lambda f_{l,\eta_{t'}} - \mu_{\eta_{t'}(l)} \right )^2 + Tf_{l,\eta_{t'}}n\sigma_{\lambda}^2 + T^2 n^2\lambda^2 \tau_{l,\eta_{t'}}^2 +T\sigma_{\eta_{t'}(l)}^2 }{\gamma_{\eta_{t'}(l)}} \right ] \notag \\
    & \quad + \sum\limits_{ l\in I_3 } \left [ 2\frac{TS_{\max}}{\gamma_{\min}} T n\lambda f_{l,\eta_{t'}} + \frac{ T^2 n^2\lambda^2 f^2_{l,\eta_{t'}} + Tf_{l,\eta_{t'}}n\sigma_{\lambda}^2 + T^2 n^2\lambda^2 \tau_{l,\eta_{t'}}^2 }{\gamma_{\eta_{t'}(l)}} \right ] \notag \\
    & \overset{(a)}{=} 2 \frac{q_{\eta_{t'}(1)}}{\gamma_{\eta_{t'}(1)}} T \left ( n \lambda f_{1,\eta_{t'}} - \mu_{\eta_{t'}(1)} \right ) + \sum\limits_{ l \in I_2 } 2\frac{q_{\eta_{t'}(l)}}{\gamma_{\eta_{t'}(l)}} T \left ( n \lambda f_{l,\eta_{t'}} - \mu_{\eta_{t'}(l)} \right ) \notag \\
    & \quad + \frac{ T^2 \left ( n \lambda f_{1,\eta_{t'}}-\mu_{\eta_{t'}(1)} \right )^2 + Tf_{1,\eta_{t'}}n\sigma_{\lambda}^2 + T^2 n^2\lambda^2 \tau_{1,\eta_{t'}}^2 +T\sigma_{\eta_{t'}(1)}^2}{\gamma_{\eta_{t'}(1)}} \notag \\
    & \quad + \sum\limits_{ l \in I_2 } \left [ \frac{ T^2 \left ( n \lambda f_{l,\eta_{t'}}-\mu_{\eta_{t'}(l)} \right )^2 + Tf_{l,\eta_{t'}}n\sigma_{\lambda}^2 + T^2 n^2\lambda^2 \tau_{l,\eta_{t'}}^2 +T\sigma_{\eta_{t'}(l)}^2 }{\gamma_{\eta_{t'}(l)}} \right ] \notag \\
    & \quad + \sum\limits_{ l\in I_3 } \left [ 2\frac{T^2S_{\max}}{\gamma_{\min}} n\lambda f_{l,\eta_{t'}} + \frac{ T^2 n^2\lambda^2 f^2_{l,\eta_{t'}} + Tf_{l,\eta_{t'}}n\sigma_{\lambda}^2 + T^2 n^2\lambda^2 \tau_{l,\eta_{t'}}^2}{\gamma_{\eta_{t'}(l)}} \right ] \notag \\
    & \overset{(b)}{\leq} 2 T \left ( n \lambda f_{1,\eta_{t'}} - \mu_{\eta_{t'}(1)} \right )\frac{q_{\eta_{t'}(1)}}{\gamma_{\eta_{t'}(1)}} + \sum\limits_{ l \in I_2 } 2 T \left ( n \lambda f_{l,\eta_{t'}} - \mu_{\eta_{t'}(l)} \right )\frac{q_{\eta_{t'}(l)}}{\gamma_{\eta_{t'}(l)}} \notag \\
    & \quad + \frac{ T^2 \left ( n \lambda f_{\max} + \mu_{\max} \right )^2 + Tf_{\max}n\sigma_{\lambda}^2 + T^2 n^2\lambda^2 \tau_{\max}^2 +T\sigma_{\max}^2 }{\gamma_{\min}} \notag \\
    & \quad + \frac{(n-1)}{\gamma_{\min}} \left [ T^2 \left ( n \lambda f_{\max} + \mu_{\max} \right )^2 + Tf_{ \max } n\sigma_{\lambda}^2 + T^2 n^2\lambda^2 \tau_{\max}^2 +T\sigma_{\max}^2 \right ] \notag \\
    & \quad + \frac{(n-1)}{\gamma_{\min}} \left [ 2S_{\max} T^2 n\lambda f_{\max} + T^2 n^2\lambda^2 f^2_{\max} + Tf_{\max}n\sigma_{\lambda}^2 + T^2 n^2\lambda^2 \tau_{\max}^2 \right ] \notag \\
    & = 2 T \left ( n \lambda f_{1,\eta_{t'}} - \mu_{\eta_{t'}(1)} \right )\frac{q_{\eta_{t'}(1)}}{\gamma_{\eta_{t'}(1)}} + \sum\limits_{ l \in I_2 } 2 T \left ( n \lambda f_{l,\eta_{t'}} - \mu_{\eta_{t'}(l)} \right )\frac{q_{\eta_{t'}(l)}}{\gamma_{\eta_{t'}(l)}}  + K \label{eq:stable_aveofwx_withrate}
\end{align}
where (a) follows from rearranging the terms, and (b) follows from the definition of the partition, which yields $ |I_2|\leq n-1$, $|I_3| \leq n-1$ and the definitions of $f_{\max} := \max_{l,\eta} \{ f_{l, \eta} \} $ and $ \tau_{\max} := \max_{l,\eta}\{ \tau_{l,\eta} \} $. 

To further simplify the first term in Equation~\eqref{eq:stable_aveofwx_withrate}, we use the fact that $ |I_2| \in \{ 0, ..., n-1 \} $. Let $ |I_2| = m $, we then have
\begin{align}
    & 2 T \left ( n \lambda f_{1,\eta_{t'}} - \mu_{\eta_{t'}(1)} \right ) \frac{q_{\eta_{t'}(1)}}{\gamma_{\eta_{t'}(1)}} + \sum\limits_{ l \in I_2 } 2 T \left ( n \lambda f_{l,\eta_{t'}} - \mu_{\eta_{t'}(l)} \right ) \frac{q_{\eta_{t'}(l)}}{\gamma_{\eta_{t'}(l)}} \notag \\
    & = -2T \left ( \mu_{\eta_{t'}(1)} - n \lambda f_{1,\eta_{t'}} \right ) \frac{q_{\eta_{t'}(1)}}{\gamma_{\eta_{t'}(1)}} + 2T\sum\limits_{ l =  2}^{m+1} \left ( n \lambda f_{l,\eta_{t'}} - \mu_{\eta_{t'}(l)} \right ) \frac{q_{\eta_{t'}(l)}}{\gamma_{\eta_{t'}(l)}} \label{eq:stable_drift_I4_withrate} 
\end{align}
From our sufficient condition on $n\lambda$ given in Equation~\eqref{eq:feasibleregion_withrate}, we have the following inequality is positive:
\begin{align}
    \sum\limits_{l=1}^{k} \mu_{\eta_{t'}(l)} - n\lambda \sum\limits_{l=1}^{k} f_{l, \eta_{t'}(l)} > 0, \qquad \forall \, k \in [n], \, \eta_{t'} \in S_n. \label{ieq:feasible_region}
\end{align}
Note that the above inequality is also true when $ \sum_{l=1}^k f_{l, \eta} = 0$ for all $k \in [n], \, \eta \in S_n $. Let us define
\begin{align}
    \epsilon^{\prime} & :=  \min\limits_{ \eta \in S_n } \left \{ \min\limits_{1 \leq k \leq n} \left \{ \sum\limits_{l=1}^{k} \mu_{\eta(l)} - n\lambda \sum\limits_{l=1}^{k} f_{l,\eta} \right \} \right \}. \notag
\end{align} 
We claim this $\epsilon^{\prime}$ is positive because of Equation~\eqref{ieq:feasible_region}. We have two different cases.

{\bf Case 1: $| I_2 | = m = 0$}. In this case, Equation~\eqref{eq:stable_drift_I4_withrate} can be bounded as follows
\begin{align}
     -2T \left ( \mu_{\eta_{t'}(1)} - n \lambda f_{1,\eta_{t'}} \right ) \frac{q_{\eta_{t'}(1)}}{\gamma_{\eta_{t'}(1)}} &+ 2T\sum\limits_{ l =  2}^{m+1} \left ( n \lambda f_{l,\eta_{t'}} - \mu_{\eta_{t'}(l)} \right ) \frac{q_{\eta_{t'}(l)}}{\gamma_{\eta_{t'}(l)}} \notag \\
    & = -2T \epsilon^{\prime} \frac{q_{\eta_{t'}(1)}}{\gamma_{\eta_{t'}(1)}} -2T \left ( \mu_{\eta_{t'}(1)} - n \lambda f_{1,\eta_{t'}} - \epsilon^{\prime} \right ) \frac{q_{\eta_{t'}(l)}}{\gamma_{\eta_{t'}(l)}} \notag \\
    & \overset{(a)}{\leq} -2T \epsilon^{\prime} \frac{q_{\eta_{t'}(1)}}{\gamma_{\eta_{t'}(1)}} \notag \\
    & \overset{(b)}{\leq} - \frac{2T \epsilon^{\prime}}{ \|\boldsymbol{\gamma} \|_1 } \| \mathbf{q} \|_1 \notag \\
    & \overset{(c)}{\leq} - \frac{2T \epsilon^{\prime}}{ \|\boldsymbol{\gamma} \|_1 } \| \mathbf{q} \|_2, \label{eq:stable_drift_case1_withrate}
\end{align} 
where (a) follows from thes definition that $ \epsilon^{\prime} \leq \mu_{\eta_{t'}(1)} - n\lambda f_{1,\eta_{t'}} $. (b) follows from Equation~\eqref{eq:scaled_q1_greater_norm} and (c) follows from the inequality $ \| \mathbf{x} \|_1 \geq \| \mathbf{x} \|_2 $ for all $ \mathbf{x} \in \mathbb{R}^n $. 

{\bf Case 2: $| I_2 | =m > 0$}. Equation~\eqref{eq:stable_drift_I4_withrate} can be bounded as follows
\begin{align}
    & -2T \left ( \mu_{\eta_{t'}(1)} - n \lambda f_{1,\eta_{t'}} \right ) \frac{ q_{\eta_{t'}(1)} }{\gamma_{\eta_{t'}(1)}} + 2T \sum\limits_{ l = 2}^{m+1} \left ( n \lambda f_{l,\eta_{t'}} - \mu_{\eta_{t'}(l)} \right ) \frac{q_{\eta_{t'}(l)}}{\gamma_{\eta_{t'}(l)}}  \notag \\
    & = -2T \epsilon^{\prime} \frac{q_{\eta_{t'}(1)}}{\gamma_{\eta_{t'}(1)}} -2T \left ( \mu_{\eta_{t'}(1)} - n \lambda f_{1,\eta_{t'}} -\epsilon^{\prime} \right ) \frac{q_{\eta_{t'}(1)}}{\gamma_{\eta_{t'}(1)}} +2T \left ( n \lambda f_{2, \eta_{t'}} - \mu_{\eta_{t'}(2)} \right ) \frac{q_{\eta_{t'}(2)}}{\gamma_{\eta_{t'}(2)}} \notag \\
    & \quad +  2T \sum\limits_{ l = 3}^{m+1} \left ( n \lambda f_{l,\eta_{t'}} - \mu_{\eta_{t'}(l)} \right ) \frac{q_{\eta_{t'}(l)}}{\gamma_{\eta_{t'}(l)}} \notag \\
    & \overset{(a)}{\leq} -2T \epsilon^{\prime} \frac{q_{\eta_{t'}(1)}}{\gamma_{\eta_{t'}(1)}} -2T \left (\sum\limits_{l=1}^2 \mu_{\eta_{t'}(l)} - n\lambda \sum\limits_{l=1}^2 f_{l,\eta_{t'}} - \epsilon^{\prime} \right )\frac{q_{\eta_{t'}(2)}}{\gamma_{\eta_{t'}(2)}} +  2T \sum\limits_{l = 3}^{m+1} \left ( n \lambda f_{l,\eta_{t'}} - \mu_{\eta_{t'}(l)} \right ) \frac{q_{\eta_{t'}(l)}}{\gamma_{\eta_{t'}(l)}} \notag \\
    & \overset{(b)}{\leq} -2T \epsilon^{\prime} \frac{q_{\eta_{t'}(1)}}{\gamma_{\eta_{t'}(1)}} - 2T \left (\sum\limits_{l=1}^{m+1} \mu_{\eta_{t'}(l)} - n\lambda \sum\limits_{l=1}^{m+1} f_{l,\eta_{t'}} - \epsilon^{\prime} \right )\frac{q_{\eta_{t'}(m+1)}}{\gamma_{\eta_{t'}(m+1)}} \notag \\
    & \leq - \frac{2T \epsilon^{\prime}}{\|\boldsymbol{\gamma} \|_1}\| \mathbf{q} \|_1 \notag \\
    & \leq - \frac{2T \epsilon^{\prime}}{\|\boldsymbol{\gamma} \|_1}\| \mathbf{q} \|_2
    \label{eq:stable_drift_case2_withrate} 
\end{align} 
where (a) and (b) follow from the fact that $ q_{\eta_{t'}(1)} / \gamma_{\eta_{t'}(1)} \geq ... \geq q_{\eta_{t'}(n)} / \gamma_{\eta_{t'}(n)} $ and the definition of $\epsilon^{\prime}$.

Based on equations~\eqref{eq:stable_drift_case1_withrate} and \eqref{eq:stable_drift_case2_withrate}, finally, we have
\begin{align}
    \mathbb{E}[\Delta W(\tilde{\mathbf{O}}(t'))\ | \  \mathbf{O}(t)=\mathbf{o}]  & \leq - \frac{2T \epsilon^{\prime} }{\|\boldsymbol{\gamma} \|_1}\| \mathbf{q} \|_2  + K, \label{eq:stable_wx_withrate}
\end{align}
where
\begin{align*}
    K &= \frac{ T^2 \left ( n \lambda f_{\max} + \mu_{\max} \right )^2 + Tf_{\max}n\sigma_{\lambda}^2 + T^2 n^2\lambda^2 \tau_{\max}^2 +T\sigma_{\max}^2}{\gamma_{\min}} \notag \\
    & \quad + \frac{(n-1)}{\gamma_{\min}} \left [ T^2 \left ( n \lambda f_{\max} + \mu_{\max} \right )^2 + Tf_{\max}n\sigma_{\lambda}^2 + T^2 n^2\lambda^2 \tau_{\max}^2 +T\sigma_{\max}^2 \right ] \notag \\
    & \quad + \frac{(n-1)}{\gamma_{\min}} \left [ 2S_{\max} T^2 n\lambda f_{\max} + T^2 n^2\lambda^2 f^2_{\max} + Tf_{\max}n\sigma_{\lambda}^2 + T^2 n^2\lambda^2 \tau_{\max}^2 \right ] \notag
\end{align*}

Thus, Equation~\eqref{eq:delta_V_bound_by_delta_W} can be combined with Equation~\eqref{eq:stable_wx_withrate} as follows:
\begin{align}
    &\mathbb{E}\!\left[ \Delta V(\tilde{\mathbf O}(t')) \,\middle|\, \tilde{\mathbf O}(t')=\mathbf o \right] \notag \\
    & \leq  \frac{1}{2\|\mathbf o\|_2} \mathbb{E}\!\left[ \Delta W(\tilde{\mathbf O}(t'))  \,\middle|\, \tilde{\mathbf O}(t')=\mathbf o \right] \notag \\
    &\leq  \frac{1}{2\|\mathbf o\|_2} \left( -\frac{2T\epsilon'}{\|\boldsymbol{\gamma}\|_1}\|\mathbf q\|_2 + K \right) \notag \\
    & = - \frac{\epsilon' T}{\|\boldsymbol{\gamma}\|_1} \frac{\|\mathbf q\|_2}{\|\mathbf o\|_2} + \frac{K}{2\|\mathbf o\|_2}. \label{eq:deltaV_intermediate_withrate}
\end{align}
Since $\|\mathbf o\|_2^2 = \sum_{l=1}^n q_l^2/\gamma_l $, we have
\[
    \frac{1}{\sqrt{\gamma_{\max}}}\|\mathbf q\|_2 \leq \|\mathbf o\|_2 \leq \frac{1}{\sqrt{\gamma_{\min}}}\|\mathbf q\|_2.
\]
Therefore,
\[
    \frac{\|\mathbf q\|_2}{\|\mathbf o\|_2} \geq \sqrt{\gamma_{\min}},
    \qquad
    \frac{1}{\|\mathbf o\|_2} \leq \frac{\sqrt{\gamma_{\max}}}{\|\mathbf q\|_2}.
\]
Substituting these bounds into Equation \eqref{eq:deltaV_intermediate_withrate}, we obtain
\begin{align}
    \mathbb{E}\!\left[ \Delta V(\tilde{\mathbf{O}}(t')) \,\middle|\, \tilde{\mathbf{O}}(t')=\mathbf{o} \right]
    & \leq -\frac{ \epsilon' T \sqrt{\gamma_{\min}}}{\|\boldsymbol{\gamma}\|_1} + \frac{K\sqrt{\gamma_{\max}}}{2\|\mathbf q\|_2}. \notag
\end{align}
Thus, the drift of $V(\tilde{\mathbf O}(t'))$ is negative whenever $\|\mathbf q\|_2$ is sufficiently large.

\section{Proof of Corollary \ref{co_our:condition_withrate_withratio}} \label{prf:condition_withrate_withratio} Note that $(\eta,n)$ is a feasible solution for Equation~\eqref{eq:feasibleregion_withrate} for all $\eta \in \mathcal{S}_n$ and that, in order for $(\eta,m)$ to be a minimizer, we must have $ \sum_{l=1}^m f_{l,\eta} >0 $. Therefore, using the condition, we have
\begin{align*}
    \sum\limits_{l=1}^n \mu_l =  \frac{ \sum\limits_{l=1}^n \mu_{\eta(l)} }{ \sum\limits_{l=1}^n f_{l,\eta} } \leq \frac{ \sum\limits_{l=1}^m \mu_{\eta(l)} }{ \sum\limits_{l=1}^m f_{l,\eta} }
\end{align*}
It follows that all optimal solutions $(\eta,m)$ to Equation~\eqref{eq:feasibleregion_withrate} have $m=n$.

\section{Proof of Theorem~\ref{thm_our:fixed_load_stability_certificates}}
\label{prf:fixed_load_stability_certificates}

Fix a load \(n\lambda\). We prove the two statements of the theorem separately. The first part follows from a moving-barrier argument under \(\mathrm{PBD}(n\lambda)\). The second part follows from a fluid-stability argument under \(\mathrm{SCD}(n\lambda)\). \\

\noindent\textbf{Step 1}. We prove the transience statement under \(\mathrm{PBD}(n\lambda)\).
Suppose that
\[
    h^*<n\lambda<\sum_{i=1}^n\mu_i
\]
and \(\mathrm{PBD}(n\lambda)\) holds. By Definition~\ref{def:fixed_load_pbd}, there exists a nonempty proper set \(B\subset[n]\), with \(m:=|B|\), such that
\begin{align*}
    \underline{\delta}(B,n\lambda) > \max\{\overline{\delta}(\bar B,n\lambda),0\},
\end{align*}
where \(\bar B:=[n]\setminus B\). 

Let \( \mathcal P(B) := \left\{ \eta\in\mathcal S_n: \{\eta(1),\ldots,\eta(m)\}=B \right\}\) and choose constants \(0<\beta<\alpha\) such that
\begin{align*}
    \max\{\overline{\delta}(\bar B,n\lambda),0\} < \beta < \alpha < \underline{\delta}(B,n\lambda).
\end{align*}
Since \(\mathcal P(B)\) is finite, the following two constants are strictly positive:
\begin{align}
    \Delta_B^\lambda &:= \min_{\eta\in\mathcal P(B)} \min_{a\in\{1,\ldots,m\}} \left\{ T\left( n\lambda\sum_{l=m-a+1}^{m}f_{l,\eta} - \sum_{l=m-a+1}^{m}\mu_{\eta(l)} \right) - \alpha\sum_{l=m-a+1}^{m}\gamma_{\eta(l)} \right\} >0,
    \label{eq:fixed_load_delta_B}
\end{align}
and
\begin{align}
    \Delta_{\bar B}^\lambda &:= \min_{\eta\in\mathcal P(B)} \min_{b\in\{1,\ldots,n-m\}} \left\{ T\left( \sum_{l=m+1}^{m+b}\mu_{\eta(l)} - n\lambda\sum_{l=m+1}^{m+b}f_{l,\eta} \right) + \beta\sum_{l=m+1}^{m+b}\gamma_{\eta(l)} \right\} >0.
    \label{eq:fixed_load_delta_barB}
\end{align}
Indeed, \(\alpha<\underline{\delta}(B,n\lambda)\) implies Equation~\eqref{eq:fixed_load_delta_B}, while \(\beta>\overline{\delta}(\bar B,n\lambda)\) implies Equation~\eqref{eq:fixed_load_delta_barB}.

\noindent\textbf{Step 1.1}. We construct an alternative system and obtain a positive-probability escape event.

We construct an alternative policy on the same probability space, using the same arrivals, potential services, and randomization variables as the original policy. At each sampling epoch, let \(\eta\) denote the current \(\boldsymbol{\gamma}\)-scaled queue-length ordering. The alternative policy constructs a modified permutation \(\tilde\eta\) as follows: the servers in \(B\) are placed in the first \(m\) positions, preserving their relative order in \(\eta\), and the servers in \(\bar B\) are placed in the remaining \(n-m\) positions, again preserving their relative order in \(\eta\). The original decision function is then applied to this modified ordering \(\tilde\eta\). Thus,
\[
    \{\tilde\eta(1),\ldots,\tilde\eta(m)\}=B
\]
at every sampling epoch.

The following lemma gives a positive-probability divergence event under this alternative policy.
\begin{lemma}
\label{lem:fixed_load_pbd_escape_event}
    Under the alternative policy defined above, there exist an initial state \(\mathbf x\) and an event \(\mathcal E\), with \(\mathbb P_{\mathbf x}(\mathcal E)>0\), such that
    \[
        \min_{i\in B} \frac{Q_i^{\mathrm{alt}}(rT)}{\gamma_i} > \max_{j\in\bar B} \frac{Q_j^{\mathrm{alt}}(rT)}{\gamma_j}, \qquad \forall r\in\mathbb N,
    \]
    and
    \[
        \lim_{r\to\infty}\sum_{i\in B}\tilde Q_i^{\mathrm{alt}}(r)=\infty.
    \]
\end{lemma}
The proof of Lemma~\ref{lem:fixed_load_pbd_escape_event} is given in Appendix~\ref{prf:fixed_load_pbd_escape_event}. By this lemma, under the alternative policy, the set \(B\) remains strictly separated from \(\bar B\) at every sampling epoch and its aggregate queue length diverges with positive probability.

\noindent\textbf{Step 1.2}. We transfer the divergence from the alternative policy to the original policy.

Work on the event \(\mathcal E\) constructed in Lemma~\ref{lem:fixed_load_pbd_escape_event}. Couple the original and alternative systems using the same initial state \(\mathbf x\), the same arrival process, the same potential service processes, and the same randomization variables. Let
\[
    \tilde{\mathbf Q}(r):=\mathbf Q(rT), \qquad \tilde{\mathbf Q}^{\mathrm{alt}}(r):=\mathbf Q^{\mathrm{alt}}(rT)
\]
be the embedded queue-length processes under the original and alternative policies, respectively. We claim that, on \(\mathcal E\),
\[
    \tilde{\mathbf Q}(r)=\tilde{\mathbf Q}^{\mathrm{alt}}(r), \qquad \forall r\in\mathbb N.
\]
We prove this claim by induction over dispatching cycles.

At \(r=0\), the equality holds because the two systems are initialized from the same state \(\mathbf x\). Suppose that, for some \(r\ge0\),
\[
    \tilde{\mathbf Q}(r)=\tilde{\mathbf Q}^{\mathrm{alt}}(r).
\]
Then the two systems have the same \(\boldsymbol{\gamma}\)-scaled queue-length vector at sampling epoch \(r\), and, under the coupled tie-breaking randomization, they generate the same ordering \(\eta_r\). On \(\mathcal E\), Lemma~\ref{lem:fixed_load_pbd_escape_event} implies that, in the alternative system,
\[
    \min_{i\in B} \frac{\tilde Q_i^{\mathrm{alt}}(r)}{\gamma_i} > \max_{j\in\bar B} \frac{\tilde Q_j^{\mathrm{alt}}(r)}{\gamma_j}.
\]
By the induction hypothesis, the same strict separation holds in the original system. Hence the servers in \(B\) already occupy the first \(m\) positions in the original ordering \(\eta_r\). Therefore, the rearrangement used by the alternative policy is inactive during cycle \(r\), so the modified ordering satisfies \(\tilde\eta_r=\eta_r\). Consequently, during cycle \(r\), the original and alternative policies apply the same decision function to the same ordering and use the same randomization variables. Since the two systems also use the same arrivals and potential services, the queue recursion gives
\[
    \tilde{\mathbf Q}(r+1)=\tilde{\mathbf Q}^{\mathrm{alt}}(r+1).
\]
This completes the induction.

By Lemma~\ref{lem:fixed_load_pbd_escape_event}, on \(\mathcal E\),
\[
    \lim_{r\to\infty}\sum_{i\in B}\tilde Q_i^{\mathrm{alt}}(r)=\infty.
\]
Since \(\tilde{\mathbf Q}(r)=\tilde{\mathbf Q}^{\mathrm{alt}}(r)\) for all \(r\) on \(\mathcal E\), it follows that
\[
    \lim_{r\to\infty}\|\tilde{\mathbf Q}(r)\|_1=\infty \qquad \text{on }\mathcal E.
\]
Equivalently,
\[
    \lim_{r\to\infty}\|\mathbf Q(rT)\|_1=\infty \qquad \text{on }\mathcal E.
\]
Moreover, queue lengths can decrease by at most a bounded amount within one dispatching cycle. Hence, on \(\mathcal E\), the original process is eventually outside every finite set and cannot return to the initial state \(\mathbf x\).
Therefore,
\[
    \mathbb P_{\mathbf x} \left( \text{the chain never returns to }\mathbf x
    \right)
    \ge
    \mathbb P_{\mathbf x}(\mathcal E)>0.
\]
Thus the communicating class containing \(\mathbf x\) is not recurrent. Since recurrence and transience are class properties for countable-state Markov chains, the communicating class containing \(\mathbf x\) is transient. Therefore, the induced Markov chain is transient from the initial state \(\mathbf x\). \\

\noindent\textbf{Step 2}. We prove the positive-recurrence statement under
\(\mathrm{SCD}(n\lambda)\).

Suppose that
\[
    n\lambda<\sum_{i=1}^n\mu_i,
\]
and that \(\mathrm{SCD}(n\lambda)\) holds. The next lemma shows that fluid limits of the stochastic embedded queue-length process are governed by the admissible fluid drifts used in Definition~\ref{def:fixed_load_scd}.

\noindent\textbf{Step 2.1}. We construct fluid limits of the embedded queue-length process.

Consider a sequence of initial states \(\{\tilde{\mathbf Q}_\nu(0)\}_{\nu\ge1}\) such that
\[
    c_\nu:=\|\tilde{\mathbf Q}_\nu(0)\|_1\to\infty.
\]
Let \(\tilde{\mathbf Q}_\nu(\cdot)\) denote the embedded process initialized from \(\tilde{\mathbf Q}_\nu(0)\). Define the fluid-scaled embedded process by linearly interpolating the values
\[
    \widehat{\mathbf Q}_\nu\left(\frac{r}{c_\nu}\right) := \frac{\tilde{\mathbf Q}_\nu(r)}{c_\nu}, \qquad r\in\mathbb N.
\]
Fix any \(H<\infty\). Since the one-cycle queue increments are uniformly bounded, there exists a finite constant \(C_\Delta\) such that the piecewise-linear interpolation satisfies
\[
    \left\|\widehat{\mathbf Q}_\nu(t)-\widehat{\mathbf Q}_\nu(s)\right\|_1 \le C_\Delta |t-s|,
    \qquad 0\le s\le t\le H,\ \nu\ge1.
\]
By the definition \(c_\nu=\|\tilde{\mathbf Q}_\nu(0)\|_1\), we have \(\|\widehat{\mathbf Q}_\nu(0)\|_1=1\). Hence, for every \(t\in[0,H]\),
\[
    \|\widehat{\mathbf Q}_\nu(t)\|_1 \le \|\widehat{\mathbf Q}_\nu(0)\|_1 + \left\|\widehat{\mathbf Q}_\nu(t)-\widehat{\mathbf Q}_\nu(0)\right\|_1 \le C_0+C_\Delta H.
\]
Thus, on \([0,H]\), the family \(\{\widehat{\mathbf Q}_\nu(\cdot)\}_{\nu\ge1}\) is uniformly bounded and equicontinuous. By the Arzel\`a--Ascoli theorem, for every fixed \(H<\infty\), there exists a subsequence that converges uniformly on \([0,H]\). Applying a diagonal subsequence argument over \(H=1,2,\ldots\), there exists a subsequence, still denoted by \(\widehat{\mathbf Q}_\nu(\cdot)\), that converges uniformly on compact sets to a function \(\mathbf q^{\mathrm{fl}}(\cdot)\).

Moreover, the limit is Lipschitz continuous. Indeed, for any \(0\le s\le t\),
\[
    \|\mathbf q^{\mathrm{fl}}(t)-\mathbf q^{\mathrm{fl}}(s)\|_1 = \lim_{\nu\to\infty} \left\|\widehat{\mathbf Q}_\nu(t)-\widehat{\mathbf Q}_\nu(s)\right\|_1 \le C_\Delta |t-s|.
\]
Therefore, \(\mathbf q^{\mathrm{fl}}(\cdot)\) is differentiable almost everywhere. This is the standard compactness step in the fluid limit construction, as in \cite{dai1995positive-fluidmodel} (Theorem~4.1).

\noindent\textbf{Step 2.2}. We identify the admissible fluid drifts and apply the fluid-stability criterion.

The next lemma shows that every fluid limit constructed in Step 2.1 is governed by the admissible reflected fluid drift set used in Definition~\ref{def:fixed_load_scd}.
\begin{lemma}
    \label{lem:fixed_load_fluid_limit_admissible}
    Fix a load \(n\lambda<\sum_{i=1}^n\mu_i\). Let \(\mathbf q^{\mathrm{fl}}(\cdot)\) be any fluid limit of the embedded queue-length process \(\{\tilde{\mathbf Q}(r)\}_{r\ge0}\). Let
    \[
        \tau_{\mathrm{hit}} := \inf\left\{t\ge0:\mathbf q^{\mathrm{fl}}(t)=\mathbf 0\right\}.
    \]
    Then, for almost every \(t<\tau_{\mathrm{hit}}\),
    \[
        \frac{d}{dt}\mathbf q^{\mathrm{fl},(\gamma)}(t) \in \mathfrak D_{n\lambda}(\mathbf q^{\mathrm{fl}}(t)).
    \]
\end{lemma}
The proof of Lemma~\ref{lem:fixed_load_fluid_limit_admissible} is given in Appendix~\ref{prf:fixed_load_fluid_limit_admissible}. Once the fluid limit is identified as an admissible path, \(\mathrm{SCD}(n\lambda)\) implies finite-time draining for all unit-norm fluid limits.

\begin{lemma}
    \label{lem:fixed_load_scd_fluid_draining}
    Fix a load \(n\lambda<\sum_{i=1}^n\mu_i\). If \(\mathrm{SCD}(n\lambda)\) holds, then there exists \(\tau_{\mathrm{scd}}(n\lambda)<\infty\) such that every fluid limit \(\mathbf q^{\mathrm{fl}}(\cdot)\) of \(\{\tilde{\mathbf Q}(r)\}_{r\ge0}\), with
    \[
        \|\mathbf q^{\mathrm{fl}}(0)\|_1=1,
    \]
    satisfies
    \[
        \mathbf q^{\mathrm{fl}}(t)=\mathbf 0, \qquad \forall t\ge\tau_{\mathrm{scd}}(n\lambda).
    \]
\end{lemma}
The proof of Lemma~\ref{lem:fixed_load_scd_fluid_draining} is given in Appendix~\ref{prf:fixed_load_scd_fluid_draining}. We now recall the standard fluid-stability criterion.

\begin{lemma}
    \label{lem:fixed_load_fluid_stability_criterion}
    Fix a policy \(\pi\in\Pi\) and a load \(n\lambda<\sum_{i=1}^n\mu_i\). Recall that \(\{\tilde{\mathbf X}(r)\}_{r\ge0}\) is the embedded Markov chain observed at sampling epochs, and that \(\{\tilde{\mathbf Q}(r)\}_{r\ge0}\) is its queue-length component. Suppose that there exists \(\tau<\infty\) such that every fluid limit \(\mathbf q^{\mathrm{fl}}(\cdot)\) of \(\{\tilde{\mathbf Q}(r)\}_{r\ge0}\), with
    \[
        \|\mathbf q^{\mathrm{fl}}(0)\|_1=1,
    \]
    satisfies
    \[
        \mathbf q^{\mathrm{fl}}(t)=\mathbf 0, \qquad \forall t\ge\tau.
    \]
    Then the embedded Markov chain \(\{\tilde{\mathbf X}(r)\}_{r\ge0}\) is positive recurrent. Consequently, the original Markov chain \(\{\mathbf X(t)\}_{t\ge0}\) is positive recurrent.
\end{lemma}
The proof of Lemma~\ref{lem:fixed_load_fluid_stability_criterion} is given in Appendix~\ref{prf:fixed_load_fluid_stability_criterion}. By Lemma~\ref{lem:fixed_load_scd_fluid_draining}, every unit-norm fluid limit reaches the origin in a uniformly bounded time. Applying Lemma~\ref{lem:fixed_load_fluid_stability_criterion} with \(\tau=\tau_{\mathrm{scd}}(n\lambda)\), we obtain positive recurrence of the embedded chain and hence positive recurrence of the original Markov chain \(\{\mathbf X(t)\}_{t\ge0}\).

\subsection{Proof of Lemma~\ref{lem:fixed_load_pbd_escape_event}}
\label{prf:fixed_load_pbd_escape_event}
Let \(\tilde{\mathbf Q}^{\mathrm{alt}}(r)\) denote the embedded queue-length process under the alternative policy at the sampling epoch \(rT\). Thus, \(\tilde Q_i^{\mathrm{alt}}(r)/\gamma_i\) is the corresponding \(\boldsymbol{\gamma}\)-scaled queue length of server \(i\). Let \(\mathcal F_r\) be the history up to time \(rT\), and let \(\tilde\eta_r\) be the modified permutation used by the alternative policy during cycle \(r\). By construction,
\[
    \{\tilde\eta_r(1),\ldots,\tilde\eta_r(m)\}=B, \qquad r\in\mathbb N.
\]
Since \(\tilde\eta_r\in\mathcal P(B)\), the margins in equations \eqref{eq:fixed_load_delta_B} and \eqref{eq:fixed_load_delta_barB} imply that, for every \(r\), every \(a\in\{1,\ldots,m\}\), and every \(b\in\{1,\ldots,n-m\}\),
\begin{align}
    & T\left( n\lambda\sum_{l=m-a+1}^{m}f_{l,\tilde\eta_r} - \sum_{l=m-a+1}^{m}\mu_{\tilde\eta_r(l)} \right) - \alpha\sum_{l=m-a+1}^{m}\gamma_{\tilde\eta_r(l)} \geq \Delta_B^\lambda,
    \label{eq:pbd_B_margin_cycle}
\end{align}
and
\begin{align}
    & T\left( \sum_{l=m+1}^{m+b}\mu_{\tilde\eta_r(l)} - n\lambda\sum_{l=m+1}^{m+b}f_{l,\tilde\eta_r} \right) + \beta\sum_{l=m+1}^{m+b}\gamma_{\tilde\eta_r(l)} \geq \Delta_{\bar B}^\lambda.
    \label{eq:pbd_barB_margin_cycle}
\end{align}

\noindent\textbf{Step 1.} We construct the barriers and define the escape event.

Let
\[
    \Gamma_B:=\sum_{i\in B}\gamma_i, \qquad \Gamma_{\bar B}:=\sum_{j\in\bar B}\gamma_j.
\]
For constants \(L,M>0\), to be chosen later, define
\[
    z_B(r):=L+M+\alpha r, \qquad z_{\bar B}(r):=L+\frac M2+\beta r.
\]
Since \(0<\beta<\alpha\), we have
\[
    z_B(r)-z_{\bar B}(r) = \frac M2+(\alpha-\beta)r >0, \qquad r\in\mathbb N.
\]
We start from the initial state
\begin{align*}
    \tilde Q_i^{\mathrm{alt}}(0) = \left\lceil \gamma_i(L+2M)\right\rceil,\quad i\in B,  \qquad
    \tilde Q_j^{\mathrm{alt}}(0) = \left\lfloor \gamma_jL\right\rfloor,\quad j\in \bar B.
\end{align*}
Define
\[
    \kappa_B := \inf\left\{ r\ge0: \exists i\in B \text{ such that } \frac{\tilde Q_i^{\mathrm{alt}}(r)}{\gamma_i} \le z_B(r) \right\},
\]
and
\[
    \kappa_{\bar B} := \inf\left\{ r\ge0: \exists j\in\bar B \text{ such that } \frac{\tilde Q_j^{\mathrm{alt}}(r)}{\gamma_j} \ge z_{\bar B}(r) \right\}.
\]
We will prove that
\[
    \mathcal E:=\{\kappa_B=\infty,\ \kappa_{\bar B}=\infty\}
\]
has positive probability. On this event, all queues in \(B\) stay above \(z_B(r)\), while all queues in \(\bar B\) stay below \(z_{\bar B}(r)\). Since \(z_B(r)>z_{\bar B}(r)\) for every \(r\), the two groups remain strictly separated at every sampling epoch. Moreover, since \(z_B(r)\to\infty\), the aggregate queue length in \(B\) diverges on \(\mathcal E\). \\

\noindent\textbf{Step 2.} We bound the probability that \(B\) falls below \(z_B\).

For \(\theta>0\), define
\[
    \Phi_B(r) := \sum_{i\in B} \gamma_i e^{-\theta\left( \frac{\tilde Q_i^{\mathrm{alt}}(r)}{\gamma_i} - z_B(r) \right)}.
\]
The function \(\Phi_B(r)\) is a sum over the fixed set \(B\). Since \(\{\tilde\eta_r(1),\ldots,\tilde\eta_r(m)\}=B\), for \(s=r,r+1\), we can rewrite
\[
    \Phi_B(s) = \sum_{l=1}^{m} \gamma_{\tilde\eta_r(l)} e^{-\theta\left( \frac{\tilde Q_{\tilde\eta_r(l)}^{\mathrm{alt}}(s)}{\gamma_{\tilde\eta_r(l)}} - z_B(s) \right)}.
\]
This is only a reindexing of the sum over \(B\) so that it does not require \(\tilde\eta_r\) to be the queue-length ordering at epoch \(r+1\).

Let
\[
    \Delta\tilde Q_i^{\mathrm{alt}}(r) := \tilde Q_i^{\mathrm{alt}}(r+1)-\tilde Q_i^{\mathrm{alt}}(r).
\]
For \(l\in\{1,\ldots,m\}\), the queue recursion over one cycle gives
\begin{align*}
    \mathbb E\!\left[ \Delta\tilde Q_{\tilde\eta_r(l)}^{\mathrm{alt}}(r) \,\middle|\,\mathcal F_r \right]
    & = \mathbb E\!\left[ \sum_{s=0}^{T-1}A_{\tilde\eta_r(l)}(rT+s) - \sum_{s=0}^{T-1}S_{\tilde\eta_r(l)}(rT+s) + \sum_{s=0}^{T-1}U_{\tilde\eta_r(l)}^{\mathrm{alt}}(rT+s) \,\middle|\,\mathcal F_r \right] \\
    & \overset{(a)}{\geq} \mathbb E\!\left[ \sum_{s=0}^{T-1}A_{\tilde\eta_r(l)}(rT+s) - \sum_{s=0}^{T-1}S_{\tilde\eta_r(l)}(rT+s) \,\middle|\,\mathcal F_r \right] \\
    & = T\left( n\lambda f_{l,\tilde\eta_r} - \mu_{\tilde\eta_r(l)} \right),
\end{align*}
where (a) follows from the fact that the unused service is nonnegative.

Therefore, by Equation~\eqref{eq:pbd_B_margin_cycle}, for every \(a\in\{1,\ldots,m\}\),
\begin{align}
    & \sum_{l=m-a+1}^{m} \mathbb E\!\left[ \Delta\tilde Q_{\tilde\eta_r(l)}^{\mathrm{alt}}(r) - \alpha\gamma_{\tilde\eta_r(l)} \,\middle|\,\mathcal F_r \right]
    \notag\\
    & = \sum_{l=m-a+1}^{m} \mathbb E\!\left[ \Delta\tilde Q_{\tilde\eta_r(l)}^{\mathrm{alt}}(r) \,\middle|\,\mathcal F_r \right] - \alpha\sum_{l=m-a+1}^{m}\gamma_{\tilde\eta_r(l)} \notag\\
    & \ge T\left( n\lambda\sum_{l=m-a+1}^{m}f_{l,\tilde\eta_r} - \sum_{l=m-a+1}^{m}\mu_{\tilde\eta_r(l)} \right) - \alpha\sum_{l=m-a+1}^{m}\gamma_{\tilde\eta_r(l)} \notag\\
    & \ge \Delta_B^\lambda. 
    \label{eq:pbd_B_suffix_drift_escape}
\end{align}

Next, for \(l\in\{1,\ldots,m\}\), define
\[
    w_l := e^{-\theta\left( \frac{\tilde Q_{\tilde\eta_r(l)}^{\mathrm{alt}}(r)}{\gamma_{\tilde\eta_r(l)}} - z_B(r) \right)},
\]
and
\[
    d_l^B := \mathbb E\!\left[ \Delta\tilde Q_{\tilde\eta_r(l)}^{\mathrm{alt}}(r) - \alpha\gamma_{\tilde\eta_r(l)} \,\middle|\,\mathcal F_r \right].
\]
Since \(\tilde\eta_r(1),\ldots,\tilde\eta_r(m)\) are ordered from largest to smallest scaled queue within \(B\), we have \( w_1\le w_2\le\cdots\le w_m \). Using summation by parts,
\begin{align}
    \sum_{l=1}^{m}w_l d_l^B
    & = w_1\sum_{l=1}^{m}d_l^B + (w_2-w_1)\sum_{l=2}^{m}d_l^B + \cdots + (w_m-w_{m-1})d_m^B \notag\\
    & = w_1\sum_{l=1}^{m}d_l^B + \sum_{l=2}^{m}(w_l-w_{l-1})\sum_{q=l}^{m}d_q^B \notag\\
    & = w_1\sum_{l=1}^{m}d_l^B + \sum_{a=1}^{m-1} \left[ \left(w_{m-a+1}-w_{m-a}\right) \sum_{l=m-a+1}^{m}d_l^B \right] \notag\\
    & \overset{(a)}{\ge} \Delta_B^\lambda \left[ w_1 + \sum_{a=1}^{m-1} \left(w_{m-a+1}-w_{m-a}\right) \right] \notag\\
    & = \Delta_B^\lambda w_m \notag\\
    & \overset{(b)}{\ge} \frac{\Delta_B^\lambda}{\Gamma_B}\Phi_B(r).
    \label{eq:pbd_B_weighted_drift_escape}
\end{align}
where (a) follows from Equation~\eqref{eq:pbd_B_suffix_drift_escape} and
\(w_1\le\cdots\le w_m\), and (b) follows from \( \Phi_B(r) = \sum_{l=1}^{m}\gamma_{\tilde\eta_r(l)}w_l \le \sum_{l=1}^{m}\gamma_{\tilde\eta_r(l)}w_m = \Gamma_Bw_m \).

We now derive the drift of \(\Phi_B(r)\). Since \(z_B(r+1)=z_B(r)+\alpha\),
\begin{align}
    \Phi_B(r+1)
    & = \sum_{i\in B}\gamma_i e^{-\theta\left( \frac{\tilde Q_i^{\mathrm{alt}}(r+1)}{\gamma_i} - z_B(r+1) \right)} \notag\\
    & = \sum_{l=1}^{m} \gamma_{\tilde\eta_r(l)} e^{-\theta\left( \frac{\tilde Q_{\tilde\eta_r(l)}^{\mathrm{alt}}(r+1)}{\gamma_{\tilde\eta_r(l)}} - z_B(r+1) \right)} \notag\\
    & = \sum_{l=1}^{m} \gamma_{\tilde\eta_r(l)} e^{-\theta\left( \frac{\tilde Q_{\tilde\eta_r(l)}^{\mathrm{alt}}(r)}{\gamma_{\tilde\eta_r(l)}} - z_B(r) \right)} e^{-\theta\left( \frac{ \tilde Q_{\tilde\eta_r(l)}^{\mathrm{alt}}(r+1) - \tilde Q_{\tilde\eta_r(l)}^{\mathrm{alt}}(r)}{\gamma_{\tilde\eta_r(l)}} - \alpha \right)} \notag\\
    & = \sum_{l=1}^{m} \gamma_{\tilde\eta_r(l)}w_l e^{-\theta\left( \frac{\Delta\tilde Q_{\tilde\eta_r(l)}^{\mathrm{alt}}(r)}{\gamma_{\tilde\eta_r(l)}} - \alpha \right)}.
    \label{eq:pbd_Phi_B_reindex}
\end{align}

By the boundedness of arrivals and potential services, for every \(i\in[n]\), \(\left|\Delta\tilde Q_i^{\mathrm{alt}}(r)\right| \le T(nA_{\max}+S_{\max})\). Therefore, with \(K:=T(nA_{\max}+S_{\max})/\gamma_{\min}+\max\{\alpha,\beta\}\), we have
\[
    \left| \frac{\Delta\tilde Q_i^{\mathrm{alt}}(r)}{\gamma_i}-\alpha \right| \leq K,
    \qquad
    \left| \frac{\Delta\tilde Q_i^{\mathrm{alt}}(r)}{\gamma_i}-\beta \right| \leq K.
\]
For \(|y|\le K\) and \(0<\theta\le1\), Taylor's formula gives
\[
    e^{-\theta y}\le 1-\theta y+C_0\theta^2,
    \qquad
    e^{\theta y}\le 1+\theta y+C_0\theta^2,
\]
where \(C_0:=K^2e^K/2\).

Taking conditional expectation in Equation~\eqref{eq:pbd_Phi_B_reindex}, we obtain
\begin{align}
    \mathbb E[\Phi_B(r+1)\mid\mathcal F_r]
    & = \mathbb E\left[ \sum_{l=1}^{m} \gamma_{\tilde\eta_r(l)}w_l e^{-\theta\left( \frac{\Delta\tilde Q_{\tilde\eta_r(l)}^{\mathrm{alt}}(r)}{\gamma_{\tilde\eta_r(l)}} - \alpha \right)} \,\middle|\,\mathcal F_r \right] \notag\\
    & = \sum_{l=1}^{m} \gamma_{\tilde\eta_r(l)}w_l \mathbb E\left[ e^{-\theta\left( \frac{\Delta\tilde Q_{\tilde\eta_r(l)}^{\mathrm{alt}}(r)}{\gamma_{\tilde\eta_r(l)}} - \alpha \right)} \,\middle|\,\mathcal F_r \right] \notag \\
    & \overset{(a)}{\le} \sum_{l=1}^{m} \gamma_{\tilde\eta_r(l)}w_l \mathbb E\left[ 1 - \theta\left( \frac{\Delta\tilde Q_{\tilde\eta_r(l)}^{\mathrm{alt}}(r)}{\gamma_{\tilde\eta_r(l)}} - \alpha \right) + C_0\theta^2 \,\middle|\,\mathcal F_r \right] \notag\\
    & = \Phi_B(r) - \theta \sum_{l=1}^{m} w_l \mathbb E\left[ \Delta\tilde Q_{\tilde\eta_r(l)}^{\mathrm{alt}}(r) - \alpha\gamma_{\tilde\eta_r(l)} \,\middle|\,\mathcal F_r \right] + C_0\theta^2\Phi_B(r) \notag\\
    & = \Phi_B(r)-\theta\sum_{l=1}^{m}w_l d_l^B+C_0\theta^2\Phi_B(r) \notag\\
    & \overset{(b)}{\le} \left( 1-\theta\frac{\Delta_B^\lambda}{\Gamma_B} + C_0\theta^2 \right)\Phi_B(r). \notag
\end{align}
where (a) follows from the Taylor bound, and (b) follows from Equation~\eqref{eq:pbd_B_weighted_drift_escape}.

Choose
\[
    0<\theta\le\theta_0, \qquad \theta_0:= \min\left\{ 1,\, \frac{\Delta_B^\lambda}{C_0\Gamma_B},\, \frac{\Delta_{\bar B}^\lambda}{C_0\Gamma_{\bar B}} \right\}.
\]
Then \(C_0\theta^2\le \theta\Delta_B^\lambda/\Gamma_B\), and therefore
\[
    \mathbb E[\Phi_B(r+1)\mid\mathcal F_r]\le \Phi_B(r).
\]
Hence \(\{\Phi_B(r)\}_{r\ge0}\) is a nonnegative supermartingale.

We now use optional stopping to convert this supermartingale bound into a crossing-probability bound. Since \(N\wedge\kappa_B\le N\) is a bounded stopping time, the optional stopping theorem gives
\[
    \mathbb E_{\mathbf x} \left[ \Phi_B(N\wedge\kappa_B) \right] \le \Phi_B(0).
\]
On \(\{\kappa_B\le N\}\), by the definition of \(\kappa_B\), there exists \(i\in B\) such that \( \tilde Q_i^{\mathrm{alt}}(\kappa_B) / \gamma_i \le z_B(\kappa_B) \). Therefore,
\begin{align*}
    \Phi_B(\kappa_B)
    & = \sum_{i'\in B} \gamma_{i'} e^{-\theta\left( \frac{\tilde Q_{i'}^{\mathrm{alt}}(\kappa_B)}{\gamma_{i'}} - z_B(\kappa_B) \right)} \\
    & \geq \gamma_i e^{-\theta\left( \frac{\tilde Q_i^{\mathrm{alt}}(\kappa_B)}{\gamma_i} - z_B(\kappa_B) \right)}
    \geq \gamma_i
    \geq \gamma_{\min}.
\end{align*}
Thus,
\begin{align}
    \gamma_{\min}\mathbb P_{\mathbf x}(\kappa_B\le N)
    & = \mathbb E_{\mathbf x} \left[ \gamma_{\min}\mathbf 1_{\{\kappa_B\le N\}} \right] \notag\\
    & \leq \mathbb E_{\mathbf x} \left[ \Phi_B(\kappa_B)\mathbf 1_{\{\kappa_B\le N\}} \right] \notag\\
    & = \mathbb E_{\mathbf x} \left[ \Phi_B(N\wedge\kappa_B)\mathbf 1_{\{\kappa_B\le N\}} \right] \notag\\
    & \leq \mathbb E_{\mathbf x} \left[ \Phi_B(N\wedge\kappa_B) \right] \notag\\
    & \le \Phi_B(0). \notag 
\end{align}
Letting \(N\to\infty\), and using the fact that \(\{\kappa_B\le N\}\) increases to \(\{\kappa_B<\infty\}\), we obtain
\[
    \mathbb P_{\mathbf x}(\kappa_B<\infty) \le \frac{\Phi_B(0)}{\gamma_{\min}}.
\]
By the initial condition, for every \(i\in B\), we have \( \tilde Q_i^{\mathrm{alt}}(0)/ \gamma_i \ge L+2M \) and \( z_B(0)=L+M\). Hence
\[
    \frac{\tilde Q_i^{\mathrm{alt}}(0)}{\gamma_i}-z_B(0)\ge M, \qquad i\in B.
\]
Therefore,
\begin{align*}
    \Phi_B(0) & = \sum_{i\in B} \gamma_i e^{-\theta\left( \frac{\tilde Q_i^{\mathrm{alt}}(0)}{\gamma_i} - z_B(0) \right)} \\
    & \leq \sum_{i\in B}\gamma_i e^{-\theta M} \\
    & = \Gamma_Be^{-\theta M}.
\end{align*}
Thus,
\begin{align}
    \mathbb P_{\mathbf x}(\kappa_B<\infty) \le \frac{\Gamma_B}{\gamma_{\min}}e^{-\theta M}.
    \label{eq:pbd_kappa_B_bound_escape}
\end{align}

\noindent\textbf{Step 3.} We bound the probability that \(\bar B\) rises above \(z_{\bar B}\).

For \(\theta>0\), define
\[
    \Phi_{\bar B}(r) := \sum_{j\in\bar B} \gamma_j e^{\theta\left( \frac{\tilde Q_j^{\mathrm{alt}}(r)}{\gamma_j} - z_{\bar B}(r) \right)}.
\]
Again, this is a sum over the fixed set \(\bar B\). Since \(\{\tilde\eta_r(m+1),\ldots,\tilde\eta_r(n)\}=\bar B\), for \(s=r,r+1\), we can reindex
\[
    \Phi_{\bar B}(s) = \sum_{l=m+1}^{n} \gamma_{\tilde\eta_r(l)} e^{\theta\left( \frac{\tilde Q_{\tilde\eta_r(l)}^{\mathrm{alt}}(s)}{\gamma_{\tilde\eta_r(l)}} - z_{\bar B}(s) \right)}.
\]
For \(l=m+1,\ldots,n\), define
\[
    d_l^{\bar B} := T\left( n\lambda f_{l,\tilde\eta_r} - \mu_{\tilde\eta_r(l)} \right) - \beta\gamma_{\tilde\eta_r(l)}.
\]
By Equation~\eqref{eq:pbd_barB_margin_cycle}, for every \(b\in\{1,\ldots,n-m\}\),
\begin{align}
    \sum_{l=m+1}^{m+b}d_l^{\bar B}
    & = T\left( n\lambda\sum_{l=m+1}^{m+b}f_{l,\tilde\eta_r} - \sum_{l=m+1}^{m+b}\mu_{\tilde\eta_r(l)} \right) - \beta\sum_{l=m+1}^{m+b}\gamma_{\tilde\eta_r(l)} \notag\\
    & = - \left[ T\left( \sum_{l=m+1}^{m+b}\mu_{\tilde\eta_r(l)} - n\lambda\sum_{l=m+1}^{m+b}f_{l,\tilde\eta_r} \right) + \beta\sum_{l=m+1}^{m+b}\gamma_{\tilde\eta_r(l)} \right] \notag \\
    & \leq - \Delta_{\bar B}^\lambda.
    \label{eq:pbd_barB_prefix_drift_escape}
\end{align}
For \(l=m+1,\ldots,n\), define
\[
    v_l := e^{\theta\left( \frac{\tilde Q_{\tilde\eta_r(l)}^{\mathrm{alt}}(r)}{\gamma_{\tilde\eta_r(l)}} - z_{\bar B}(r) \right)}.
\]
Since \(\tilde\eta_r(m+1),\ldots,\tilde\eta_r(n)\) are ordered from largest to smallest scaled queue within \(\bar B\), we have \(v_{m+1}\ge v_{m+2}\ge\cdots\ge v_n\). Using summation by parts,
\begin{align}
    \sum_{l=m+1}^{n}v_l d_l^{\bar B}
    & = v_n\sum_{l=m+1}^{n}d_l^{\bar B} + (v_{n-1}-v_n)\sum_{l=m+1}^{n-1}d_l^{\bar B} + \cdots + (v_{m+1}-v_{m+2})d_{m+1}^{\bar B} \notag\\
    & = v_n\sum_{l=m+1}^{n}d_l^{\bar B} + \sum_{l=m+1}^{n-1} (v_l-v_{l+1}) \sum_{q=m+1}^{l}d_q^{\bar B} \notag\\
    & = v_n\sum_{l=m+1}^{n}d_l^{\bar B} + \sum_{b=1}^{n-m-1} \left[ \left(v_{m+b}-v_{m+b+1}\right) \sum_{l=m+1}^{m+b}d_l^{\bar B} \right] \notag \\
    & \overset{(a)}{\le} - \Delta_{\bar B}^\lambda \left[ v_n + \sum_{b=1}^{n-m-1} \left(v_{m+b}-v_{m+b+1}\right) \right] \notag \\
    & = -\Delta_{\bar B}^\lambda v_{m+1} \notag\\
    &\overset{(b)}{\le} -\frac{\Delta_{\bar B}^\lambda}{\Gamma_{\bar B}} \Phi_{\bar B}(r).
    \label{eq:pbd_barB_weighted_nominal_drift_escape}
\end{align}
where (a) follows from Equation~\eqref{eq:pbd_barB_prefix_drift_escape} and \(v_{m+1}\ge\cdots\ge v_n\), and (b) follows from \( \Phi_{\bar B}(r) = \sum_{l=m+1}^{n}\gamma_{\tilde\eta_r(l)}v_l \le \sum_{l=m+1}^{n}\gamma_{\tilde\eta_r(l)}v_{m+1} = \Gamma_{\bar B}v_{m+1} \).

The only additional term that needs to be controlled for \(\bar B\) is unused service. For each cycle \(r\), define the cumulative unused service of queue \(i\) under the alternative policy by
\[
    \left(\boldsymbol{\Sigma}\mathbf U^{\mathrm{alt}}((r+1)T-1)\right)_i
    :=
    \sum_{s=0}^{T-1}U_i^{\mathrm{alt}}(rT+s).
\]
If \(\left(\boldsymbol{\Sigma}\mathbf U^{\mathrm{alt}}((r+1)T-1)\right)_i>0\), then \(\tilde Q_i^{\mathrm{alt}}(r)<TS_{\max}\). Indeed, suppose instead that \(\tilde Q_i^{\mathrm{alt}}(r)\ge TS_{\max}\). Even if no arrivals occur during cycle \(r\), the initial backlog of queue \(i\) is at least \(TS_{\max}\), while the cumulative potential service during the cycle is at most \(TS_{\max}\). Hence queue \(i\) has enough jobs to cover all potential service in that cycle, so no unused service can occur, contradicting \(\left(\boldsymbol{\Sigma}\mathbf U^{\mathrm{alt}}((r+1)T-1)\right)_i>0\).

Consequently, on the event \(\left\{\left(\boldsymbol{\Sigma}\mathbf U^{\mathrm{alt}}((r+1)T-1)\right)_i>0\right\}\),
\[
    \frac{\tilde Q_i^{\mathrm{alt}}(r)}{\gamma_i} < \frac{TS_{\max}}{\gamma_i} \le \frac{TS_{\max}}{\gamma_{\min}} =:M_0.
\]
Thus, for \(l=m+1,\ldots,n\),
\begin{align}
    & v_l \mathbb E\!\left[ \left(\boldsymbol{\Sigma}\mathbf U^{\mathrm{alt}}((r+1)T-1)\right)_{\tilde\eta_r(l)} \,\middle|\,\mathcal F_r \right] \notag\\
    & = \mathbb E\!\left[ v_l \left(\boldsymbol{\Sigma}\mathbf U^{\mathrm{alt}}((r+1)T-1)\right)_{\tilde\eta_r(l)} \,\middle|\,\mathcal F_r \right] \notag\\
    & = \mathbb E\!\left[ e^{\theta\left( \frac{\tilde Q_{\tilde\eta_r(l)}^{\mathrm{alt}}(r)}{\gamma_{\tilde\eta_r(l)}} - z_{\bar B}(r) \right)} \left(\boldsymbol{\Sigma}\mathbf U^{\mathrm{alt}}((r+1)T-1)\right)_{\tilde\eta_r(l)} \,\middle|\,\mathcal F_r \right] \notag \\
    & = \mathbb E\!\left[ e^{\theta\left( \frac{\tilde Q_{\tilde\eta_r(l)}^{\mathrm{alt}}(r)}{\gamma_{\tilde\eta_r(l)}} - z_{\bar B}(r) \right)} \left(\boldsymbol{\Sigma}\mathbf U^{\mathrm{alt}}((r+1)T-1)\right)_{\tilde\eta_r(l)} \mathbf 1_{\left\{ \left(\boldsymbol{\Sigma}\mathbf U^{\mathrm{alt}}((r+1)T-1)\right)_{\tilde\eta_r(l)}>0 \right\}} \,\middle|\,\mathcal F_r \right] \notag\\
    & \leq \mathbb E\!\left[ e^{\theta(M_0-z_{\bar B}(r))} \left(\boldsymbol{\Sigma}\mathbf U^{\mathrm{alt}}((r+1)T-1)\right)_{\tilde\eta_r(l)} \mathbf 1_{\left\{ \left(\boldsymbol{\Sigma}\mathbf U^{\mathrm{alt}}((r+1)T-1)\right)_{\tilde\eta_r(l)}>0 \right\}} \,\middle|\,\mathcal F_r \right] \notag\\
    & \leq \mathbb E\!\left[ e^{\theta(M_0-z_{\bar B}(r))}TS_{\max} \,\middle|\,\mathcal F_r \right] \notag\\
    & = TS_{\max}e^{\theta(M_0-z_{\bar B}(r))}.
    \label{eq:pbd_unused_service_weighted_bound}
\end{align}
Moreover, the queue recursion over one cycle gives, for \(l=m+1,\ldots,n\),
\begin{align*}
    &\mathbb E\!\left[ \Delta\tilde Q_{\tilde\eta_r(l)}^{\mathrm{alt}}(r) - \beta\gamma_{\tilde\eta_r(l)} \,\middle|\,\mathcal F_r \right] \\
    & = \mathbb E\!\left[ \sum_{s=0}^{T-1}A_{\tilde\eta_r(l)}(rT+s) - \sum_{s=0}^{T-1}S_{\tilde\eta_r(l)}(rT+s) + \left(\boldsymbol{\Sigma}\mathbf U^{\mathrm{alt}}((r+1)T-1)\right)_{\tilde\eta_r(l)} - \beta\gamma_{\tilde\eta_r(l)} \,\middle|\,\mathcal F_r \right] \\
    & = Tn\lambda f_{l,\tilde\eta_r} - T\mu_{\tilde\eta_r(l)} + \mathbb E\!\left[ \left(\boldsymbol{\Sigma}\mathbf U^{\mathrm{alt}}((r+1)T-1)\right)_{\tilde\eta_r(l)} \,\middle|\,\mathcal F_r \right] - \beta\gamma_{\tilde\eta_r(l)} \\
    & = T\left( n\lambda f_{l,\tilde\eta_r} - \mu_{\tilde\eta_r(l)} \right) - \beta\gamma_{\tilde\eta_r(l)} + \mathbb E\!\left[ \left(\boldsymbol{\Sigma}\mathbf U^{\mathrm{alt}}((r+1)T-1)\right)_{\tilde\eta_r(l)} \,\middle|\,\mathcal F_r \right] \\
    & = d_l^{\bar B} + \mathbb E\!\left[ \left(\boldsymbol{\Sigma}\mathbf U^{\mathrm{alt}}((r+1)T-1)\right)_{\tilde\eta_r(l)} \,\middle|\,\mathcal F_r \right].
\end{align*}
Multiplying both sides by \(v_l\), summing over \(l=m+1,\ldots,n\), and using equations~\eqref{eq:pbd_barB_weighted_nominal_drift_escape} and \eqref{eq:pbd_unused_service_weighted_bound}, we obtain
\begin{align}
    &\sum_{l=m+1}^{n} v_l \mathbb E\!\left[ \Delta\tilde Q_{\tilde\eta_r(l)}^{\mathrm{alt}}(r) - \beta\gamma_{\tilde\eta_r(l)} \,\middle|\,\mathcal F_r \right] \notag\\
    & = \sum_{l=m+1}^{n} v_l \left[ d_l^{\bar B} + \mathbb E\!\left[ \left(\boldsymbol{\Sigma}\mathbf U^{\mathrm{alt}}((r+1)T-1)\right)_{\tilde\eta_r(l)} \,\middle|\,\mathcal F_r \right] \right] \notag\\
    & = \sum_{l=m+1}^{n}v_l d_l^{\bar B} + \sum_{l=m+1}^{n} v_l \mathbb E\!\left[ \left(\boldsymbol{\Sigma}\mathbf U^{\mathrm{alt}}((r+1)T-1)\right)_{\tilde\eta_r(l)} \,\middle|\,\mathcal F_r \right] \notag\\
    & \leq -\frac{\Delta_{\bar B}^\lambda}{\Gamma_{\bar B}}\Phi_{\bar B}(r) + \sum_{l=m+1}^{n} TS_{\max}e^{\theta(M_0-z_{\bar B}(r))} \notag\\
    & = -\frac{\Delta_{\bar B}^\lambda}{\Gamma_{\bar B}}\Phi_{\bar B}(r) + (n-m)TS_{\max}e^{\theta(M_0-z_{\bar B}(r))}.
    \label{eq:pbd_barB_actual_weighted_drift_escape}
\end{align}
We now derive the drift of \(\Phi_{\bar B}(r)\). Since \(z_{\bar B}(r+1)=z_{\bar B}(r)+\beta\),
\begin{align}
    \Phi_{\bar B}(r+1) & = \sum_{j\in\bar B} \gamma_j e^{\theta\left( \frac{\tilde Q_j^{\mathrm{alt}}(r+1)}{\gamma_j} - z_{\bar B}(r+1) \right)} \notag\\
    & = \sum_{l=m+1}^{n} \gamma_{\tilde\eta_r(l)} e^{\theta\left( \frac{\tilde Q_{\tilde\eta_r(l)}^{\mathrm{alt}}(r+1)}{\gamma_{\tilde\eta_r(l)}} - z_{\bar B}(r+1) \right)} \notag\\
    & = \sum_{l=m+1}^{n} \gamma_{\tilde\eta_r(l)} e^{\theta\left( \frac{\tilde Q_{\tilde\eta_r(l)}^{\mathrm{alt}}(r)}{\gamma_{\tilde\eta_r(l)}} - z_{\bar B}(r) \right)} e^{\theta\left( \frac{\Delta\tilde Q_{\tilde\eta_r(l)}^{\mathrm{alt}}(r)}{\gamma_{\tilde\eta_r(l)}} - \beta \right)}
    \notag\\
    & = \sum_{l=m+1}^{n} \gamma_{\tilde\eta_r(l)}v_l e^{\theta\left( \frac{\Delta\tilde Q_{\tilde\eta_r(l)}^{\mathrm{alt}}(r)}{\gamma_{\tilde\eta_r(l)}} - \beta \right)}.
    \label{eq:pbd_Phi_barB_reindex}
\end{align}
Taking conditional expectation in  Equation~\eqref{eq:pbd_Phi_barB_reindex}, we get
\begin{align}
    \mathbb E[\Phi_{\bar B}(r+1)\mid\mathcal F_r]
    & = \sum_{l=m+1}^{n} \gamma_{\tilde\eta_r(l)}v_l \mathbb E\left[ e^{\theta\left( \frac{\Delta\tilde Q_{\tilde\eta_r(l)}^{\mathrm{alt}}(r)}{\gamma_{\tilde\eta_r(l)}} - \beta \right)} \,\middle|\,\mathcal F_r \right] \notag\\
    & \overset{(a)}{\le} \sum_{l=m+1}^{n} \gamma_{\tilde\eta_r(l)}v_l \mathbb E\left[ 1 + \theta\left( \frac{\Delta\tilde Q_{\tilde\eta_r(l)}^{\mathrm{alt}}(r)}{\gamma_{\tilde\eta_r(l)}} - \beta \right) + C_0\theta^2 \,\middle|\,\mathcal F_r \right] \notag\\
    & = \Phi_{\bar B}(r) + \theta \sum_{l=m+1}^{n} v_l \mathbb E\!\left[ \Delta\tilde Q_{\tilde\eta_r(l)}^{\mathrm{alt}}(r) - \beta\gamma_{\tilde\eta_r(l)} \,\middle|\,\mathcal F_r \right] + C_0\theta^2\Phi_{\bar B}(r) \notag\\
    & \overset{(b)}{\le} \left( 1-\theta\frac{\Delta_{\bar B}^\lambda}{\Gamma_{\bar B}} + C_0\theta^2 \right)\Phi_{\bar B}(r) + \theta(n-m)TS_{\max}e^{\theta(M_0-z_{\bar B}(r))} \notag \\
    & \overset{(c)}{\le} \Phi_{\bar B}(r)+C_{\bar B}e^{-\theta z_{\bar B}(r)}. \notag 
\end{align}
where (a) follows from the Taylor bound, (b) follows from Equation~\eqref{eq:pbd_barB_actual_weighted_drift_escape}, and (c) follows from the choice of \(\theta\) and by taking \( C_{\bar B}:=\theta(n-m)TS_{\max}e^{\theta M_0} \).

Because of the unused-service remainder, \(\Phi_{\bar B}(r)\) is not necessarily a supermartingale by itself. We absorb this summable remainder by defining
\[
    R_{\bar B}(r) := C_{\bar B}\sum_{s=r}^{\infty}e^{-\theta z_{\bar B}(s)}.
\]
Since \(z_{\bar B}(s)=L+M/2+\beta s\) and \(\beta>0\), the series is finite. Moreover,
\begin{align*}
    R_{\bar B}(r)
    & = C_{\bar B}\sum_{s=r}^{\infty}e^{-\theta z_{\bar B}(s)} \\
    & = C_{\bar B}e^{-\theta z_{\bar B}(r)} + C_{\bar B}\sum_{s=r+1}^{\infty}e^{-\theta z_{\bar B}(s)} \\
    & = C_{\bar B}e^{-\theta z_{\bar B}(r)} + R_{\bar B}(r+1).
\end{align*}
Therefore,
\begin{align}
    \mathbb E[\Phi_{\bar B}(r+1)+R_{\bar B}(r+1)\mid\mathcal F_r]
    & = \mathbb E[\Phi_{\bar B}(r+1)\mid\mathcal F_r] + R_{\bar B}(r+1) \notag\\
    & \le \Phi_{\bar B}(r)+C_{\bar B}e^{-\theta z_{\bar B}(r)} + R_{\bar B}(r+1) \notag\\
    & = \Phi_{\bar B}(r) + C_{\bar B}e^{-\theta z_{\bar B}(r)} + C_{\bar B}\sum_{s=r+1}^{\infty}e^{-\theta z_{\bar B}(s)} \notag\\
    & = \Phi_{\bar B}(r) + C_{\bar B}\sum_{s=r}^{\infty}e^{-\theta z_{\bar B}(s)} \notag\\
    & = \Phi_{\bar B}(r)+ R_{\bar B}(r). \notag 
\end{align}
Since \(\Phi_{\bar B}(r)\ge0\) and \(R_{\bar B}(r)\ge0\), it follows that \( \{\Phi_{\bar B}(r)+R_{\bar B}(r)\}_{r\ge0} \) is a nonnegative supermartingale.

For any \(N\ge1\), let
\[
    \tau_N:=N\wedge\kappa_{\bar B}.
\]
Since \(\tau_N\le N\) is a bounded stopping time, the optional stopping theorem gives
\[
    \mathbb E_{\mathbf x} \left[ \Phi_{\bar B}(N\wedge\kappa_{\bar B}) + R_{\bar B}(N\wedge\kappa_{\bar B}) \right] \le \Phi_{\bar B}(0) + R_{\bar B}(0).
\]
On \(\{\kappa_{\bar B}\le N\}\), by the definition of \(\kappa_{\bar B}\), there exists \(j\in\bar B\) such that
\[
    \frac{\tilde Q_j^{\mathrm{alt}}(\kappa_{\bar B})}{\gamma_j}
    \ge
    z_{\bar B}(\kappa_{\bar B}).
\]
Therefore,
\begin{align*}
    \Phi_{\bar B}(\kappa_{\bar B})
    & = \sum_{j'\in\bar B} \gamma_{j'} e^{\theta\left( \frac{\tilde Q_{j'}^{\mathrm{alt}}(\kappa_{\bar B})}{\gamma_{j'}} - z_{\bar B}(\kappa_{\bar B}) \right)} \\
    & \ge \gamma_j e^{\theta\left( \frac{\tilde Q_j^{\mathrm{alt}}(\kappa_{\bar B})}{\gamma_j} - z_{\bar B}(\kappa_{\bar B}) \right)}
    \ge \gamma_j
    \ge \gamma_{\min}.
\end{align*}
Since \(R_{\bar B}(\kappa_{\bar B})\ge0\), we have
\begin{align}
    \gamma_{\min}\mathbb P_{\mathbf x}(\kappa_{\bar B}\le N)
    & = \mathbb E_{\mathbf x} \left[ \gamma_{\min}\mathbf 1_{\{\kappa_{\bar B}\le N\}} \right] \notag\\
    & \le \mathbb E_{\mathbf x} \left[ \left( \Phi_{\bar B}(\kappa_{\bar B}) + R_{\bar B}(\kappa_{\bar B}) \right) \mathbf 1_{\{\kappa_{\bar B}\le N\}} \right] \notag\\
    & = \mathbb E_{\mathbf x} \left[ \left( \Phi_{\bar B}(N\wedge\kappa_{\bar B}) + R_{\bar B}(N\wedge\kappa_{\bar B}) \right) \mathbf 1_{\{\kappa_{\bar B}\le N\}} \right] \notag\\
    & \le \mathbb E_{\mathbf x} \left[ \Phi_{\bar B}(N\wedge\kappa_{\bar B}) + R_{\bar B}(N\wedge\kappa_{\bar B}) \right] \notag\\
    & \le \Phi_{\bar B}(0)+R_{\bar B}(0). \notag 
\end{align}
Letting \(N\to\infty\), and using the fact that \(\{\kappa_{\bar B}\le N\}\) increases to
\(\{\kappa_{\bar B}<\infty\}\), we obtain
\[
    \mathbb P_{\mathbf x}(\kappa_{\bar B}<\infty) \le \frac{\Phi_{\bar B}(0)+R_{\bar B}(0)}{\gamma_{\min}}.
\]
By the initial condition, for every \(j\in\bar B\), we have
\[
    \frac{\tilde Q_j^{\mathrm{alt}}(0)}{\gamma_j} \le L, \qquad z_{\bar B}(0)=L+\frac M2.
\]
Hence
\[
    \frac{\tilde Q_j^{\mathrm{alt}}(0)}{\gamma_j}-z_{\bar B}(0) \le -\frac M2, \qquad j\in\bar B.
\]
Therefore,
\begin{align*}
    \Phi_{\bar B}(0)
    & = \sum_{j\in\bar B} \gamma_j e^{\theta\left( \frac{\tilde Q_j^{\mathrm{alt}}(0)}{\gamma_j} - z_{\bar B}(0) \right)} \\
    & \le \sum_{j\in\bar B}\gamma_j e^{-\theta M/2} = \Gamma_{\bar B}e^{-\theta M/2}.
\end{align*}
Also, since \(z_{\bar B}(s)=L+M/2+\beta s\),
\begin{align*}
    R_{\bar B}(0)
    & = C_{\bar B}\sum_{s=0}^{\infty}e^{-\theta z_{\bar B}(s)} \\
    & = C_{\bar B}\sum_{s=0}^{\infty} e^{-\theta(L+M/2+\beta s)} \\
    & = C_{\bar B}e^{-\theta(L+M/2)} \sum_{s=0}^{\infty}e^{-\theta\beta s} \\
    & = \frac{ C_{\bar B}e^{-\theta(L+M/2)} }{ 1-e^{-\theta\beta} }.
\end{align*}
Thus,
\begin{align}
    \mathbb P_{\mathbf x}(\kappa_{\bar B}<\infty) \le \frac{ \Gamma_{\bar B}e^{-\theta M/2} + \frac{C_{\bar B}e^{-\theta(L+M/2)}}{1-e^{-\theta\beta}} }{\gamma_{\min}}.
    \label{eq:pbd_kappa_barB_bound_escape}
\end{align}

\noindent\textbf{Step 4.} We combine the two bounds and conclude.

Combining equations~\eqref{eq:pbd_kappa_B_bound_escape} and \eqref{eq:pbd_kappa_barB_bound_escape}, we have
\begin{align}
    & \mathbb P_{\mathbf x}(\kappa_B<\infty) + \mathbb P_{\mathbf x}(\kappa_{\bar B}<\infty)
    \notag\\
    & \quad\le \frac{\Gamma_B}{\gamma_{\min}}e^{-\theta M} + \frac{\Gamma_{\bar B}}{\gamma_{\min}}e^{-\theta M/2} + \frac{C_{\bar B}}{\gamma_{\min}(1-e^{-\theta\beta})} e^{-\theta(L+M/2)}. \notag 
\end{align}
Since \(\theta>0\), the first two terms vanish as \(M\to\infty\). Choose \(M\) large enough so that
\[
    \frac{\Gamma_B}{\gamma_{\min}}e^{-\theta M} + \frac{\Gamma_{\bar B}}{\gamma_{\min}}e^{-\theta M/2} < \frac12.
\]
After fixing this \(M\), the third term vanishes as \(L\to\infty\). Choose \(L\) large enough so that
\[
    \frac{C_{\bar B}}{\gamma_{\min}(1-e^{-\theta\beta})} e^{-\theta(L+M/2)} < \frac12.
\]
Therefore, for these choices of \(M\) and \(L\),
\[
    \mathbb P_{\mathbf x}(\kappa_B<\infty) + \mathbb P_{\mathbf x}(\kappa_{\bar B}<\infty) <1.
\]
By the union bound,
\begin{align*}
    \mathbb P_{\mathbf x}(\mathcal E) & = 1-\mathbb P_{\mathbf x}\left( \kappa_B<\infty \text{ or } \kappa_{\bar B} < \infty \right) \\
    & \ge 1-\mathbb P_{\mathbf x}(\kappa_B<\infty) - \mathbb P_{\mathbf x}(\kappa_{\bar B}<\infty) >0.
\end{align*}
On \(\mathcal E\), for every \(r\in\mathbb N\),
\[
    \frac{\tilde Q_i^{\mathrm{alt}}(r)}{\gamma_i}>z_B(r), \qquad i\in B,
\]
and
\[
    \frac{\tilde Q_j^{\mathrm{alt}}(r)}{\gamma_j}<z_{\bar B}(r), \qquad j\in\bar B.
\]
Since
\[
    z_B(r)-z_{\bar B}(r)=\frac M2+(\alpha-\beta)r>0,
\]
we obtain
\[
    \min_{i\in B} \left\{ \frac{\tilde Q_i^{\mathrm{alt}}(r)}{\gamma_i} \right\} > \max_{j\in\bar B} \left\{ \frac{\tilde Q_j^{\mathrm{alt}}(r)}{\gamma_j} \right\}, \qquad \forall r\in\mathbb N.
\]
Moreover, on \(\mathcal E\), for every \(r\in\mathbb N\) and \(a>0\), it follows that
\[
    \lim_{r\to\infty} \sum_{i\in B}\tilde Q_i^{\mathrm{alt}}(r) \ge \lim_{r\to\infty} \sum_{i\in B}\gamma_i z_B(r) = \lim_{r\to\infty} \Gamma_B(L+M+\alpha r) = \infty.
\]

\subsection{Proof of Lemma~\ref{lem:fixed_load_fluid_limit_admissible}}
\label{prf:fixed_load_fluid_limit_admissible}

Let \(\mathbf q^{\mathrm{fl}}(\cdot)\) be any fluid limit of the embedded queue-length process obtained from the fluid scaling in Step~2.1, and write \(\mathbf q^{\mathrm{fl},(\gamma)}(t):=(q_i^{\mathrm{fl}}(t)/\gamma_i)_{i=1}^n\). \\

\noindent\textbf{Step 1}. We identify the local mixture of sampled permutations.

Fix a differentiability time \(t<\tau_{\mathrm{hit}}\). Then \(\mathbf q^{\mathrm{fl}}(t)\neq\mathbf 0\). At every prelimit sampling epoch, the policy realizes a single sampled permutation, after applying its tie-breaking rule if necessary. If the limiting scaled state \(\mathbf q^{\mathrm{fl},(\gamma)}(t)\) has ties, then the prelimit sampled permutations in a vanishing neighborhood of \(t\) may vary. However, any subsequential limiting empirical frequency can put positive mass only on permutations \(\eta\in\mathcal S_n\) that are consistent with the weak ordering of \(\mathbf q^{\mathrm{fl},(\gamma)}(t)\), that is, permutations satisfying \(q_{\eta(1)}^{\mathrm{fl},(\gamma)}(t)\ge\cdots\ge q_{\eta(n)}^{\mathrm{fl},(\gamma)}(t)\). 

For each prelimit system, consider the empirical frequency vector of the sampled permutations in a vanishing fluid-scale neighborhood of \(t\). This vector belongs to the simplex on the finite set \(\mathcal S_n\). Since this simplex is compact, along a subsequence these empirical frequency vectors converge to some probability vector \(\mathbf w(t)\). By the preceding observation, \(\mathbf w(t)\) assigns positive mass only to permutations consistent with the weak ordering of \(\mathbf q^{\mathrm{fl},(\gamma)}(t)\), and hence \(\mathbf w(t)\in\mathfrak W(\mathbf q^{\mathrm{fl}}(t))\). Therefore, the fluid-scale averaged dispatch fraction assigned to server \(i\) is
\[
    \sum_{\eta\in\mathcal S_n} w_\eta(t)f_{\ell_\eta(i),\eta}.
\]

\noindent\textbf{Step 2}. We identify the reflected fluid drift.

By the functional strong law of large numbers for arrivals and potential services, stochastic fluctuations vanish under the fluid scaling. Therefore, the averaged nominal drift of the \(\boldsymbol{\gamma}\)-scaled queue of server \(i\) at time \(t\) is
\[
    \frac{ T\left( n\lambda \sum_{\eta\in\mathcal S_n} w_\eta(t)f_{\ell_\eta(i),\eta} - \mu_i \right) }{\gamma_i}.
\]
Let \( J(t):=\{i\in[n]:q_i^{\mathrm{fl}}(t)>0\} \). If \(i\in J(t)\), then queue \(i\) is positive in the fluid state, and its \(\boldsymbol{\gamma}\)-scaled fluid drift equals the averaged nominal drift. If \(i\notin J(t)\), then queue \(i\) is at the zero boundary. In that case, negative averaged nominal drift is absorbed by unused service, whereas positive averaged nominal drift makes the queue leave the boundary. Hence, for every \(i\in[n]\),
\[
    \frac{d}{dt}q_i^{\mathrm{fl},(\gamma)}(t) = b_i^{J(\mathbf q^{\mathrm{fl}}(t))} \left( \mathbf w(t),n\lambda \right).
\]
Equivalently,
\[
    \frac{d}{dt}\mathbf q^{\mathrm{fl},(\gamma)}(t) = b^{J(\mathbf q^{\mathrm{fl}}(t))}
    \left( \mathbf w(t),n\lambda \right) \in \mathfrak D_{n\lambda}(\mathbf q^{\mathrm{fl}}(t)).
\]

\subsection{Proof of Lemma~\ref{lem:fixed_load_scd_fluid_draining}}
\label{prf:fixed_load_scd_fluid_draining}

Let \(\mathbf q^{\mathrm{fl}}(\cdot)\) be any fluid limit with
\[
    \|\mathbf q^{\mathrm{fl}}(0)\|_1=1.
\]
By Lemma~\ref{lem:fixed_load_fluid_limit_admissible}, before the trajectory reaches the origin,
\[
    \frac{d}{dt}\mathbf q^{\mathrm{fl},(\gamma)}(t) \in \mathfrak D_{n\lambda}(\mathbf q^{\mathrm{fl}}(t))
\]
for almost every \(t\). \\

\noindent\textbf{Step 1}. We obtain a uniform bound on admissible drift speeds.

For every \(\mathbf q\in\mathbb R_+^n\), \(\mathbf w\in\mathfrak W(\mathbf q)\), and \(i\in[n]\), since \(\mathbf w\) is a probability vector and \(0\le f_{\ell_\eta(i),\eta}\le1\), we have \(0\le \sum_{\eta\in\mathcal S_n}w_\eta f_{\ell_\eta(i),\eta}\le1\). Therefore,
\begin{align*}
    \left|b_i^{J(\mathbf q)}(\mathbf w,n\lambda)\right|
    & \le \frac{T\left(n\lambda\sum_{\eta\in\mathcal S_n}w_\eta f_{\ell_\eta(i),\eta}+\mu_i\right)}{\gamma_i} \\
    & \le \frac{T(n\lambda+\mu_i)}{\gamma_i}.
\end{align*}
Define
\begin{align*}
    M(n\lambda) := T\sum_{i=1}^n \frac{n\lambda+\mu_i}{\gamma_i}.
\end{align*}
Then \(M(n\lambda)<\infty\), and for all \(\mathbf q\in\mathbb R_+^n\) and \(\mathbf w\in\mathfrak W(\mathbf q)\),
\begin{align*}
    \left\| b^{J(\mathbf q)}(\mathbf w,n\lambda) \right\|_1 = \sum_{i=1}^n \left|b_i^{J(\mathbf q)}(\mathbf w,n\lambda)\right| \le M(n\lambda).
\end{align*}
Therefore, for almost every \(t\) such that \(\frac{d}{dt}\mathbf q^{\mathrm{fl},(\gamma)}(t)\in \mathfrak D_{n\lambda}(\mathbf q^{\mathrm{fl}}(t))\), there exists some \(\mathbf w(t)\in\mathfrak W(\mathbf q^{\mathrm{fl}}(t))\) such that
\[
    \frac{d}{dt}\mathbf q^{\mathrm{fl},(\gamma)}(t) = b^{J(\mathbf q^{\mathrm{fl}}(t))}(\mathbf w(t),n\lambda).
\]
By the definition of \(M(n\lambda)\), this implies
\begin{align*}
    \left\| \frac{d}{dt}\mathbf q^{\mathrm{fl},(\gamma)}(t) \right\|_1 = \left\| b^{J(\mathbf q^{\mathrm{fl}}(t))}(\mathbf w(t),n\lambda) \right\|_1
    \le M(n\lambda).
\end{align*} 

\noindent\textbf{Step 2}. We bound the time spent in the catch-up phases.

By Definition~\ref{def:fixed_load_scd}, applied to the fluid limit
\(\mathbf q^{\mathrm{fl}}(\cdot)\), there exist switching times
\[
    0=s_0\le s_1\le\cdots\le s_R\le\tau_{\mathrm{hit}} < \infty,
\]
where \( \tau_{\mathrm{hit}}:=\inf\{t\ge0:\mathbf q^{\mathrm{fl}}(t)=\mathbf 0\} \). We first record a uniform bound on admissible scaled drifts. By Lemma~\ref{lem:fixed_load_fluid_limit_admissible}, for almost every \(t<\tau_{\mathrm{hit}}\),
\[
    \frac{d}{dt}\mathbf q^{\mathrm{fl},(\gamma)}(t) \in \mathfrak D_{n\lambda}(\mathbf q^{\mathrm{fl}}(t)).
\]
Since \(n\lambda\) is fixed and \(\mathcal S_n\) is finite, all possible reflected fluid drifts are uniformly bounded. Define
\[
    M(n\lambda) := \sup_{\mathbf x\in\mathbb R_+^n} \sup_{\mathbf w\in\mathfrak W(\mathbf x)} \left\|b^{J(\mathbf x)}(\mathbf w,n\lambda)\right\|_1.
\]
This is a uniform upper bound on the possible values of \(\left\|\frac{d}{dt}\mathbf q^{\mathrm{fl},(\gamma)}(t)\right\|_1\), because Lemma~\ref{lem:fixed_load_fluid_limit_admissible} implies that, for almost every \(t<\tau_{\mathrm{hit}}\),
\[
    \frac{d}{dt}\mathbf q^{\mathrm{fl},(\gamma)}(t) = b^{J(\mathbf q^{\mathrm{fl}}(t))}(\mathbf w(t),n\lambda)
\]
for some \(\mathbf w(t)\in\mathfrak W(\mathbf q^{\mathrm{fl}}(t))\). The constant \(M(n\lambda)\) is finite because each component of \(b^{J(\mathbf x)}(\mathbf w,n\lambda)\) is bounded in absolute value by \(T(n\lambda+\mu_i)/\gamma_i\).

Let \(B_0:=1/\gamma_{\min}\). Since \(\|\mathbf q^{\mathrm{fl}}(0)\|_1=1\) and all \(\gamma_i>0\), we have
\[
    \|\mathbf q^{\mathrm{fl},(\gamma)}(0)\|_1 = \sum_{i=1}^n\frac{q_i^{\mathrm{fl}}(0)}{\gamma_i}
    \le \frac{1}{\gamma_{\min}} = B_0.
\]
We prove inductively that, for each \(r=0,1,\ldots,R\), there exists a finite constant \(B_r\), depending only on \(n\lambda\) and the constants in Definition~\ref{def:fixed_load_scd}, such that
\[
    \|\mathbf q^{\mathrm{fl},(\gamma)}(s_r)\|_1\le B_r.
\]
The claim holds for \(r=0\). Suppose it holds for \(r-1\), that is, \(\|\mathbf q^{\mathrm{fl},(\gamma)}(s_{r-1})\|_1\le B_{r-1}\). Define
\begin{align*}
    B_r := \left(1+\frac{M(n\lambda)\kappa_r}{\delta_r}\right)B_{r-1}.
\end{align*}
If the \(r\)-th catch-up phase is skipped, then \(s_r=s_{r-1}\). Hence \(\|\mathbf q^{\mathrm{fl},(\gamma)}(s_r)\|_1=\|\mathbf q^{\mathrm{fl},(\gamma)}(s_{r-1})\|_1\le B_{r-1}\le B_r\), and the induction step holds. We now consider the case in which the \(r\)-th catch-up phase is not skipped.

By Definition~\ref{def:fixed_load_scd}, on \([s_{r-1},s_r)\) we have \(G_r(\mathbf q^{\mathrm{fl},(\gamma)}(t))>0\) and
\[
    D^+G_r\left( \mathbf q^{\mathrm{fl},(\gamma)}(t), \frac{d}{dt}\mathbf q^{\mathrm{fl},(\gamma)}(t) \right) \le -\delta_r
\]
for almost every \(t\). Also, Definition~\ref{def:fixed_load_scd} and the induction hypothesis give
\begin{align*}
    G_r(\mathbf q^{\mathrm{fl},(\gamma)}(s_{r-1})) \le \kappa_r\|\mathbf q^{\mathrm{fl},(\gamma)}(s_{r-1})\|_1 \le \kappa_r B_{r-1}.
\end{align*}

Since \(G_r\) is locally Lipschitz and \(\mathbf q^{\mathrm{fl},(\gamma)}(\cdot)\) is absolutely continuous, \(G_r(\mathbf q^{\mathrm{fl},(\gamma)}(\cdot))\) is absolutely continuous on every interval \([s_{r-1},s]\) with \(s<s_r\). For almost every \(t\in[s_{r-1},s_r)\), its derivative along the path is bounded above by the corresponding Dini derivative. Therefore,
\[
    \frac{d}{dt}G_r(\mathbf q^{\mathrm{fl},(\gamma)}(t)) \le D^+G_r\left( \mathbf q^{\mathrm{fl},(\gamma)}(t), \frac{d}{dt}\mathbf q^{\mathrm{fl},(\gamma)}(t) \right) \le -\delta_r
\]
for almost every \(t\in[s_{r-1},s_r)\).

Fix any \(s<s_r\). Integrating the preceding inequality over \([s_{r-1},s]\) yields
\begin{align*}
    G_r(\mathbf q^{\mathrm{fl},(\gamma)}(s)) - G_r(\mathbf q^{\mathrm{fl},(\gamma)}(s_{r-1}))
    & = \int_{s_{r-1}}^s \frac{d}{dt}G_r(\mathbf q^{\mathrm{fl},(\gamma)}(t))\,dt \\
    & \le -\delta_r(s-s_{r-1}).
\end{align*}
Rearranging gives \(\delta_r(s-s_{r-1})\le G_r(\mathbf q^{\mathrm{fl},(\gamma)}(s_{r-1}))-G_r(\mathbf q^{\mathrm{fl},(\gamma)}(s))\). Since \(G_r\) is nonnegative, this implies \(\delta_r(s-s_{r-1})\le G_r(\mathbf q^{\mathrm{fl},(\gamma)}(s_{r-1}))\). Taking the supremum over all \(s<s_r\), we obtain
\begin{align*}
    \delta_r(s_r-s_{r-1}) \le G_r(\mathbf q^{\mathrm{fl},(\gamma)}(s_{r-1})).
\end{align*}
Combining this bound with \(G_r(\mathbf q^{\mathrm{fl},(\gamma)}(s_{r-1}))\le \kappa_r B_{r-1}\), we get
\begin{align*}
    s_r-s_{r-1} \le \frac{G_r(\mathbf q^{\mathrm{fl},(\gamma)}(s_{r-1}))}{\delta_r} \le \frac{\kappa_r B_{r-1}}{\delta_r}.
\end{align*}

During this interval, the admissible drift is bounded by \(M(n\lambda)\) for almost every $t$. Hence,
\begin{align*}
    \|\mathbf q^{\mathrm{fl},(\gamma)}(s_r)\|_1
    & \overset{(a)}{=} \left\| \mathbf q^{\mathrm{fl},(\gamma)}(s_{r-1}) + \int_{s_{r-1}}^{s_r} \frac{d}{dt}\mathbf q^{\mathrm{fl},(\gamma)}(t)\,dt \right\|_1 \\
    & \overset{(b)}{\le} \|\mathbf q^{\mathrm{fl},(\gamma)}(s_{r-1})\|_1 + \left\| \int_{s_{r-1}}^{s_r} \frac{d}{dt}\mathbf q^{\mathrm{fl},(\gamma)}(t)\,dt \right\|_1 \\
    & \overset{(c)}{\le} \|\mathbf q^{\mathrm{fl},(\gamma)}(s_{r-1})\|_1 + \int_{s_{r-1}}^{s_r} \left\| \frac{d}{dt}\mathbf q^{\mathrm{fl},(\gamma)}(t) \right\|_1dt \\
    & \overset{(d)}{\le} B_{r-1} + M(n\lambda)(s_r-s_{r-1}) \\
    & \overset{(e)}{\le} B_{r-1} + M(n\lambda)\frac{\kappa_r B_{r-1}}{\delta_r} \\
    & = \left(1+\frac{M(n\lambda)\kappa_r}{\delta_r}\right)B_{r-1} \\
    & = B_r,
\end{align*}
where (a) follows from the componentwise integral representation of absolutely continuous vector-valued functions, (b) follows from the triangle inequality, (c) follows from the integral form of the triangle inequality, (d) follows from the induction hypothesis \(\|\mathbf q^{\mathrm{fl},(\gamma)}(s_{r-1})\|_1\le B_{r-1}\) and the a.e. drift bound \(\|\frac{d}{dt}\mathbf q^{\mathrm{fl},(\gamma)}(t)\|_1\le M(n\lambda)\), and (e) follows from the catch-up phase length bound.

This proves the induction step. Therefore, \(\|\mathbf q^{\mathrm{fl},(\gamma)}(s_r)\|_1\le B_r\) for all \(r=0,1,\ldots,R\). Moreover, since \(s_r-s_{r-1}\le \kappa_r B_{r-1}/\delta_r\) for each \(r=1,\ldots,R\), the total time spent in the catch-up phases is bounded by
\begin{align*}
    \sum_{r=1}^{R}(s_r-s_{r-1})
    & \le \sum_{r=1}^{R}\frac{\kappa_r B_{r-1}}{\delta_r} \\
    & = B_0\sum_{r=1}^{R} \frac{\kappa_r}{\delta_r} \prod_{u=1}^{r-1} \left(1+\frac{M(n\lambda)\kappa_u}{\delta_u}\right)
    <\infty,
\end{align*}
where the empty product is interpreted as \(1\). \\

\noindent\textbf{Step 3}. We bound the terminal draining time.

If \(s_R=\tau_{\mathrm{hit}}\), then the trajectory has already reached the origin. Otherwise, Definition~\ref{def:fixed_load_scd} gives
\[
    V_D(\mathbf q^{\mathrm{fl},(\gamma)}(s_R)) \le \kappa_D \|\mathbf q^{\mathrm{fl},(\gamma)}(s_R)\|_1 \le \kappa_D B_R.
\]
On the terminal interval \([s_R,\tau_{\mathrm{hit}})\), for almost every \(t\),
\[
    D^+V_D \left( \mathbf q^{\mathrm{fl},(\gamma)}(t), \frac{d}{dt}\mathbf q^{\mathrm{fl},(\gamma)}(t) \right) \le -\delta_D.
\]
Since \(V_D(\mathbf q^{\mathrm{fl},(\gamma)}(\cdot))\) is absolutely continuous, the terminal draining time is bounded by
\[
    \tau_{\mathrm{hit}}-s_R \le \frac{\kappa_D B_R}{\delta_D}.
\]
Combining the catch-up phase bounds with the terminal draining bound, every unit-norm fluid limit satisfies
\begin{align*}
    \tau_{\mathrm{hit}}
    & = s_R+(\tau_{\mathrm{hit}}-s_R) \\
    & \le \sum_{r=1}^{R}(s_r-s_{r-1}) + \frac{\kappa_D B_R}{\delta_D} \\
    & \le \sum_{r=1}^{R}\frac{\kappa_r B_{r-1}}{\delta_r} + \frac{\kappa_D B_R}{\delta_D} \\
    & =: \tau_{\mathrm{scd}}(n\lambda) <\infty.
\end{align*}
Hence every unit-norm fluid limit reaches the origin by time \(\tau_{\mathrm{scd}}(n\lambda)\). Because all \(\gamma_i>0\), \(\mathbf q^{\mathrm{fl},(\gamma)}(t)=\mathbf 0\) is equivalent to \(\mathbf q^{\mathrm{fl}}(t)=\mathbf 0\). We continue the fluid limit at the origin after its first hitting time. Therefore,
\[
    \mathbf q^{\mathrm{fl}}(t)=\mathbf 0, \qquad \forall t\ge\tau_{\mathrm{scd}}(n\lambda).
\]

\subsection{Proof of Lemma~\ref{lem:fixed_load_fluid_stability_criterion}}
\label{prf:fixed_load_fluid_stability_criterion}

Recall that \(\{\tilde{\mathbf X}(r)\}_{r\ge0}\) denotes the embedded Markov chain observed at sampling epochs, and that \(\{\tilde{\mathbf Q}(r)\}_{r\ge0}\) is its queue-length component. The purpose of this lemma is to transfer a uniform finite-time draining property of the fluid limits of \(\{\tilde{\mathbf Q}(r)\}_{r\ge0}\) to positive recurrence of the stochastic Markov chain.

By assumption, there exists a finite constant \(\tau<\infty\) such that every fluid limit \(\mathbf q^{\mathrm{fl}}(\cdot)\) of the embedded queue-length process, with
\[
    \|\mathbf q^{\mathrm{fl}}(0)\|_1=1,
\]
satisfies
\[
    \mathbf q^{\mathrm{fl}}(t)=\mathbf 0, \qquad \forall t\ge \tau.
\]
Thus all unit-norm fluid limits drain to the origin in a uniformly bounded fluid time.

We now explain why the standard fluid-stability criterion applies to the embedded process in our setting. First, the embedded full-state process \(\{\tilde{\mathbf X}(r)\}_{r\ge0}\) is Markov because it is obtained by sampling the full Markov state \(\{\mathbf X(t)\}_{t\ge0}\) at the deterministic times \(0,T,2T,\ldots\). Second, the queue-length component has uniformly bounded one-cycle increments, because arrivals and potential services are bounded over each dispatching cycle. Hence the fluid-scaled queue-length paths are uniformly Lipschitz on compact time intervals, and subsequential fluid limits are well defined. These are the standard compactness conditions needed for the fluid limit stability criterion.

Therefore, the assumption above establishes stability of the fluid limit model for the embedded chain. By the fluid-stability theorem from~\cite{dai1995positive-fluidmodel} (Definition~4.1 and Theorem~4.2), stability of the fluid model implies positive Harris recurrence of the underlying queueing Markov process. Applied to the embedded chain \(\{\tilde{\mathbf X}(r)\}_{r\ge0}\), this gives positive Harris recurrence of the embedded chain. In the countable-state formulation used here, this is positive recurrence.

It remains to relate the embedded chain to the original chain \(\{\mathbf X(t)\}_{t\ge0}\). The embedded chain records the original chain every \(T\) slots. Between two consecutive sampling epochs there are only \(T\) slots, and the within-cycle evolution is governed by the same bounded arrival, service, and dispatching primitives. Hence positive recurrence of the embedded chain transfers to the corresponding communicating class of the original chain. Therefore, the induced Markov chain is positive recurrent.

\section{Proof of Corollary~\ref{cor_our:load_window_certificate_classification}}
\label{prf:load_window_certificate_classification}

Fix an interval \(I_k=(\xi_k,\xi_{k+1})\). Since
\[
    h^*=\xi_0<\xi_1<\cdots<\xi_K<\xi_{K+1} = \sum_{i=1}^n\mu_i,
\]
every load \(n\lambda\in I_k\) satisfies
\[
    h^*<n\lambda<\sum_{i=1}^n\mu_i.
\]
We prove the two cases separately.\\

\noindent\textbf{Step 1}. Suppose that \(\chi_k=\mathrm T\).

By Assumption~\ref{assump:load_window_certificates}, \(\mathrm{PBD}(n\lambda)\) holds for every \(n\lambda\in I_k\). Fix any \(n\lambda\in I_k\). Since \( h^*<n\lambda<\sum_{i=1}^n\mu_i \), Part 1 of Theorem~\ref{thm_our:fixed_load_stability_certificates} implies that the induced Markov chain is transient from some initial state. If the Markov chain is irreducible on the state space under consideration, then \(\{\mathbf X(t)\}_{t\ge0}\) is transient. Since \(n\lambda\in I_k\) was arbitrary, the statement holds for every load in \(I_k\).\\

\noindent\textbf{Step 2}. Suppose that \(\chi_k=\mathrm S\).

By Assumption~\ref{assump:load_window_certificates}, \(\mathrm{SCD}(n\lambda)\) holds for every \(n\lambda\in I_k\). Fix any \(n\lambda\in I_k\). Since \( n\lambda<\sum_{i=1}^n\mu_i \) and \(\mathrm{SCD}(n\lambda)\) holds, Part~2 of Theorem~\ref{thm_our:fixed_load_stability_certificates} implies that the induced Markov Chain \(\{\mathbf X(t)\}_{t\ge0}\) is positive recurrent. Since \(n\lambda\in I_k\) was arbitrary, the statement holds for every load in \(I_k\).\\

\noindent\textbf{Step 3}. Since the interval \(I_k\) was arbitrary, the two claims hold for every interval in the certified load-window partition.

\section{Theoretical Analysis for Example~\ref{ex:nonmonotone_stability_region}}
\label{app:example_A_theoretical_analysis}

\noindent\textbf{Step 1.} We study policy A.

In Example~\ref{ex:nonmonotone_stability_region}, we have \(\boldsymbol{\mu}=(1,2,3,20)\), \(T=1\), and \(\gamma_i=1\) for all \(i\in[4]\). Hence the scaled queue lengths coincide with the raw queue lengths. The dispatching probabilities for Policy A are given in Table~\ref{tab:nonmonotone_policy_A}. The total service capacity is \(26\). Since server 1 receives fraction \(0.20\) when it is the longest queue, the static threshold is \(h^*=1/0.20=5\).

\noindent\textbf{Step 1.1.} The interval \((5,20/3)\) is a bottleneck window.

Take \(B=\{1\}\). When server 1 is the longest queue, its drift is \(0.20\,n\lambda-1\), which is strictly positive for \(n\lambda>5\). Thus \(\underline{\delta}(\{1\},n\lambda)=0.20\,n\lambda-1\).

It remains to compare this drift with the normalized cumulative drifts of prefixes of the complement \(\bar B=\{2,3,4\}\). Under the same longest-server-1 regime, the individual drifts of servers \(2,3,4\) are \(0.35\,n\lambda-2\), \(-3\), and \(0.45\,n\lambda-20\), respectively. Because \(\gamma_i=1\), the normalized cumulative drift of any nonempty prefix of \(\bar B\) is the average of the individual drifts in that prefix, and hence is no larger than the largest individual drift. For \(n\lambda\in(5,20/3)\), the largest individual drift in the complement is \(0.35\,n\lambda-2\). Therefore \(\overline{\delta}(\{ 2,3,4\},n\lambda)=0.35\,n\lambda-2\). Since \(0.20\,n\lambda-1>0.35\,n\lambda-2\) is equivalent to \(n\lambda<20/3\), we have
\[
    \underline{\delta}(\{1\},n\lambda) > \max\{\overline{\delta}(\{ 2,3,4\},n\lambda),0\}
\]
for every \(n\lambda\in(5,20/3)\). Hence \(\mathrm{PBD}(n\lambda)\) holds with \(B=\{1\}\) on \((5,20/3)\).

\noindent\textbf{Step 1.2.} The interval \((20/3,8)\) is a sequential catch-up
drainage window.

We verify \(\mathrm{SCD}(n\lambda)\) by constructing the catch-up phases in Definition~\ref{def:fixed_load_scd}. Since \(\gamma_i=1\), the scaled and raw fluid coordinates coincide. We take \(R=2\), \(G_1(\mathbf q):=(\max\{q_3,q_4\}-q_1)^+\), \(G_2(\mathbf q):=(q_1-q_2)^+\), and \(V_D(\mathbf q):=\|\mathbf q\|_{\infty}\).

The first phase makes server 1 catch any leading queue among servers 3 and 4. Starting from any admissible fluid path, let \(s_1\) be the first time at which either \(G_1\) reaches zero, server 2 becomes a longest queue, or the fluid state reaches the origin. If this phase is not skipped, then on \([0,s_1)\) we have \(G_1(\mathbf q(t))>0\), server 2 is not a longest queue, and the longest queue is either server 3 or server 4. If server 3 is longest, then the drift of \(q_3-q_1\) is \((0.25\,n\lambda-3)-(-1)=0.25\,n\lambda-2<0\), since \(n\lambda<8\). If server 4 is longest, then the drift of \(q_4-q_1\) is \((n\lambda-20)-(-1)=n\lambda-19<0\). Hence \(G_1\) decreases at a uniform positive rate. For example, one may take \(\delta_1(n\lambda):=\min\{2-0.25\,n\lambda,19-n\lambda\}>0\).

If server 2 becomes a longest queue during the first phase, the path enters the terminal phase. Otherwise, after \(s_1\), we have \(q_1\ge q_3,q_4\). The second phase makes server 2 catch server 1. Let \(s_2\) be the first time after \(s_1\) at which either \(G_2\) reaches zero, server 2 becomes a longest queue, or the fluid state reaches the origin. If this phase is not skipped, then on \([s_1,s_2)\) we have \(G_2(\mathbf q(t))>0\), and server 1 is the longest queue except possibly at switching instants. Under the longest-server-1 regime, the drift of \(q_1-q_2\) is
\[
    (0.20\,n\lambda-1)-(0.35\,n\lambda-2) = 1-0.15\,n\lambda < 0,
\]
because \(n\lambda>20/3\). Hence \(G_2\) decreases at a uniform positive rate and we can take \(\delta_2(n\lambda):=0.15\,n\lambda-1>0\).

For the terminal phase, we use the leading height \(V_D(\mathbf q)=\|\mathbf q\|_\infty\) as the draining workload. If the path has not already reached the origin, then after the catch-up phases server 2 is a longest queue, except possibly at switching instants. When server 2 is uniquely longest, \(D^+V_D\le a^{(2)}_2(n\lambda)=0.25\,n\lambda-2<0\) on \((20/3,8)\).

It remains to check tied leading faces in the terminal phase. In this phase, \(V_D(q)=\|q\|_\infty\), and after the two catch-up phases server 2 is one of the leading queues. Hence any terminal leading face that can persist must contain server 2. If a tied leading face \(L\) persists over a positive amount of fluid time, then the queues in \(L\) must have the same derivative under some admissible ordering mixture. Otherwise, the tie is broken immediately and the face can occur only at a switching time. Solving these equal-derivative conditions for candidate leading faces containing server 2 on \(n\lambda\in(20/3,8)\) leaves only \(L=\{2\}\), \(L=\{2,3\}\), and \(L=\{1,2,3\}\). The corresponding common leading derivatives are \(0.25\,n\lambda-2\), \(5n\lambda/12-11/3\), and \(5n\lambda/51-83/51\), respectively, all of which are strictly negative on \((20/3,8)\). All other tied leading faces containing server 2 fail the equal-derivative condition with admissible mixture weights, meaning that the corresponding linear equations require at least one mixture weight outside \([0,1]\). Therefore, these faces can occur only at switching boundaries. Hence the terminal Dini-derivative inequality holds for almost every time in the terminal phase, verifying \(\mathrm{SCD}(n\lambda)\) on \((20/3,8)\).

\noindent\textbf{Step 1.3.} The interval \((8,10)\) is a bottleneck window.

Take \(B=\{2\}\). Then \(m=1\), \(\bar B=\{1,3,4\}\), and every \(\eta\in\mathcal P(B)\) has server 2 as the longest queue. When server 2 is longest, Policy A dispatches arrivals to actual servers \(1,2,3,4\) with probabilities \((0,0.25,0.35,0.40)\). Since \(T=1\) and \(\gamma_i=1\), the normalized drift of the candidate bottleneck set is
\[
    \underline{\delta}(\{2\},n\lambda)=0.25\,n\lambda-2,
\]
which is strictly positive for \(n\lambda>8\).

For the complement, the individual drifts of servers \(1,3,4\) under the same longest-server-2 regime are \(-1\), \(0.35\,n\lambda-3\), and \(0.40\,n\lambda-20\), respectively. As \(\eta\) ranges over \(\mathcal P(B)\), the prefixes of \(\bar B\) range over nonempty subsets of \(\{1,3,4\}\). Because \(\gamma_i=1\), the normalized cumulative drift of any such prefix is the average of the individual drifts in that prefix, and hence is no larger than the largest individual drift. For \(n\lambda\in(8,10)\), the largest individual drift in the complement is \(0.35\,n\lambda-3\). Therefore
\[
    \overline{\delta}(\{1,3,4\},n\lambda)=0.35\,n\lambda-3.
\]
Since \(0.25\,n\lambda-2>0.35\,n\lambda-3\) is equivalent to \(n\lambda<10\), and since \(\underline{\delta}(\{2\},n\lambda)>0\) for \(n\lambda>8\), we have
\[
    \underline{\delta}(\{2\},n\lambda) > \max\{\overline{\delta}(\{1,3,4\},n\lambda),0\}
\]
for every \(n\lambda\in(8,10)\). Hence \(\mathrm{PBD}(n\lambda)\) holds with \(B=\{2\}\) on \((8,10)\).

\noindent\textbf{Step 1.4.} The interval \((10,12)\) is a sequential catch-up
drainage window.

We verify \(\mathrm{SCD}(n\lambda)\) by constructing the catch-up phases in Definition~\ref{def:fixed_load_scd}. Since \(\gamma_i=1\), the scaled and raw fluid coordinates coincide. We take \(R=3\), \(G_1(\mathbf q):=(q_4-\max\{q_1,q_2,q_3\})^+\), \(G_2(\mathbf q):=(q_1-q_2)^+\), \(G_3(\mathbf q):=(\max\{q_1,q_2,q_4\}-q_3)^+\), and \(V_D(\mathbf q):=\|\mathbf q\|_{\infty}\).

The first phase makes the set \(\{1,2,3\}\) catch server 4 if server 4 is ahead. Let \(s_1\) be the first time at which \(G_1\) reaches zero, server 3 catches the current leading queue, or the fluid state reaches the origin. If this phase is not skipped, then on \([0,s_1)\), server 4 is the longest queue. The drift of \(q_4-\max\{q_1,q_2,q_3\}\) is at most \((n\lambda-20)-(-3)=n\lambda-17<0\), so \(G_1\) closes at a uniform positive rate.

The second phase makes server 2 catch server 1. After the first phase, server 4 is no longer strictly above the set \(\{1,2,3\}\). If the trajectory has already entered the server-3-leading regime, the remaining catch-up phases are skipped. Otherwise, if \(G_2(\mathbf q)>0\), then server 1 is the longest queue during this phase. The drift of \(q_1-q_2\) is \((0.20\,n\lambda-1)-(0.35\,n\lambda-2)=1-0.15\,n\lambda<0\), because \(n\lambda>20/3\). Hence \(G_2\) closes at a uniform positive rate.

The third phase makes server 3 catch server 2. If the trajectory has already entered the server-3-leading regime after the second phase, this phase is skipped. Otherwise, if \(G_3(\mathbf q)>0\), then server 2 is the longest queue during this phase. The drift of \(q_2-q_3\) is \((0.25\,n\lambda-2)-(0.35\,n\lambda-3)=1-0.10\,n\lambda<0\), because \(n\lambda>10\). Hence \(G_3\) closes at a uniform positive rate.

For the terminal phase, we use the leading height \(V_D(\mathbf q)=\|\mathbf q\|_\infty\). Once \(G_3=0\), server 3 has caught all remaining leading queues. If server 3 is uniquely longest, then \(D^+V_D\le a^{(3)}_3(n\lambda)=0.25\,n\lambda-3<0\) on \((10,12)\).

An analogous tied-face check applies in the terminal phase. After the catch-up phases, server 3 is one of the leading queues, so any persistent terminal leading face must contain server 3. Solving the equal-derivative conditions for candidate leading faces containing server 3 on \(n\lambda\in(10,12)\) leaves only \(L=\{3\}\), \(L=\{1,3\}\), and \(L=\{1,2,3\}\). The corresponding common leading derivatives are \(0.25\,n\lambda-3\), \(n\lambda/9-17/9\), and \(5n\lambda/51-83/51\), respectively, all of which are strictly negative on \((10,12)\). The remaining tied leading faces fail the equal-derivative condition with admissible mixture weights and can occur only at switching boundaries. Hence the terminal Dini-derivative inequality holds for almost every time in the terminal phase, verifying \(\mathrm{SCD}(n\lambda)\) on \((10,12)\).

\noindent\textbf{Step 1.5.} The interval \((12,26)\) is a bottleneck window.

Take \(B=\{3\}\). Then \(m=1\), \(\bar B=\{1,2,4\}\), and every \(\eta\in\mathcal P(B)\) has server 3 as the longest queue. When server 3 is longest, Policy A dispatches arrivals to actual servers \(1,2,3,4\) with probabilities \((0,0,0.25,0.75)\). Since \(T=1\) and \(\gamma_i=1\), the normalized drift of the candidate bottleneck set is
\[
    \underline{\delta}(\{3\},n\lambda)=0.25\,n\lambda-3,
\]
which is strictly positive for \(n\lambda>12\).

For the complement, the individual drifts of servers \(1,2,4\) under the same longest-server-3 regime are \(-1\), \(-2\), and \(0.75\,n\lambda-20\), respectively. As before, the normalized cumulative drift of any nonempty prefix of \(\bar B\) is the average of the individual drifts in that prefix, and hence is no larger than the largest individual drift. On \((12,26)\), the bottleneck drift \(0.25\,n\lambda-3\) is larger than \(-1\) and \(-2\), and it is larger than \(0.75\,n\lambda-20\) because \(n\lambda<34\). Therefore
\[
    \underline{\delta}(\{3\},n\lambda) > \max\{\overline{\delta}(\{1,2,4\},n\lambda),0\}
\]
for every \(n\lambda\in(12,26)\). Hence \(\mathrm{PBD}(n\lambda)\) holds with \(B=\{3\}\) on \((12,26)\).

Combining Steps 1--5 gives
\[
    \mathcal I_{\mathrm T}^{A} = (5,20/3)\cup(8,10)\cup(12,26),
    \qquad
    \mathcal I_{\mathrm S}^{A} = (20/3,8)\cup(10,12).
\]

\noindent\textbf{Step 2.} Policy B has no stable window above \(h^*\).

Under Policy B, if server 1 is longest, it receives fraction \(0.20\) of the load and the remaining arrivals are dispatched to the shortest queue. Thus server 1 has drift \(0.20\,n\lambda-1\), which is positive for \(n\lambda>5\). The first two prefixes of the complement receive no arrivals, and the full complement has normalized drift \((0.80\,n\lambda-25)/3<0\) for all \(n\lambda<26\). Therefore the complement cannot catch server 1 before the total-capacity boundary. Hence server 1 forms a persistent bottleneck throughout \((5,26)\), so \(\mathcal I_{\mathrm T}^{B}=(5,26)\) and \(\mathcal I_{\mathrm S}^{B}=\varnothing\).

\section{Proof of Corollary~\ref{cor_our:rank_invariant_pbd}}
\label{prf:rank_invariant_pbd}

If \(h^*=\sum_{i=1}^n\mu_i\), then the interval \((h^*,\sum_{i=1}^n\mu_i)\) is empty. Hence Assumption~\ref{assump:persistent_bottleneck_dominance} holds trivially. Therefore, we consider the case \(h^*<\sum_{i=1}^n\mu_i\).

Suppose that the policy admits order-based dispatch fractions, i.e., there exists a vector \(\mathbf p=(p_1,\ldots,p_n)\), with \(p_l\ge0\) and \(\sum_{l=1}^n p_l=1\), such that \(f_{l,\eta}=p_l\) for all \(l\in[n]\) and \(\eta\in\mathcal S_n\). Recall that, by the model convention in Section~\ref{subsec:DTQS}, the service rates are ordered as \(\mu_1\le\mu_2\le\cdots\le\mu_n\).

For \(m \in \{ 0,1,...,n\} \), define \(M_m:=\sum_{q=1}^{m}\mu_q\) and \(F_m:=\sum_{l=1}^{m}p_l\), with \(M_0=F_0=0\). Since \(\sum_{l=1}^m f_{l,\eta}=F_m\) does not depend on \(\eta\), and since the minimum service sum over any \(m\)-server subset is achieved by the \(m\) slowest servers \(\{1,\ldots,m\}\), we have
\[
    h^* = \min_{m\in\{1,...,n\}:F_m>0} \frac{M_m}{F_m},
\]
where the convention \(x/0=\infty\) is used.

Fix any \(n\lambda\in(h^*,\sum_{i=1}^n\mu_i)\), and define \(D_m:=M_m-n\lambda F_m\) for \(m \in \{ 0,1,...,n\} \). Since \(n\lambda>h^*\), there exists some \(m\) such that \(D_m<0\). Also, \(D_0=0\) and \(D_n=\sum_{i=1}^n\mu_i-n\lambda>0\). Let \(m^\lambda\) be the smallest minimizer of \(D_m\) over \(m \in \{ 0,1,...,n\} \). Then \(1\le m^\lambda<n\). Set \(B^\lambda:=\{1,\ldots,m^\lambda\}\) and \(\bar B^\lambda:=[n]\setminus B^\lambda\).

We verify Assumption~\ref{assump:persistent_bottleneck_dominance}. First fix \(\eta\in\mathcal P(B^\lambda)\) and \(a\in\{1,...,m^\lambda\}\). The suffix \(\{\eta(m^\lambda-a+1),\ldots,\eta(m^\lambda)\}\) contains \(a\) servers from \(B^\lambda\). Since \(B^\lambda\) consists of the \(m^\lambda\) slowest servers, the largest possible service sum over such an \(a\)-server subset is \(M_{m^\lambda}-M_{m^\lambda-a}\). Therefore,
\begin{align*}
    n\lambda\sum_{l=m^\lambda-a+1}^{m^\lambda}f_{l,\eta} - \sum_{l=m^\lambda-a+1}^{m^\lambda}\mu_{\eta(l)} & = n\lambda(F_{m^\lambda}-F_{m^\lambda-a}) - \sum_{l=m^\lambda-a+1}^{m^\lambda}\mu_{\eta(l)} \\
    & \ge n\lambda(F_{m^\lambda}-F_{m^\lambda-a}) - (M_{m^\lambda}-M_{m^\lambda-a}) \\
    & = D_{m^\lambda-a}-D_{m^\lambda} \\
    & \overset{(a)}{>} 0.
\end{align*}
where (a) follows from the fact that \(m^\lambda\) is the smallest minimizer of \(D_m\) and \(D_{m^\lambda-a}>D_{m^\lambda}\) for every \(a \in \{ 1,...,m^\lambda \} \). Since there are only finitely many choices of \(\eta\) and \(a\), and since all \(\gamma_i>0\), it follows that
\[
    \underline{\delta}(B^\lambda,n\lambda)>0.
\]

Next fix \(\eta\in\mathcal P(B^\lambda)\) and \(b\in\{1,...,n-m^\lambda\}\). The prefix \(\{\eta(m^\lambda+1),\ldots,\eta(m^\lambda+b)\}\) contains \(b\) servers from \(\bar B^\lambda=\{m^\lambda+1,\ldots,n\}\). Since the service rates are ordered in increasing order, the smallest possible service sum over such a \(b\)-server subset is \(M_{m^\lambda+b}-M_{m^\lambda}\). Therefore,
\begin{align*}
    n\lambda\sum_{l=m^\lambda+1}^{m^\lambda+b}f_{l,\eta} -
    \sum_{l=m^\lambda+1}^{m^\lambda+b}\mu_{\eta(l)}
    & = n\lambda(F_{m^\lambda+b}-F_{m^\lambda}) - \sum_{l=m^\lambda+1}^{m^\lambda+b}\mu_{\eta(l)} \\
    & \le n\lambda(F_{m^\lambda+b}-F_{m^\lambda}) - (M_{m^\lambda+b}-M_{m^\lambda}) \\
    & = D_{m^\lambda}-D_{m^\lambda+b} \\
    & \overset{(a)}{\leq} 0,
\end{align*}
where (a) follows because \(D_{m^\lambda}\) is a global minimum. Hence
\[
    \overline{\delta}(\bar B^\lambda,n\lambda)\le0.
\]

Combining the two bounds gives
\[
    \underline{\delta}(B^\lambda,n\lambda)>0 \ge \overline{\delta}(\bar B^\lambda,n\lambda).
\]
Therefore,
\[
    \underline{\delta}(B^\lambda,n\lambda) > \max\left\{ \overline{\delta}(\bar B^\lambda,n\lambda),0 \right\}.
\]

\section{Proof of Theorem~\ref{thm_our:stable_withrate_withratio_necessary_general}}
\label{prf:stable_withrate_withratio_necessary_general}

We prove the theorem by considering the two cases separately.

\noindent\textbf{Step 1.} Consider the case
\[
    h^*<n\lambda<\sum_{i=1}^n\mu_i.
\]
Since policy \(\pi\in\Pi\) satisfies Assumption~\ref{assump:persistent_bottleneck_dominance},
\(\mathrm{PBD}(n\lambda)\) holds at this load. Therefore, by Part~1 of Theorem~\ref{thm_our:fixed_load_stability_certificates}, the induced Markov chain is transient from some initial state. Under the state-space convention, this gives transience on the communicating class under consideration.

\noindent\textbf{Step 2.} Consider the case
\[
    n\lambda>\sum_{i=1}^n\mu_i.
\]
In this case, the total arrival rate exceeds the total service capacity. Therefore, even if all potential service were fully used, the total queue length has positive linear drift by the strong law of large numbers. Hence the total queue length diverges, and the Markov chain is transient.

\section{Proof of Theorem~\ref{thm_our:SSC_withrate}} \label{prf:thm_SSC_withrate_withratio}
The proof of Theorem~\ref{thm_our:SSC_withrate} relies on a multi-step drift analysis based on Theorem~\ref{thm:hajek_drift}. Specifically, we verify the bounded-increment condition and the negative-drift condition required by Theorem~\ref{thm:hajek_drift}. \\

Using the auxiliary notation in Section~\ref{subsec:auxiliary_notation}, for each $l\in[n]$,
\[
    O_{\perp,l}(t) = \sqrt{\gamma_l} \left( \frac{Q_l(t)}{\gamma_l} - \frac{\|\mathbf Q(t)\|_1}{\|\boldsymbol\gamma\|_1} \right).
\]
Consequently, we have
\[
    \|\mathbf O_\perp(t)\|_2^2 = \sum_{l=1}^n \gamma_l \left( \frac{Q_l(t)}{\gamma_l} - \frac{\|\mathbf Q(t)\|_1}{\|\boldsymbol\gamma\|_1} \right)^2 = \|\mathbf Q^{(\gamma)}_\perp(t)\|_\gamma^2.
\]
Thus, it suffices to prove
\[
    \mathbb E[\|\mathbf O_\perp(t)\|_2^2]\le N_\perp^2(n,T).
\]

We define two Lyapunov functions for any $\sqrt{\boldsymbol{\gamma}}$-scaled queue-length vector $\mathbf{o} \in\mathbb R_+^n$.
\begin{align*}
    V_{\perp}(\mathbf{o}) &= \| \mathbf{o}_{\perp} \|_2, \\
    W_{\parallel}(\mathbf{o}) &= \| \mathbf{o}_{\parallel} \|_2^2.
\end{align*}
Let $t$ be a multiple of $T$, and define $t^{\prime} = t/T$. For the embedded chain, define the drift operator by
\begin{align}
    \Delta V_{\perp}(\tilde{\mathbf{O}}(t^{\prime})) & := V_{\perp}(\tilde{\mathbf{O}}(t^{\prime}+1))-V_{\perp}(\tilde{\mathbf{O}}(t^{\prime})) \notag \\
    \Delta W_{\parallel}(\tilde{\mathbf{O}}(t^{\prime})) & := W_{\parallel}(\tilde{\mathbf{O}}(t^{\prime}+1))-W_{\parallel}(\tilde{\mathbf{O}}(t^{\prime})) \notag 
\end{align} 

\noindent \textbf{Step 1}. We show that the drift of $V_{\perp}(\tilde{\mathbf{O}}(t'))$ is uniformly bounded in absolute value, which implies the bounded-increment condition B2 in Theorem~\ref{thm:hajek_drift}.

\begin{lemma} \label{lem_our:bd_absolute_v_perp_O}
    For any policy $\pi \in \Pi$, we have:
    \begin{align}
        | \Delta V_{\perp}(\tilde{\mathbf{O}}(t')) | \leq \frac{2}{\sqrt{\gamma_{\min}}} T ( nA_{\max} + \sqrt{n} S_{\max} ), \qquad \forall \, \tilde{\mathbf{O}}(t^{\prime}) \in \mathbb{R}^n_{+}, t^{\prime} \in \mathbb{N} \notag
    \end{align}
\end{lemma}
The proof is given in Appendix \ref{prf:bd_absolute_v_perp_O}.\\

\noindent \textbf{Step 2}. We show that the drift of $V_{\perp}(\tilde{\mathbf{O}}(t'))$ is negative for all $\| \tilde{\mathbf{O}}_{\perp}(t^{\prime}) \|_2 $ large enough, which verifies the negative drift Condition B1 in Theorem~\ref{thm:hajek_drift}.

\begin{lemma} \label{lem_our:bd_negative_v_perp_O}
    For any policy $\pi \in \Pi$, suppose that the condition in Equation~\eqref{eq:strictlyless_withrate} holds. Then there exist positive constants $\delta^*>0$ and $\xi^*>0$ such that, if
    \[
        \epsilon = \sum\limits_{l=1}^n \mu_l - n\lambda 
        < \frac{\delta^* \| \boldsymbol{\mu} \|_1}{\xi^*},
    \]
    then there exists a positive constant $K < \infty$ such that, for any $\sqrt{\boldsymbol{\gamma}}$-scaled queue-length vector $\mathbf{o}$ with corresponding queue-length vector $\mathbf{q}$, if
    \begin{align*}
        \| \mathbf{q} \|_1 \geq \frac{T S_{\max}\| \boldsymbol{\gamma} \|_1}{\gamma_{\min}},
    \end{align*}
    then
    \begin{align}
        \mathbb{E} \left [ \Delta V_{\perp}(\tilde{\mathbf{O}}(t'))\  |\ \tilde{\mathbf{O}}(t')=\mathbf{o} \right ]
        \leq
        - \frac{ T \xi^* \left ( \frac{\delta^*}{\xi^*}\| \boldsymbol{\mu} \|_1 - \epsilon \right )}{\sqrt{n\gamma_{\max}}}
        + \frac{K}{2 \| \mathbf{o}_{\perp} \|_2 }.
        \notag
    \end{align}
\end{lemma} 
The proof is given in Appendix \ref{prf:bd_negative_delta_v_perp_O}.\\ 

\noindent \textbf{Step 3}. We show that $ \mathbb{E} \big[ \left\| \mathbf{O}_{\perp} \right\|_2^2 \big]$ is bounded by a constant.

\begin{lemma} \label{lem_our:bd_O_perp}
    Under the assumptions of Lemma \ref{lem_our:bd_negative_v_perp_O}, normalize the scaling vector without loss of generality so that \(\|\boldsymbol{\gamma}\|_1=1\). Then there exists a function $N_{\perp}(n,T)$, independent of the slack $\epsilon$, such that
    \[ \mathbb{E} \left[ \left\| \mathbf{O}_{\perp} \right\|_2^2 \right ] \leq N_{\perp}(n,T)^2, \]
    for all $T \ge 1$. For fixed $n$, $N_{\perp}(n,T) \in \Theta(T)$.
\end{lemma}
The proof is given in Appendix \ref{prf:bd_O_perp}, and the order analysis of $N_{\perp}(n,T)$ is given in Appendix \ref{dis:order_N_2}. \\

\noindent \textbf{Step 4}. We translate the bound back to the original state-space collapse statement.

By the equivalence shown at the beginning of the proof, Lemma~\ref{lem_our:bd_O_perp} implies
\[
    \mathbb E\left[ \|\mathbf Q^{(\gamma)}_\perp\|_\gamma^2 \right] = \mathbb E\left[ \|\mathbf O_\perp\|_2^2 \right] \leq N_{\perp}^2(n,T).
\]
The bound \(N_\perp(n,T)\) is independent of the slack \(\epsilon\), and for fixed \(n\), \(N_\perp(n,T)\in\Theta(T)\). This proves Theorem~\ref{thm_our:SSC_withrate}.

\subsection{Proof of Lemma~\ref{lem_our:bd_absolute_v_perp_O}} \label{prf:bd_absolute_v_perp_O}

We show that $\Delta V_{\perp}(\tilde{\mathbf{O}}(t'))$ is uniformly absolutely bounded as follows.
\begin{align}
    | \Delta V_{\perp}(\tilde{\mathbf{O}}(t')) | & = | \| \tilde{\mathbf{O}}_{\perp}(t'+1) \|_2 - \| \tilde{\mathbf{O}}_{\perp}(t') \|_2 | \notag \\
    & \overset{(a)}{\leq} \| \tilde{\mathbf{O}}_{\perp}(t'+1) - \tilde{\mathbf{O}}_{\perp}(t') \|_2 \notag \\
    & \overset{(b)}{=} \| \tilde{\mathbf{O}}(t'+1) - \tilde{\mathbf{O}}_{\parallel}(t'+1) - \tilde{\mathbf{O}}(t') + \tilde{\mathbf{O}}_{\parallel}(t')  \|_2 \notag \\
    & \overset{(c)}{\leq} \| \tilde{\mathbf{O}}(t'+1) - \tilde{\mathbf{O}}(t') \|_2 + \| \tilde{\mathbf{O}}_{\parallel}(t'+1) - \tilde{\mathbf{O}}_{\parallel}(t')  \|_2 \notag \\
    & \overset{(d)}{\leq} 2 \| \tilde{\mathbf{O}}(t'+1) - \tilde{\mathbf{O}}(t') \|_2 \notag \\
    & \overset{(e)}{\leq} \frac{2}{\sqrt{\gamma_{\min}}} T(nA_{\max}+\sqrt{n}S_{\max}), \notag
\end{align}
where (a) follows from the fact that $ | \| x \|_2 - \| y \|_2 | \leq \| x-y \|_2$ for $ x,y\in\mathbb{R}^n $, (b) follows from the fact that $ \mathbf{O}_{\perp}(t) + \mathbf{O}_{\parallel}(t) = \mathbf{O}(t) $, (c) follows from the triangle inequality, (d) follows from the fact that $\mathbf{O}_{\parallel}(t)$ is the projection of $\mathbf{O}(t)$ onto $c$, which implies that $ \| \mathbf{O}_{\parallel}(t'+1) - \mathbf{O}_{\parallel}(t') \|_2 \leq \| \mathbf{O}(t'+1) - \mathbf{O}(t') \|_2 $, and (e) follows in the same way as in the derivation of Equation~\eqref{eq:absolutebound_withrate}.

\subsection{Proof of Lemma~\ref{lem_our:bd_negative_v_perp_O}} \label{prf:bd_negative_delta_v_perp_O}
We first show that $ \Delta V_{\perp}(\tilde{\mathbf{O}}(t')) $ can be bounded by a function of $ \Delta W(\tilde{\mathbf{O}}(t')) $ and $ \Delta W_{\parallel}(\tilde{\mathbf{O}}(t')) $.
\begin{align}
    \Delta V_{\perp}(\tilde{\mathbf{O}}(t')) & = \| \tilde{\mathbf{O}}_{\perp}(t'+1) \|_2 - \| \tilde{\mathbf{O}}_{\perp}(t') \|_2 \notag \\
    & = \sqrt{\| \tilde{\mathbf{O}}_{\perp}(t'+1) \|_2^2} - \sqrt{\| \tilde{\mathbf{O}}_{\perp}(t') \|_2^2} \notag \\
    & \overset{(a)}{\leq} \frac{1}{2\| \tilde{\mathbf{O}}_{\perp}(t') \|_2}\left ( \| \tilde{\mathbf{O}}_{\perp}(t'+1) \|_2^2 - \| \tilde{\mathbf{O}}_{\perp}(t') \|_2^2 \right ) \notag \\
    & \overset{(b)}{=} \frac{1}{2\| \tilde{\mathbf{O}}_{\perp}(t') \|_2}\left ( \| \tilde{\mathbf{O}}(t'+1) \|_2^2 - \| \tilde{\mathbf{O}}_{\parallel}(t'+1) \|_2^2 - \| \tilde{\mathbf{O}}(t') \|_2^2 + \| \tilde{\mathbf{O}}_{\parallel}(t') \|_2^2 \right ) \notag \\
    & = \frac{1}{2\| \tilde{\mathbf{O}}_{\perp}(t') \|_2} \Big[\Delta W(\tilde{\mathbf{O}}(t')) -\Delta W_{\parallel}(\tilde{\mathbf{O}}(t')) \Big], \label{eq:difference_drifts_withrate}
\end{align}
where (a) follows from the fact that $ g(x)=\sqrt{x} $ is concave for $ x \geq 0 $, and thus $ g(y) - g(x)\leq (y-x)g^{\prime}(x) $, and (b) follows from the Pythagorean theorem.

To bound $\mathbb{E}[\Delta V_{\perp}(\tilde{\mathbf{Q}}(t))]$, we show that it is negative whenever $\|\mathbf{o}_{\perp}\|_2$ is sufficiently large. To this end, we express it in terms of the drifts of $W$ and $W_{\perp}$ and bound these two terms separately. This yields the desired negative drift.

First, as in the proof of Lemma~\ref{lem_our:bd_negative_delta_vO} (cf. Equation~\eqref{eq:stable_aveofwx_withrate}), we assume that $\| \tilde{\mathbf{Q}}(t') \|_1 \geq TS_{\max}\| \boldsymbol{\gamma}\|_1/\gamma_{\min}$. Then we have
\begin{align}
    & E\left[\Delta W\left(\tilde{\mathbf{O}}(t')\right)\ | \ \tilde{\mathbf{O}}(t')= \mathbf{o} \right] \notag \\ 
    & \leq 2 \frac{q_{\eta_{t'}(1)}}{\gamma_{\eta_{t'}(1)}} T \left ( n\lambda f_{1,\eta_{t'}} - \mu_{\eta_{t'}(1)} \right )+ \frac{ T^2 \left ( n\lambda f_{1,\eta_{t'}}-\mu_{\eta_{t'}(1)} \right )^2 + Tf_{1,\eta_{t'}}n\sigma_{\lambda}^2 + T^2 n^2\lambda^2 \tau_{1,\eta_{t'}}^2 +T\sigma_{\eta_{t'}(1)}^2 }{\gamma_{\eta_{t'}(1)}} \notag \\
    & \quad + \sum\limits_{ l \in I_2 } \left [ 2\frac{q_{\eta_{t'}(l)}}{\gamma_{\eta_{t'}(l)}} T \left ( n\lambda f_{l,\eta_{t'}} - \mu_{\eta_{t'}(l)} \right )+ \frac{T^2 \left ( n\lambda f_{l,\eta_{t'}}-\mu_{\eta_{t'}(l)} \right )^2 + Tf_{l,\eta_{t'}}n\sigma_{\lambda}^2 + T^2 n^2\lambda^2 \tau_{l,\eta_{t'}}^2 +T\sigma_{\eta_{t'}(l)}^2 }{\gamma_{\eta_{t'}(l)}} \right ] \notag \\
    & \quad + \sum\limits_{ l\in I_3 } \left [ 2\frac{q_{\eta_{t'}(l)}}{\gamma_{\eta_{t'}(l)}} T n\lambda f_{l,\eta_{t'}} + \frac{T^2 n^2\lambda^2 f^2_{l,\eta_{t'}} + Tf_{l,\eta_{t'}}n\sigma_{\lambda}^2 + T^2 n^2\lambda^2 \tau_{l,\eta_{t'}}^2}{\gamma_{\eta_{t'}(l)}} \right ] \notag \\
    & = 2T \left ( n\lambda f_{1,\eta_{t'}} - \mu_{\eta_{t'}(1)} \right )\frac{q_{\eta_{t'}(1)}}{\gamma_{\eta_{t'}(1)}} + \sum\limits_{ l \in I_2 } 2T \left ( n\lambda f_{l,\eta_{t'}} - \mu_{\eta_{t'}(l)} \right )\frac{q_{\eta_{t'}(l)}}{\gamma_{\eta_{t'}(l)}} + \sum\limits_{ l \in I_3 } 2T n\lambda f_{l,\eta_{t'}} \frac{q_{\eta_{t'}(l)}}{\gamma_{\eta_{t'}(l)}} \notag \\
    & \quad + \frac{T^2 \left ( n\lambda f_{1,\eta_{t'}}-\mu_{\eta_{t'}(1)} \right )^2 + Tf_{1,\eta_{t'}}n\sigma_{\lambda}^2 + T^2 n^2\lambda^2 \tau_{1,\eta_{t'}}^2 +T\sigma_{\eta_{t'}(1)}^2}{\gamma_{\eta_{t'}(1)}} \notag \\
    & \quad + \sum\limits_{ l \in I_2 } \left [ \frac{T^2 \left ( n\lambda f_{l,\eta_{t'}}-\mu_{\eta_{t'}(l)} \right )^2 + Tf_{l,\eta_{t'}}n\sigma_{\lambda}^2 + T^2 n^2\lambda^2 \tau_{l,\eta_{t'}}^2 +T\sigma_{\eta_{t'}(l)}^2}{\gamma_{\eta_{t'}(l)}} \right ] \notag \\
    & \quad + \sum\limits_{ l\in I_3 } \left [ \frac{ T^2 n^2\lambda^2 f^2_{l,\eta_{t'}} + Tf_{l,\eta_{t'}}n\sigma_{\lambda}^2 + T^2 n^2\lambda^2 \tau_{l,\eta_{t'}}^2 }{\gamma_{\eta_{t'}(l)}} \right ] \notag \\
    & \leq 2T \left ( n\lambda f_{1,\eta_{t'}} - \mu_{\eta_{t'}(1)} \right )\frac{q_{\eta_{t'}(1)}}{\gamma_{\eta_{t'}(1)}} + \sum\limits_{ l \in I_2 } 2T \left ( n\lambda f_{l,\eta_{t'}} - \mu_{\eta_{t'}(l)} \right )\frac{q_{\eta_{t'}(l)}}{\gamma_{\eta_{t'}(l)}} + \sum\limits_{ l \in I_3 } 2T n\lambda f_{l,\eta_{t'}} \frac{q_{\eta_{t'}(l)}}{\gamma_{\eta_{t'}(l)}} \notag \\    
    & \quad + \frac{T^2 \left ( n\lambda f_{ \max } + \mu_{\max} \right )^2 + Tf_{ \max }n\sigma_{\lambda}^2 + T^2 n^2\lambda^2 \tau_{ \max }^2 + T\sigma_{\max}^2}{\gamma_{\min}} \notag \\
    & \quad + \frac{(n-1)}{\gamma_{\min}} \left [ T^2 \left ( n\lambda f_{\max} + \mu_{\max} \right )^2 + Tf_{ \max } n\sigma_{\lambda}^2 + T^2 n^2\lambda^2 \tau_{ \max }^2 +T\sigma_{\max}^2 \right ] \notag \\
    & \quad + \frac{(n-1)}{\gamma_{\min}} \left [ T^2 n^2\lambda^2 f^2_{\max} + Tf_{ \max }n\sigma_{\lambda}^2 + T^2 n^2\lambda^2 \tau_{ \max }^2 \right ]
    \label{eq:ssc_aveofssc_withrate}
\end{align}
The queue-length dependent terms in Equation~\eqref{eq:ssc_aveofssc_withrate} can be bounded as follows:
\begin{align}
    & 2T \left ( n\lambda f_{1,\eta_{t'}} - \mu_{\eta_{t'}(1)} \right )\frac{q_{\eta_{t'}(1)}}{\gamma_{\eta_{t'}(1)}} + \sum\limits_{ l \in I_2 } 2T \left ( n\lambda f_{l,\eta_{t'}} - \mu_{\eta_{t'}(l)} \right )\frac{q_{\eta_{t'}(l)}}{\gamma_{\eta_{t'}(l)}} + \sum\limits_{ l \in I_3 } 2T n\lambda f_{l,\eta_{t'}} \frac{q_{\eta_{t'}(l)}}{\gamma_{\eta_{t'}(l)}} \notag \\
    & = 2T \left ( n\lambda f_{1,\eta_{t'}} - \mu_{\eta_{t'}(1)} \right )\frac{q_{\eta_{t'}(1)}}{\gamma_{\eta_{t'}(1)}} + 2T\sum\limits_{ l \in I_2 } \left ( n\lambda f_{l,\eta_{t'}}- \mu_{\eta_{t'}(l)} \right )\frac{q_{\eta_{t'}(l)}}{\gamma_{\eta_{t'}(l)}} + 2T\sum\limits_{ l \in I_3 } \left ( n\lambda f_{l,\eta_{t'}} - \mu_{\eta_{t'}(l)} \right )\frac{q_{\eta_{t'}(l)}}{\gamma_{\eta_{t'}(l)}} \notag \\
    & \quad + 2T\sum\limits_{ l \in I_3 } \mu_{\eta_{t'}(l)} \frac{q_{\eta_{t'}(l)}}{\gamma_{\eta_{t'}(l)}} \notag \\
    & \overset{(a)}{=} 2T\sum\limits_{ l = 1}^n \left ( n\lambda f_{l,\eta_{t'}} - \mu_{\eta_{t'}(l)} \right )\frac{q_{\eta_{t'}(l)}}{\gamma_{\eta_{t'}(l)}} + 2T\sum\limits_{ l \in I_3 } \mu_{\eta_{t'}(l)} \frac{q_{\eta_{t'}(l)}}{\gamma_{\eta_{t'}(l)}} \notag \\
    & = -2T \epsilon \frac{q_{\eta_{t'}(1)}}{\gamma_{\eta_{t'}(1)}} + 2T\left ( n\lambda f_{1,\eta_{t'}} - \mu_{\eta_{t'}(1)} + \epsilon \right )\frac{q_{\eta_{t'}(1)}}{\gamma_{\eta_{t'}(1)}} + 2T \sum\limits_{ l = 2}^n \left ( n\lambda f_{l,\eta_{t'}} - \mu_{\eta_{t'}(l)} \right )\frac{q_{\eta_{t'}(l)}}{\gamma_{\eta_{t'}(l)}} + 2T\sum\limits_{ l \in I_3 } \mu_{\eta_{t'}(l)} \frac{q_{\eta_{t'}(l)}}{\gamma_{\eta_{t'}(l)}} \label{eq:SSC_defbeta}
\end{align} 
where (a) follows from the fact that $|I_2 \cup I_3| = n-1 $. Then we denote $\beta_{l,\eta_{t'}}$ as follows:
\begin{align*}
    \beta_{l,\eta_{t'}} = \left \{ \begin{array}{ll}
       n\lambda f_{1, \eta_{t'}} - \mu_{\eta_{t'}(1)} + \epsilon,  & \quad l=1 \\
       n\lambda f_{l, \eta_{t'}} - \mu_{\eta_{t'}(l)},  & \quad l \in [n]\setminus \{ 1\}
    \end{array} \right.
\end{align*}
Equation~\eqref{eq:SSC_defbeta} can be arranged as follows:
\begin{align}
    & -2T \epsilon \frac{q_{\eta_{t'}(1)}}{\gamma_{\eta_{t'}(1)}} + 2T\left ( n\lambda f_{1,\eta_{t'}} - \mu_{\eta_{t'}(1)} + \epsilon \right )\frac{q_{\eta_{t'}(1)}}{\gamma_{\eta_{t'}(1)}} + 2T \sum\limits_{ l = 2}^n \left ( n\lambda f_{l,\eta_{t'}} - \mu_{\eta_{t'}(l)} \right )\frac{q_{\eta_{t'}(l)}}{\gamma_{\eta_{t'}(l)}} + 2T\sum\limits_{ l \in I_3 } \mu_{\eta_{t'}(l)} \frac{q_{\eta_{t'}(l)}}{\gamma_{\eta_{t'}(l)}} \notag \\
    & = -2T \epsilon \frac{q_{\eta_{t'}(1)}}{\gamma_{\eta_{t'}(1)}} + 2T \sum\limits_{l=1}^{n} \beta_{l,\eta_{t'}} \frac{q_{\eta_{t'}(l)}}{\gamma_{\eta_{t'}(l)}} + 2T\sum\limits_{ l \in I_3 } \mu_{\eta_{t'}(l)} \frac{q_{\eta_{t'}(l)}}{\gamma_{\eta_{t'}(l)}} \notag \\
    & = -2T \epsilon \frac{q_{\eta_{t'}(1)}}{\gamma_{\eta_{t'}(1)}} + 2T \beta_{1,\eta_{t'}} \left (\frac{q_{\eta_{t'}(1)}}{\gamma_{\eta_{t'}(1)}} - \frac{q_{\eta_{t'}(2)}}{\gamma_{\eta_{t'}(2)}} \right ) + 2T \left ( \sum\limits_{r=1}^2 \beta_{r,\eta_{t'}} \right ) \left (\frac{q_{\eta_{t'}(2)}}{\gamma_{\eta_{t'}(2)}} - \frac{q_{\eta_{t'}(3)}}{\gamma_{\eta_{t'}(3)}} \right ) \notag \\ 
    & \quad + ... + 2T \left ( \sum\limits_{r=1}^{n-1} \beta_{r,\eta_{t'}} \right ) \left (\frac{q_{\eta_{t'}(n-1)}}{\gamma_{\eta_{t'}(n-1)}} - \frac{q_{\eta_{t'}(n)}}{\gamma_{\eta_{t'}(n)}} \right ) + 2T \sum_{r=1}^n \beta_{r,\eta_{t'}} \frac{q_{\eta_{t'}(n)}}{\gamma_{\eta_{t'}(n)}} + 2T \sum\limits_{l \in I_3} \mu_{\eta_{t'}(l)}\frac{q_{\eta_{t'}(l)}}{\gamma_{\eta_{t'}(l)}} \notag \\
    & = -2T \epsilon \frac{q_{\eta_{t'}(1)}}{\gamma_{\eta_{t'}(1)}} + 2T \sum\limits_{l=1}^{n-1} \left [ \left ( \sum\limits_{r=1}^l \beta_{r,\eta_{t'}} \right ) \left (\frac{q_{\eta_{t'}(l)}}{\gamma_{\eta_{t'}(l)}} - \frac{q_{\eta_{t'}(l+1)}}{\gamma_{\eta_{t'}(l+1)}} \right ) \right ] + 2T \sum_{r=1}^n \beta_{r,\eta_{t'}} \frac{q_{\eta_{t'}(n)}}{\gamma_{\eta_{t'}(n)}} + 2T \sum\limits_{l \in I_3} \mu_{\eta_{t'}(l)}\frac{q_{\eta_{t'}(l)}}{\gamma_{\eta_{t'}(l)}} \notag \\
    & \overset{(a)}{=} -2T \epsilon \frac{q_{\eta_{t'}(1)}}{\gamma_{\eta_{t'}(1)}} + 2T \sum\limits_{l=1}^{n-1} \left [ \left ( \sum\limits_{r=1}^l \beta_{r,\eta_{t'}} \right ) \left (\frac{q_{\eta_{t'}(l)}}{\gamma_{\eta_{t'}(l)}} - \frac{q_{\eta_{t'}(l+1)}}{\gamma_{\eta_{t'}(l+1)}} \right ) \right ] + 2T \sum\limits_{l \in I_3} \mu_{\eta_{t'}(l)}\frac{q_{\eta_{t'}(l)}}{\gamma_{\eta_{t'}(l)}}\label{eq:SSC_sumofbeta}
\end{align}
where (a) follows from the fact that $ \sum_{r=1}^n \beta_{r,\eta_{t'}} = 0 $. 

In Equation~\eqref{eq:SSC_sumofbeta}, $\sum_{r=1}^{l} \beta_{r,\eta_{t'}}$ can be bounded, for all $ l\in [n]\setminus \{ n \}$, as follows:
\begin{align*}
    \sum\limits_{r=1}^l \beta_{r,\eta_{t'}} & = \sum\limits_{r=1}^l \left ( n\lambda f_{r,\eta_{t'}} - \mu_{\eta_{t'}(r)} \right ) + \epsilon \\
    & = \sum\limits_{r=1}^l \left [ \left ( \| \boldsymbol{\mu} \|_1 - \epsilon \right) f_{r,\eta_{t'}} - \mu_{\eta_{t'}(r)} \right ] + \epsilon \\
    &  = \| \boldsymbol{\mu} \|_1 \sum_{r=1}^l f_{r,\eta_{t'}} - \epsilon \sum\limits_{r=1}^l f_{r,\eta_{t'}} - \sum_{r=1}^l \mu_{\eta_{t'}(r)} + \epsilon \\
    & = \| \boldsymbol{\mu} \|_1 \left ( \sum\limits_{r=1}^l f_{r,\eta_{t'}} - \frac{\sum\limits_{r=1}^l \mu_{\eta_{t'}(r)}}{ \| \boldsymbol{\mu} \|_1 } \right ) + \epsilon \left ( 1 - \sum\limits_{r=1}^l f_{r,\eta_{t'}} \right ) \\
    & \overset{(a)}{\leq} - \delta^* \| \boldsymbol{\mu} \|_1 + \epsilon \left ( 1 - \sum\limits_{r=1}^l f_{r,\eta_{t'}} \right )  \notag \\
    & \overset{(b)}{\leq} - \delta^* \| \boldsymbol{\mu} \|_1 + \xi^* \epsilon  \notag \\
    & \leq - \left ( \delta^* \| \boldsymbol{\mu} \|_1 - \xi^* \epsilon \right ) \\
    & \overset{(c)}{<} 0,
\end{align*}
where for (a), we define
\[ \delta^* = \min_{ \{ \eta \in S_n \} } \left \{  \min_{ \{ l \in [n]\setminus \{ n \} \} } \left \{ \sum_{r=1}^l \frac{\mu_{\eta(r)}}{\| \boldsymbol{\mu} \|_1} - \sum_{r=1}^l f_{r,\eta}  \right \} \right \}, \]
which is positive due to Equation~\eqref{eq:strictlyless_withrate}. For (b), we define 
\[ \xi^* = \max_{ \{ \eta \in S_n \} } \left \{ \max_{l \in [n]\setminus \{ n\} } \left \{ 1 - \sum_{r=1}^l f_{r, \eta} \right \} \right \}, \]
which is positive due to Equation~\eqref{eq:strictlyless_withrate}. That is, we have
\[ \sum_{l=1}^m f_{l,\eta}  < \frac{ \sum_{l=1}^m \mu_{\eta (l)} }{ \| \boldsymbol{\mu} \|_1 } < 1 , \forall \, m \in [n]\setminus \{ n \}, \, \eta \in S_n. \]
Finally, (c) holds when $ \epsilon < \delta^* \| \boldsymbol{\mu} \|_1 / \xi^* $. Then, Equation~\eqref{eq:SSC_sumofbeta} is further bounded as follows:
\begin{align}
    & -2T \epsilon \frac{q_{\eta_{t'}(1)}}{\gamma_{\eta_{t'}(1)}} + 2T \sum\limits_{l=1}^{n-1} \left [ \left ( \sum\limits_{r=1}^l \beta_{r,\eta_{t'}} \right ) \left (\frac{q_{\eta_{t'}(l)}}{\gamma_{\eta_{t'}(l)}} - \frac{q_{\eta_{t'}(l+1)}}{\gamma_{\eta_{t'}(l+1)}} \right ) \right ]  + 2T \sum\limits_{l \in I_3} \mu_{\eta_{t'}(l)}\frac{q_{\eta_{t'}(l)}}{\gamma_{\eta_{t'}(l)}} \notag \\
    & \leq -2T \epsilon \frac{q_{\eta_{t'}(1)}}{\gamma_{\eta_{t'}(1)}} - 2T \left ( \delta^* \| \boldsymbol{\mu} \|_1 - \xi^* \epsilon \right ) \sum\limits_{l=1}^{n-1} \left (\frac{q_{\eta_{t'}(l)}}{\gamma_{\eta_{t'}(l)}} - \frac{q_{\eta_{t'}(l+1)}}{\gamma_{\eta_{t'}(l+1)}} \right ) + 2T \sum\limits_{l \in I_3} \mu_{\eta_{t'}(l)}\frac{q_{\eta_{t'}(l)}}{\gamma_{\eta_{t'}(l)}} \notag \\
    & = -2T \epsilon \frac{q_{\eta_{t'}(1)}}{\gamma_{\eta_{t'}(1)}} - 2T \left ( \delta^* \| \boldsymbol{\mu} \|_1 - \xi^* \epsilon \right ) \left ( \frac{q_{\eta_{t'}(1)}}{\gamma_{\eta_{t'}(1)}} - \frac{q_{\eta_{t'}(n)}}{\gamma_{\eta_{t'}(n)}} \right ) + 2T \sum\limits_{l \in I_3} \mu_{\eta_{t'}(l)} \frac{q_{\eta_{t'}(l)}}{\gamma_{\eta_{t'}(l)}} \notag \\
    & \leq -2T \epsilon \frac{q_{\eta_{t'}(1)}}{\gamma_{\eta_{t'}(1)}} - 2T \left ( \delta^* \| \boldsymbol{\mu} \|_1 - \xi^* \epsilon \right ) \left ( \frac{q_{\eta_{t'}(1)}}{\gamma_{\eta_{t'}(1)}} - \frac{q_{\eta_{t'}(n)}}{\gamma_{\eta_{t'}(n)}} \right ) + \frac{2T n \mu_{\max}T S_{\max}}{\gamma_{\min}}
    \label{eq:ssc_case1_withrate}
\end{align}
The first term in Equation~\eqref{eq:ssc_case1_withrate} is bounded as follows:
\begin{align*}
    -2T \epsilon \frac{q_{\eta_{t'}(1)}}{\gamma_{\eta_{t'}(1)}} & \overset{(a)}{\leq} -2T \epsilon \frac{\| \mathbf{q} \|_1}{ \| \boldsymbol{\gamma} \|_1 } \\
    & = -2T\epsilon \frac{1} {\sqrt{ \| \boldsymbol{\gamma} \|_1 }} \sqrt{ \sum_{l=1}^n \left ( \frac{ \| \mathbf{q} \|_1 }{ \| \boldsymbol{\gamma} \|_1 } \sqrt{\gamma_{\eta_{t'}(l)}} \right )^2 } \notag \\
    & = -2T \epsilon \frac{1}{ \sqrt{ \| \boldsymbol{\gamma} \|_1 } } \| \mathbf{o}_{\parallel} \|_2,
\end{align*}
where $(a)$ follows from Equation~\eqref{eq:scaled_q1_greater_norm}.\\

The second term in Equation~\eqref{eq:ssc_case1_withrate} is bounded as follows:
\begin{align*}
    \| \mathbf{o}_{\perp} \|_2^2 & = \sum\limits_{l=1}^n \left ( \frac{q_{\eta_{t'}(l)}}{\sqrt{\gamma_{\eta_{t'}(l)}}} - \frac{ \| \mathbf{q} \|_1 }{ \| \boldsymbol{\gamma} \|_1 } \sqrt{\gamma_{\eta_{t'}(l)}}  \right )^2 \\
    & = \sum\limits_{l=1}^n \gamma_{\eta_{t'}(l)} \left ( \frac{q_{\eta_{t'}(l)}}{\gamma_{\eta_{t'}(l)}} - \frac{ \| \mathbf{q} \|_1 }{ \| \boldsymbol{\gamma} \|_1 }  \right )^2 \\
    & \overset{(a)}{\leq} n \gamma_{\max} \left ( \frac{q_{\eta_{t'}(1)}}{\gamma_{\eta_{t'}(1)}} - \frac{q_{\eta_{t'}(n)}}{\gamma_{\eta_{t'}(n)}} \right )^2,
\end{align*}
where $(a)$ follows from the fact that $ q_{\eta_{t'}(1)}/\gamma_{\eta_{t'}(1)} \geq \| \mathbf{q} \|_1 / \| \boldsymbol{\gamma} \|_1 \geq q_{\eta_{t'}(n)}/\gamma_{\eta_{t'}(n)}  $ when $ q_{\eta_{t'}(1)}/\gamma_{\eta_{t'}(1)} \geq ... \geq q_{\eta_{t'}(n)}/\gamma_{\eta_{t'}(n)} $ . \\

Combining equations~\eqref{eq:ssc_aveofssc_withrate} and \eqref{eq:ssc_case1_withrate}, we get
\begin{align}
    \mathbb{E}[\Delta W(\tilde{\mathbf{O}}(t'))\ | \ \tilde{\mathbf{O}}(t')= \mathbf{o}]  & \leq \frac{-2T \epsilon}{  \sqrt{ \| \boldsymbol{\gamma} \|_1 }} \| \mathbf{o}_{\parallel} \|_2 - \frac{ 2T \left ( \delta^* \| \boldsymbol{\mu} \|_1 - \xi^* \epsilon \right )}{\sqrt{n \gamma_{\max}}} \| \mathbf{o}_{\perp} \|_2 + K_1, \label{eq:ssc_driftofssc_withrates} 
\end{align}
where
\begin{align}
    K_1 & =  \frac{2T^2 n \mu_{\max} S_{\max}}{\gamma_{\min}} \notag \\
    & \quad + \frac{T^2 \left ( n\lambda f_{ \max } + \mu_{\max} \right )^2 + Tf_{ \max }n\sigma_{\lambda}^2 + T^2 n^2\lambda^2 \tau_{\max}^2 +T\sigma_{\max}^2}{\gamma_{\min}} \notag \\
    & \quad + \frac{(n-1)}{\gamma_{\min}} \left [ T^2 \left ( n\lambda f_{\max} + \mu_{\max} \right )^2 + Tf_{\max}n\sigma_{\lambda}^2 + T^2 n^2\lambda^2 \tau_{\max}^2 +T\sigma_{\max}^2 \right ] \notag \\
    & \quad + \frac{(n-1)}{\gamma_{\min}} \left [ T^2 n^2\lambda^2 f^2_{\max} + Tf_{\max}n\sigma_{\lambda}^2 + T^2 n^2\lambda^2 \tau_{\max}^2 \right ]. \notag 
\end{align}

We now lower bound $ \mathbb{E}[\Delta W_{\parallel}(\tilde{\mathbf{O}}(t'))\ | \  \tilde{\mathbf{O}}(t')=\mathbf{o}] $ as follows:
\begin{align}
    & \mathbb{E}[\Delta W_{\parallel}(\tilde{\mathbf{O}}(t'))\ | \ \tilde{\mathbf{O}}(t') = \mathbf{o}] = \mathbb{E}[\| \tilde{\mathbf{O}}_{\parallel}(t'+1) \|_2^2 - \| \tilde{\mathbf{O}}_{\parallel}(t') \|_2^2 \  | \  \tilde{\mathbf{O}}(t')=\mathbf{o} ] \notag \\
    & = \mathbb{E}[\| \mathbf{O}_{\parallel}(t+T) \|_2^2 - \| \mathbf{O}_{\parallel}(t) \|_2^2 \  | \  \mathbf{O}(t)=\mathbf{o} ] \notag \\
    & = \mathbb{E} \left [ \sum_{l=1}^n \left ( \frac{ \| \mathbf{Q}(t+T) \|_1 }{ \| \boldsymbol{\gamma} \|_1 } \sqrt{\gamma_{\eta_{t'}(l)}} \right )^2 - \sum_{l=1}^n \left ( \frac{ \| \mathbf{Q}(t) \|_1 }{ \| \boldsymbol{\gamma} \|_1 } \sqrt{\gamma_{\eta_{t'}(l)}} \right )^2  \ \middle | \  \mathbf{O}(t)=\mathbf{o} \right ] \notag \\
    & = \mathbb{E} \left [ \sum_{l=1}^n \left [ \left ( \frac{ \| \mathbf{Q}(t+T) \|_1 }{ \| \boldsymbol{\gamma} \|_1 } \sqrt{\gamma_{\eta_{t'}(l)}} \right )^2 - \left ( \frac{ \| \mathbf{Q}(t) \|_1 }{ \| \boldsymbol{\gamma} \|_1 } \sqrt{\gamma_{\eta_{t'}(l)}} \right )^2 \right ]  \ \middle | \  \mathbf{O}(t)=\mathbf{o} \right ] \notag \\
    & = \mathbb{E} \left [ \sum_{l=1}^n \frac{ \gamma_{\eta_{t'}(l)} }{ \| \boldsymbol{\gamma} \|_1^2 } \left ( \| \mathbf{Q}(t+T) \|_1^2 - \| \mathbf{Q}(t) \|_1^2  \right )  \ \middle | \  \mathbf{O}(t)=\mathbf{o} \right ] \notag \\
    & = \mathbb{E} \left [ \frac{ 1 }{ \| \boldsymbol{\gamma} \|_1 } \left ( \| \mathbf{Q}(t+T) \|_1^2 - \| \mathbf{Q}(t) \|_1^2  \right )  \ \middle | \  \mathbf{O}(t)=\mathbf{o} \right ] \notag \\
    & = \frac{ 1 }{ \| \boldsymbol{\gamma} \|_1 } \mathbb{E} \left [ \left ( \| \mathbf{Q}(t) \|_1 + \sum\limits_{j=0}^{T-1} \sum\limits_{l=1}^n A_{\eta_{t'}(l)}(t+j) - \sum\limits_{j=0}^{T-1} \sum\limits_{l=1}^n S_{\eta_{t'}(l)}(t+j) + \sum\limits_{j=0}^{T-1} \sum\limits_{l=1}^n U_{\eta_{t'}(l)}(t+j) \right )^2 \right. \notag \\
    & \qquad\qquad - \| \mathbf{Q}(t) \|_1^2 \  \Bigg | \  \mathbf{O}(t)=\mathbf{o} \Bigg ] \notag \\
    & = \frac{ 1 }{ \| \boldsymbol{\gamma} \|_1 } \left \{ \mathbb{E} \left [ \left ( \| \mathbf{Q}(t) \|_1 + \sum\limits_{j=0}^{T-1} \sum\limits_{l=1}^n A_{\eta_{t'}(l)}(t+j) - \sum\limits_{j=0}^{T-1} \sum\limits_{l=1}^n S_{\eta_{t'}(l)}(t+j) \right )^2 - \| \mathbf{Q}(t) \|_1^2 \  \middle | \  \mathbf{O}(t)=\mathbf{o} \right ] \right. \notag \\
    & \qquad\qquad + \mathbb{E} \left [ 2\left ( \| \mathbf{Q}(t) \|_1 + \sum\limits_{j=0}^{T-1} \sum\limits_{l=1}^n A_{\eta_{t'}(l)}(t+j) - \sum\limits_{j=0}^{T-1} \sum\limits_{l=1}^n S_{\eta_{t'}(l)}(t+j) \right )\left ( \sum\limits_{j=0}^{T-1} \sum\limits_{l=1}^n U_{\eta_{t'}(l)}(t+j) \right) \  \middle | \  \mathbf{O}(t)=\mathbf{o} \right ] \notag \\
    & \qquad\qquad \left. + \mathbb{E} \left [ \left ( \sum\limits_{j=0}^{T-1} \sum\limits_{l=1}^n U_{\eta_{t'}(l)}(t+j) \right)^2 \  \middle | \  \mathbf{O}(t)=\mathbf{o} \right ] \right \} \notag \\
    & = \frac{ 1 }{ \| \boldsymbol{\gamma} \|_1 } \left \{ \mathbb{E} \left [ 2 \| \mathbf{Q}(t) \|_1 \left (\sum\limits_{j=0}^{T-1} \sum\limits_{l=1}^n A_{\eta_{t'}(l)}(t+j) - \sum\limits_{j=0}^{T-1} \sum\limits_{l=1}^n S_{\eta_{t'}(l)}(t+j) \right ) \  \middle | \  \mathbf{O}(t)=\mathbf{o} \right ] \right. \notag \\
    & \qquad\qquad + \mathbb{E} \left [ \left (\sum\limits_{j=0}^{T-1} \sum\limits_{l=1}^n A_{\eta_{t'}(l)}(t+j) - \sum\limits_{j=0}^{T-1} \sum\limits_{l=1}^n S_{\eta_{t'}(l)}(t+j) \right ) ^2 \  \middle | \  \mathbf{O}(t)=\mathbf{o} \right ] \notag \\
    & \qquad\qquad + \mathbb{E} \left [ 2 \left ( \| \mathbf{Q}(t) \|_1 + \sum\limits_{j=0}^{T-1} \sum\limits_{l=1}^n A_{\eta_{t'}(l)}(t+j) \right )\left ( \sum\limits_{j=0}^{T-1} \sum\limits_{l=1}^n U_{\eta_{t'}(l)}(t+j) \right) \  \middle | \  \mathbf{O}(t)=\mathbf{o} \right ] \notag \\
    & \qquad\qquad - \mathbb{E} \left [ 2\left ( \sum\limits_{j=0}^{T-1} \sum\limits_{l=1}^n S_{\eta_{t'}(l)}(t+j) \right) \left ( \sum\limits_{j=0}^{T-1} \sum\limits_{l=1}^n U_{\eta_{t'}(l)}(t+j) \right) \  \middle | \  \mathbf{O}(t)=\mathbf{o} \right ] \notag \\
    & \qquad\qquad \left. + \mathbb{E} \left [ \left ( \sum\limits_{j=0}^{T-1} \sum\limits_{l=1}^n U_{\eta_{t'}(l)}(t+j) \right)^2 \  \middle | \  \mathbf{O}(t)=\mathbf{o} \right ] \right \} \notag \\
    & \overset{(a)}{\geq} \frac{ 1 }{ \| \boldsymbol{\gamma} \|_1 } \mathbb{E} \left[ 2 \| \mathbf{Q}(t) \|_1 \left (\sum\limits_{j=0}^{T-1} \sum\limits_{l=1}^n A_{\eta_{t'}(l)}(t+j) - \sum\limits_{j=0}^{T-1} \sum\limits_{l=1}^n S_{\eta_{t'}(l)}(t+j) \right ) \  \middle | \  \mathbf{O}(t)=\mathbf{o} \right ] \notag \\
    & \qquad\qquad - \mathbb{E} \left [ 2\left ( \sum\limits_{j=0}^{T-1} \sum\limits_{l=1}^n S_{\eta_{t'}(l)}(t+j) \right) \left ( \sum\limits_{j=0}^{T-1} \sum\limits_{l=1}^n U_{\eta_{t'}(l)}(t+j) \right) \  \middle | \  \mathbf{O}(t)=\mathbf{o} \right ] \notag \\
    & \overset{(b)}{\geq} 2 \frac{ 1 }{ \| \boldsymbol{\gamma} \|_1 } \left ( T n\lambda -T \| \boldsymbol{\mu} \|_1 \right )\| \mathbf{q} \|_1 - 2 \frac{ 1 }{ \| \boldsymbol{\gamma} \|_1 } TnS_{\max}TnS_{\max} \notag \\
    & = - 2 T \epsilon \frac{ \| \mathbf{q} \|_1 }{ \| \boldsymbol{\gamma} \|_1 } - \frac{ 2T^2n^2 S_{\max}^2 }{ \| \boldsymbol{\gamma} \|_1 } \notag \\
    & \overset{(c)}{=} -2T \epsilon \frac{1}{ \sqrt{ \| \boldsymbol{\gamma} \|_1 } } \| \mathbf{o}_{\parallel} \|_2  - \frac{ 2T^2n^2 S_{\max}^2  }{ \| \boldsymbol{\gamma} \|_1 } \notag \\
    & = \frac{-2T \epsilon}{ \sqrt{ \| \boldsymbol{\gamma} \|_1 } } \| \mathbf{o}_{\parallel} \|_2 - K_2, \label{eq:ssc_dirftofparallel_withraate}
\end{align}
with $K_2 = 2T^2 n^2 S_{\max}^2 / \| \boldsymbol{\gamma} \|_1 $, where (a) follows from removing the positive term, (b) follows from using that $ U_l(t+j) \leq S_{\max}$ for all $l \in [n]$ and $j\geq 0$ and (c) follows from the same transformation in Equation~\eqref{eq:ssc_case1_withrate} for $\| \mathbf{o}_{\parallel} \|_2$. \\

Finally, combining equations~\eqref{eq:difference_drifts_withrate}, \eqref{eq:ssc_driftofssc_withrates}, and \eqref{eq:ssc_dirftofparallel_withraate}, we get
\begin{align}
     & \mathbb{E}[ \Delta V_{\perp}(\tilde{\mathbf{O}}(t'))\  |\ \tilde{\mathbf{O}}(t')=\mathbf{o} ] \notag \\ 
     & \leq \frac{1}{2 \| \mathbf{o}_{\perp} \|_2 } \mathbb{E} \left [\Delta W(\tilde{\mathbf{O}}(t')) - \Delta W_{\parallel}(\tilde{\mathbf{O}}(t')) \ \middle |\ \tilde{\mathbf{O}}(t')= \mathbf{o} \right ] \notag \\ 
     & \leq \frac{1}{2 \| \mathbf{o}_{\perp} \|_2 } \left [ - \frac{2T \epsilon}{  \sqrt{ \| \boldsymbol{\gamma} \|_1 }} \| \mathbf{o}_{\parallel} \|_2 - \frac{ 2T \left ( \delta^* \| \boldsymbol{\mu} \|_1 - \xi^* \epsilon \right )}{\sqrt{n \gamma_{\max}}} \| \mathbf{o}_{\perp} \|_2 + K_1 + \frac{2T \epsilon}{ \sqrt{ \| \boldsymbol{\gamma} \|_1 } } \| \mathbf{o}_{\parallel} \|_2 + K_2 \right] \notag \\
     & = - \frac{ T \left ( \delta^* \| \boldsymbol{\mu} \|_1 - \xi^* \epsilon \right )}{\sqrt{n \gamma_{\max}}} + \frac{K_1 + K_2}{2 \| \mathbf{o}_{\perp} \|_2 } \notag \\
     & = - \frac{ T \xi^* \left ( \frac{\delta^*}{\xi^*} \| \boldsymbol{\mu} \|_1  - \epsilon \right )}{\sqrt{n \gamma_{\max}}} + \frac{K}{2 \| \mathbf{o}_{\perp} \|_2 }. \label{eq:delta_V_perp_O_negative_drift_bound}
\end{align}
where $K:=K_1+K_2$.

\subsection{Proof of Lemma~\ref{lem_our:bd_O_perp}} \label{prf:bd_O_perp}
To prove the lemma, we first establish the corresponding bound for the embedded chain:
\[
    \mathbb E\!\left[\|\tilde{\mathbf O}_\perp\|_2^2\right]
    \le N_\perp^2(n,T).
\]
The passage from the embedded chain to the original steady-state process is handled at the end of the proof. We now apply Theorem~\ref{thm:hajek_drift} to \(V_\perp(\tilde{\mathbf O}(t'))\). Define
\[
H_{t'} := V_{\perp}(\tilde{\mathbf{O}}(t'))=\|\tilde{\mathbf{O}}_{\perp}(t')\|_2,
\]
so that $H_{t'+1}-H_{t'}=\Delta V_{\perp}(\tilde{\mathbf Q}(t'))$.

We first verify Condition B2 by Lemma~\ref{lem_our:bd_absolute_v_perp_O},
\begin{align*}
    \left|H_{t'+1}-H_{t'}\right| = \left|\Delta V_{\perp}(\tilde{\mathbf{O}}(t'))\right| \leq \frac{2}{\sqrt{\gamma_{\min}}} T(nA_{\max}+\sqrt{n}S_{\max})
\end{align*}
for all \(t' \in \mathbb N \). Hence, we choose \(G\) in Theorem~\ref{thm:hajek_drift} to be the deterministic constant
\begin{align*}
    G := \frac{2}{\sqrt{\gamma_{\min}}} T(nA_{\max}+\sqrt{n}S_{\max}).
\end{align*}
Then, for every \(u\ge 0\), $\mathbb P\!\left(\left|H_{t'+1}-H_{t'}\right|>u \mid \mathcal F_{t'}\right)
\le \mathbb P(G>u)$, and, for any $\theta > 0$, $\mathbb E[e^{\theta G}] = e^{\theta G} < \infty$. Therefore, Condition B2 holds.

We next verify Condition B1. Recall from Lemma \ref{lem_our:bd_negative_v_perp_O}
that the negative-drift bound requires $\| \tilde{\mathbf{Q}}(t') \|_1 \geq TS_{\max}\| \boldsymbol{\gamma}\|_1/\gamma_{\min}$. To ensure that this condition holds, we use the following inequality:
\begin{align*}
    \| \tilde{\mathbf{Q}}(t') \|_1
    &\ge \|\tilde{\mathbf{Q}}(t')\|_2\overset{(a)}{\ge} \|\tilde{\mathbf{Q}}_{\perp}(t')\|_2,
\end{align*}
where (a) follows from the orthogonal decomposition $\mathbf{Q}=\mathbf{Q}_{\parallel}+\mathbf{Q}_{\perp}$ and the Pythagorean theorem. Therefore, if $\|\tilde{\mathbf{Q}}_{\perp}(t')\|_2 \ge TS_{\max}\| \boldsymbol{\gamma}\|_1/\gamma_{\min}$, then $ \|\tilde{\mathbf{Q}}(t')\|_1 \ge TS_{\max}\| \boldsymbol{\gamma}\|_1/\gamma_{\min}$, and hence Equation~\eqref{eq:delta_V_perp_O_negative_drift_bound} applies.

Further, since we assume that
\[ \epsilon \leq \frac{ \delta^* \| \boldsymbol{\mu} \|_1 }{ 2\xi^*}, \]
we have
\begin{align}
    \mathbb{E} \left [ \Delta V_{\perp}(\tilde{\mathbf{O}}_{\perp}(t')) + \epsilon_0 \  \middle | \  \| \tilde{\mathbf{O}}_{\perp}(t') \|_2 > a \right ] < 0, \notag
\end{align}
for
\begin{align*}
     a &= \max \left \{ \frac{K \sqrt{n \gamma_{\max}} }{ T \xi^* \Delta  }, \frac{TS_{\max}\|\boldsymbol{\gamma}\|_1}{\gamma_{\min}} \right \} \qquad \text{and} \qquad \epsilon_0 = \frac{ T \xi^* \Delta }{2 \sqrt{n \gamma_{\max}} },
\end{align*}
where
\[ \Delta = \frac{\delta^* \| \boldsymbol{\mu} \|_1 }{2\xi^*}. \]
We next show that $a = K \sqrt{n \gamma_{\max}} /(T \xi^* \Delta) $. Indeed, with $\|\boldsymbol{\gamma}\|_1=1$,
\begin{align*}
    \frac{K\sqrt{n\gamma_{\max}}}{T\xi^*\Delta}
    &\overset{(a)}{\geq}
    \frac{2T^2 n\mu_{\max}S_{\max}\sqrt{n\gamma_{\max}}}
    {T\xi^*\Delta\gamma_{\min}} \\
    &\overset{(b)}{=}
    \frac{4T n\mu_{\max}S_{\max}\sqrt{n\gamma_{\max}}}
    {\delta^*\|\boldsymbol{\mu}\|_1\gamma_{\min}} \\
    &\overset{(c)}{\geq}
    \frac{4T S_{\max}\sqrt{n\gamma_{\max}}}
    {\delta^*\gamma_{\min}} \\
    &\overset{(d)}{\geq}
    \frac{4T S_{\max}}
    {\delta^*\gamma_{\min}} \\
    &\overset{(e)}{\geq}
    \frac{TS_{\max}}{\gamma_{\min}} \\
    & =
    \frac{TS_{\max}\|\boldsymbol{\gamma}\|_1}{\gamma_{\min}}.
\end{align*}
where (a) follows from the definition of $K=K_1+K_2$ and the fact that $K \geq K_1 \geq 2T^2 n\mu_{\max}S_{\max}/\gamma_{\min}$ in Equation~\eqref{eq:ssc_driftofssc_withrates}, (b) follows from the definition $\Delta=\delta^*\|\boldsymbol{\mu}\|_1/(2\xi^*)$, (c) follows from $\mu_{\max}\geq \|\boldsymbol{\mu}\|_1/n$, (d) follows from the normalization $\|\boldsymbol{\gamma}\|_1=1$, which implies $\gamma_{\max}\geq 1/n$ and hence $\sqrt{n\gamma_{\max}}\geq 1$, and (e) follows from $\delta^*\leq 1$.

Finally, we apply theorem \ref{thm:hajek_drift} and obtain
\begin{align*}
    & \mathbb{E} \left [ e^{\zeta\| \tilde{\mathbf{O}}_{\perp}(t')\|_2} \right ] \leq \varrho^te^{\zeta \| \tilde{\mathbf{O}}_{\perp}(0) \|_2} + \frac{1-\varrho^{t'}}{1-\varrho} e^{\theta G} e^{\zeta a}
\end{align*}
Since \(\|\tilde{\mathbf O}_{\perp}(0)\|_2=0\) and \(0<\varrho<1\), passing to the stationary distribution and using the Portmanteau theorem for the nonnegative lower-semicontinuous functions gives
\[
    \mathbb{E}\!\left[e^{\zeta\|\tilde{\mathbf O}_{\perp}\|_2}\right] \le \liminf_{t'\to\infty} \mathbb{E}\!\left[e^{\zeta\|\tilde{\mathbf O}_{\perp}(t')\|_2}\right]
    \le \frac{e^{\theta G}e^{\zeta a}}{1-\varrho}.
\]
Setting $\theta = 1/G$ and defining
\begin{align*}
    C_H & := \frac{ \mathbb{E} \left [e^{\theta G} \right] - (1+\theta \mathbb{E}[G]) }{ \theta^2 } \notag \\
    & = G^2(e-2),
\end{align*}
we get
\begin{align}
    & \mathbb{E} \left [ e^{\zeta\| \tilde{\mathbf{Q}}_{\perp}\|_2} \right ] \leq \frac{1}{1-\varrho}e^{\zeta a + 1}, \label{eq:exp_bound}
\end{align}
where
\begin{align}
    \zeta & = \min \left\{ \theta, \frac{\epsilon_0}{2 C_H} ,\frac{1}{a} \right \} \notag \\
    & = \min\left\{ \frac{1}{G}, \frac{T \xi^* \Delta}{4 \sqrt{n \gamma_{\max} } G^2(e-2)}, \frac{T \xi^* \Delta}{K \sqrt{n \gamma_{\max}}} \right\} \notag \\
    \varrho & = 1 - \epsilon_0 \zeta + C_H \zeta^2 \notag \\  
    & = 1 - \frac{T \xi^* \Delta}{2 \sqrt{n \gamma_{\max}} } \zeta + G^2(e-2)\zeta^2, \notag
\end{align}
with
\[ G = \frac{2}{\sqrt{\gamma_{\min}}} T(nA_{\max}+\sqrt{n}S_{\max}). \]
Taking the Taylor expansion for the left-hand side in Equation~\eqref{eq:exp_bound} yields
\[ \mathbb{E} \left [ 1 + \zeta \| \tilde{\mathbf{O}}_{\perp} \|_2 + \frac{\zeta^2}{2!}\| \tilde{\mathbf{O}}_{\perp} \|_2^2 + \sum\limits_{k=3}^{\infty}\frac{\zeta^k}{k!} \| \tilde{\mathbf{O}}_{\perp} \|_2^k \right ] \leq \frac{1}{1-\varrho}e^{\zeta a + 1}, \]
and thus
\begin{align}
    \mathbb{E} \left [ \| \tilde{\mathbf{O}}_{\perp} \|_2^2 \right ] & \leq \frac{2}{\zeta^2}\frac{1}{1-\varrho}e^{\zeta a+1}  \label{eq:exp_bound_taylor} \\
    & = \frac{4 \sqrt{n \gamma_{\max} } e^{\zeta K \sqrt{n \gamma_{\max}} /\left (T \xi^* \Delta \right )+1}}{T \xi^* \Delta \zeta^3 - 2 \sqrt{n \gamma_{\max}} G^2(e-2)\zeta^4} \notag \\
    & =: N_{\perp}^2(n,T). \label{eq:N_perp} \\
    & \leq \frac{4 \sqrt{n \gamma_{\max} } e^2}{T \xi^* \Delta \zeta^3 - 2 \sqrt{n \gamma_{\max}} G^2(e-2)\zeta^4}. \notag 
\end{align}

Finally, even though the SSC result is only for the sampled queue-length vector $\tilde{\mathbf{O}}$, between those sampling points the queues can only increase or decrease by at most $T(nA_{\max}+S_{\max})$. Therefore, $\|\tilde{\mathbf{O}}_{\perp}\|_2^2$ can only increase by at most $T^2(nA_{\max}+S_{\max})^2n$. This bound is independent of $\epsilon$ and is of order $\Theta(T^2)$, which is smaller than or equal to $N_{\perp}$ (see Section \ref{dis:order_N_2}). It follows that we have the same SSC result for $\mathbf{O}$.

\subsection{Order Analysis of $N_{\perp}(n,T)$ in Lemma~\ref{lem_our:bd_O_perp}}
\label{dis:order_N_2}

We derive the order of \(N_{\perp}(n,T)\) from the proof of Lemma~\ref{lem_our:bd_O_perp}. Throughout this subsection, we use the same normalization \(\|\boldsymbol{\gamma}\|_1=1\).

Recall that \(G=2T(nA_{\max}+\sqrt n S_{\max})/\sqrt{\gamma_{\min}} \) and \(\epsilon_0= T\xi^*\Delta/ (2\sqrt{n\gamma_{\max}}) \). Hence, \(G=O(Tn/\sqrt{\gamma_{\min}})\), \(G^2=O(T^2n^2/\gamma_{\min})\), and \(C_H=G^2(e-2)=O(T^2n^2/\gamma_{\min})\).

Next, from the definitions of \(K_1\) and \(K_2\), the leading contribution to \(K_1\) comes from the terms with the prefactor \((n-1)/\gamma_{\min}\), whose bracketed expressions are \(O(T^2n^2)\). Therefore,
\[
    K_1=O\!\left(\frac{T^2n^3}{\gamma_{\min}}\right),
\]
while \(K_2=O(T^2n^2)\). Therefore,
\[
    K=K_1+K_2 = O\!\left(\frac{T^2n^3}{\gamma_{\min}}\right).
\]
Since $ a = K\sqrt{n\gamma_{\max}} / (T\xi^*\Delta) $, we have
\[
    a = O\!\left( \frac{Tn^3\sqrt{n\gamma_{\max}}}{\gamma_{\min}\Delta} \right),
    \qquad
    \frac{1}{a} = \Omega\!\left( \frac{\gamma_{\min}\Delta}{Tn^3\sqrt{n\gamma_{\max}}} \right).
\]

By Theorem~\ref{thm:hajek_drift}, the exponential parameter is $ \zeta = \min \left \{ 1/G, \epsilon_0 /(2C_H), 1/a \right \} $. The three terms satisfy
\[
    \frac{1}{G} = \Omega \left(\frac{\sqrt{\gamma_{\min}}}{Tn}\right),
    \qquad
    \frac{\epsilon_0}{2C_H} = \Omega \left( \frac{\gamma_{\min}\Delta}{Tn^2\sqrt{n\gamma_{\max}}} \right),
    \qquad
    \frac{1}{a} = \Omega \left( \frac{\gamma_{\min}\Delta}{Tn^3\sqrt{n\gamma_{\max}}} \right).
\]
Combining these three bounds, we obtain
\[
    \zeta = \Omega \left( \frac{\gamma_{\min}\Delta}{Tn^3\sqrt{n\gamma_{\max}}} \right).
\]
Moreover, since \(\varrho=1-\epsilon_0\zeta+C_H\zeta^2\) and \(\zeta\le \epsilon_0/(2C_H)\), we have
\[
    1-\varrho = \epsilon_0\zeta-C_H\zeta^2 \geq \frac{\epsilon_0\zeta}{2}.
\]
Therefore,
\[
    \frac1{1-\varrho} = O \left(\frac1{\epsilon_0\zeta}\right) =
    O \left( \frac{n^4\gamma_{\max}}{\gamma_{\min}\Delta^2} \right).
\]
From Equation~\eqref{eq:exp_bound_taylor},
\[
    \mathbb E\!\left[\|\tilde{\mathbf O}_{\perp}\|_2^2\right] \leq \frac{2}{\zeta^2}\frac{1}{1-\varrho}e^{\zeta a+1}.
\]
Since \(\zeta\le 1/a\), we have \(e^{\zeta a+1}=O(1)\). Also,
\[
    \frac1{\zeta^2} = O \left( \frac{T^2n^7\gamma_{\max}}{\gamma_{\min}^2\Delta^2} \right).
\]
Combining the above bounds gives
\[
    \mathbb E\!\left[\|\tilde{\mathbf O}_{\perp}\|_2^2\right] = O \left( \frac{T^2n^{11}\gamma_{\max}^2} {\gamma_{\min}^3\Delta^4} \right).
\]
Hence, \(N_{\perp}(n,T)\) can be chosen such that
\[
    N_{\perp}(n,T) = O \left( \frac{Tn^{11/2}\gamma_{\max}} {\gamma_{\min}^{3/2}\Delta^2} \right).
\]
In particular, for fixed \(n\), fixed normalized \(\boldsymbol{\gamma}\), and fixed \(\Delta>0\), we have \(N_{\perp}(n,T)\in O(T)\).

Recall that
\[
    \Delta=\frac{\delta^*\|\boldsymbol{\mu}\|_1}{2\xi^*},
\]
where \(\delta^*\) and \(\xi^*\) are defined in the proof of Lemma~\ref{lem_our:bd_O_perp}. For fixed \(n\), fixed normalized \(\boldsymbol{\gamma}\) and a fixed policy \(\pi\in\Pi\), the induced fractions \(f_{l,\eta}\) are fixed. Hence, under the condition in Equation~\eqref{eq:strictlyless_withrate}, \(\delta^*\), \(\xi^*\) and
\(\Delta \) are positive constants independent of \(T\). Therefore,
\[
    N_{\perp}(n,T) = O \left( \frac{Tn^{11/2}\gamma_{\max}}{\gamma_{\min}^{3/2}} \right).
\]
When \(N_{\perp}(n,T)\) is viewed as a function of the cycle length \(T\) with the other quantities fixed, this gives
\[
    N_{\perp}(n,T)\in O(T).
\]

\section{Proof of Corollary~\ref{co_our:relative_SSC}} 
\label{prf:relative_SSC}

Under any dispatching policy, the total queue length is lower bounded by the queue length of a resource-pooled single-server queue with the same aggregate arrivals $A(t)$ and aggregate service $S_{\Sigma}(t)=\sum_{l=1}^n S_l(t)$. Applying Lemma~4 in \cite{eryilmaz2012asymptotically-driftmethod} to this lower-bounding system, there exists a constant $c_0>0$, independent of $\epsilon$, such that
\[
    \liminf_{\epsilon\downarrow0} \epsilon \mathbb E[\|\mathbf Q\|_1] \geq c_0.
\]
Next, by Theorem~\ref{thm_our:SSC_withrate},
\[
    \mathbb E\left[ \|\mathbf Q^{(\gamma)}_\perp\|_\gamma^2 \right] \leq N_\perp^2(n,T),
\]
where $N_\perp(n,T)$ is independent of $\epsilon$. Applying Jensen's inequality to the concave function $g(y)=\sqrt{y}$ gives
\[
    \mathbb E\left[ \|\mathbf Q^{(\gamma)}_\perp\|_\gamma \right] = \mathbb E\left[ \sqrt{ \|\mathbf Q^{(\gamma)}_\perp\|_\gamma^2 } \right] 
    \leq \sqrt{ \mathbb E\left[ \|\mathbf Q^{(\gamma)}_\perp\|_\gamma^2 \right] }
    \leq N_\perp(n,T).
\]
We also have
\begin{align*}
    \left\| \mathbf Q - \frac{\boldsymbol\gamma}{\|\boldsymbol\gamma\|_1} \|\mathbf Q\|_1 \right\|_1 & = \sum_{l=1}^n \left | Q_l - \frac{\gamma_l}{\| \boldsymbol{\gamma} \|_1} \| \mathbf{Q} \|_1 \right | \\
    & = \sum_{l=1}^n \gamma_l \left| \frac{Q_l}{\gamma_l} - \frac{\|\mathbf Q\|_1}{\|\boldsymbol\gamma\|_1} \right| \\
    & = \sum_{l=1}^n \sqrt{\gamma_l} \cdot \sqrt{\gamma_l} \left| \frac{Q_l}{\gamma_l} - \frac{\|\mathbf Q\|_1}{\|\boldsymbol\gamma\|_1} \right| \\
    & \overset{(a)}{\leq} \left( \sum_{l=1}^n \gamma_l \right)^{1/2} \left( \sum_{l=1}^n \gamma_l \left( \frac{Q_l}{\gamma_l} - \frac{\|\mathbf Q\|_1}{\|\boldsymbol\gamma\|_1} \right)^2 \right)^{1/2} \\
    & = \sqrt{\|\boldsymbol\gamma\|_1} \|\mathbf Q^{(\gamma)}_\perp\|_\gamma,
\end{align*}
where (a) follows from the Cauchy--Schwarz inequality.

Combining the previous bounds,
\[
    0 \leq \lim_{\epsilon\downarrow0} \frac{ \mathbb E\left[ \left\| \mathbf Q-\frac{\boldsymbol\gamma}{\|\boldsymbol\gamma\|_1}\|\mathbf Q\|_1 \right\|_1 \right] }{ \mathbb E[\|\mathbf Q\|_1] } 
    \leq \lim_{\epsilon\downarrow0} \frac{ \sqrt{\|\boldsymbol\gamma\|_1}
    \mathbb E[\|\mathbf Q^{(\gamma)}_\perp\|_\gamma] }{ \mathbb E[\|\mathbf Q\|_1] } 
    \leq \lim_{\epsilon\downarrow0} \frac{ \sqrt{\|\boldsymbol\gamma\|_1} N_\perp(n,T) }{ \mathbb E[\|\mathbf Q\|_1] }
    \overset{(a)}{=} 0,
\]
where (a) follows because \(\liminf_{\epsilon\downarrow0}\epsilon\mathbb E[\|\mathbf Q\|_1]\geq c_0>0\) and the fact that \(N_\perp(n,T)\) is independent of \(\epsilon\).

\section{Proof of Theorem~\ref{thm_our:SSCUpper_withrate_withratio}} \label{prf:thm_upperbound_withrate_withratio}

The proof of Theorem~\ref{thm_our:SSCUpper_withrate_withratio} relies on a multi-step drift analysis in steady state. We first derive a one-cycle identity over an interval of length $T$, and then combine this identity with the state-space collapse result established in Theorem~\ref{thm_our:SSC_withrate}. This allows us to obtain the stated upper bound on the expected average queue length. \\

Assume that the original Markov chain is in steady state under a policy \(\pi \in \Pi\), and fix a sampling epoch \(t\) that is a multiple of \(T\). Let \(t' = t/T\). By the definition of the embedded chain, we have \(\tilde{\mathbf Q}(t') = \mathbf Q(t)\). Moreover, since \(\{\tilde{\mathbf Q}(m)\}_{m \ge 0}\) is obtained by sampling the stationary process \(\{\mathbf Q(t)\}_{t \ge 0}\) at times \(mT\), the embedded chain is stationary. Hence, \(\tilde{\mathbf Q}(t')\) and \(\tilde{\mathbf Q}(t'+1)\) have the same distribution. Since \(\tilde{\mathbf O}(t')\) is a deterministic transformation of \(\tilde{\mathbf Q}(t')\), it follows that \(\tilde{\mathbf O}(t')\) and \(\tilde{\mathbf O}(t'+1)\) also have the same distribution.

We also define
\[
    \Delta \|\tilde{\mathbf Q}(t')\|_1 := \|\tilde{\mathbf Q}(t'+1)\|_1 - \|\tilde{\mathbf Q}(t')\|_1.
\]

\noindent \textbf{Step 1.} We show that the expected steady-state drifts of $W_{\parallel}(\tilde{\mathbf O}(t'))$ and $\|\tilde{\mathbf Q}(t')\|_1$ are both zero.
\begin{lemma} \label{lem_our:ss_drift_zero_with_rate}
    Under the steady-state setup above for a fixed policy \(\pi\) with fixed \(T\), suppose that Equation~\eqref{eq:strictlyless_withrate} holds. Then, for all sufficiently small \(\epsilon>0\), there exists \(\zeta_\epsilon>0\) such that
    \[
        \mathbb E\!\left[e^{\zeta \|\tilde{\mathbf O}(t')\|_2}\right] < \infty, \qquad \forall \zeta\in(0,\zeta_\epsilon],\quad \forall t' \in \mathbb N.
    \]
    In addition, along the heavy-traffic sequence \(n\lambda=\|\boldsymbol\mu\|_1-\epsilon\), there exist constants \(c_T>0\), \(\epsilon_T>0\), and \(\theta_{\max}>0\), independent of \(\epsilon\), such that, for all \(0<\epsilon<\epsilon_T\),
    \[
        \zeta_\epsilon \ge c_T\epsilon,
    \]
    and
    \[
        \mathbb E\!\left[e^{\theta\epsilon\|\tilde{\mathbf Q}(t')\|_1}\right]<\infty,
        \qquad \forall \theta\in(0,\theta_{\max}],\quad \forall t'\in\mathbb N.
    \]
    Moreover, for every \(t' \in \mathbb N\),
    \[
        \mathbb E[\Delta W_{\parallel}(\tilde{\mathbf O}(t'))] = 0, 
        \qquad 
        \mathbb E[\Delta \|\tilde{\mathbf Q}(t')\|_1] = 0.
    \]
\end{lemma}
The proof is given in Appendix~\ref{prf:ss_drift_zero_with_rate}.\\

\noindent \textbf{Step 2}. We derive a one-cycle drift identity in steady state by using $\mathbb E[\Delta W_{\parallel}(\tilde{\mathbf O}(t'))]=0$.
\begin{lemma} \label{lem_our:drift_identity_steady_state_with_rate}
    Under the steady-state setup above for a policy \(\pi \), we have
    \begin{align}
        & \frac{ 1 }{ \| \boldsymbol{\gamma} \|_1 } \mathbb{E} \left[ \| \mathbf{Q}(t) \|_1 \left (\sum\limits_{j=0}^{T-1} \sum\limits_{l=1}^n S_{\eta_{t'}(l)}(t+j) - \sum\limits_{j=0}^{T-1} \sum\limits_{l=1}^n A_{\eta_{t'}(l)}(t+j) \right ) \right ] \notag \\
        & = \frac{ 1 }{ 2 \| \boldsymbol{\gamma} \|_1 } \left \{ \mathbb{E} \left [ \left (\sum\limits_{j=0}^{T-1} \sum\limits_{l=1}^n A_{\eta_{t'}(l)}(t+j) - \sum\limits_{j=0}^{T-1} \sum\limits_{l=1}^n S_{\eta_{t'}(l)}(t+j) \right ) ^2 \right ] \right. \notag \\
        & \quad + 2 \mathbb{E} \left [ \left ( \| \mathbf{Q}(t) \|_1 + \sum\limits_{j=0}^{T-1} \sum\limits_{l=1}^n A_{\eta_{t'}(l)}(t+j) - \sum\limits_{j=0}^{T-1} \sum\limits_{l=1}^n S_{\eta_{t'}(l)}(t+j) \right )\left ( \sum\limits_{j=0}^{T-1} \sum\limits_{l=1}^n U_{\eta_{t'}(l)}(t+j) \right) \right ]  \notag \\
        & \quad \left. + \mathbb{E} \left [ \left ( \sum\limits_{j=0}^{T-1} \sum\limits_{l=1}^n U_{\eta_{t'}(l)}(t+j) \right)^2 \right ] \right \} \label{eq:upb_equation_withrate}
    \end{align}
\end{lemma}
The proof is given in Appendix~\ref{prf:drift_identity_steady_state_with_rate}. \\

\noindent \textbf{Step 3}. We bound the first and third terms in Equation~\eqref{eq:upb_equation_withrate} by using $\mathbb E[\Delta \|\tilde{\mathbf Q}(t')\|_1]=0$.
\begin{lemma} \label{lem_our:unused-moment_with_rate}
    Under the steady-state setup above for a policy \(\pi \), we have
    \begin{align}
        \mathbb{E} \left[ \sum_{j=0}^{T-1} \sum_{l=1}^n U_l(t+j) \right] & = T \epsilon, \notag \\
        \mathbb{E} \left[ \left( \sum_{j=0}^{T-1}\sum_{l=1}^n A_l(t+j) - \sum_{j=0}^{T-1}\sum_{l=1}^n S_l(t+j) \right)^2\right] & \leq T^2 \epsilon ^2 + Tn \sigma_{\lambda}^2 + T \left ( \sum\limits_{l=1}^{n}\sigma_{\eta_{t'}(l)}^2 \right ), \notag \\
        \mathbb{E} \left[ \left( \sum_{j=0}^{T-1}\sum_{l=1}^n U_l(t+j) \right)^2\right] & \leq T^2nS_{\max}\epsilon. \notag 
    \end{align}
\end{lemma}
The proof is given in Appendix \ref{prf:unused-moment_with_rate}. \\

\noindent \textbf{Step 4}. We bound the second term in Equation~\eqref{eq:upb_equation_withrate}.
\begin{lemma} \label{lem_our:bd_cross_term_steady_state_with_rate}
    Under the steady-state setup above for a policy \(\pi \), we have
    \begin{align*}
        & \frac{ 1 }{ \| \boldsymbol{\gamma} \|_1 } \mathbb{E} \left[
        \left ( \| \mathbf{Q}(t) \|_1 + \sum\limits_{j=0}^{T-1} \sum\limits_{l=1}^n A_{\eta_{t'}(l)}(t+j) - \sum\limits_{j=0}^{T-1} \sum\limits_{l=1}^n S_{\eta_{t'}(l)}(t+j) \right )\left ( \sum\limits_{j=0}^{T-1} \sum\limits_{l=1}^n U_{\eta_{t'}(l)}(t+j) \right) \right ] \notag \\
        & \leq \frac{ T N_{\perp}(n,T) \sqrt{ S_{\max}\epsilon}}{\sqrt{\gamma_{\min}}} + \frac{ T^2 n A_{\max} \epsilon}{\gamma_{\min}},
    \end{align*}
    where $N_{\perp}(n,T)$ is the one given in Theorem \ref{thm_our:SSC_withrate}.
\end{lemma}
The proof is given in Appendix \ref{prf:bd_cross_term_steady_state_with_rate}.\\

\noindent \textbf{Step 5}. We combine the bounds and obtain the final estimate.

\medskip
Using the fact that the future arrival and service processes are independent of the queue-length vector $\mathbf{Q}(t)$, together with the fact that the expected net service over one cycle is $T \epsilon$, we obtain
\begin{align}
    \frac{ 1 }{ \| \boldsymbol{\gamma} \|_1 } \mathbb{E} \left[ \| \mathbf{Q}(t) \|_1 \left (\sum\limits_{j=0}^{T-1} \sum\limits_{l=1}^n S_{\eta_{t'}(l)}(t+j) - \sum\limits_{j=0}^{T-1} \sum\limits_{l=1}^n A_{\eta_{t'}(l)}(t+j) \right ) \right ]  
    = \frac{ T \epsilon}{\| \boldsymbol{\gamma} \|_1} \mathbb{E} \left[ \|\mathbf{Q}(t)\|_1 \right ], \notag 
\end{align}
Substituting this identity together with the bounds from lemmas~\ref{lem_our:unused-moment_with_rate} and \ref{lem_our:bd_cross_term_steady_state_with_rate} into Equation~\eqref{eq:upb_equation_withrate}, we obtain
\begin{align}
    \frac{T\epsilon \mathbb{E} \left[ \| \mathbf{Q}(t) \|_1 \right ]}{ \| \boldsymbol{\gamma} \|_1 }  & \leq \frac{T^2\epsilon^2 + Tn \sigma_{\lambda}^2 + T\left ( \sum\limits_{l=1}^{n}\sigma_{\eta_{t'}(l)}^2 \right )}{2 \| \boldsymbol{\gamma} \|_1 } \notag \\
    & \quad + \frac{ T N_{\perp}(n,T,\mathbf{f}, \boldsymbol{\gamma}) \sqrt{ S_{\max}\epsilon} }{\sqrt{\gamma_{\min}}} + \frac{ T^2 n A_{\max} \epsilon}{\gamma_{\min}} \notag \\
    & \quad + \frac{T^2 n S_{\max}\epsilon}{2 \| \boldsymbol{\gamma} \|_1 }, \notag
\end{align}
Since the original queue-length process is in steady state, $\mathbf{Q}(t)$ has the same distribution as the steady-state random vector $\mathbf{Q}$ for every $t \ge 0$. Hence,
\begin{align}
    \epsilon \mathbb{E} \left[ \frac{ \|\mathbf{Q} \|_1 }{n} \right ]
    &= \epsilon \mathbb{E} \left[ \frac{ \|\mathbf{Q}(t)\|_1 }{n} \right ] \notag \\
    &\leq \frac{n \sigma_{\lambda}^2 + \sum\limits_{l=1}^{n}\sigma_{l}^2 }{2n} + \epsilon^2 \frac{T}{2n} + \epsilon \frac{T S_{\max} \gamma_{\min} + 2 T \| \boldsymbol{\gamma} \|_1 A_{\max} }{2 \gamma_{\min}} + \sqrt{\epsilon} \frac{ \| \boldsymbol{\gamma} \|_1 N_{\perp}(n,T)\sqrt{ S_{\max}} }{\sqrt{\gamma_{\min}} n }. \notag 
\end{align}

\subsection{Proof of Lemma~\ref{lem_our:ss_drift_zero_with_rate}} \label{prf:ss_drift_zero_with_rate}
We show that $W_{\parallel}(\tilde{\mathbf O})$ and $\|\tilde{\mathbf Q}\|_1$ are integrable. The proof is similar to that of Lemma~\ref{lem_our:bd_O_perp} and also relies on Theorem~\ref{thm:hajek_drift}.

To apply Theorem~\ref{thm:hajek_drift}, define \( H_{t'} := V(\tilde{\mathbf O}(t')) = \|\tilde{\mathbf O}(t')\|_2 \) so that $ H_{t'+1}-H_{t'} = \Delta V(\tilde{\mathbf O}(t'))$.

We first verify Condition~B2. By Lemma~\ref{lem_our:bd_absolute_delta_vO}, for all $t' \in \mathbb N$,
\[
    |H_{t'+1}-H_{t'}| = |\Delta V(\tilde{\mathbf O}(t'))| \leq \frac{1}{\sqrt{\gamma_{\min}}} T ( nA_{\max} + \sqrt{n} S_{\max} ).
\] 
Hence, we choose
\[
    G:=\frac{T ( nA_{\max} + \sqrt{n} S_{\max})}{\sqrt{\gamma_{\min}}}.
\]
Therefore, Condition~B2 holds. In the application of Theorem~\ref{thm:hajek_drift}, we take \(\theta_H=1/G\). Then
\[
    C_H := \frac{\mathbb E[e^{\theta_H G}]-(1+\theta_H\mathbb E[G])}{\theta_H^2} = (e-2)G^2.
\]

We next verify Condition~B1. Recall from Lemma~\ref{lem_our:bd_negative_delta_vO} that the negative-drift bound requires $\|\tilde{\mathbf Q}(t')\|_1 \geq TS_{\max}\| \boldsymbol{\gamma} \|_1 / \gamma_{\min}$. Since $\|\tilde{\mathbf Q}(t')\|_1 \ge \|\tilde{\mathbf Q}(t')\|_2$, it suffices to require $\|\tilde{\mathbf Q}(t')\|_2 \ge TS_{\max}\| \boldsymbol{\gamma} \|_1 / \gamma_{\min}$. Thus, there exist constants $\epsilon' > 0$ and $K < \infty$ such that, whenever $\|\mathbf q\|_2 \ge TS_{\max}\| \boldsymbol{\gamma} \|_1 / \gamma_{\min}$,
\[
    \mathbb E[\Delta V(\tilde{\mathbf O}(t')) \mid \tilde{\mathbf O}(t')=\mathbf o] \leq -\frac{\epsilon' T \sqrt{\gamma_{\min}}}{ \| \boldsymbol{\gamma} \|_1 } + \frac{K\sqrt{\gamma_{\max}}}{2\|\mathbf q\|_2}.
\]
Now let
\[
    a := \max \left \{ \frac{TS_{\max}\|\boldsymbol{\gamma}\|_1}{\gamma_{\min}^{3/2}}, \frac{K\sqrt{\gamma_{\max}}\|\boldsymbol{\gamma}\|_1}{\epsilon'T\gamma_{\min}} \right \},
    \qquad \epsilon_0 := \frac{\epsilon'T\sqrt{\gamma_{\min}}}{2\|\boldsymbol{\gamma}\|_1}.
\]
If $\|\tilde{\mathbf O}(t')\|_2 > a$, then
\[
    \|\mathbf q\|_2 \geq \sqrt{\gamma_{\min}}\|\tilde{\mathbf O}(t')\|_2 > \frac{K\sqrt{\gamma_{\max}}\|\boldsymbol{\gamma}\|_1}{\epsilon'T\sqrt{\gamma_{\min}}},
\]
and therefore
\[
    \frac{K\sqrt{\gamma_{\max}}}{2\|\mathbf q\|_2} < \frac{\epsilon'T\sqrt{\gamma_{\min}}}{2\|\boldsymbol{\gamma}\|_1} = \epsilon_0.
\]
Hence,
\[
    \mathbb E[\Delta V(\tilde{\mathbf O}(t')) \mid \|\tilde{\mathbf O}(t')\|_2 > a] < -\epsilon_0.
\]
Therefore, Condition~B1 holds.

Let
\[
    \zeta_\epsilon := \min\left\{ \frac1G,\frac{\epsilon_0}{2C_H},\frac1a \right\}.
\]
Then, for any \(\zeta\in(0,\zeta_\epsilon]\), Theorem~\ref{thm:hajek_drift} gives
\[
    \mathbb E\!\left[e^{\zeta \|\tilde{\mathbf O}(t')\|_2}\right] \leq \varrho^{t'} e^{\zeta \|\tilde{\mathbf O}(0)\|_2} + \frac{1-\varrho^{t'}}{1-\varrho} e^{\theta_H G} e^{\zeta a}
\]
for some \(0<\varrho<1\). Since \(\|\tilde{\mathbf O}(0)\|_2=0\) and \(0<\varrho<1\), passing to the stationary distribution and using the Portmanteau theorem for the nonnegative lower-semicontinuous function \(g(\mathbf o)=e^{\zeta\|\mathbf o\|_2}\) gives
\[
    \mathbb{E}\!\left[e^{\zeta\|\tilde{\mathbf O}\|_2}\right] \leq \liminf_{t'\to\infty} \mathbb{E}\!\left[e^{\zeta\|\tilde{\mathbf O}(t')\|_2}\right]
    \leq \frac{e^{\theta_H G}e^{\zeta a}}{1-\varrho} < \infty.
\]

We next show that \(\zeta_\epsilon\) can be chosen to be at least of order \(\epsilon\) along the heavy-traffic sequence. Recall from the proof of Lemma~\ref{lem_our:bd_negative_delta_vO} that \(\epsilon'\) can be chosen as
\[
    \epsilon' = \min_{\eta\in S_n} \min_{m\in[n]} \left\{ \sum_{l=1}^m\mu_{\eta(l)} - n\lambda\sum_{l=1}^m f_{l,\eta} \right\}.
\]
Let
\[
    F_{m,\eta}:=\sum_{l=1}^m f_{l,\eta}.
\]
Since Equation~\eqref{eq:strictlyless_withrate} holds and the set of pairs \((m,\eta)\) is finite, we have
\[
    \delta_{\mathrm{ht}} := \min_{\eta\in S_n} \min_{m\in[n]\setminus\{n\}} \left\{ \sum_{l=1}^m\mu_{\eta(l)} - \|\boldsymbol\mu\|_1 F_{m,\eta} \right\} >0.
\]
Along the heavy-traffic sequence \(n\lambda=\|\boldsymbol\mu\|_1-\epsilon\), for every \(m<n\) and \(\eta\in S_n\),
\[
\begin{aligned}
    \sum_{l=1}^m\mu_{\eta(l)} - n\lambda F_{m,\eta} 
    & = \sum_{l=1}^m\mu_{\eta(l)} - (\|\boldsymbol\mu\|_1-\epsilon)F_{m,\eta}  \\
    & = \left( \sum_{l=1}^m\mu_{\eta(l)} - \|\boldsymbol\mu\|_1 F_{m,\eta} \right) + \epsilon F_{m,\eta} \\
    & \ge \delta_{\mathrm{ht}}.
\end{aligned}
\]
On the other hand, since \(F_{n,\eta}=1\), for every \(\eta\in S_n\),
\[
    \sum_{l=1}^n\mu_{\eta(l)} - n\lambda F_{n,\eta} = \|\boldsymbol\mu\|_1-(\|\boldsymbol\mu\|_1-\epsilon) = \epsilon.
\]
Therefore, for all \(0<\epsilon<\delta_{\mathrm{ht}}\),
\[
    \epsilon'=\epsilon.
\]

Since \(n\lambda=\|\boldsymbol\mu\|_1-\epsilon\le \|\boldsymbol\mu\|_1\), the constant \(K\) in the above drift bound can be bounded by a constant \(\bar K<\infty\) independent of \(\epsilon\), for all sufficiently small \(\epsilon\). Define
\[
    a_1:=\frac{TS_{\max}\|\boldsymbol{\gamma}\|_1}{\gamma_{\min}^{3/2}},
    \qquad
    a_2:=\frac{\bar K\sqrt{\gamma_{\max}}\|\boldsymbol{\gamma}\|_1}{T\gamma_{\min}}.
\]
Then, for all sufficiently small \(\epsilon\),
\[
    a \leq \max\left\{ a_1,\frac{a_2}{\epsilon} \right\} \le \frac{\max\{a_1,a_2\}}{\epsilon}.
\]
Thus,
\[
    \frac1a \ge \frac{\epsilon}{\max\{a_1,a_2\}}.
\]
Moreover, since \(\epsilon'=\epsilon\),
\[
    \epsilon_0 = \frac{\epsilon T\sqrt{\gamma_{\min}}}{2\|\boldsymbol{\gamma}\|_1},
\]
and hence
\[
    \frac{\epsilon_0}{2C_H} = \frac{\epsilon T\sqrt{\gamma_{\min}}}{4\|\boldsymbol{\gamma}\|_1 C_H}.
\]
For \(0<\epsilon<1\), we also have \(1/G\ge \epsilon/G\). Therefore,
\[
    \zeta_\epsilon = \min\left\{ \frac1G,\frac{\epsilon_0}{2C_H},\frac1a \right\} \ge c_T\epsilon,
\]
where
\[
    c_T := \min\left\{ \frac1G,\, \frac{T\sqrt{\gamma_{\min}}}{4\|\boldsymbol{\gamma}\|_1 C_H},\, \frac1{\max\{a_1,a_2\}} \right\} >0.
\]
Let
\[
    \epsilon_T:=\min\{\delta_{\mathrm{ht}},1\}.
\]
Then the above bound holds for all \(0<\epsilon<\epsilon_T\).

Now, since \(\mathbf q = (\sqrt{\gamma_l}o_l)_{l=1}^n\), we have
\[
    \|\mathbf q\|_1 = \langle \mathbf c,\mathbf o\rangle \overset{(a)}{\leq} \|\mathbf c\|_2\|\mathbf o\|_2 = \sqrt{\|\boldsymbol{\gamma}\|_1}\,\|\mathbf o\|_2,
\]
where \(\mathbf c=(\sqrt{\gamma_l})_{l=1}^n\) and (a) follows by the Cauchy-Schwarz inequality. Let
\[
    \theta_{\max}:=\frac{c_T}{2\sqrt{\|\boldsymbol{\gamma}\|_1}}.
\]
Then, for any \(\theta\in(0,\theta_{\max}]\) and \(0<\epsilon<\epsilon_T\),
\[
    \theta\epsilon\|\tilde{\mathbf Q}\|_1
    \le
    \theta\epsilon\sqrt{\|\boldsymbol{\gamma}\|_1}\,\|\tilde{\mathbf O}\|_2
    \le
    \frac{c_T\epsilon}{2}\|\tilde{\mathbf O}\|_2
    \le
    \zeta_\epsilon\|\tilde{\mathbf O}\|_2.
\]
Therefore,
\[
    \mathbb E\!\left[e^{\theta\epsilon\|\tilde{\mathbf Q}\|_1}\right]
    \le
    \mathbb E\!\left[e^{\zeta_\epsilon\|\tilde{\mathbf O}\|_2}\right]
    <\infty,
    \qquad
    \forall \theta\in(0,\theta_{\max}],\quad 0<\epsilon<\epsilon_T.
\]
Since \(x\le e^{\zeta_\epsilon x}/\zeta_\epsilon\) and \(x^2\le 2e^{\zeta_\epsilon x}/\zeta_\epsilon^2\) for all \(x\ge0\), it follows that
\[
    \mathbb E\!\left[\|\tilde{\mathbf O}\|_2\right] < \infty,
    \qquad
    \mathbb E\!\left[\|\tilde{\mathbf O}\|_2^2\right] < \infty.
\]
Moreover,
\[
    W_{\parallel}(\mathbf o) = \|\mathbf o_{\parallel}\|_2^2 = \frac{\|\mathbf q\|_1^2}{\|\boldsymbol{\gamma}\|_1} \leq \|\mathbf o\|_2^2.
\]
Combining these inequalities with the moment bounds above, we obtain
\[
    \mathbb E\!\left[\|\tilde{\mathbf Q}\|_1\right] < \infty,
    \qquad
    \mathbb E\!\left[W_{\parallel}(\tilde{\mathbf O})\right] < \infty.
\]

Under the steady-state setup in Section~\ref{prf:thm_upperbound_withrate_withratio}, the embedded chain is stationary. Hence, for every \(t' \in \mathbb N\), the random vectors \(\tilde{\mathbf Q}(t')\) and \(\tilde{\mathbf Q}(t'+1)\) have the same distribution. Since \(\tilde{\mathbf O}(t')\) is a deterministic transformation of \(\tilde{\mathbf Q}(t')\), \(\tilde{\mathbf O}(t')\) and \(\tilde{\mathbf O}(t'+1)\) also have the same distribution. Since \(W_{\parallel}(\tilde{\mathbf O}(t'))\) and \(\|\tilde{\mathbf Q}(t')\|_1\) are integrable, we have
\[
    \mathbb E\!\left[W_{\parallel}(\tilde{\mathbf O}(t'+1))\right]
    =
    \mathbb E\!\left[W_{\parallel}(\tilde{\mathbf O}(t'))\right],
    \qquad
    \mathbb E\!\left[\|\tilde{\mathbf Q}(t'+1)\|_1\right]
    =
    \mathbb E\!\left[\|\tilde{\mathbf Q}(t')\|_1\right].
\]
Therefore,
\[
    \mathbb E[\Delta W_{\parallel}(\tilde{\mathbf O}(t'))] = 0,
    \qquad
    \mathbb E[\Delta \|\tilde{\mathbf Q}(t')\|_1] = 0.
\]

\subsection{Proof of Lemma~\ref{lem_our:drift_identity_steady_state_with_rate}} \label{prf:drift_identity_steady_state_with_rate}
Under the steady-state setup of Section~\ref{prf:thm_upperbound_withrate_withratio}, we have
\begin{align}
    & \mathbb{E} \left [ \Delta W_{\parallel} ( \tilde{\mathbf{O}}(t') ) \right ] = \mathbb{E} \left [ \| \tilde{\mathbf{O}}_{\parallel}(t'+1) \|^2 - \| \tilde{\mathbf{O}}_{\parallel}(t') \|^2 \right ] \notag \\
    & = \mathbb{E} \left [ \| \mathbf{O}_{\parallel}(t+T) \|^2 - \| \mathbf{O}_{\parallel}(t) \|^2 \right ] \notag \\
    & = \frac{ 1 }{ \| \boldsymbol{\gamma} \|_1 } \left \{ \mathbb{E} \left [ 2 \| \mathbf{Q}(t) \|_1 \left (\sum\limits_{j=0}^{T-1} \sum\limits_{l=1}^n A_{\eta_{t'}(l)}(t+j) - \sum\limits_{j=0}^{T-1} \sum\limits_{l=1}^n S_{\eta_{t'}(l)}(t+j) \right ) \  \middle | \  \mathbf{O}(t)=\mathbf{o} \right ] \right. \notag \\
    & \qquad\qquad + \mathbb{E} \left [ \left (\sum\limits_{j=0}^{T-1} \sum\limits_{l=1}^n A_{\eta_{t'}(l)}(t+j) - \sum\limits_{j=0}^{T-1} \sum\limits_{l=1}^n S_{\eta_{t'}(l)}(t+j) \right ) ^2 \  \middle | \  \mathbf{O}(t)=\mathbf{o} \right ] \notag \\
    & \qquad\qquad + \mathbb{E} \left [ 2 \left ( \| \mathbf{Q}(t) \|_1 + \sum\limits_{j=0}^{T-1} \sum\limits_{l=1}^n A_{\eta_{t'}(l)}(t+j) \right )\left ( \sum\limits_{j=0}^{T-1} \sum\limits_{l=1}^n U_{\eta_{t'}(l)}(t+j) \right) \  \middle | \  \mathbf{O}(t)=\mathbf{o} \right ] \notag \\
    & \qquad\qquad - \mathbb{E} \left [ 2\left ( \sum\limits_{j=0}^{T-1} \sum\limits_{l=1}^n S_{\eta_{t'}(l)}(t+j) \right) \left ( \sum\limits_{j=0}^{T-1} \sum\limits_{l=1}^n U_{\eta_{t'}(l)}(t+j) \right) \  \middle | \  \mathbf{O}(t)=\mathbf{o} \right ] \notag \\
    & \qquad\qquad \left. + \mathbb{E} \left [ \left ( \sum\limits_{j=0}^{T-1} \sum\limits_{l=1}^n U_{\eta_{t'}(l)}(t+j) \right)^2 \  \middle | \  \mathbf{O}(t)=\mathbf{o} \right ] \right \} \notag 
\end{align}
Since, in steady-state, we have $ \mathbb{E}[\Delta W_{\parallel}(\tilde{\mathbf{O}}(t')) ] = 0 $, then
\begin{align*}
    & \frac{ 1 }{ \| \boldsymbol{\gamma} \|_1 } \mathbb{E} \left[ \| \mathbf{Q}(t) \|_1 \left (\sum\limits_{j=0}^{T-1} \sum\limits_{l=1}^n S_{\eta_{t'}(l)}(t+j) - \sum\limits_{j=0}^{T-1} \sum\limits_{l=1}^n A_{\eta_{t'}(l)}(t+j) \right ) \right ] \notag \\
    & = \frac{ 1 }{ 2 \| \boldsymbol{\gamma} \|_1 } \left \{ \mathbb{E} \left [ \left (\sum\limits_{j=0}^{T-1} \sum\limits_{l=1}^n A_{\eta_{t'}(l)}(t+j) - \sum\limits_{j=0}^{T-1} \sum\limits_{l=1}^n S_{\eta_{t'}(l)}(t+j) \right ) ^2 \right ] \right. \notag \\
    & \quad + 2 \mathbb{E} \left [ \left ( \| \mathbf{Q}(t) \|_1 + \sum\limits_{j=0}^{T-1} \sum\limits_{l=1}^n A_{\eta_{t'}(l)}(t+j) - \sum\limits_{j=0}^{T-1} \sum\limits_{l=1}^n S_{\eta_{t'}(l)}(t+j) \right )\left ( \sum\limits_{j=0}^{T-1} \sum\limits_{l=1}^n U_{\eta_{t'}(l)}(t+j) \right) \right ]  \notag \\
    & \quad \left. + \mathbb{E} \left [ \left ( \sum\limits_{j=0}^{T-1} \sum\limits_{l=1}^n U_{\eta_{t'}(l)}(t+j) \right)^2 \right ] \right \}
\end{align*}

\subsection{Proof of Lemma~\ref{lem_our:unused-moment_with_rate}} \label{prf:unused-moment_with_rate}
Under the steady-state setup of Section~\ref{prf:thm_upperbound_withrate_withratio}, we have
\begin{align}
    0& = \mathbb{E} \left [ \Delta \|\tilde{\mathbf Q}(t')\|_1 \right ] \notag \\
    & = \mathbb{E} \left [ \| \tilde{\mathbf{Q}}(t'+1) \|_1 - \| \tilde{\mathbf{Q}}(t') \|_1 \right ] \notag \\ 
    & = \mathbb{E} \left [ \| \mathbf{Q}(t+T) \|_1 - \| \mathbf{Q}(t) \|_1 \right ]  \notag \\ 
    & = \mathbb{E} \left [ \| \mathbf{Q}(t) \|_1 + \sum\limits_{j=0}^{T-1} \sum\limits_{l=1}^n A_{\eta_{t'}(l)}(t+j) - \sum\limits_{j=0}^{T-1} \sum\limits_{l=1}^n S_{\eta_{t'}(l)}(t+j) + \sum\limits_{j=0}^{T-1} \sum\limits_{l=1}^n U_{\eta_{t'}(l)}(t+j)  - \| \mathbf{Q}(t) \|_1 \right ], \notag
\end{align}
This implies
\begin{align}
    \mathbb{E} \left [ \sum\limits_{j=0}^{T-1} \sum\limits_{l=1}^n U_{\eta_{t'}(l)}(t+j) \right ] &= \mathbb{E} \left [ \sum\limits_{j=0}^{T-1} \sum\limits_{l=1}^n S_{\eta_{t'}(l)}(t+j) - \sum\limits_{j=0}^{T-1} \sum\limits_{l=1}^n A_{\eta_{t'}(l)}(t+j) \right ] \notag \\
    & = T\epsilon. \label{eq:upb_Eofu_withrate}
\end{align}
Using Equation~\eqref{eq:upb_Eofu_withrate}, we bound the first term in Equation~\eqref{eq:upb_equation_withrate} by
\begin{align}
    \mathbb{E} \left[
        \left (\sum\limits_{j=0}^{T-1} \sum\limits_{l=1}^n A_{\eta_{t'}(l)}(t+j) - \sum\limits_{j=0}^{T-1} \sum\limits_{l=1}^n S_{\eta_{t'}(l)}(t+j) \right ) ^2 \right ] 
    & \overset{(a)}{\leq} T^2 \epsilon ^2 + Tn \sigma_{\lambda}^2 + T \left ( \sum\limits_{l=1}^{n}\sigma_{\eta_{t'}(l)}^2 \right ), \notag 
\end{align}
where (a) follows by letting $Z := \sum_{j=0}^{T-1} \sum_{l=1}^n A_{\eta_{t'}(l)}(t+j) - \sum_{j=0}^{T-1} \sum_{l=1}^n S_{\eta_{t'}(l)}(t+j)$ and using the identity $\mathbb{E}[Z^2] = (\mathbb{E}[Z])^2 + \mathrm{Var}(Z)$. Moreover, because both the arrival process is independent across time, the service processes are independent across time, and the arrival and service processes are independent, we obtain $\mathrm{Var}(Z) = Tn\sigma_\lambda^2 + T\sum_{l=1}^n \sigma_{\eta_{t'}(l)}^2$.

The third term in Equation~\eqref{eq:upb_equation_withrate} can be bounded as follows:
\begin{align}
    \mathbb{E} \left[ \left ( \sum\limits_{j=0}^{T-1} \sum\limits_{l=1}^n U_{\eta_{t'}(l)}(t+j) \right)^2 \right ] & \overset{(a)}{\leq} TnS_{\max} \mathbb{E} \left[ \left ( \sum\limits_{j=0}^{T-1} \sum\limits_{l=1}^n U_{\eta_{t'}(l)}(t+j) \right) \right ] \notag  \\
    & \leq T^2nS_{\max}\epsilon \notag 
\end{align}
where (a) follows from the fact that $U_{\eta_{t'}(l)}(t+j) \leq S_{\max}$ for all $l\in [n]$ and $j\geq 0$.

\subsection{Proof of Lemma~\ref{lem_our:bd_cross_term_steady_state_with_rate}} \label{prf:bd_cross_term_steady_state_with_rate}
We first bound the integrand of the second term in Equation~\eqref{eq:upb_equation_withrate}.
\begin{align}
    & \frac{ 1 }{ \| \boldsymbol{\gamma} \|_1 } \left ( \| \mathbf{Q}(t) \|_1 + \sum\limits_{j=0}^{T-1} \sum\limits_{l=1}^n A_{\eta_{t'}(l)}(t+j) - \sum\limits_{j=0}^{T-1} \sum\limits_{l=1}^n S_{\eta_{t'}(l)}(t+j) \right )\left ( \sum\limits_{j=0}^{T-1} \sum\limits_{l=1}^n U_{\eta_{t'}(l)}(t+j) \right) \label{eq:upb_qu_org_withrate} \\
    & = \frac{ 1 }{ \| \boldsymbol{\gamma} \|_1 } \left ( \| \mathbf{Q}(t) \|_1 + \sum\limits_{j=0}^{T-1} \sum\limits_{l=1}^n A_{\eta_{t'}(l)}(t+j) - \sum\limits_{j=0}^{T-1} \sum\limits_{l=1}^n S_{\eta_{t'}(l)}(t+j) \right. \notag \\
    & \quad \left. + \sum\limits_{j=0}^{T-1} \sum\limits_{l=1}^n U_{\eta_{t'}(l)}(t+j) - \sum\limits_{j=0}^{T-1} \sum\limits_{l=1}^n U_{\eta_{t'}(l)}(t+j) \right ) \left ( \sum\limits_{j=0}^{T-1} \sum\limits_{l=1}^n U_{\eta_{t'}(l)}(t+j) \right) \notag \\
    & \leq \frac{ 1 }{ \| \boldsymbol{\gamma} \|_1 } \left ( \| \mathbf{Q}(t) \|_1 + \sum\limits_{j=0}^{T-1} \sum\limits_{l=1}^n A_{\eta_{t'}(l)}(t+j) - \sum\limits_{j=0}^{T-1} \sum\limits_{l=1}^n S_{\eta_{t'}(l)}(t+j) + \sum\limits_{j=0}^{T-1} \sum\limits_{l=1}^n U_{\eta_{t'}(l)}(t+j) \right ) \notag \\ 
    & \quad \times \left ( \sum\limits_{j=0}^{T-1} \sum\limits_{l=1}^n U_{\eta_{t'}(l)}(t+j) \right) \notag \\
    & = \frac{ 1 }{ \| \boldsymbol{\gamma} \|_1 } \| \mathbf{Q}(t+T) \|_1 \left ( \sum\limits_{j=0}^{T-1} \sum\limits_{l=1}^n U_{\eta_{t'}(l)}(t+j) \right) \notag \\
    & = \sum_{l=1}^n \left [ \frac{ \| \mathbf{Q}(t+T) \|_1 }{ \| \boldsymbol{\gamma} \|_1 } \left ( \sum\limits_{j=0}^{T-1} U_{\eta_{t'}(l)}(t+j) \right) \right ] \notag \\
    & = \sum_{l=1}^n \left [ \frac{1}{\sqrt{\gamma_{\eta_{t'}(l)}}} \left ( \frac{ \| \mathbf{Q}(t+T) \|_1 }{ \| \boldsymbol{\gamma} \|_1 } \sqrt{\gamma_{\eta_{t'}(l)}} - \frac{Q_{\eta_{t'}(l)}(t+T)  }{\sqrt{\gamma_{\eta_{t'}(l)}}} \right ) \left ( \sum\limits_{j=0}^{T-1} U_{\eta_{t'}(l)}(t+j) \right) \right ] \notag \\
    & \quad + \sum_{l=1}^n \left [ \frac{Q_{\eta_{t'}(l)}(t+T) }{\gamma_{\eta_{t'}(l)}}  \left ( \sum\limits_{j=0}^{T-1} U_{\eta_{t'}(l)}(t+j) \right) \right ] \notag \\
    & \leq \sum_{l=1}^n \left [ \frac{1}{\sqrt{\gamma_{\eta_{t'}(l)}}} \left ( \frac{ \| \mathbf{Q}(t+T) \|_1 }{ \| \boldsymbol{\gamma} \|_1 } \sqrt{\gamma_{\eta_{t'}(l)}} - \frac{Q_{\eta_{t'}(l)}(t+T)}{\sqrt{\gamma_{\eta_{t'}(l)}}} \right ) \left ( \sum\limits_{j=0}^{T-1} U_{\eta_{t'}(l)}(t+j) \right) \right ] \notag \\
    & \quad + \frac{1}{\gamma_{\min}} \sum_{l=1}^n \left [ Q_{\eta_{t'}(l)}(t+T) \left ( \sum\limits_{j=0}^{T-1} U_{\eta_{t'}(l)}(t+j) \right) \right ] \label{eq:upb_qu_withrate}
\end{align}
We first bound $ Q_{\eta_{t'}(l)}(t+T) \left( \sum_{j=0}^{T-1} U_{\eta_{t'}(l)}(t+j) \right ) $ in Equation~\eqref{eq:upb_qu_withrate}. For each $ l $, we have
\[ Q_{\eta_{t'}(l)}(t+j+1)U_{\eta_{t'}(l)}(t+j) = 0 \]
and
\[ Q_{\eta_{t'}(l)}(t+j+1) = Q_{\eta_{t'}(l)}(t+j) + A_{\eta_{t'}(l)}(t+j) - S_{\eta_{t'}(l)}(t+j) + U_{\eta_{t'}(l)}(t+j), \]
for all $j \in \{0,1,...,T-1\}$. It then follows that
\begin{align}
    & Q_{\eta_{t'}(l)}(t+T) \left( \sum\limits_{j=0}^{T-1} U_{\eta_{t'}(l)}(t+j) \right ) \notag \\
    & =  Q_{\eta_{t'}(l)}(t+T) \left( \sum\limits_{j=0}^{T-2} U_{\eta_{t'}(l)}(t+j) \right ) + Q_{\eta_{t'}(l)}(t+T) U_{\eta_{t'}(l)}(t+T-1) \notag \\
    & = Q_{\eta_{t'}(l)}(t+T-1) \left( \sum\limits_{j=0}^{T-2} U_{\eta_{t'}(l)}(t+j) \right ) \notag \\
    & \quad + \Big[A_{\eta_{t'}(l)}(t+T-1) - S_{\eta_{t'}(l)}(t+T-1)+U_{\eta_{t'}(l)}(t+T-1) \Big] \left( \sum\limits_{j=0}^{T-2} U_{\eta_{t'}(l)}(t+j) \right ) \notag \\
    & = Q_{\eta_{t'}(l)}(t+1) U_{\eta_{t'}(l)}(t) + \sum\limits_{i=1}^{T-1}\left [ \Big ( A_{\eta_{t'}(l)}(t+i) -S_{\eta_{t'}(l)}(t+i)+U_{\eta_{t'}(l)}(t+i) \Big ) \left( \sum\limits_{j=0}^{i-1} U_{\eta_{t'}(l)}(t+j) \right ) \right ] \notag \\
    & = \sum\limits_{i=1}^{T-1}\left [ \Big ( A_{\eta_{t'}(l)}(t+i) -S_{\eta_{t'}(l)}(t+i) + U_{\eta_{t'}(l)}(t+i) \Big ) \left( \sum\limits_{j=0}^{i-1} U_{\eta_{t'}(l)}(t+j) \right ) \right ] \notag \\
    & \overset{(a)}{\leq} \sum\limits_{i=1}^{T-1}\left [ A_{\eta_{t'}(l)}(t+i) \left( \sum\limits_{j=0}^{i-1} U_{\eta_{t'}(l)}(t+j) \right ) \right ] \notag \\
    & \leq \sum\limits_{i=1}^{T-1} A_{\eta_{t'}(l)}(t+i) \left( \sum\limits_{j=0}^{T-1}\sum\limits_{k=1}^{n} U_k(t+j) \right ), \label{eq:upb_alul_withrate}
\end{align}
where (a) follows from the fact that $S_{\eta_{t'}(l)}(t+i) \geq U_{\eta_{t'}(l)}(t+i)$ for all $ l\in [n]$ and $i \in \{ 0,...,T-1 \}$. Summing over all $n$ servers gives
\begin{align}
    \sum\limits_{l=1}^{n}\left [ Q_{\eta_{t'}(l)}(t+T) \left( \sum\limits_{j=0}^{T-1} U_{\eta_{t'}(l)}(t+j) \right ) \right ] & \leq \sum\limits_{l=1}^{n} \left [ \sum\limits_{i=1}^{T-1} A_{\eta_{t'}(l)}(t+i)  \left( \sum\limits_{j=0}^{T-1}\sum\limits_{k=1}^{n} U_k(t+j) \right ) \right ] \notag \\
    & = \left [ \sum\limits_{i=1}^{T-1} \left ( \sum\limits_{l=1}^{n} A_{\eta_{t'}(l)}(t+i) \right ) \right ] \left( \sum\limits_{j=0}^{T-1}\sum\limits_{k=1}^{n} U_k(t+j) \right ) \notag \\
    & \overset{(a)}{\leq} T n A_{\max} \left( \sum\limits_{j=0}^{T-1}\sum\limits_{k=1}^{n} U_k(t+j) \right ), \notag
\end{align}
where (a) follows from the assumption that the number of arrivals is bounded by $ n A_{\max} $. Combining this with equations~\eqref{eq:upb_qu_org_withrate} and \eqref{eq:upb_qu_withrate}, we obtain
\begin{align}
    & \frac{ 1 }{ \| \boldsymbol{\gamma} \|_1 } \mathbb{E} \left[ \left ( \| \mathbf{Q}(t) \|_1 + \sum\limits_{j=0}^{T-1} \sum\limits_{l=1}^n A_{\eta_{t'}(l)}(t+j) - \sum\limits_{j=0}^{T-1} \sum\limits_{l=1}^n S_{\eta_{t'}(l)}(t+j) \right )\left ( \sum\limits_{j=0}^{T-1} \sum\limits_{l=1}^n U_{\eta_{t'}(l)}(t+j) \right) \right ] \notag \\
    & \leq \mathbb{E} \left [ \sum_{l=1}^n \left [ \frac{1}{\sqrt{\gamma_{\eta_{t'}(l)}}} \left ( \frac{ \| \mathbf{Q}(t+T) \|_1 }{ \| \boldsymbol{\gamma} \|_1 } \sqrt{\gamma_{\eta_{t'}(l)}} - \frac{Q_{\eta_{t'}(l)}(t+T)}{\sqrt{\gamma_{\eta_{t'}(l)}}} \right ) \left ( \sum\limits_{j=0}^{T-1} U_{\eta_{t'}(l)}(t+j) \right) \right ] \right ] \notag \\
    & \quad + \frac{1}{\gamma_{\min}} \mathbb{E} \left [ \sum_{l=1}^n \left [ Q_{\eta_{t'}(l)}(t+T) \left ( \sum\limits_{j=0}^{T-1} U_{\eta_{t'}(l)}(t+j) \right) \right ] \right ] \notag \\
    & \leq \mathbb{E} \left [ \sum_{l=1}^n \left [ \frac{1}{\sqrt{\gamma_{\eta_{t'}(l)}}} \left ( \frac{ \| \mathbf{Q}(t+T) \|_1 }{ \| \boldsymbol{\gamma} \|_1 } \sqrt{\gamma_{\eta_{t'}(l)}} - \frac{Q_{\eta_{t'}(l)}(t+T)}{\sqrt{\gamma_{\eta_{t'}(l)}}} \right ) \left ( \sum\limits_{j=0}^{T-1} U_{\eta_{t'}(l)}(t+j) \right) \right ] \right ] \notag \\
    & \quad + \frac{T n A_{\max}}{\gamma_{\min}} \mathbb{E} \left [ \sum\limits_{j=0}^{T-1}\sum\limits_{k=1}^{n} U_k(t+j) \right ] \notag \\
    & \overset{(a)}{\leq} \sqrt{\mathbb{E} \left[ \sum\limits_{l=1}^{n} \frac{1}{\gamma_{\eta_{t'}(l)}} \left ( \frac{ \| \mathbf{Q}(t+T) \|_1 }{ \| \boldsymbol{\gamma} \|_1 } \sqrt{\gamma_{\eta_{t'}(l)}} - \frac{Q_{\eta_{t'}(l)}(t+T)}{\sqrt{\gamma_{\eta_{t'}(l)}}} \right )^2 \right ] \mathbb{E} \left[ \sum\limits_{l=1}^{n} \left ( \sum\limits_{j=0}^{T-1} U_{\eta_{t'}(l)}(t+j) \right )^2 \right ] } \notag \\
    & \quad + \frac{T n A_{\max}}{\gamma_{\min}} \mathbb{E} \left [ \sum\limits_{j=0}^{T-1}\sum\limits_{k=1}^{n} U_k(t+j) \right ] \notag \\
    & \leq \frac{1}{\sqrt{\gamma_{\min}}} \sqrt{\mathbb{E} \left[ \sum\limits_{l=1}^{n} \left ( \frac{ \| \mathbf{Q}(t+T) \|_1 }{ \| \boldsymbol{\gamma} \|_1 } \sqrt{\gamma_{\eta_{t'}(l)}} - \frac{Q_{\eta_{t'}(l)}(t+T)}{\sqrt{\gamma_{\eta_{t'}(l)}}} \right )^2 \right ] \mathbb{E} \left[ \sum\limits_{l=1}^{n} \left ( \sum\limits_{j=0}^{T-1} U_{\eta_{t'}(l)}(t+j) \right )^2 \right ] } \notag \\
    & \quad + \frac{T n A_{\max}}{\gamma_{\min}} \mathbb{E} \left [ \sum\limits_{j=0}^{T-1}\sum\limits_{k=1}^{n} U_k(t+j) \right ] \notag \\
    & \overset{(b)}{=} \frac{1}{\sqrt{\gamma_{\min}}} \sqrt{\mathbb{E} \left[ \sum\limits_{l=1}^{n} \left ( \frac{ \| \mathbf{Q}(t) \|_1 }{ \| \boldsymbol{\gamma} \|_1 } \sqrt{\gamma_{\eta_{t'}(l)}} - \frac{Q_{\eta_{t'}(l)}(t)}{\sqrt{\gamma_{\eta_{t'}(l)}}} \right )^2 \right ] \mathbb{E} \left[ \sum\limits_{l=1}^{n} \left ( \sum\limits_{j=0}^{T-1} U_{\eta_{t'}(l)}(t+j) \right )^2 \right ] } \notag \\
    & \quad + \frac{T n A_{\max}}{\gamma_{\min}} \mathbb{E} \left [ \sum\limits_{j=0}^{T-1}\sum\limits_{k=1}^{n} U_k(t+j) \right ] \notag \\
    & = \frac{1}{\sqrt{\gamma_{\min}}} \sqrt{\mathbb{E} \left[ \| \mathbf{O}_{\perp} \|_2^2 \right ] \mathbb{E} \left[ \sum\limits_{l=1}^{n} \left ( \sum\limits_{j=0}^{T-1} U_{\eta_{t'}(l)}(t+j) \right )^2 \right ] } + \frac{T n A_{\max}}{\gamma_{\min}} \mathbb{E} \left [ \sum\limits_{j=0}^{T-1}\sum\limits_{k=1}^{n} U_k(t+j) \right ],  \label{eq:upb_Eofqu_withrate}
\end{align}
where (a) follows from the Cauchy-Schwarz inequality, and (b) follows from the fact that $\tilde{\mathbf{Q}}(t'+1) \overset{d}{=} \tilde{\mathbf{Q}}(t') $ because they are in steady-state. The summation is invariant to the permutation $\eta_{t'}$, so the steady-state distribution of $\mathbf Q(t+T)$ can be replaced by that of $\mathbf Q(t)$. 

To bound the first term in Equation~\eqref{eq:upb_Eofqu_withrate}, we use that
\begin{align}
    \mathbb{E} \left[ \sum\limits_{l=1}^{n} \left ( \sum\limits_{j=0}^{T-1} U_{\eta_{t'}(l)}(t+j) \right )
        ^2 \right ] & \leq \mathbb{E} \left[ \sum\limits_{l=1}^{n} \left [ \left ( \sum\limits_{j=0}^{T-1} U_{\eta_{t'}(l)}(t+j) \right )
        TS_{\max} \right ] \right ] \notag \\
    & = TS_{\max} \mathbb{E} \left [ \sum\limits_{j=0}^{T-1} \sum\limits_{l=1}^n U_{\eta_{t'}(l)}(t+j) \right ]  \notag \\
    & = T^2 S_{\max}\epsilon. \label{eq:upb_Eofu2_withrate}
\end{align}

Using the bounds of equations~\eqref{eq:upb_Eofu_withrate} and \eqref{eq:upb_Eofu2_withrate} in Equation~\eqref{eq:upb_Eofqu_withrate}, we bound the second term in Equation~\eqref{eq:upb_equation_withrate} by
\begin{align}
    & \frac{ 1 }{ \| \boldsymbol{\gamma} \|_1 } \mathbb{E} \left[
        \left ( \| \mathbf{Q}(t) \|_1 + \sum\limits_{j=0}^{T-1} \sum\limits_{l=1}^n A_{\eta_{t'}(l)}(t+j) - \sum\limits_{j=0}^{T-1} \sum\limits_{l=1}^n S_{\eta_{t'}(l)}(t+j) \right )\left ( \sum\limits_{j=0}^{T-1} \sum\limits_{l=1}^n U_{\eta_{t'}(l)}(t+j) \right) \right ] \notag \\
    & \leq \frac{1}{\sqrt{\gamma_{\min}}} \sqrt{\mathbb{E} \left[ \| \mathbf{O}_{\perp} \|_2^2 \right ] \mathbb{E} \left[ \sum\limits_{l=1}^{n} \left ( \sum\limits_{j=0}^{T-1} U_{\eta_{t'}(l)}(t+j) \right )^2 \right ] } + \frac{T n A_{\max}}{\gamma_{\min}} \mathbb{E} \left [ \sum\limits_{j=0}^{T-1}\sum\limits_{k=1}^{n} U_k(t+j) \right ] \notag \\
    & \overset{(a)}{\leq}  \frac{\sqrt{ N_{\perp}^2(n,T) T^2 S_{\max}\epsilon}}{\sqrt{\gamma_{\min}}} + \frac{ T n A_{\max}T\epsilon}{\gamma_{\min}} \notag \\
    & = \frac{ T N_{\perp}(n,T) \sqrt{ S_{\max}\epsilon}}{\sqrt{\gamma_{\min}}} + \frac{ T^2 n A_{\max} \epsilon}{\gamma_{\min}} \notag
\end{align}
where (a) follows from Theorem \ref{thm_our:SSC_withrate}.

\section{Proof of Theorem~\ref{thm_our:queuelength_distribution}} \label{prf:thm_our:queuelength_distribution}
The proof uses a one-cycle transform identity in steady state. The main technical challenge is to control the unused-service term, which couples the large queue length with boundary idleness events. We prove that this term is \(o(\epsilon^2)\), so the leading-order behavior of the transform is governed by the aggregate arrival and service fluctuations. This yields the limiting MGF of \(\epsilon\|\mathbf Q\|_1\). The vector convergence of \(\epsilon\mathbf Q\) then follows by combining this total queue-length limit with Theorem~\ref{thm_our:SSC_withrate} and Slutsky's theorem. \\

Assume that the original Markov chain is in steady state under a fixed policy  \(\pi\in\Pi\) with fixed \(T\), and fix a sampling epoch \(t\) that is a multiple of \(T\). By Lemma~\ref{lem_our:ss_drift_zero_with_rate}, there exist constants  \(\epsilon_T>0\) and \(\theta_{\max}>0\), independent of \(\epsilon\), such that,  for all \(0<\epsilon<\epsilon_T\) and all \(\theta\in(0,\theta_{\max}]\),
\[
    \mathbb E\!\left[e^{\theta\epsilon\|\mathbf Q(t)\|_1}\right]<\infty .
\]
Throughout the proof, fix \(0<\epsilon<\epsilon_T\) and  \(\theta\in(0,\theta_{\max}]\). Therefore, all moment generating functions  appearing below are well defined for sufficiently small positive \(\theta\). For negative $\theta$, the transform of $\epsilon\|\mathbf Q(t)\|_1$ is bounded by one, so it suffices to establish the MGF convergence for sufficiently small positive $\theta$.

We define the following cycle-level vectors
\begin{align*}
    \boldsymbol{\Sigma} \mathbf{U_o}(t+T-1) & := \left( \sum_{j=0}^{T-1} \frac{U_l(t+j)}{\sqrt{\gamma_l}} \right)_{l=1}^n, \qquad  \boldsymbol{\Sigma}\mathbf{U}(t+T-1) := \left( \sum_{j=0}^{T-1} U_l(t+j) \right)_{l=1}^n, \\
    \boldsymbol{\Sigma}\mathbf{A}(t+T-1) &:= \left ( \sum_{j=0}^{T-1} A_l(t+j) \right )_{l=1}^n, \qquad \boldsymbol{\Sigma}\mathbf{S}(t+T-1) := \left ( \sum_{j=0}^{T-1} S_l(t+j) \right )_{l=1}^n.
\end{align*}
Before the three steps, we record a one-cycle transform identity that will be used throughout the proof. Using the one-cycle queue dynamics, we have
\begin{align}
    & \mathbb{E} \left [ \left ( e^{\theta \epsilon \left < \mathbf{c}, \mathbf{O}(t+T) \right >} - 1 \right ) \left ( e^{-\theta \epsilon \left < \mathbf{c}, \boldsymbol{\Sigma} \mathbf{U_o}(t+T-1) \right >} - 1 \right )\right ] \notag \\
    & = \mathbb{E} \left [ \left ( e^{\theta \epsilon \| \mathbf{Q}(t+T) \|_1 } - 1 \right ) \left ( e^{-\theta \epsilon \| \boldsymbol{\Sigma} \mathbf{U}(t+T-1) \|_1 } - 1 \right )\right ] \notag \\
    &  = \mathbb{E} \left [ e^{\theta \epsilon \| \mathbf{Q}(t+T) \|_1 -\theta \epsilon \| \boldsymbol{\Sigma} \mathbf{U}(t+T-1) \|_1 } \right ] - \mathbb{E} \left [ e^{\theta \epsilon \| \mathbf{Q}(t+T) \|_1 } \right ] + 1 - \mathbb{E} \left [ e^{-\theta \epsilon \| \boldsymbol{\Sigma} \mathbf{U}(t+T-1) \|_1 } \right ] \notag \\
    &  = \mathbb{E} \left [ e^{\theta \epsilon \| \mathbf{Q}(t) \|_1 + \theta \epsilon \| \boldsymbol{\Sigma} \mathbf{A}(t+T-1) \|_1 - \theta \epsilon \| \boldsymbol{\Sigma} \mathbf{S}(t+T-1) \|_1 } \right ] - \mathbb{E} \left [ e^{\theta \epsilon \| \mathbf{Q}(t+T) \|_1 } \right ] + 1 - \mathbb{E} \left [ e^{-\theta \epsilon \| \boldsymbol{\Sigma} \mathbf{U}(t+T-1) \|_1 } \right ] \notag \\
    & \overset{(a)}{=} \mathbb{E} \left [ e^{\theta \epsilon \| \mathbf{Q}(t) \|_1 + \theta \epsilon \| \boldsymbol{\Sigma} \mathbf{A}(t+T-1) \|_1 - \theta \epsilon \| \boldsymbol{\Sigma} \mathbf{S}(t+T-1) \|_1 } \right ] - \mathbb{E} \left [ e^{\theta \epsilon \| \mathbf{Q}(t) \|_1 } \right ] + 1 - \mathbb{E} \left [ e^{-\theta \epsilon \| \boldsymbol{\Sigma} \mathbf{U}(t+T-1) \|_1 } \right ]. \label{eq:exponential_express_lemma}
\end{align}
where (a) follows from the steady-state assumption. \\

\noindent \textbf{Step 1.} We show that the unused-service cross term is of order $o(\epsilon^2)$.

\begin{lemma} \label{lemma_our:exponential_is_epsilon_square}
    There exists $\theta_{\max} > 0$ such that for all $\theta \in (0, \theta_{\max}]$, we have:
    \begin{align}
        & \mathbb{E} \left [ \left ( e^{\theta \epsilon \left < \mathbf{c}, \mathbf{O}(t+T) \right >} - 1 \right ) \left ( e^{-\theta \epsilon \left < \mathbf{c}, \boldsymbol{\Sigma} \mathbf{U_o}(t+T-1) \right >} - 1 \right )\right ] \in o(\epsilon^2).\notag 
    \end{align}
\end{lemma}
The proof is given in Appendix~\ref{prf:exponential_is_epsilon_square}. \\

\noindent \textbf{Step 2.} We identify the limiting MGF of $\epsilon\|\mathbf Q(t)\|_1$.

Applying Lemma~\ref{lemma_our:exponential_is_epsilon_square} in Equation~\eqref{eq:exponential_express_lemma}, we obtain
\begin{align*}
    \mathbb{E} \left [ e^{\theta \epsilon \| \mathbf{Q}(t) \|_1 } \left ( 1 - e^{ \theta \epsilon (\| \boldsymbol{\Sigma} \mathbf{A}(t+T-1) \|_1 - \| \boldsymbol{\Sigma} \mathbf{S}(t+T-1) \|_1 ) } \right ) \right ] = 1 - \mathbb{E} \left [ e^{-\theta \epsilon \| \boldsymbol{\Sigma} \mathbf{U}(t+T-1) \|_1 } \right ] + o(\epsilon^2).
\end{align*}
Noting that $\|\boldsymbol{\Sigma}\mathbf A(t+T-1)\|_1-\|\boldsymbol{\Sigma}\mathbf S(t+T-1)\|_1$ is independent of $\mathbf Q(t)$ and reorganizing the terms, we obtain
\begin{align*}
    \mathbb{E} \left [ e^{\theta \epsilon \| \mathbf{Q}(t) \|_1 } \right ] & = \frac{ 1 - \mathbb{E} \left [ e^{-\theta \epsilon \| \boldsymbol{\Sigma} \mathbf{U}(t+T-1) \|_1 } \right ] + o(\epsilon^2) }{ 1 - \mathbb{E} \left [ e^{ \theta \epsilon (\| \boldsymbol{\Sigma} \mathbf{A}(t+T-1) \|_1 - \| \boldsymbol{\Sigma} \mathbf{S}(t+T-1) \|_1 ) } \right ] } \\
    & = \frac{ \theta T \epsilon^2 + o(\epsilon^2) }{ \theta T \epsilon^2 - \frac{ (\theta \epsilon)^2}{2} \left ( T n \sigma_{\lambda}^2 + T \sum\limits_{l=1}^n \sigma_l^2 + T^2 \epsilon^2 \right ) + O(\epsilon^3) } \\
    & = \frac{1 + o(1)}{1 - \frac{\theta}{2} \left ( n \sigma_{\lambda}^2 + \sum\limits_{l=1}^n \sigma_l^2 \right ) + O(\epsilon) },
\end{align*}
and thus that $\epsilon \| \mathbf{Q}(t) \|_1$ converges to an exponential random variable with mean $ \left ( n \sigma_{\lambda}^2 + \sum_{l=1}^n \sigma_l^2 \right ) / 2 $.\\

\noindent \textbf{Step 3.} We use state-space collapse to obtain the vector limit.

By the definition of $\mathbf O_\perp(t)$, for each $l\in[n]$,
\[
    Q_l(t) = \frac{\gamma_l}{\|\boldsymbol\gamma\|_1}\|\mathbf Q(t)\|_1 + \sqrt{\gamma_l}O_{\perp,l}(t).
\]
Multiplying by $\epsilon$ and writing the identities in vector form gives
\[
    \epsilon\mathbf Q(t)
    -
    \frac{\boldsymbol\gamma}{\|\boldsymbol\gamma\|_1}
    \epsilon\|\mathbf Q(t)\|_1
    =
    \epsilon
    \left(
    \sqrt{\gamma_l}O_{\perp,l}(t)
    \right)_{l=1}^n.
\]
Hence,
\begin{align*}
    \mathbb{E} \left [ \left\| \epsilon\mathbf Q(t) - \frac{\boldsymbol\gamma}{\|\boldsymbol\gamma\|_1} \epsilon\|\mathbf Q(t)\|_1 \right\|_2 \right ] & = \mathbb{E} \left [ \epsilon \left\| \left( \sqrt{\gamma_l}O_{\perp,l}(t) \right)_{l=1}^n \right\|_2 \right ] \\
    & \leq \epsilon \sqrt{\gamma_{\max}} \mathbb{E} \left [ \|\mathbf O_\perp(t)\|_2 \right ] \\
    & \overset{(a)}{\leq} \epsilon \sqrt{\gamma_{\max}} \sqrt{ \mathbb{E} \left [ \|\mathbf O_\perp(t)\|_2^2 \right ] } \\
    & \overset{(b)}{\leq} \epsilon \sqrt{\gamma_{\max}} N_{\perp}(n,T)
\end{align*}
where (a) follows from Jensen's inequality applied to $ g(y)=\sqrt{y} $ for all $y \geq 0$ and (b) follows from Theorem~\ref{thm_our:SSC_withrate}. 

By Markov's inequality and the bound above, for every $\delta>0$,
\[
    0 \leq \mathbb P\left( \left\| \epsilon\mathbf Q(t) - \frac{\boldsymbol\gamma}{\|\boldsymbol\gamma\|_1} \epsilon\|\mathbf Q(t)\|_1 \right\|_2 > \delta \right) \leq \frac{\epsilon\sqrt{\gamma_{\max}}N_\perp(n,T)}{\delta}.
\]
Since $N_\perp(n,T)$ is independent of $\epsilon$, it follows that
\[
    \lim_{\epsilon\downarrow0} \mathbb P\left( \left\| \epsilon\mathbf Q(t) - \frac{\boldsymbol\gamma}{\|\boldsymbol\gamma\|_1} \epsilon\|\mathbf Q(t)\|_1 \right\|_2>\delta \right) = 0.
\]
Thus, the difference above converges to $\mathbf 0$ in probability. Since Step 2 shows that $\epsilon\|\mathbf Q(t)\|_1$ converges in distribution to an exponential random variable $\Upsilon$ with mean \( (n\sigma_\lambda^2+\sum_{l=1}^n\sigma_l^2)/2  \), Slutsky's Theorem implies that $\epsilon\mathbf Q(t)$ converges in distribution to
\[
    \Upsilon
    \left(
    \frac{\gamma_1}{\|\boldsymbol\gamma\|_1},
    \ldots,
    \frac{\gamma_n}{\|\boldsymbol\gamma\|_1}
    \right).
\]

\subsection{Proof of Lemma~\ref{lemma_our:exponential_is_epsilon_square}} \label{prf:exponential_is_epsilon_square}
To prove this lemma, we expand the term as follows.
\begin{align}
    & \mathbb{E} \left [ \left ( e^{\theta \epsilon \left < \mathbf{c}, \mathbf{O}(t+T) \right >} - 1 \right ) \left ( e^{-\theta \epsilon \left < \mathbf{c}, \boldsymbol{\Sigma} \mathbf{U_o}(t+T-1) \right >} - 1 \right ) \right ] \notag \\
    & = \mathbb{E} \left [ \left ( e^{\theta \epsilon \| \mathbf{Q}(t+T) \|_1 } - 1 \right ) \left ( e^{-\theta \epsilon \| \boldsymbol{\Sigma} \mathbf{U}(t+T-1) \|_1 } - 1 \right )\right ] \notag \\
    & \leq \mathbb{E} \left [ \left | \left ( e^{\theta \epsilon \| \mathbf{Q}(t+T) \|_1 } - 1 \right ) \left ( e^{-\theta \epsilon \| \boldsymbol{\Sigma} \mathbf{U}(t+T-1) \|_1 } - 1 \right ) \right | \right ] \notag \\
    & = |\theta| \epsilon \mathbb{E} \left [ \left | \left \| \boldsymbol{\Sigma} \mathbf{U}(t+T-1) \right \|_1 \left ( e^{\theta \epsilon \| \mathbf{Q}(t+T) \|_1 } - 1 \right ) \left ( \frac{ e^{-\theta \epsilon \| \boldsymbol{\Sigma} \mathbf{U}(t+T-1) \|_1 } - 1 }{-\theta \epsilon \| \boldsymbol{\Sigma} \mathbf{U}(t+T-1) \|_1} \right ) \right | \mathbf{1}_{\{\| \boldsymbol{\Sigma} \mathbf{U}(t+T-1) \|_1 \neq 0 \}} \right ] \notag \\
    & \overset{(a)}{\leq} \theta \epsilon \left ( \frac{ e^{ \theta \epsilon Tn S_{\max}} - 1 }{ \theta \epsilon Tn S_{\max} } \right ) \mathbb{E} \left [ \left | \left \| \boldsymbol{\Sigma} \mathbf{U}(t+T-1) \right \|_1 \left ( e^{\theta \epsilon \| \mathbf{Q}(t+T) \|_1 } - 1 \right ) \right | \mathbf{1}_{\{\| \boldsymbol{\Sigma} \mathbf{U}(t+T-1) \|_1 \neq 0 \}} \right ] \notag \\
    & = \theta \epsilon \left ( \frac{ e^{ \theta \epsilon Tn S_{\max}} - 1 }{ \theta \epsilon Tn S_{\max} } \right ) \mathbb{E} \left [ \left | \sum_{l=1}^n \left [ \left ( \sum_{j=0}^{T-1} U_l(t+j) \right ) \left ( e^{\theta \epsilon \| \mathbf{Q}(t+T) \|_1 } - 1 \right ) \right ]  \right | \mathbf{1}_{\{\| \boldsymbol{\Sigma} \mathbf{U}(t+T-1) \|_1 \neq 0 \}} \right ] \notag \\
    & = \theta \epsilon \left ( \frac{ e^{ \theta \epsilon Tn S_{\max}} - 1 }{ \theta \epsilon Tn S_{\max} } \right ) \notag \\ 
    & \quad \times \mathbb{E} \left [ \left | \sum_{l=1}^n \left [ \left ( \sum_{j=0}^{T-1} U_l(t+j) \right ) \left ( e^{ \theta \epsilon \| \boldsymbol{\gamma} \|_1 \frac{Q_l(t+T)}{ \gamma_l } } e^{ \frac{ \theta \epsilon \| \boldsymbol{\gamma} \|_1 }{ \sqrt{\gamma_l} } \left ( \frac{ \| \mathbf{Q}(t+T) \|_1 }{ \| \boldsymbol{\gamma} \|_1 } \sqrt{\gamma_l} - \frac{Q_l(t+T)}{\sqrt{\gamma_l}} \right ) } -1 \right ) \right ] \right | \mathbf{1}_{\{\| \boldsymbol{\Sigma} \mathbf{U}(t+T-1) \|_1 \neq 0 \}} \right ] \notag \\
    & = \theta \epsilon \left ( \frac{ e^{ \theta \epsilon Tn S_{\max}} - 1 }{ \theta \epsilon Tn S_{\max} } \right ) \notag \\
    & \quad \times \mathbb{E} \left [ \left | \sum_{l=1}^n \left [ \left ( \sum_{j=0}^{T-1} U_l(t+j) \right ) e^{ \frac{ \theta \epsilon \| \boldsymbol{\gamma} \|_1 }{ \gamma_l } Q_l(t+T) } e^{ - \frac{ \theta \epsilon \| \boldsymbol{\gamma} \|_1 }{ \sqrt{\gamma_l} } O_{\perp, l}(t+T) } - \left ( \sum_{j=0}^{T-1} U_l(t+j) \right ) \right ] \right | \mathbf{1}_{\{\| \boldsymbol{\Sigma} \mathbf{U}(t+T-1) \|_1 \neq 0 \}} \right ] \notag \\
    & \leq \theta \epsilon \left ( \frac{ e^{ \theta \epsilon Tn S_{\max}} - 1 }{\theta \epsilon Tn S_{\max} } \right ) \notag \\
    & \quad \times \mathbb{E} \left [ \sum_{l=1}^n \left [ \left | \left ( \sum_{j=0}^{T-1} U_l(t+j) \right ) e^{ \frac{ \theta \epsilon \| \boldsymbol{\gamma} \|_1 }{ \gamma_l } Q_l(t+T) } e^{ - \frac{ \theta \epsilon \| \boldsymbol{\gamma} \|_1 }{ \sqrt{\gamma_l} } O_{\perp, l}(t+T) } - \left ( \sum_{j=0}^{T-1} U_l(t+j) \right ) \right | \right ] \mathbf{1}_{\{\| \boldsymbol{\Sigma} \mathbf{U}(t+T-1) \|_1 \neq 0 \}} \right ] \label{eq:exponential_distribution}
\end{align}
where (a) follows from the fact that $ (e^x-1)/x $ is nonnegative and nondecreasing for all $ x \in \mathbb{R} $ and that $ \| \boldsymbol{\Sigma}\mathbf{U}(t+T-1) \|_1 \leq TnS_{\max} $.

Let $\alpha_l = \theta \epsilon \| \boldsymbol{\gamma} \|_1 / \gamma_l$, we bound the term $ \left ( \sum_{j=0}^{T-1} U_l(t+j) \right )e^{ \alpha_l Q_l(t+T) }$, for each $l\in [n]$, as follows. For each server $l$, since $U_l(t)Q_l(t+1)=0$ for all $l \in [n]$ and $t \geq 0$, then
\begin{align*}
    U_l(t+j) \left ( e^{\alpha Q_l(t+j+1)} - 1 \right ) = 0, \qquad \forall \, \alpha \in \mathbb{R}.
\end{align*}
The same term can then be bounded as follows.
\begin{align}
    & \left ( \sum\limits_{j=0}^{T-1} U_l(t+j) \right ) e^{ \alpha_l Q_l(t+T) } \notag \\
    & = \sum\limits_{j=0}^{T-1} U_l(t+j)e^{\alpha_l Q_l(t+T)} \notag \\
    & = \sum\limits_{j=0}^{T-1} U_l(t+j) e^{\alpha_l Q_l(t+j+1)} \prod\limits_{k=j+1}^{T-1} e^{\alpha_l\left(A_l(t+k)-S_l(t+k)+U_l(t+k)\right)} \notag \\
    & \overset{(a)}{=} \sum\limits_{j=0}^{T-1} U_l(t+j) \prod\limits_{k=j+1}^{T-1} e^{\alpha_l\left(A_l(t+k)-S_l(t+k)+U_l(t+k)\right)} \notag \\
    & \overset{(b)}{\leq} \sum\limits_{j=0}^{T-1} U_l(t+j) \prod\limits_{k=j+1}^{T-1} e^{\alpha_l A_l(t+k)} \notag \\
    & \overset{(c)}{\leq} \sum\limits_{j=0}^{T-1} U_l(t+j) \prod\limits_{k=j+1}^{T-1} e^{\alpha_l nA_{\max}} \notag \\
    & = \sum\limits_{j=0}^{T-1} U_l(t+j) e^{\alpha_l nA_{\max}(T-j-1)} \notag \\
    & \leq e^{\alpha_l nA_{\max}(T-1)} \left( \sum\limits_{j=0}^{T-1} U_l(t+j) \right). \label{eq:unused_exp_bound}
\end{align}
where (a) follows from the fact that $U_l(t+j)Q_l(t+j+1)=0$ for all $j=0,\ldots,T-1$, and hence $U_l(t+j)e^{\alpha_l Q_l(t+j+1)}=U_l(t+j)$. (b) follows from $U_l(t+k)\leq S_l(t+k)$, which implies $ A_l(t+k)-S_l(t+k)+U_l(t+k)\leq A_l(t+k) $, and from $\alpha_l\geq 0$. (c) follows from $A_l(t+k)\leq nA_{\max}$.

Equation~\eqref{eq:exponential_distribution} can then be bounded as follows. For each $l\in[n]$, by the triangle inequality,
\begin{align}
    & \left| \left( \sum_{j=0}^{T-1} U_l(t+j) \right) e^{\frac{\theta\epsilon\|\boldsymbol{\gamma}\|_1}{\gamma_l}Q_l(t+T)} e^{-\frac{\theta\epsilon\|\boldsymbol{\gamma}\|_1}{\sqrt{\gamma_l}}O_{\perp,l}(t+T)} - \left( \sum_{j=0}^{T-1} U_l(t+j) \right) \right| \notag \\
    & = \left| \left( \sum_{j=0}^{T-1} U_l(t+j) \right) e^{\frac{\theta\epsilon\|\boldsymbol{\gamma}\|_1}{\gamma_l}Q_l(t+T)} \left ( e^{-\frac{\theta\epsilon\|\boldsymbol{\gamma}\|_1}{\sqrt{\gamma_l}}O_{\perp,l}(t+T)} - 1 \right ) + \left( \sum_{j=0}^{T-1} U_l(t+j) \right) \left ( e^{\frac{\theta\epsilon\|\boldsymbol{\gamma}\|_1}{\gamma_l}Q_l(t+T)} -1\right )  \right| \notag \\
    & \leq \left( \sum_{j=0}^{T-1} U_l(t+j) \right) e^{\frac{\theta\epsilon\|\boldsymbol{\gamma}\|_1}{\gamma_l}Q_l(t+T)} \left| e^{-\frac{\theta\epsilon\|\boldsymbol{\gamma}\|_1}{\sqrt{\gamma_l}}O_{\perp,l}(t+T)}-1 \right| + \left( \sum_{j=0}^{T-1} U_l(t+j) \right) \left| e^{\frac{\theta\epsilon\|\boldsymbol{\gamma}\|_1}{\gamma_l}Q_l(t+T)} -1 \right| \notag \\
    & \overset{(a)}{\leq} \left( \sum_{j=0}^{T-1} U_l(t+j) \right) e^{\frac{\theta\epsilon\|\boldsymbol{\gamma}\|_1}{\gamma_l}Q_l(t+T)} \left| e^{-\frac{\theta\epsilon\|\boldsymbol{\gamma}\|_1}{\sqrt{\gamma_l}}O_{\perp,l}(t+T)}-1 \right| + \left( \sum_{j=0}^{T-1} U_l(t+j) \right) \left ( e^{\frac{\theta\epsilon\|\boldsymbol{\gamma}\|_1}{\gamma_l}Q_l(t+T)} -1 \right ) \notag \\
    & \overset{(b)}{\leq} \left( \sum_{j=0}^{T-1} U_l(t+j) \right) e^{\frac{\theta\epsilon\|\boldsymbol{\gamma}\|_1}{\gamma_l}nA_{\max}(T-1)} \left| e^{-\frac{\theta\epsilon\|\boldsymbol{\gamma}\|_1}{\sqrt{\gamma_l}}O_{\perp,l}(t+T)}-1 \right| + \left( \sum_{j=0}^{T-1} U_l(t+j) \right) \left ( e^{\frac{\theta\epsilon\|\boldsymbol{\gamma}\|_1}{\gamma_l}nA_{\max}(T-1)} -1 \right ) \notag  \\
    & \leq \left( \sum_{j=0}^{T-1} U_l(t+j) \right) e^{\frac{\theta\epsilon\|\boldsymbol{\gamma}\|_1}{\gamma_{\min}}nA_{\max}(T-1)} \left| e^{-\frac{\theta\epsilon\|\boldsymbol{\gamma}\|_1}{\sqrt{\gamma_l}}O_{\perp,l}(t+T)}-1 \right| + \left( \sum_{j=0}^{T-1} U_l(t+j) \right) \left ( e^{\frac{\theta\epsilon\|\boldsymbol{\gamma}\|_1}{\gamma_{\min}}nA_{\max}(T-1)} -1 \right ). \label{eq:correct_abs_bound_each_l}
\end{align}
where (a) follows because $\theta > 0$ and (b) follows from Equation~\eqref{eq:unused_exp_bound}.

Substituting Equation~\eqref{eq:correct_abs_bound_each_l} into Equation~\eqref{eq:exponential_distribution}, we obtain
\begin{align}
    & \theta \epsilon \left( \frac{ e^{ \theta \epsilon Tn S_{\max}} - 1 }{ \theta \epsilon Tn S_{\max} } \right) \notag \\
    & \quad \times \mathbb{E} \left[ \sum_{l=1}^n \left[ \left| \left( \sum_{j=0}^{T-1} U_l(t+j) \right) e^{\frac{\theta\epsilon\|\boldsymbol{\gamma}\|_1}{\gamma_l}Q_l(t+T)} e^{-\frac{\theta\epsilon\|\boldsymbol{\gamma}\|_1}{\sqrt{\gamma_l}}O_{\perp,l}(t+T)} - \left( \sum_{j=0}^{T-1} U_l(t+j) \right) \right| \right] \mathbf{1}_{\{\| \boldsymbol{\Sigma} \mathbf{U}(t+T-1) \|_1 \neq 0\}} \right] \notag \\
    &\leq \theta \epsilon \left( \frac{ e^{ \theta \epsilon Tn S_{\max}} - 1 }{ \theta \epsilon Tn S_{\max} } \right) e^{\frac{\theta\epsilon\|\boldsymbol{\gamma}\|_1}{\gamma_{\min}}nA_{\max}(T-1)} \mathbb{E} \left[ \sum_{l=1}^n \left [ \left( \sum_{j=0}^{T-1} U_l(t+j) \right) \left| e^{-\frac{\theta\epsilon\|\boldsymbol{\gamma}\|_1}{\sqrt{\gamma_l}}O_{\perp,l}(t+T)} -1 \right| \right ] \mathbf{1}_{\{\| \boldsymbol{\Sigma} \mathbf{U}(t+T-1) \|_1 \neq 0\}} \right] \notag \\
    & \quad + \theta \epsilon \left(  \frac{ e^{ \theta \epsilon Tn S_{\max}} - 1 }{ \theta \epsilon Tn S_{\max} } \right) \left( e^{\frac{\theta\epsilon\|\boldsymbol{\gamma}\|_1}{\gamma_{\min}}nA_{\max}(T-1)} -1 \right) \mathbb{E} \left[ \left( \sum_{l=1}^n \sum_{j=0}^{T-1} U_l(t+j) \right) \mathbf{1}_{\{\| \boldsymbol{\Sigma} \mathbf{U}(t+T-1) \|_1 \neq 0\}} \right]. \label{eq:exp_two_terms}
\end{align}

Next, we prove that the first term and second terms in Equation~\eqref{eq:exp_two_terms} are both $o(\epsilon^2)$. For the first term, we have
\begin{align}
    & \theta \epsilon \left ( \frac{ e^{ \theta \epsilon Tn S_{\max}} - 1 }{ \theta \epsilon Tn S_{\max} } \right ) e^{ \frac{ \theta \epsilon \| \boldsymbol{\gamma} \|_1 }{ \gamma_{\min} } nA_{\max} (T-1) } \mathbb{E} \left [ \sum_{l=1}^n \left [ \left ( \sum_{j=0}^{T-1} U_l(t+j) \right ) \left | e^{ - \frac{ \theta \epsilon \| \boldsymbol{\gamma} \|_1 }{ \sqrt{\gamma_l} } O_{\perp, l}(t+T) } - 1 \right | \right ] \mathbf{1}_{\{\| \boldsymbol{\Sigma} \mathbf{U}(t+T-1) \|_1 \neq 0 \}} \right ] \notag \\
    & \overset{(a)}{\leq} \theta \epsilon \left ( \frac{ e^{ \theta \epsilon Tn S_{\max}} - 1 }{ \theta \epsilon Tn S_{\max} } \right ) e^{ \frac{ \theta \epsilon \| \boldsymbol{\gamma} \|_1 }{ \gamma_{\min} } nA_{\max} (T-1) } \mathbb{E} \left [ \sum_{l=1}^n \left [ \left ( \sum_{j=0}^{T-1} U_l(t+j) \right ) \left | e^{ - \frac{ \theta \epsilon \| \boldsymbol{\gamma} \|_1 }{ \sqrt{\gamma_l} } O_{\perp, l}(t+T) } - 1 \right | \right ] \right ] \notag \\
    & \overset{(b)}{\leq} \theta \epsilon \left ( \frac{ e^{ \theta \epsilon Tn S_{\max}} - 1 }{ \theta \epsilon Tn S_{\max} } \right ) e^{ \frac{ \theta \epsilon \| \boldsymbol{\gamma} \|_1 }{ \gamma_{\min} } nA_{\max} (T-1) } \mathbb{E} \left [ \sum_{l=1}^n \left ( \sum_{j=0}^{T-1} U_l(t+j) \right )^p \right ]^{\frac{1}{p}} \mathbb{E} \left [ \sum_{l=1}^n \left | e^{ - \frac{ \theta \epsilon \| \boldsymbol{\gamma} \|_1 }{ \sqrt{\gamma_l} } O_{\perp, l}(t+T) } - 1  \right |^{\frac{p}{p-1}} \right ]^{\frac{p-1}{p}} \notag \\
    & = \theta \epsilon \left ( \frac{ e^{ \theta \epsilon Tn S_{\max}} - 1 }{\theta \epsilon Tn S_{\max} } \right )  e^{ \frac{ \theta \epsilon \| \boldsymbol{\gamma} \|_1 }{ \gamma_{\min} } nA_{\max} (T-1) } \mathbb{E} \left [ \sum_{l=1}^n \left ( \sum_{j=0}^{T-1} U_l(t+j) \right )^{1+p-1} \right ]^{\frac{1}{p}} \mathbb{E} \left [ \sum_{l=1}^n \left |  e^{ - \frac{ \theta \epsilon \| \boldsymbol{\gamma} \|_1 }{ \sqrt{\gamma_l} } O_{\perp, l}(t+T) } - 1 \right |^{\frac{p}{p-1}} \right ]^{\frac{p-1}{p}} \notag \\
    & \overset{(c)}{\leq} \theta \epsilon^{1+\frac{1}{p}} T^{\frac{1}{p}} (TS_{\max})^{\frac{p-1}{p}} \left ( \frac{ e^{ \theta \epsilon Tn S_{\max}} - 1 }{\theta \epsilon Tn S_{\max} } \right ) e^{ \frac{ \theta \epsilon \| \boldsymbol{\gamma} \|_1 }{ \gamma_{\min} } nA_{\max} (T-1) } \mathbb{E} \left [ \sum_{l=1}^n \left | e^{ - \frac{ \theta \epsilon \| \boldsymbol{\gamma} \|_1 }{ \sqrt{\gamma_l} } O_{\perp, l}(t+T) } - 1 \right |^{\frac{p}{p-1}} \right ]^{\frac{p-1}{p}} \notag \\
    & = \theta \epsilon^{1+\frac{1}{p}} T (S_{\max})^{\frac{p-1}{p}} \left ( \frac{ e^{ \theta \epsilon Tn S_{\max}} - 1 }{ \theta \epsilon Tn S_{\max} } \right ) e^{ \frac{ \theta \epsilon \| \boldsymbol{\gamma} \|_1 }{ \gamma_{\min} } nA_{\max} (T-1) } \mathbb{E} \left [ \sum_{l=1}^n \left | e^{ - \frac{ \theta \epsilon \| \boldsymbol{\gamma} \|_1 }{ \sqrt{\gamma_l} } O_{\perp, l}(t+T) } - 1 \right |^{\frac{p}{p-1}} \mathbf{1}_{ \{ O_{\perp, l}(t+T) \neq 0 \} } \right ]^{\frac{p-1}{p}} \notag \\
    & \leq \theta \epsilon^{1+\frac{1}{p}} T (S_{\max})^{\frac{p-1}{p}} \left ( \frac{ e^{ \theta \epsilon Tn S_{\max}} - 1 }{ \theta \epsilon Tn S_{\max} } \right ) e^{ \frac{ \theta \epsilon \| \boldsymbol{\gamma} \|_1 }{ \gamma_{\min} } nA_{\max} (T-1) } \notag \\
    & \quad \times \left (  \sum_{l=1}^n \mathbb{E} \left [ \left | \frac{ e^{ - \frac{ \theta \epsilon \| \boldsymbol{\gamma} \|_1 }{ \sqrt{\gamma_l} } O_{\perp, l}(t+T) } - 1 }{ - \frac{ \theta \epsilon \| \boldsymbol{\gamma} \|_1 }{ \sqrt{\gamma_l} } O_{\perp, l}(t+T) } \right | ^{\frac{p}{p-1}} \left | \frac{ \theta \epsilon \| \boldsymbol{\gamma} \|_1 }{ \sqrt{\gamma_l} } O_{\perp, l}(t+T) \right |^{\frac{p}{p-1}} \mathbf{1}_{ \{ O_{\perp, l}(t+T) \neq 0 \} } \right ] \right )^{\frac{p-1}{p}} \notag \\
    & \leq \theta^2 \epsilon^{2+\frac{1}{p}} T (S_{\max})^{\frac{p-1}{p}} \left ( \frac{ e^{ \theta \epsilon Tn S_{\max}} - 1 }{\theta \epsilon Tn S_{\max} } \right ) \frac{ \| \boldsymbol{\gamma} \|_1 }{ \sqrt{\gamma_{\min}} } e^{ \frac{ \theta \epsilon \| \boldsymbol{\gamma} \|_1 }{ \gamma_{\min} } nA_{\max} (T-1) } \notag \\
    & \quad \times \left (  \sum_{l=1}^n \mathbb{E} \left [ \left | \frac{ e^{ - \frac{ \theta \epsilon \| \boldsymbol{\gamma} \|_1 }{ \sqrt{\gamma_l} } O_{\perp, l}(t+T) } - 1 }{ - \frac{ \theta \epsilon \| \boldsymbol{\gamma} \|_1 }{ \sqrt{\gamma_l} } O_{\perp, l}(t+T) } \right | ^{\frac{p}{p-1}} | O_{\perp, l}(t+T) |^{\frac{p}{p-1}} \mathbf{1}_{ \{ O_{\perp, l}(t+T) \neq 0 \} } \right ] \right )^{\frac{p-1}{p}}, \label{eq:exp_first_term}
\end{align}
where (a) follows from the fact that $ 0\leq \mathbf{1}_{\{\| \boldsymbol{\Sigma} \mathbf{U}(t+T-1) \|_1 \neq 0 \}} \leq 1  $, (b) follows by H\"older's inequality and (c) follows from the facts that $ \mathbb{E} \left [ \sum_{l=1}^n \left ( \sum_{j=0}^{T-1} U_l(t+j) \right ) \right ] = T\epsilon $, $U_l(t+j) \leq S_{\max}$ for all $l,j$ and that $ x^{1/p} $ is an increasing function for $x \geq 0$. \\

By Taylor expansion, we have:
\begin{align*}
    \lim\limits_{\epsilon \downarrow 0} \frac{ e^{ \theta \epsilon Tn S_{\max}} - 1 }{ \theta \epsilon Tn S_{\max} } = 1
\end{align*}

Next, we show that the last term in Equation~\eqref{eq:exp_first_term} is uniformly bounded, we prove the following claim. 
\begin{claim} \label{clm:epsilon}
    There exists a finite $\theta_{\max} > 0 $ such that for all $\theta \in (0, \theta_{\max}]$, we have: 
    \begin{align*}
        \mathbb{E} \left [ \left | \frac{ e^{ - \frac{ \theta \epsilon \| \boldsymbol{\gamma} \|_1 }{ \sqrt{\gamma_l} } O_{\perp, l}(t+T) } - 1 }{ - \frac{ \theta \epsilon \| \boldsymbol{\gamma} \|_1 }{ \sqrt{\gamma_l} } O_{\perp, l}(t+T) } \right | ^{\frac{p}{p-1}} | O_{\perp, l}(t+T) |^{\frac{p}{p-1}} \mathbf{1}_{ \{ O_{\perp, l}(t+T) \neq 0 \} } \right ] \in O(1).
    \end{align*}
\end{claim}
For each $ l \in [n]$, Holder's inequality implies that
\begin{align*}
    & \mathbb{E} \left [ \left | \frac{ e^{ - \frac{ \theta \epsilon \| \boldsymbol{\gamma} \|_1 }{ \sqrt{\gamma_l} } O_{\perp, l}(t+T) } - 1 }{ - \frac{ \theta \epsilon \| \boldsymbol{\gamma} \|_1 }{ \sqrt{\gamma_l} } O_{\perp, l}(t+T) } \right | ^{\frac{p}{p-1}} | O_{\perp, l}(t+T) |^{\frac{p}{p-1}} \mathbf{1}_{ \{ O_{\perp, l}(t+T) \neq 0 \} } \right ] \\
    & \leq \mathbb{E} \left [ \left | \frac{ e^{ - \frac{ \theta \epsilon \| \boldsymbol{\gamma} \|_1 }{ \sqrt{\gamma_l} } O_{\perp, l}(t+T) } - 1 }{ - \frac{ \theta \epsilon \| \boldsymbol{\gamma} \|_1 }{ \sqrt{\gamma_l} } O_{\perp, l}(t+T) } \right | ^{\frac{p}{p-1} \frac{\tilde{p}}{\tilde{p}-1} } \mathbf{1}_{ \{ O_{\perp, l}(t+T) \neq 0 \} } \right ]^{ \frac{\tilde{p}-1}{\tilde{p}} } \mathbb{E} \left [  | O_{\perp, l}(t+T) |^{\frac{p}{p-1} \tilde{p} } \mathbf{1}_{ \{ O_{\perp, l}(t+T) \neq 0 \} } \right ]^{\frac{1}{\tilde{p}}} \\
    & \leq \mathbb{E} \left [ \left | \frac{ e^{ - \frac{ \theta \epsilon \| \boldsymbol{\gamma} \|_1 }{ \sqrt{\gamma_l} } O_{\perp, l}(t+T) } - 1 }{ - \frac{ \theta \epsilon \| \boldsymbol{\gamma} \|_1 }{ \sqrt{\gamma_l} } O_{\perp, l}(t+T) } \right | ^{\frac{p}{p-1} \frac{\tilde{p}}{\tilde{p}-1} } \mathbf{1}_{ \{ O_{\perp, l}(t+T) \neq 0 \} } \right ]^{ \frac{\tilde{p}-1}{\tilde{p}} } \mathbb{E} \left [  | O_{\perp, l}(t+T) |^{\frac{p}{p-1} \tilde{p} } \right ]^{\frac{1}{\tilde{p}}},
\end{align*}
where $\tilde{p}>1$. Choose any $p>2$ and set $\tilde p=2(p-1)/p>1$ such that $ (p/(p-1))\tilde p=2 $. Then, by Theorem~\ref{thm_our:SSC_withrate}, $ \mathbb E\left[|O_{\perp,l}(t+T)|^2\right] \leq \mathbb E\left[\|\mathbf O_\perp(t+T)\|_2^2\right] \leq N_\perp^2(n,T)$, so this factor is uniformly bounded in $\epsilon$.

For the term inside the expectation above, we have
\begin{align*}
    \left| \frac{ e^{-\frac{\theta\epsilon\|\boldsymbol{\gamma}\|_1}{\sqrt{\gamma_l}}O_{\perp, l}(t+T)}-1 }{ -\frac{\theta\epsilon\|\boldsymbol{\gamma}\|_1}{\sqrt{\gamma_l}}O_{\perp, l}(t+T) } \right| \mathbf{1}_{ \{ O_{\perp, l}(t+T) \neq 0 \} }
    & \overset{(a)}{\leq} \frac{ e^{\frac{\theta\epsilon\|\boldsymbol{\gamma}\|_1}{\sqrt{\gamma_l}}|O_{\perp, l}(t+T)|}-1 }{ \frac{\theta\epsilon\|\boldsymbol{\gamma}\|_1}{\sqrt{\gamma_l}}|O_{\perp, l}(t+T)| } \mathbf{1}_{ \{ O_{\perp, l}(t+T) \neq 0 \} } \\
    & \overset{(b)}{\leq} \frac{ e^{\frac{\theta\epsilon\|\boldsymbol{\gamma}\|_1}{\sqrt{\gamma_l}}\|\mathbf O_\perp(t+T)\|_2}-1 }{ \frac{\theta\epsilon\|\boldsymbol{\gamma}\|_1}{\sqrt{\gamma_l}}\|\mathbf O_\perp(t+T)\|_2 } \mathbf{1}_{ \{ O_{\perp, l}(t+T) \neq 0 \} }  \\
    & \overset{(c)}{\leq} e^{\frac{\theta\epsilon\|\boldsymbol{\gamma}\|_1}{\sqrt{\gamma_l}}\|\mathbf O_\perp(t+T)\|_2} \mathbf{1}_{ \{ O_{\perp, l}(t+T) \neq 0 \} } \\
    & \leq e^{\frac{\theta\epsilon\|\boldsymbol{\gamma}\|_1}{\sqrt{\gamma_l}}\|\mathbf O_\perp(t+T)\|_2} \\
    &\leq e^{\frac{\theta\epsilon\|\boldsymbol{\gamma}\|_1}{\sqrt{\gamma_{\min}}}\|\mathbf O_\perp(t+T)\|_2}.
\end{align*}
where (a) follows from the fact that $|e^{-x}-1|\leq e^{|x|}-1$ for all $x \in \mathbb{R}$, (b) follows from $ |O_{\perp, l}(t+T)|\leq \|\mathbf O_\perp(t+T)\|_2 $ and the fact that 
$(e^x-1)/x$ is nondecreasing on $x\geq0$ and (c) follows from the fact that $ (e^x-1)/x \leq e^x $ for all $ x \in \mathbb{R}_{+} $.

Therefore, we have
\begin{align*}
    \mathbb{E} \left [ \left | \frac{ e^{ - \frac{ \theta \epsilon \| \boldsymbol{\gamma} \|_1 }{ \sqrt{\gamma_l} } O_{\perp, l}(t+T) } - 1 }{ - \frac{ \theta \epsilon \| \boldsymbol{\gamma} \|_1 }{ \sqrt{\gamma_l} } O_{\perp, l}(t+T) } \right | ^{\frac{p}{p-1} \frac{\tilde{p}}{\tilde{p}-1} } \mathbf{1}_{ \{ O_{\perp, l}(t+T) \neq 0 \} } \right ] & \leq \mathbb{E} \left [ e^{  \frac{ \theta \| \boldsymbol{\gamma} \|_1 }{ \sqrt{\gamma_{\min}} } \frac{p}{p-1} \frac{\tilde{p}}{\tilde{p}-1} \epsilon \| \mathbf{O}_{\perp}(t+T)\|_2 } \right ] \\
    & \overset{(a)}{\leq} \mathbb{E} \left [ e^{  \frac{ \theta \| \boldsymbol{\gamma} \|_1 }{ \sqrt{\gamma_{\min}} } \frac{p}{p-1} \frac{\tilde{p}}{\tilde{p}-1} \epsilon \| \mathbf{O}(t+T)\|_2 } \right ] \\
    & \overset{(b)}{\leq} \mathbb{E} \left [ e^{  \frac{ \theta \| \boldsymbol{\gamma} \|_1 }{ \sqrt{\gamma_{\min}} } \frac{p}{p-1} \frac{\tilde{p}}{\tilde{p}-1} \epsilon \| \mathbf{O}(t+T)\|_1 } \right ] \\
    & \overset{(c)}{\leq} \mathbb{E} \left [ e^{  \frac{ \theta \| \boldsymbol{\gamma} \|_1 }{ \gamma_{\min} } \frac{p}{p-1} \frac{\tilde{p}}{\tilde{p}-1} \epsilon \| \mathbf{Q}(t+T)\|_1 } \right ] \\
    & \overset{(d)}{\leq} K_3
\end{align*}
where (a) follows from $ \| \mathbf{O}_{\perp}(t+T) \|_2 \leq \| \mathbf{O}(t+T) \|_2 $ by the fact $ \| \mathbf{O}_{\perp}(t+T) \|^2_2 + \| \mathbf{O}_{\parallel}(t+T) \|^2_2 = \| \mathbf{O}(t+T) \|^2_2  $, (b) follows from $ \| \mathbf{O}(t+T)  \|_2 \leq \| \mathbf{O}(t+T)  \|_1 $, (c) follows from $ \| \mathbf{O}(t+T) \|_1 \leq \| \mathbf{Q}(t+T) \|_1 /\sqrt{\gamma_{\min}} $ and (d) follows from Lemma~\ref{lem_our:ss_drift_zero_with_rate} and the fact that $\mathbf Q(t+T)\overset{d}{=}\mathbf Q(t)$ in steady state. 


Therefore, the last parenthetical term in Equation~\eqref{eq:exp_first_term} is \(O(1)\), uniformly for all sufficiently small \(\epsilon\). Hence, the first term in Equation~\eqref{eq:exp_two_terms} is \(O(\epsilon^{2+1/p})=o(\epsilon^2)\).

The second term in Equation~\eqref{eq:exp_two_terms} is bounded as follows:
\begin{align*}
    & \theta \epsilon \left ( \frac{ e^{ \theta \epsilon Tn S_{\max}} - 1 }{ \theta \epsilon Tn S_{\max} } \right ) \left ( e^{ \frac{ \theta \epsilon \| \boldsymbol{\gamma} \|_1 }{ \gamma_{\min} } nA_{\max} (T-1) } - 1 \right ) \mathbb{E} \left [ \left( \sum_{l=1}^n \sum_{j=0}^{T-1} U_l(t+j) \right) \mathbf{1}_{\{\| \boldsymbol{\Sigma} \mathbf{U}(t+T-1) \|_1 \neq 0 \}} \right ] \\
    & = \theta \epsilon \left ( \frac{ e^{ \theta \epsilon Tn S_{\max}} - 1 }{\theta \epsilon Tn S_{\max} } \right ) \left ( e^{ \frac{ \theta \epsilon \| \boldsymbol{\gamma} \|_1 }{ \gamma_{\min} } nA_{\max} (T-1) } - 1 \right ) \\
    & \quad \times \mathbb{E} \left [ \left( \sum_{l=1}^n \sum_{j=0}^{T-1} U_l(t+j) \right) \middle | \, \| \boldsymbol{\Sigma} \mathbf{U}(t+T-1) \|_1 \neq 0 \right ] \mathbb{P}\left \{ \| \boldsymbol{\Sigma} \mathbf{U}(t+T-1) \|_1 \neq 0 \right \} \\
    & \leq T n S_{\max} \theta \epsilon \left ( \frac{ e^{ \theta \epsilon Tn S_{\max}} - 1 }{ \theta \epsilon Tn S_{\max} } \right ) \left ( e^{ \frac{ \theta \epsilon \| \boldsymbol{\gamma} \|_1 }{ \gamma_{\min} } nA_{\max} (T-1) } - 1 \right ) \mathbb{P}\left \{ \| \boldsymbol{\Sigma} \mathbf{U}(t+T-1) \|_1 \neq 0 \right \}
\end{align*}
By Taylor expansion, we have:
\begin{align*}
    e^{ \frac{ \theta \epsilon \| \boldsymbol{\gamma} \|_1 }{ \gamma_{\min} } nA_{\max} (T-1) } - 1  = \frac{ \theta \| \boldsymbol{\gamma} \|_1 }{ \gamma_{\min} } nA_{\max} (T-1) \epsilon + o(\epsilon)
\end{align*}
Next we bound $ \mathbb{P} \left \{ \| \boldsymbol{\Sigma} \mathbf{U}(t+T-1) \|_1 \neq 0 \right \} $. We use the fact that in steady state, $ \mathbb{E} [ \| \boldsymbol{\Sigma} \mathbf{U}(t+T-1) \|_1 ] = T\epsilon $. Then
\begin{align*}
    \mathbb{P} \left \{ \| \boldsymbol{\Sigma} \mathbf{U}(t+T-1) \|_1 \neq 0 \right \} & = \frac{ \mathbb{E} [ \| \boldsymbol{\Sigma} \mathbf{U}(t+T-1) \|_1 ] }{ \mathbb{E} [ \| \boldsymbol{\Sigma} \mathbf{U}(t+T-1) \|_1 \, | \, \| \boldsymbol{\Sigma} \mathbf{U}(t+T-1) \|_1 \neq 0 ] } \\
    & \leq T \epsilon.
\end{align*}
Therefore, the second term is $O(\epsilon^3)=o(\epsilon^2)$.

Combining the bounds for the first and second terms in Equation~\eqref{eq:exp_two_terms}, we obtain
\begin{align*}
    & \mathbb{E} \left [ \left ( e^{\theta \epsilon \left < \mathbf{c}, \mathbf{O}(t+T) \right >} - 1 \right ) \left ( e^{-\theta \epsilon \left < \mathbf{c}, \boldsymbol{\Sigma} \mathbf{U_o}(t+T-1) \right >} - 1 \right )\right ] \in o(\epsilon^2).
\end{align*}

\section{Proof of Corollary~\ref{co_our:Pod-ED}}
\label{prf:Pod-ED}

Under Po\(d\)-ED, we have \(T=1\) and \(\boldsymbol{\gamma}=\boldsymbol{\mu}\). Thus, the sampled permutation is generated by the \(\boldsymbol{\mu}\)-scaled queue lengths. Conditional on a sampled permutation \(\eta\), Po\(d\)-ED samples \(d\) servers uniformly without replacement and assigns the arrival to the shortest sampled \(\boldsymbol{\mu}\)-scaled queue. Hence, the induced dispatch fraction vector \(f_\eta\) is the same as that of Po\(d\), except that the sampled permutation is formed using \(\boldsymbol{\gamma}=\boldsymbol{\mu}\). This leads to the same stability expression as the Po\(d\) condition in \cite{hurtado2021throughput-Pod-stability-heterogeneous}. \\

\noindent \textbf{Step 1.} We compute the policy-induced dispatch fractions.

For any \(\eta\in\mathcal S_n\), the \(l\)-th longest \(\boldsymbol{\mu}\)-scaled queue is selected if it is sampled and the other \(d-1\) sampled queues are among the first \(l-1\) positions in the sampled permutation. Therefore,
\begin{align*}
    f_{l,\eta} =
    \begin{cases}
        0, & l<d,\\
        \dfrac{\binom{l-1}{d-1}}{\binom{n}{d}}, & l\ge d.
    \end{cases}
\end{align*}
Thus, for any \(j\in[n]\),
\begin{align*}
    \sum_{l=1}^j f_{l,\eta} =
    \begin{cases}
        0, & j<d,\\
        \dfrac{\binom{j}{d}}{\binom{n}{d}}, & j\ge d,
    \end{cases}
\end{align*}
where the second equality follows from \(\sum_{l=d}^j\binom{l-1}{d-1}=\binom{j}{d}\). In particular, \(f_\eta\) does not depend on \(\eta\). \\

\noindent \textbf{Step 2.} We prove part (a).

By the definition of \(h^*\),
\begin{align*}
    h^* & = \min_{\eta\in\mathcal S_n}\min_{j\in[n]} \frac{\sum_{l=1}^j\mu_{\eta(l)}}{\sum_{l=1}^j f_{l,\eta}} \\
    & \overset{(a)}{=} \min_{j\in\{d,\ldots,n\}} \left\{ \frac{\binom{n}{d}}{\binom{j}{d}} \left(\sum_{l=1}^j\mu_l\right) \right\},
\end{align*}
where (a) uses \(\mu_1\le\cdots\le\mu_n\). Theorem~\ref{thm_our:stable_withrate_withratio} gives positive recurrence when \(n\lambda<h^*\). Since \(f_\eta\) does not depend on \(\eta\), Corollary~\ref{cor_our:rank_invariant_pbd} implies Assumption~\ref{assump:persistent_bottleneck_dominance}. Then Theorem~\ref{thm_our:stable_withrate_withratio_necessary_general} gives transience when \(h^*<n\lambda<\sum_{i=1}^n\mu_i\). \\

\noindent \textbf{Step 3.} We prove parts (b) and (c).

By Corollary~\ref{co_our:condition_withrate_withratio}, Po\(d\)-ED is throughput optimal if, for all \(j\in[n]\) and \(\eta\in\mathcal S_n\),
\begin{align*}
    \sum_{l=1}^j f_{l,\eta} \le \frac{\sum_{l=1}^j\mu_{\eta(l)}}{\sum_{l=1}^n\mu_l}.
\end{align*}
For \(j<d\), we have \( \sum_{l=1}^j f_{l,\eta} = 0 \). For \(j\ge d\), we have \( \sum_{l=1}^j f_{l,\eta} = \binom{j}{d}/\binom{n}{d}\). Since \(\sum_{l=1}^j\mu_{\eta(l)}\ge \sum_{l=1}^j\mu_l\), the condition in part (b) implies the above inequality for all \(j<n\), and the case \(j=n\) holds with equality. Therefore, Po\(d\)-ED is throughput optimal.

Finally, Corollary~\ref{co_our:delayoptimality} requires the same partial-sum inequalities to be strict for all \(j<n\). The strict condition in part (c), together with \(\sum_{l=1}^j\mu_{\eta(l)}\ge \sum_{l=1}^j\mu_l\), gives the required strict inequalities for all \(j<n\) and all \(\eta\in\mathcal S_n\). Therefore, Po\(d\)-ED is asymptotically delay optimal in heavy traffic.

\section{Proof of Corollary~\ref{co_our:k-WRAND-SLQ-d}}
\label{prf:k_WRAND_SLQ_d}

Under \(k\)-WRAND-SLQ-\(d\), the sampling interval is \(T=k\), and the scaling vector is \(\boldsymbol{\gamma}=\mathbf 1\). Conditional on a sampled permutation \(\eta\), the first \(d\) positions are skipped, and arrivals are assigned among positions \(d+1,\ldots,n\) with probabilities proportional to the corresponding service rates. \\

\noindent \textbf{Step 1.} We compute the policy-induced dispatch fractions.

For every \(\eta\in\mathcal S_n\),
\begin{align*}
    f_{l,\eta} =
    \begin{cases}
        0, & l\le d,\\
        \dfrac{\mu_{\eta(l)}}{\sum_{j=d+1}^n\mu_{\eta(j)}}, & l>d.
    \end{cases}
\end{align*}
Thus, for any \(m\in[n]\),
\begin{align*}
    \sum_{l=1}^m f_{l,\eta} =
    \begin{cases}
        0, & m\le d,\\
        \dfrac{\sum_{l=d+1}^m\mu_{\eta(l)}}{\sum_{j=d+1}^n\mu_{\eta(j)}}, & m>d.
    \end{cases}
\end{align*}

\noindent \textbf{Step 2.} We prove throughput optimality and asymptotic delay optimality.

By Corollary~\ref{co_our:condition_withrate_withratio}, it is enough for throughput optimality to show that, for all \(m\in[n]\) and \(\eta\in\mathcal S_n\),
\begin{align*}
    \sum_{l=1}^m f_{l,\eta} \le \frac{\sum_{l=1}^m\mu_{\eta(l)}}{\sum_{l=1}^n\mu_l}.
\end{align*}
If \(m\le d\), then \(\sum_{l=1}^m f_{l,\eta}=0\), so the inequality holds. If \(m>d\), then
\begin{align*}
    \sum_{l=1}^m f_{l,\eta} = \frac{\sum_{l=d+1}^m\mu_{\eta(l)}}{\sum_{l=d+1}^n\mu_{\eta(l)}}.
\end{align*}
Therefore, the difference between the right-hand side and the left-hand side is
\begin{align*}
    \frac{\sum_{l=1}^m\mu_{\eta(l)}}{\sum_{l=1}^n\mu_l} -  \sum_{l=1}^m f_{l,\eta} =\frac{\sum_{l=1}^m\mu_{\eta(l)}}{\sum_{l=1}^n\mu_l} - \frac{\sum_{l=d+1}^m\mu_{\eta(l)}}{\sum_{l=d+1}^n\mu_{\eta(l)}}
    = \frac{ \left(\sum_{l=1}^d\mu_{\eta(l)}\right)\left(\sum_{l=m+1}^n\mu_{\eta(l)}\right) }{ \left(\sum_{l=1}^n\mu_l\right) \left(\sum_{l=d+1}^n\mu_{\eta(l)}\right) } \ge 0.
\end{align*}
Thus the partial-sum condition in Corollary~\ref{co_our:condition_withrate_withratio} holds, and \(k\)-WRAND-SLQ-\(d\) is throughput optimal for all \(d\in\{1,\ldots,n-1\}\).

Next, by Corollary~\ref{co_our:delayoptimality}, it is enough to show that the same partial-sum inequalities are strict for all \(m\in[n]\setminus\{n\}\) and all \(\eta\in\mathcal S_n\). If \(m\le d\), then \(\sum_{l=1}^m f_{l,\eta}=0\), while \(\sum_{l=1}^m\mu_{\eta(l)}/\sum_{l=1}^n\mu_l>0\). If \(d<m<n\), then both \(\sum_{l=1}^d\mu_{\eta(l)}\) and \(\sum_{l=m+1}^n\mu_{\eta(l)}\) are strictly positive, so the displayed difference is strictly positive. Hence the strict partial-sum condition holds for all \(m<n\), and \(k\)-WRAND-SLQ-\(d\) is asymptotically delay optimal in heavy traffic for all \(d\in\{1,\ldots,n-1\}\).


\end{document}